\documentclass[12pt]{guelphthesis}
\usepackage{afterpage} 
\usepackage{alltt}
\usepackage{amsmath}
\usepackage{amsfonts}
\usepackage{amssymb}
\usepackage{amsthm}
\usepackage{array} 
\usepackage[usenames]{color}
\usepackage{enumerate}
\usepackage{graphicx}
\usepackage{latexsym}
\usepackage{longtable}
\usepackage{multirow}
\usepackage{url}
\usepackage[numbers]{natbib}

\newtheorem{thm}{Theorem}[section]
\newtheorem{claim}[thm]{Claim}
\newtheorem{cor}[thm]{Corollary}

\newtheorem{remark}[thm]{Remark}
\newtheorem{conj}[thm]{Conjecture}
\newtheorem{lemma}[thm]{Lemma}
\newtheorem{prop}[thm]{Proposition}

\newtheorem{exam}[thm]{Example}
\newtheorem{defn}[thm]{Definition}
\newcommand{\bra}[1]{\langle #1 |}
\newcommand{\ket}[1]{| #1 \rangle}
\newcommand{\braket}[2]{\langle #1 | #2 \rangle}
\newcommand{\ketbra}[2]{| #1 \rangle\langle #2 |}
\newcommand{\Tr}{\mathrm{Tr}}
\newcommand{\bb}[1]{\mathbb{#1}}
\newcommand{\cl}[1]{\mathcal{#1}}

\makeindex 

\newcommand{\mnorm}[1]{%
\left\vert\kern-0.9pt\left\vert\kern-0.9pt\left\vert #1
\right\vert\kern-0.9pt\right\vert\kern-0.9pt\right\vert}
\newcommand{\bmnorm}[1]{%
\big\vert\kern-0.9pt\big\vert\kern-0.9pt\big\vert #1
\big\vert\kern-0.9pt\big\vert\kern-0.9pt\big\vert}

\def\dsp{\def\baselinestretch{2.0}\large\normalsize}
\def\ssp{\def\baselinestretch{1.0}\large\normalsize}
\def\sspp{\def\baselinestretch{1.1}\large\normalsize}
\def\hsp{\def\baselinestretch{1.6}\large\normalsize}
\dsp

\begin{document}
\title{Norms and Cones in the Theory of Quantum Entanglement}
\author{Nathaniel Johnston}
\degreeyear{2012}
\degree{Ph.D.~Mathematics}
\chair{Dr.~D.~W.~Kribs}
\othermembers{Rajesh Pereira, Bei Zeng, John Watrous}
\numberofmembers{3}
\maketitle
%
%

\begin{frontmatter}\renewcommand{\thepage}{\roman{page}}
\begin{abstract}
\vspace{0.5in}

\begin{center}{\bf \large NORMS AND CONES IN THE THEORY OF \\ QUANTUM ENTANGLEMENT}\end{center}
\vspace{0.5in}

\noindent Nathaniel Johnston \hfill Advisor:\hspace*{0.72in}\\
\noindent University of Guelph, 2012 \hfill Professor D. Kribs\\ 
\vspace{0.5in}

\indent There are various notions of positivity for matrices and linear matrix-valued maps that play important roles in quantum information theory. The cones of positive semidefinite matrices and completely positive linear maps, which represent quantum states and quantum channels respectively, are the most ubiquitous positive cones. There are also many natural cones that can been regarded as ``more'' or ``less'' positive than these standard examples. In particular, entanglement theory deals with the cones of separable operators and entanglement witnesses, which satisfy very strong and weak positivity properties respectively.

Rather complementary to the various cones that arise in entanglement theory are norms. The trace norm (or operator norm, depending on context) for operators and the diamond norm (or completely bounded norm) for superoperators are the typical norms that are seen throughout quantum information theory. In this work our main goal is to develop a family of norms that play a role analogous to the cone of entanglement witnesses. We investigate the basic mathematical properties of these norms, including their relationships with other well-known norms, their isometry groups, and their dual norms. We also make the place of these norms in entanglement theory rigorous by showing that entanglement witnesses arise from minimal operator systems, and analogously our norms arise from minimal operator spaces.

Finally, we connect the various cones and norms considered here to several seemingly unrelated problems from other areas. We characterize the problem of whether or not non-positive partial transpose bound entangled states exist in terms of one of our norms, and provide evidence in favour of their existence. We also characterize the minimum gate fidelity of a quantum channel, the maximum output purity and its completely bounded counterpart, and the geometric measure of entanglement in terms of these norms.
\end{abstract}

\thispagestyle{plain}
\begin{dedication}
\vspace*{1.5in}
\begin{center}
\textit{{\bf To My Parents}\\ for instilling in me the desire to learn}
\end{center}
\end{dedication}

\begin{acknowledgements}\label{acknow}

First and foremost, thank you to my family for teaching me the value and joy of constantly learning. My parents, Bill and Betty, nurtured my love of mathematics from a young age by occupying me with logic puzzles on long car rides and by using dinner time as an excuse to explain mathematical mysteries that were way beyond my level of understanding. They taught me that if a problem seems to be too difficult, then that's a good sign that it's worth doing. This work could never have come together without this lesson, and I would not be the man I am today if not for the shining example my parents set for me.

My two older brothers have been fantastic role models throughout my life as well. Matthew is a fellow academic who has constantly let me know what lies around the next corner of my education. Now that I have caught up with him, we frequently commiserate with each other over our academia-related struggles. David keeps me grounded in reality and is great at reminding me of the value of family. He is also kind enough to occasionally humour my mathematical tendencies -- he is the one who taught me the Pythagorean theorem (perhaps my first theorem!). Thank you also to my wife, Kathryn, for putting up with me while I'm zoned out in ``math mode'' (and for being awesome in general).

Many thanks of course go to my advisor, David Kribs, and my committees. David seems to take pleasure in pushing me past my comfort zone, and for that he has my gratitude. He has a seemingly encyclopedic knowledge of both operator theory and quantum information theory, and in the rare instance when he doesn't know something, he knows exactly who will. He has been a great inspiration over the past six years(!), and I couldn't have asked for a better mentor.

Thank you to my collaborators and everyone who has shared ideas with me over the past few years. Thank you to Gilad Gour, Vern Paulsen, and Andreas Winter for kind hospitality as I visited their groups, and for exchanging ideas with me while I was with them. Thank you Fernando Brand\~{a}o for giving a presentation that led to my initial interest in the norms discussed throughout this thesis. Thank you also to Moritz Ernst, Sevag Gharibian, Christian Gogolin, Marius Junge, Chi-Kwong Li, Easwar Magesan, Rajesh Pereira, Mary-Beth Ruskai, {\L}ukasz Skowronek, Erling St{\o}rmer, Stanislaw Szarek, John Watrous, and Li Xin for various e-mails and conversations related to the work of this thesis.

Finally, my time as a graduate student would not have been possible without the generous financial support of my advisor, the Natural Sciences and Engineering Research Council of Canada, and Bill and Anne Brock, to whom I will always be indebted. Their generosity is truly extraordinary.

\end{acknowledgements}

\tableofcontents

\listoffigures

\listoftables

\end{frontmatter}


\chapter{Introduction}\label{ch:intro}

Quantum information theory is the study of how information can be stored, communicated, and manipulated via the laws of quantum mechanics. While quantum information differs from classical information in many ways, the two key components that seem to lead to its most interesting and useful properties are superpositions and entanglement. A superposition of quantum states is, mathematically, nothing more than a linear combination of vectors, and is thus very well-understood mathematically. Entanglement, on the other hand, deals with positive operators on the tensor product of matrix spaces, and is much more difficult to manipulate.

One of the most fundamental questions that can be asked in this setting is whether or not a given quantum state is entangled. However, even if we have a complete mathematical description of the state, determining whether or not it is entangled seems to be a very difficult task \cite{G03,HHH09}. In recent years, there has been a surge of interest in this problem, and several partial results are known. The most well-known characterizations of separability make use of the fact that the set of separable operators forms a cone that has simple relationships with other well-known cones of operators and linear maps \cite{DPS04,HLVC00}. One of the landmark results in this area says that the set of states that are separable (i.e., not entangled) is dual to the set of positive matrix-valued maps in a natural sense \cite{HHH96,P96}.

Many other results characterize the set of separable states in terms of norms. For example, there is a certain norm on the space of density operators that completely determines whether or not a given quantum state is entangled \cite{R00}, but this norm is difficult to compute. There are also easily-computed norms that give criteria for separability, but these conditions are only necessary, not sufficient \cite{CW03,R03}. Another separability criterion based on norms makes use of matrix-valued maps that are contractive \cite{HHH06}.

In this work, we generalize and unify these results. We introduce a family of norms that completely characterize positive linear maps, and we thoroughly investigate their properties. We characterize their isometry groups, we derive several inequalities in order to help bound them, and we show that their dual norms characterize separability. We show that the separable operators arise naturally from abstract operator systems, and that our norms arise analogously from abstract operator spaces. In this way, we show how the various cone-based and norm-based criteria for separability are related to each other.

We also discuss how our results generalize to cones other than the cone of separable operators. We show that they key property that drives most of our results is something that we call ``right CP-invariance''. We discuss the structure of right CP-invariant cones in general, and we show that every abstract operator system can be associated with a right CP-invariant cone (and vice-versa). We similarly show that some nice properties of separable states follow from the cone of positive maps being a semigroup (i.e., being closed under composition). We thus examine the role of semigroup cones in this setting, and show that semigroup cones can give rise to operator systems in a natural way as well.

Finally, we discuss several applications of our results to other areas of quantum information theory. In particular, we use the separability problem and our norms to approach the NPPT bound entanglement conjecture as well as the computation of minimum gate fidelity, maximum output purity, and geometric measure of entanglement.

\section{Organization of the Thesis}\label{sec:thesis_layout}

This work consists largely of work originally presented in \cite{JK10,JK11,JK11b,JKPP11,Joh11,JS11}, but many results are expanded upon and additional examples are provided. Several original results that have not appeared in any of these works are included in Sections~\ref{sec:cp_invariant}, \ref{sec:realign}, \ref{sec:sk_vector_dual_norm}, \ref{sec:SP_shareable}, \ref{sec:min_gate_fid_complexity}, \ref{sec:max_output_purity}, \ref{sec:tri_geo_ent}, and \ref{sec:dual_op_norm}.

The layout of the thesis is generally linear in that each chapter assumes knowledge of the content of the previous chapter. The main exception to this rule is that Chapters~\ref{ch:computation} and~\ref{ch:optheory} are independent of each other. A brief chapter-by-chapter breakdown of the thesis follows.

\noindent {\bf Chapter~\ref{ch:basics}.} Here we introduce the mathematical basics necessary for dealing with quantum information: state vectors, density operators, superoperators, and so on. We give a thorough overview of quantum entanglement and various methods for detecting it. We also introduce several sets of superoperators (i.e., linear maps on matrices) that are used throughout this work, and we demonstrate how various properties of superoperators relate to well-known properties of matrices via either the vector-operator isomorphism or the Choi--Jamio{\l}kowski isomorphism.

\noindent {\bf Chapter~\ref{ch:geometry}.} We present general results about cones of operators and superoperators that are relevant in entanglement theory. We consider properties that many specific cones of interest share, and discuss what those properties can tell us. We also introduce most of the norms that we use throughout the thesis and general results about linear preservers and isometry groups. We find that the group of operators that preserve separable states are exactly the local unitary and swap operators.

\noindent {\bf Chapter~\ref{chap:Sk_norms}.} We introduce a family of vector norms and a family of operator norms that arise naturally when considering entanglement. These norms characterize the cone of block positive operators (i.e., the operators that are dual to separable states) in a natural way, so they allow us to transform questions about separability into questions about norms, and vice-versa. We discuss numerous inequalities involving these norms, as well as their basic mathematical properties such as their dual norms and isometries.

\noindent {\bf Chapter~\ref{ch:computation}.} Here we consider computational problems and applications of our results. We introduce semidefinite and conic programming, and demonstrate how they can be used to compute the norms of Chapter~\ref{chap:Sk_norms}. We show that these norms appear and are useful in many unexpected areas of quantum information theory. We demonstrate how our norms can be used to answer questions about bound entanglement, minimum gate fidelity, maximum output purity, and the geometric measure of entanglement. For example, we use NP-hardness of the separability problem to show that computing minimum gate fidelity is also NP-hard.

\noindent {\bf Chapter~\ref{ch:optheory}.} In this chapter we give the norms of Chapter~\ref{chap:Sk_norms} a solid mathematical foundation and further motivate them as somehow the ``right'' norms to be studying in entanglement theory. To this end, we introduce abstract operator spaces and operator systems. We show that the maximal operator system gives rise to the cone of separable (i.e., non-entangled) operators, and similarly minimal operator spaces give rise to the norms we have been investigating. As an application of these connections, we are finally able to derive expressions for the duals of these norms, and we see that they exactly characterize separability. We also generalize recent results that show that separability can be characterized in terms of maps that are contractive in the trace norm.

We close the chapter by showing that not only do separable operators arise from a natural operator system, but so do many other cones that are considered throughout the thesis. We characterize exactly which cones occur as the maps that are completely positive into an abstract operator system, and present examples to illustrate our results.

\chapter{Quantum Information Theory and Entanglement}\label{ch:basics}

This chapter is devoted to developing the basic tools of quantum information theory that will be necessary throughout this work. We will focus on the interplay between vectors, matrices, and linear maps on matrices. We begin by defining our notation and recalling some basic notions from linear algebra.

We will use $\bb{C}^n$ to denote $n$-dimensional complex Euclidean space, $M_{n,m}$ to denote the space of $n \times m$ complex matrices, and for brevity we use the shorthand $M_n := M_{n,n}$. We will freely switch between thinking about $M_{n,m}$ as a space of matrices and thinking about it as the space of operators from $\bb{C}^m$ to $\bb{C}^n$ that they represent. We will make use of bra-ket notation from quantum mechanics as follows: we will use \emph{kets} $\ket{v} \in \bb{C}^n$ to represent unit (column) vectors and \emph{bras} $\bra{v} := \ket{v}^\dagger$ to represent the dual (row) vectors, where $(\cdot)^\dagger$ represents the conjugate transpose. The standard basis of $\mathbb{C}^m$ will be represented by $\big\{\ket{i}\big\}_{i=1}^{m}$. In the few instances where we work with a vector with norm different than $1$, we will represent it by a lowercase letter such as $v$, and it will be made clear that it is a vector.

We will use $I \in M_n$ to represent the identity matrix, and we may instead write it as $I_n$ if we wish to emphasize its size. A matrix $X \in M_n$ is called \emph{Hermitian} if $X^\dagger = X$ and it is called \emph{unitary} if $X^\dagger X = I$. The set of unitary matrices in $M_n$ forms a group called the \emph{unitary group}, which we will denote $U(n)$. A Hermitian matrix $X$ is called \emph{positive semidefinite} if $\bra{v}X\ket{v} \geq 0$ for all $\ket{v} \in \bb{C}^n$ and it is called \emph{positive definite} if that inequality is strict (which is equivalent to $X$ being positive semidefinite and having full rank). If $X$ is positive semidefinite or positive definite, we will write $X \geq 0$ or $X > 0$, respectively. We will denote the set of positive semidefinite matrices in $M_n$ by $M_n^+$.

A \emph{superoperator} is a linear map $\Phi : M_m \rightarrow M_n$. A fundamental fact about superoperators is that there always exist families of matrices $\big\{A_\ell\big\},\big\{B_\ell\big\} \subset M_{n,m}$ such that
\begin{align}\label{eq:CK_form}
	\Phi(X) = \sum_\ell A_\ell X B_\ell^\dagger \quad \forall \, X \in M_m.
\end{align}
We call a representation of $\Phi$ of this form a \emph{generalized Choi--Kraus representation} and we refer to the operators $\big\{A_\ell\big\}$ and $\big\{B_\ell\big\}$ as \emph{left} and \emph{right generalized Choi--Kraus operators} for $\Phi$, respectively. This terminology will make more sense after Section~\ref{sec:cp_maps}, where we introduce the (not generalized) Kraus representation and Kraus operators for completely positive maps.

To see that every linear map $\Phi : M_m \rightarrow M_n$ admits a representation of the form~\eqref{eq:CK_form}, write $\Phi\big(\ketbra{i}{j}\big) = \displaystyle\sum_{\ell=1}^n c_{ij\ell}\ketbra{x_{ij\ell}}{y_{ij\ell}}$ for all $1 \leq i,j \leq m$. If we define $A_{ab\ell},B_{ab\ell} \in M_{n,m}$ for $1 \leq a,b \leq m, 1 \leq \ell \leq n$ by
\begin{align*}
	A_{ab\ell} := c_{ab\ell}\ketbra{x_{ab\ell}}{a} \quad \text{ and } \quad B_{ab\ell} := \ketbra{y_{ab\ell}}{b},
\end{align*}
then the map $\displaystyle\Psi(X) := \sum_{a,b=1}^{m}\sum_{\ell=1}^n A_{ab\ell} X B_{ab\ell}^{\dagger}$ satisfies
\begin{align*}
	\Psi\big(\ketbra{i}{j}\big) & = \sum_{a,b=1}^{m}\sum_{\ell=1}^n A_{ab\ell} \ketbra{i}{j} B_{ab\ell}^{\dagger} \\
	& = \sum_{a,b=1}^{m} \sum_{\ell=1}^n c_{ab\ell}\ketbra{x_{ab\ell}}{a}\ketbra{i}{j} \ketbra{b}{y_{ab\ell}} \\
	& = \sum_{\ell=1}^n c_{ij\ell}\ketbra{x_{ij\ell}}{y_{ij\ell}} \\
	& = \Phi\big(\ketbra{i}{j}\big).
\end{align*}
By noting that the map $\Phi$ is completely determined by its action on the basis $\big\{\ketbra{i}{j}\big\}_{i,j=1}^{m} \subset M_m$, it follows that $\Phi = \Psi$, so in particular $\Phi$ can be written in the operator-sum form~\eqref{eq:CK_form}. In fact, we have written it in this form using only rank one generalized Choi--Kraus operators.

Notice that this ``na\"{i}ve'' construction of the generalized Choi--Kraus operators gives us a family of $m^2 n$ such operators. We will see in Section~\ref{sec:choi_jamiolkowski_isomorphism} how to construct a much more efficient decomposition consisting of only $mn$ generalized Choi--Kraus operators, and we will see that this is optimal in the sense that a general map requires at least $mn$ such operators.

We will often think of $M_n$ itself as a Hilbert space with the so-called \emph{Hilbert--Schmidt inner product} defined by $\langle X | Y \rangle := \Tr(X^\dagger Y)$ for $X,Y \in M_n$. With this inner product in mind, we can define the \emph{dual} map $\Phi^\dagger$ of a superoperator $\Phi : M_m \rightarrow M_n$ as the unique map satisfying $\Tr\big( \Phi(X)^\dagger Y \big) = \Tr\big( X^\dagger \Phi^\dagger(Y) \big)$ for all $X \in M_m$, $Y \in M_n$. If we write $\displaystyle\Phi(X) = \sum_\ell A_\ell X B_\ell^\dagger$, then we have
\begin{align*}
	\Tr\big( \Phi(X)^\dagger Y \big) = \Tr\left( \Big(\sum_\ell A_\ell X B_\ell^\dagger\Big)^\dagger Y \right) = \Tr\left( X^\dagger \Big( \sum_\ell A_\ell^\dagger Y B_\ell\Big) \right).
\end{align*}
It follows that the dual map has the generalized Choi--Kraus representation $\displaystyle\Phi^\dagger(Y) = \sum_\ell A_\ell^\dagger Y B_\ell$.

Many interesting properties of a superoperator $\Phi$ will depend on whether or not we can write it in the generalized Choi--Kraus form~\eqref{eq:CK_form} with families of operators $\big\{A_\ell\big\}$ and $\big\{B_\ell\big\}$ with certain properties. For example, the case when $A_\ell = B_\ell$ for all $\ell$ corresponds exactly to $\Phi$ being completely positive -- a property that we will introduce in Section~\ref{sec:cp_maps}. Other important properties of $\Phi$ that follow from elementary properties of the generalized Choi--Kraus operators will be explored in Section~\ref{sec:choi_jamiolkowski_basic}.

\section{Representing Quantum Information}\label{sec:quantum_info}

	We now introduce how quantum information is represented and manipulated mathematically in finite-dimensional systems. We will assume a familiarity with most basic concepts from linear algebra, such as the singular value and spectral decompositions, and tensor and Kronecker products. Other introductions to quantum information theory from perspectives similar to ours can be found in \cite{NC00,Wat04}.

\subsection{State Vectors and Density Operators}\label{sec:quantum_states}

In the Schr\"{o}dinger picture of quantum mechanics, quantum information is contained in quantum \emph{states}, which come in two varieties: \emph{pure} and \emph{mixed}. Pure quantum states are represented mathematically by unit vectors $\ket{v} \in \mathbb{C}^n$ and are typically the states that one wishes to work with. Note that the vector defining a pure state is defined only up to ``global phase'' -- that is, $\ket{v}$ and $e^{i\theta}\ket{v}$ represent the same quantum state regardless of the value of $\theta \in [0,2\pi)$.

Although pure states are desirable most of the time, once quantum states are measured and manipulated, they can decohere and become mixed. A mixed quantum state is represented via a \emph{density matrix} $\displaystyle\rho := \sum_\ell p_\ell \ketbra{v_\ell}{v_\ell}$, where $\{p_\ell\}$ is a set of real numbers such that $0 \leq p_\ell \leq 1$ and $\displaystyle\sum_\ell p_\ell = 1$ (i.e., $\{p_\ell\}$ forms a \emph{probability distribution}). It is not difficult to see that any matrix of this form is Hermitian, positive semidefinite, and has trace is equal to one. Conversely, the spectral decomposition theorem shows that every positive semidefinite matrix with trace one can be written as such a sum and thus represents a mixed quantum state.

Representing pure states by vectors and mixed states by matrices perhaps seems unnatural at first. However, we note that a pure state $\ket{v}$ can also be represented by the density matrix $\ketbra{v}{v}$. In fact, representing a pure state in this way highlights the non-uniqueness of vector representations up to global phase, as $\ketbra{v}{v} = (e^{i\theta}\ket{v})(\overline{e^{i\theta}}\bra{v})$ for all $\theta \in [0,2\pi)$. In general, if we refer to a quantum state without additional qualification, we are allowing for it to be mixed and thinking of it as a density matrix. If it is important that the state is pure then we will either specify that it is pure or make it clear that we are using a vector representation of the state.

\subsection{Positive and Completely Positive Maps}\label{sec:cp_maps}

A linear map $\Phi : M_m \rightarrow M_n$ such that $\Phi(X) \geq 0$ whenever $X \geq 0$ is said to be \emph{positive}, and we will see that maps of this type are ubiquitous in quantum information theory. If we let $id_k : M_k \rightarrow M_k$ (or simply $id$, if we do not wish to emphasize the dimension) denote the identity map, then $\Phi$ is called \emph{$k$-positive} if the map $id_k \otimes \Phi$ is positive. Finally, $\Phi$ is called \emph{completely positive} if it is $k$-positive for all $k \geq 1$.

One particularly important family of completely positive maps are maps of the form ${\rm Ad}_A(X) := AXA^\dagger$, where $A \in M_{n,m}$ is fixed -- we call such a map ${\rm Ad}_A$ an \emph{adjoint map}. To see that adjoint maps are completely positive, simply note that $(id_k \otimes {\rm Ad}_A)(X) = (I_k \otimes A)X(I_k \otimes A)^\dagger$, which is positive semidefinite whenever $X$ is positive semidefinite. Because the set of completely positive maps is easily seen to be convex, we similarly see that any map of the form $\sum_k {\rm Ad}_{A_k}$ is completely positive. We will see in Theorem~\ref{thm:choi_cp} that in fact all completely positive maps can be written in this form.

A completely positive map that is trace-preserving (i.e., $\Tr(\Phi(X)) = \Tr(X)$ for all $X$) is called a \emph{quantum channel}, as such maps represent the evolution of quantum states in the Schr\"{o}dinger picture of quantum dynamics \cite{NC00} -- a fact that is fairly intuitive, as we saw that density operators are characterized by being positive semidefinite and having trace one, so a quantum channel should preserve at least these two properties. The reason that $\Phi$ must be completely positive (rather than just positive) is because $\Phi$ should preserve positive semidefiniteness even if it is only applied to part of a quantum state (after all, the system that is evolving may be entangled with another system that we don't have direct access to).

The following characterization theorem for completely positive maps \cite{C75,K71} is fundamental in quantum information theory. Condition (c) provides a simple test for determining whether or not a given map is completely positive, while condition (d) gives a simple structure for such maps. The interested reader is directed to \cite{P03} for a detailed discussion of the structure of completely positive maps and for infinite-dimensional analogues of this result.
\begin{thm}\label{thm:choi_cp}
  Let $\Phi : M_m \rightarrow M_n$ be a linear map and consider the pure state $\ket{\psi_{+}} := \frac{1}{\sqrt{m}}\sum_{i=1}^{m} \ket{i} \otimes \ket{i}$. The following are equivalent:
  \begin{enumerate}[(a)]
  	\item $\Phi$ is completely positive;
  	\item $\Phi$ is $m$-positive;
  	\item the operator $C_\Phi := m(id_m \otimes \Phi)(\ketbra{\psi_{+}}{\psi_{+}}\big)$ is positive semidefinite; and
  	\item there exist operators $\{A_k\}_{k=1}^{mn}$ such that $\Phi = \sum_{k=1}^{mn} {\rm Ad}_{A_k}$.
  \end{enumerate}
\end{thm}
\begin{proof}
	The implications (d) $\implies$ (a) $\implies$ (b) $\implies$ (c) follow easily from the relevant definitions, so all we need to prove is (c) $\implies$ (d).
	
	To this end, note that because $C_\Phi$ is positive semidefinite we can use the spectral decomposition theorem to write
	\begin{align}\label{eq:choi_spectral}
		C_\Phi & = \sum_{k=1}^{mn} \lambda_k \ketbra{v_k}{v_k}.
	\end{align}
	We can then write each $\ket{v_k}$ as a linear combination of elementary tensors: $\ket{v_k} = \sum_{j=1}^{m} c_{kj}\ket{j} \otimes \ket{v_{kj}}$. If we multiply $C_\Phi$ on the left by $\bra{i} \otimes I$ and on the right by $\ket{j} \otimes I$ (abusing notation slightly), then from the definition of $C_\Phi$ we have
	\begin{align}\label{eq:choi_reduce1}
		(\bra{i} \otimes I)C_\Phi(\ket{j} \otimes I) & = \Phi\big(\ketbra{i}{j}\big).
	\end{align}
	
	Similarly, from Equation~\eqref{eq:choi_spectral} we have
	\begin{align}\label{eq:choi_reduce2}
		(\bra{i} \otimes I)C_\Phi(\ket{j} \otimes I) & = \sum_{k=1}^{mn} \lambda_k c_{ki}\overline{c_{kj}} \ketbra{v_{ki}}{v_{kj}} \\
		& = \sum_{k=1}^{mn} \lambda_k \big(\sum_{\ell=1}^{m}c_{k\ell}\ketbra{v_{k\ell}}{\ell}\big) \ketbra{i}{j} \big( \sum_{\ell=1}^{m}\overline{c_{k\ell}}\ketbra{\ell}{v_{k\ell}} \big).
	\end{align}
	
	If we define $A_k := \sqrt{\lambda_k}\sum_{\ell=1}^{m}c_{k\ell}\ketbra{v_{k\ell}}{\ell}$ then it follows by equating Equations~\eqref{eq:choi_reduce1} and~\eqref{eq:choi_reduce2} that $\Phi\big(\ketbra{i}{j}\big) = \sum_{k=1}^{mn} A_k \ketbra{i}{j} A_k^\dagger$. Extending by linearity shows that $\Phi(X) = \sum_{k=1}^{mn} A_k X A_k^\dagger$ for all $X \in M_m$, which completes the proof.
\end{proof}

\begin{remark}\label{rem:cp_unitary_freedom}
	{\rm The operators $\{A_k\}_{k=1}^{mn}$ of condition (d) are called \emph{Kraus operators} for $\Phi$ after \cite{K71,Kra83}, where they were extensively studied. Kraus operators are not in general unique, but two sets of Kraus operators are related to each other via a unitary matrix \cite[Theorem 8.2]{NC00}. More specifically, two sets of Kraus operators $\{A_\ell\}_{\ell=1}^{mn}$ and $\{B_\ell\}_{\ell=1}^{mn}$ correspond to the same completely positive map if and only if there is a unitary matrix $(u_{i,j})$ such that $A_\ell = \sum_{j=1}^{mn} u_{\ell,j} B_j$.

	Nonetheless, the proof of Theorem~\ref{thm:choi_cp} demonstrated the existence of a particular family of Kraus operators that arise from the eigenvectors of the operator $C_\Phi$. We will refer to this family of Kraus operators as the \emph{canonical} set of Kraus operators, and we note that they are mutually orthogonal in the Hilbert--Schmidt inner product (i.e., $\Tr(A_i^\dagger A_j) = 0$ if $i \neq j$).
}\end{remark}

Condition (d) of Theorem~\ref{thm:choi_cp} says that the extreme points of the convex set of completely positive maps are the adjoint maps. Indeed, the adjoint maps are exactly the maps $\Phi$ such that ${\rm rank}(C_\Phi) \leq 1$ -- because the zero operator and the rank one positive semidefinite operators are the extreme points of the set of positive semidefinite operators, the adjoint maps are the extreme points of set of completely positive maps.

If $\Phi$ is a completely positive linear map with Kraus operators $\{A_k\}_{k=1}^{mn}$, then we observe that it is trace-preserving (i.e., a quantum channel) if and only if
\begin{align*}
	\Tr(X) = \Tr\big(\Phi(X)\big) = \Tr\Big( \sum_{k=1}^{mn}A_k X A_k^\dagger \Big) = \Tr\Big( X \sum_{k=1}^{mn}A_k^\dagger A_k \Big) \quad \forall \, X \in M_m.
\end{align*}
It follows that $\Phi$ is trace-preserving if and only if $\sum_{k=1}^{mn}A_k^\dagger A_k = I$. On the other hand, the condition $\sum_{k=1}^{mn}A_k A_k^\dagger = I$ corresponds to $\Phi$ being \emph{unital} (i.e., $\Phi(I) = I$). Trace-preserving maps and unital maps are dual in the sense that $\Phi$ is trace-preserving if and only if $\Phi^\dagger$ is unital, and vice-versa.

The matrix $C_\Phi$ of condition (c) of Theorem~\ref{thm:choi_cp} is called the \emph{Choi matrix} of $\Phi$ -- a concept that we will explore in much more depth in Section~\ref{sec:choi_jamiolkowski_isomorphism}. For now, we present some examples to make use of Theorem~\ref{thm:choi_cp}.
\begin{exam}\label{exam:transpose_not_cp}{\rm
 Let $T : M_2 \rightarrow M_2$ be the $2 \times 2$ transpose map. It is easy to see that the eigenvalues of $X$ are exactly the eigenvalues of $T(X)$ for any $X \in M_2$, so $T$ is a positive map. To determine whether or not $T$ is completely positive we use condition (c) of Theorem~\ref{thm:choi_cp}:
\begin{align*}\sspp
	2(id_2 \otimes T)(\ketbra{\psi_+}{\psi_+}) = (id_2 \otimes T)\left(\begin{bmatrix} 1 & 0 & 0 & 1 \\ 0 & 0 & 0 & 0 \\ 0 & 0 & 0 & 0 \\ 1 & 0 & 0 & 1 \end{bmatrix}\right) = \begin{bmatrix} 1 & 0 & 0 & 0 \\ 0 & 0 & 1 & 0 \\ 0 & 1 & 0 & 0 \\ 0 & 0 & 0 & 1 \end{bmatrix},
\dsp\end{align*}
	which is not positive (its eigenvalues are $1, 1, 1,$ and $-1$). It follows that $T$ is not completely positive. In fact, we can embed this example in higher dimensions to see that the transpose map in any dimension is always positive but not $2$-positive. We will see another method of showing that the transpose map is never $2$-positive in Example~\ref{ex:trans_not_2pos}. We will see a map that is $k$-positive but not $(k+1)$-positive for any fixed $k$ in Example~\ref{exam:k_pos}.}
\end{exam}

\begin{exam}\label{exam:tr_cp}{\rm
	Define $\Phi : M_n \rightarrow M_n$ by $\Phi(X) = \frac{1}{n}\Tr(X)I$. To see that $\Phi$ is trace-preserving, observe that $\Tr(\Phi(X)) = \frac{1}{n}\Tr(X)\Tr(I) = \Tr(X)$. To see that it is completely positive, we construct its Choi matrix:
	\begin{align*}
		C_\Phi = \sum_{i,j=1}^{n} \ketbra{i}{j} \otimes \Phi(\ketbra{i}{j}) = \frac{1}{n}\sum_{i=1}^{n} \ketbra{i}{i} \otimes I = \frac{1}{n}I \otimes I \geq 0.
	\end{align*}
	It follows from Theorem~\ref{thm:choi_cp} that $\Phi$ is completely positive and thus is a quantum channel. In fact, $\Phi$ is known as the \emph{completely depolarizing channel} because it turns any density matrix into $\frac{1}{n}I$.
 
  Because $\Phi$ is completely positive, it must have a family of Kraus operators. One such family of operators is $\{\tfrac{1}{\sqrt{n}}\ketbra{i}{j}\}_{i,j=1}^{n}$, which can be seen as follows:
  \begin{align*}
  	\frac{1}{n}\sum_{i,j=1}^{n} \ketbra{i}{j} X \ketbra{j}{i} = \frac{1}{n}\sum_{i=1}^{n}\ketbra{i}{i} \sum_{j=1}^{n}\bra{j}X\ket{j} = \frac{1}{n}\Tr(X)I = \Phi(X) \quad \forall \, X \in M_n.
  \end{align*}
  
	To highlight the non-uniqueness of families of Kraus operators, we now show that $\big\{ \tfrac{1}{\sqrt{n}}A_{k} \big\}_{k=1}^{n^2}$ is a family of Kraus operators for $\Phi$ whenever $\big\{ A_k \big\}_{k=1}^{n^2}$ is a family of matrices that form an orthonormal basis for $M_n$ under the Hilbert--Schmidt inner product. That is, whenever
\begin{align}\label{eq:trace01}
  \Tr(A_{k}^{\dagger}A_{\ell}) = \delta_{k,\ell} \quad \forall \, 1 \leq k,\ell \leq n^2,
\end{align}
where $\delta_{k,\ell}$ is the Kronecker delta. To this end, let $\big\{ B_k \big\}_{k=1}^{n^2}$ be any other orthonormal basis for $M_n$ and write each $A_k$ as a linear combination of elements from $\big\{ B_k \big\}_{k=1}^{n^2}$:
\begin{align*}
  A_{k} = \sum_{i=1}^{n^2} u_{ik}B_{i} \quad \forall \, 1 \leq k \leq n^2,
\end{align*}
where $\{u_{ij}\}_{i,j=1}^{n^2} \subset \mathbb{C}$ is a family of constants. By Equation~\eqref{eq:trace01} we then have
\begin{align*}
	\delta_{k,\ell} = \Tr(A_{k}^{\dagger}A_{\ell}) = \sum_{i,j=1}^{n^2} \overline{u_{i k}} u_{j \ell} \Tr\big(B_{i}^\dagger B_j \big) = \sum_{i=1}^{n^2} \overline{u_{i k}} u_{i \ell},
\end{align*}
from which it follows that $(u_{ij})$ is a unitary matrix. Remark~\ref{rem:cp_unitary_freedom} then says that $\big\{ \tfrac{1}{\sqrt{n}}A_k \big\}_{k=1}^{n^2}$ and $\big\{ \tfrac{1}{\sqrt{n}}B_k \big\}_{k=1}^{n^2}$ represent the same completely positive map. Finally, choosing $\big\{ B_k \big\}_{k=1}^{n^2} = \big\{\ketbra{i}{j}\big\}_{i,j=1}^{n}$ shows that $\big\{ \frac{1}{\sqrt{n}}A_{k} \big\}_{k=1}^{n^2}$ is a family of Kraus operators for $\Phi$ whenever $\big\{A_k\big\}_{k=1}^{n^2}$ is an orthonormal basis of $M_n$.}
\end{exam}

Even though Theorem~\ref{thm:choi_cp} provides a simple characterization of completely positive maps as well as a simple test for determining whether or not a given map is completely positive, the closely related problem of characterizing positive maps (or $k$-positive maps for some $k < m$) is much more difficult. It is known that if $m = 2$ and $n \in \{2,3\}$ then $\Phi$ is positive if and only if it can be written in the form $\Phi = \Psi_1 + T \circ \Psi_2$, where $\Psi_1,\Psi_2 : M_m \rightarrow M_n$ are completely positive maps \cite{S63,W76}. In higher dimensions, many partial results are known \cite{BFP04,CK09,CK11,Hou10,Maj11,MM01,Ter00,TT83,Tom85}, but the structure of the set of positive maps is still not well understood.

\subsection{The Stinespring Form and Complementary Maps}\label{sec:stinespring}

Although we will most frequently use the characterization provided by Theorem~\ref{thm:choi_cp} when dealing with completely positive maps, we now present another characterization that has one very specific advantage for our purposes -- it allows us to define complementary maps.

Before proceeding, recall the partial trace map $\Tr_i$ that traces out the $i$-th subsystem of $M_m \otimes M_n$. For example, when $i = 2$ the map $\Tr_2$ is the linear map that acts on elementary tensors as $\Tr_2(A \otimes B) = \Tr(B) A$. Notice that the trace map $\Tr : M_m \rightarrow \mathbb{C}$ is completely positive, since its Choi matrix is $I_m \geq 0$. It follows that the partial trace map $\Tr_i$ is also completely positive, because (for example, when $i = 2$) $id_k \otimes Tr_2 = id_k \otimes id_m \otimes \Tr = id_{km} \otimes \Tr$ is positive for any $k,m \in \mathbb{N}$.

The following result of Stinespring says that all completely positive maps can be written as a composition of a partial trace map and an adjoint map. In this sense, the partial trace is one of the most fundamental completely positive maps, much like the adjoint maps that played a key role in the previous section. This result was originally proved in the infinite-dimensional case \cite{P03,Sti55}, but we state and prove it only in finite dimensions.
\begin{thm}[Stinespring]\label{thm:stinespring}
	Let $\Phi : M_m \rightarrow M_n$ be a linear map. Then $\Phi$ is completely positive if and only if there exists $A : \mathbb{C}^m \rightarrow \mathbb{C}^{mn} \otimes \mathbb{C}^n$ such that $\Phi = \Tr_1 \circ {\rm Ad}_A$.
\end{thm}
\begin{proof}
	We already showed that adjoint maps and the partial trace map are completely positive. Because the composition of two completely positive maps is again completely positive, the ``if'' direction of the result follows immediately.
	
	For the ``only if'' direction, suppose $\Phi : M_m \rightarrow M_n$ is completely positive. By Theorem~\ref{thm:choi_cp} we know that there exists a family of operators $\big\{A_\ell\big\}_{\ell=1}^{mn}$ we can write $\Phi = \sum_{\ell=1}^{mn}{\rm Ad}_{A_\ell}$. Now define an operator $A : \mathbb{C}^m \rightarrow \mathbb{C}^{mn} \otimes \mathbb{C}^n$ by
	\begin{align*}
		A\ket{i} = \sum_{\ell=1}^{mn} \ket{\ell} \otimes A_\ell \ket{i}.
	\end{align*}
	Then
	\begin{align*}
		(\Tr_1 \circ {\rm Ad}_A)(\ketbra{i}{j}) = \sum_{k,\ell=1}^{mn} \Tr_1 \left( \ketbra{k}{\ell} \otimes A_k \ketbra{i}{j} A_{\ell}^\dagger \right) = \sum_{\ell=1}^{mn} A_\ell \ketbra{i}{j} A_{\ell}^\dagger = \Phi(\ketbra{i}{j})
	\end{align*}
	for all $1 \leq i,j \leq m$. Linearity then shows that $\Tr_1 \circ {\rm Ad}_A = \Phi$, as desired.
\end{proof}

We will refer to the form $\Phi = \Tr_1 \circ {\rm Ad}_A$ as a \emph{Stinespring representation} of the completely positive map $\Phi$. It is worth dwelling on the construction of the operator $A$ a little bit. Through the appropriate (na\"{i}ve) identification of spaces, we can represent the operator $A$ of Theorem~\ref{thm:stinespring} as the following block matrix in $M_{mn,1}(M_{n,m}) \cong M_{mn,1} \otimes M_{n,m}$:
\begin{align*}\sspp
	A = \begin{bmatrix}A_1 \\ A_2 \\ \vdots \\ A_{mn}\end{bmatrix},
\dsp\end{align*}
where $\big\{A_\ell\big\}_{\ell=1}^{mn}$ is a Kraus representation of $\Phi$, as before. This representation of $A$ makes is clear how to go back and forth between a Stinespring and a Kraus representation of a completely positive map. Notice that trace-preservation of $\Phi$ corresponds to $A$ being an isometry (i.e., $A^\dagger A = I_m$).

The Stinespring form makes it clearer where the unitary freedom in Kraus operators, discussed in Remark~\ref{rem:cp_unitary_freedom}, comes from. If $\Tr_1 \circ {\rm Ad}_{A}$ is a Stinespring representation of a completely positive map $\Phi$, then $\Tr_1 \circ {\rm Ad}_{(U \otimes I_n)A}$ (with $U = (u_{i,j}) \in \cl{M}_{mn}$ a unitary matrix) is another Stinespring representation of $\Phi$, since $U$ is traced out by the partial trace. By constructing the operator $A$ as above, we see that the operators
\begin{align*}\sspp
	A = \begin{bmatrix}A_1 \\ A_2 \\ \vdots \\ A_{mn}\end{bmatrix} \quad \text{ and } \quad (U \otimes I_n)A = \begin{bmatrix}\sum_{j}u_{1,j}A_j \\ \sum_{j}u_{2,j}A_j \\ \vdots \\ \sum_{j}u_{mn,j}A_j\end{bmatrix}
\dsp\end{align*}
provide Stinespring representations for the same completely positive map. Thus the families of Kraus operators $\big\{A_\ell\big\}_{\ell=1}^{mn}$ and $\big\{\sum_{j=1}^{mn}u_{\ell,j}A_j\big\}_{\ell=1}^{mn}$ represent the same map, as was discussed earlier. 

Given a Stinespring representation of the map $\Phi = \Tr_1 \circ {\rm Ad}_A$, if we take the partial trace over the second subsystem rather than the first, we obtain the \emph{complementary map} $\Phi^{C} := \Tr_{2} \circ {\rm Ad}_A$. It is easily-verified that $\Phi$ is a quantum channel if and only if $\Phi^C$ is a quantum channel, in which case complementary maps have a well-defined physical interpretation. If Alice sends quantum information to Bob via the quantum channel $\Phi$, then the complementary channel $\Phi^C$ describes the information that is leaked during that transmission.

Because of the unitary-invariance of Stinespring representations, complementary maps are not uniquely defined, but rather are only defined up to unitary conjugation. That is, if $\Psi$ is a complementary map of $\Phi$, then so is ${\rm Ad}_U \circ \Psi$ for any unitary matrix $U \in M_{mn}$. This freedom up to unitary conjugation does not affect most uses of complementary maps, however, so we will ignore this technicality when possible and still speak of \emph{the} complementary map $\Phi^C$. It is easily-verified that complementary maps are dual in the sense that $\Phi$ is a complementary map of $\Phi^{C}$.

We close this section with a simple lemma that describes how adjoint maps and complementary maps behave under composition.
\begin{lemma}\label{lem:comp_adjoint}
	Let $\Phi : M_m \rightarrow M_n$ be completely positive and let $B \in M_m$. Then $(\Phi \circ {\rm Ad}_B)^C = \Phi^C \circ {\rm Ad}_B$.
\end{lemma}
\begin{proof}
	Suppose $\Phi$ has Stinespring representation $\Phi = \Tr_{1} \circ {\rm Ad}_A$. Then $\Phi \circ {\rm Ad}_B = \Tr_1 \circ {\rm Ad}_A \circ {\rm Ad}_B = \Tr_1 \circ {\rm Ad}_{AB}$, which is in Stinespring form. Thus $(\Phi \circ {\rm Ad}_B)^C = \Tr_2 \circ {\rm Ad}_{AB} = \Tr_2 \circ {\rm Ad}_{A} \circ {\rm Ad}_{B} = \Phi^C \circ {\rm Ad}_{B}$.
\end{proof}

\section{Representing Quantum Entanglement}\label{sec:quantum_ent}

	Within quantum information theory, the theory of entanglement \cite{BZ06,EPR35,HHH09,Sch35} is one of the most important and active areas of research. Entanglement leads to many of the most counter-intuitive and important properties and protocols of quantum information, such as superdense coding \cite{BW92} and quantum teleportation \cite{BBCJPW93,Vai94}. In this section we will introduce the mathematical formulation of entanglement in quantum systems.

A pure state $\ket{v} \in \mathbb{C}^m \otimes \mathbb{C}^n$ is called \emph{separable} if it can be written as an elementary tensor: $\ket{v} = \ket{a} \otimes \ket{b}$ for some $\ket{a} \in \mathbb{C}^m$ and $\ket{b} \in \mathbb{C}^n$. Otherwise, $\ket{v}$ is said to be \emph{entangled}. In the case of mixed stated, we say that $\rho \in M_m \otimes M_n$ is separable if it can be written as a convex combination of separable pure states \cite{W89}:
\begin{align}\label{eq:sep_form}
	\rho = \sum_{\ell} p_\ell \ketbra{a_\ell}{a_\ell} \otimes \ketbra{b_\ell}{b_\ell},
\end{align}
where $\{p_\ell\}$ forms a probability distribution. Otherwise, $\rho$ is called entangled. Slightly more generally, we will refer to an operator $X \geq 0$ (not necessarily with trace one) as separable if it can be written in the form $X = \displaystyle\sum_\ell Y_\ell \otimes Z_\ell$ with $Y_\ell, Z_\ell \geq 0$ for all $\ell$.

It should be pointed out that in general there is no relationship between the form~\eqref{eq:sep_form} of a separable operator and its spectral decomposition. If a density operator has separable eigenvectors then it certainly is separable, but the converse is not true -- there are separable density operators with no basis of separable eigenvectors.

The problem of determining whether or not a density matrix is separable is a problem that has received a lot of attention in recent years. While it is known that this problem is hard in general \cite{G10,G03,Ioa07}, many tests have been derived that work in certain special cases \cite{CW03,DPS04,HHH96,HHH06,P96,R03}. We will investigate some of these methods in this section, as well as in Sections~\ref{sec:realign} and~\ref{sec:contrac_sep_crit}.

\subsection{Vector-Operator Isomorphism}\label{sec:vector_operator_isomorphism}

The \emph{vector-operator isomorphism} is a valuable tool that will be used throughout this thesis to introduce many concepts from entanglement theory via fundamental and well-known results from linear algebra. It will also allow us to use classical linear preserver problems to help us answer questions about preservers and isometry groups that are relevant in entanglement theory. The key idea of the vector-operator isomorphism is that we can bring matrices and superoperators ``down a level'' by thinking about matrices as vectors and by thinking about superoperators as matrices, which makes them easier to deal with in many situations.

Consider the linear map $\Gamma : \bb{C}^m \otimes \bb{C}^n \rightarrow M_{n,m}$ defined on the standard basis by $\Gamma(\ket{i} \otimes \ket{j}) = \ketbra{j}{i}$. Because $\big\{\Gamma(\ket{i} \otimes \ket{j})\big\}$ is a basis of $M_{n,m}$, and it is easily-verified that $\braket{v}{w} = \Tr \big( \Gamma(\ket{v})^\dagger \Gamma(\ket{w}) \big)$, this map is a isomorphism -- the vector-operator isomorphism. By linearity, the vector-operator isomorphism associates an elementary tensor $\ket{a} \otimes \ket{b} \in \bb{C}^m \otimes \bb{C}^n$ with the rank-$1$ matrix $\ket{b}\overline{\bra{a}} \in M_{n,m}$, and associates a general bipartite vector $\ket{v} := \sum_i c_i \ket{a_i} \otimes \ket{b_i}$ ($c_i \in \bb{R}$) with the matrix $\sum_i c_i\ket{b_i}\overline{\bra{a_i}}$, which is called the \emph{matricization} of $\ket{v}$ and will be denoted by ${\rm mat}(\ket{v})$. In fact, we have already seen this isomorphism in action: in the proof of Theorem~\ref{thm:choi_cp} we defined the Kraus operator $A_k$ to be (up to scaling) the matricization of the eigenvector $\ket{v_k}$ of $C_\Phi$.

When thinking of the vector-operator isomorphism in reverse, the term \emph{vectorization} is often used. That is, $\ket{v}$ is called the vectorization of ${\rm mat}(\ket{v})$, and we denote the vectorization operator by ${\rm vec}(\cdot)$. It is worth noting that, in the standard basis, the vectorization of a matrix $X \in M_{n,m}$ is the $mn$-dimensional vector obtained by stacking the columns of $X$ on top of one another. Conversely, the matricization of a vector $\ket{v} \in \bb{C}^m \otimes \bb{C}^n$ is $n \times m$ matrix obtained by placing the first $n$ entries of $\ket{v}$ in its first column, the next $n$ entries of $\ket{v}$ in its second column, and so on.

The vector-operator isomorphism is isometric if the norm on $\bb{C}^m \otimes \bb{C}^n$ is the Euclidean norm and the norm on $M_{n,m}$ is taken to be the Frobenius norm $\big\|(x_{ij})\big\|_F := \sqrt{\sum_{i=1}^n\sum_{j=1}^m x_{ij}^2} = \sqrt{\sum_{i=1}^{\min\{m,n\}} \sigma_i^2}$, where $\{\sigma_i\}_{i=1}^{\min\{m,n\}}$ are the singular values of $(x_{ij})$.

\begin{exam}\label{ex:vec_op_iso}
  {\rm Consider the pure state $\ket{\psi_+} := \frac{1}{\sqrt{m}}\sum_{i=1}^{m} \ket{i} \otimes \ket{i} \in \bb{C}^m \otimes \bb{C}^m$ from Theorem~\ref{thm:choi_cp}. The vector-operator isomorphism gives ${\rm mat}(\ket{\psi_+}) = \frac{1}{\sqrt{m}}\sum_{i=1}^{m} \ketbra{i}{i} = \frac{1}{\sqrt{m}}I$ -- a scaled identity matrix. To illustrate the opposite direction of the isomorphism, let us fix $m = 2$. The vectorization of $\frac{1}{\sqrt{2}}I$ is obtained by stacking its first column on top of its second column in the standard basis, which gives ${\rm vec}\big(\frac{1}{\sqrt{2}}I\big) = \frac{1}{\sqrt{2}}(1,0,0,1)^T = \ket{\psi_+}$.}
\end{exam}

If we wish to think about $X \in M_{n,m}$ as a vector ${\rm vec}(X) \in \mathbb{C}^m \otimes \mathbb{C}^n$ via the vector-operator isomorphism, it would be beneficial to understand how a superoperator $\Phi : M_{n,m} \rightarrow M_{n,m}$ appears when represented as an operator in $M_m \otimes M_n$. That is, what is the form of the operator $M_\Phi \in M_m \otimes M_n$ with the property that $M_\Phi{\rm vec}(X) = {\rm vec}(\Phi(X))$ for all $X \in M_{n,m}$? To answer this question, write $\Phi(X) = \sum_k A_k X B_k^\dagger$ for some generalized Choi--Kraus operators $\big\{A_k\big\} \subset M_n$ and $\big\{B_k\big\} \subset M_m$. Then
\begin{align*}
	{\rm vec}(\Phi(\ketbra{i}{j})) & = \sum_k {\rm vec}(A_k\ketbra{i}{j}B_k^\dagger) \\
	& = \sum_k \overline{B_k}\ket{j} \otimes A_k\ket{i} \\
	& = \left(\sum_k \overline{B_k} \otimes A_k\right) {\rm vec}(\ketbra{i}{j}) \quad \forall \, 0 \leq i < n, 0 \leq j < m.
\end{align*}

Extending by linearity then shows that $\left(\sum_k \overline{B_k} \otimes A_k\right) {\rm vec}(X) = {\rm vec}(\Phi(X))$ for all $X \in M_{n,m}$, so the operator we seek is $M_\Phi := \sum_k \overline{B_k} \otimes A_k$. The association between $\Phi$ and $M_\Phi$ is an isomorphism, which we will generally just consider part of the vector-operator isomorphism itself.

\subsection{Schmidt Rank and Pure State Entanglement}\label{sec:schmidt_rank}

	We have already seen that a pure state $\ket{v} \in \bb{C}^m \otimes \bb{C}^n$ is called separable if it can be written in the form $\ket{v} = \ket{a} \otimes \ket{b}$, and it is called entangled otherwise. The notion of Schmidt rank extends that of separability: the \emph{Schmidt rank} of a pure state $\ket{v} \in \bb{C}^m \otimes \bb{C}^n$, written $SR(\ket{v})$, is defined as the least $k$ such that we can write $\ket{v}$ as a linear combination of $k$ separable pure states. Although this definition perhaps seems difficult to use at first glance, the Schmidt decomposition theorem \cite[Theorem 2.7]{NC00} provides a simple method of computing Schmidt rank. It also provides a useful orthogonal form for all bipartite pure states. As will be seen in its proof, the Schmidt decomposition theorem is essentially the singular value decomposition theorem in disguise.
	\begin{thm}[Schmidt decomposition]\label{thm:schmidt}
		For any unit vector $\ket{v} \in \bb{C}^m \otimes \bb{C}^n$ there exists $k \leq \min\{m,n\}$, non-negative real scalars $\{\alpha_i\}_{i=1}^k$ with $\sum_{i=1}^k \alpha_i^2 = 1$, and orthonormal sets of vectors $\{ \ket{a_i} \}_{i=1}^k \subset \bb{C}^m$ and $\{ \ket{b_i} \}_{i=1}^k \subset \bb{C}^n$ such that
		\begin{align*}
			\ket{v} = \sum_{i=1}^{k}\alpha_i \ket{a_i} \otimes \ket{b_i}.
		\end{align*}
	\end{thm}
	\begin{proof}
		Assume that $n \leq m$, as it will be clear how to modify the proof if the opposite inequality holds. By the singular value decomposition, there exist unitaries $U \in M_n$, $V \in M_m$, and a positive semidefinite diagonal matrix $D \in M_n$ such that
	\begin{align*}\sspp
		{\rm mat}(\ket{v}) = U\begin{bmatrix}D & 0\end{bmatrix}V.
	\dsp\end{align*}
	Performing this matrix multiplication gives
	\begin{align*}
		{\rm mat}(\ket{v}) = \sum_{i=1}^{n}\alpha_i \ket{a_i}\overline{\bra{b_i}},
	\end{align*}
	where $\alpha_i$ is the $i$-th diagonal entry of $D$, $\ket{a_i}$ is the $i$-th column of $U$, and $\overline{\bra{b_i}}$ is the $i$-th row of $V$. Because the set $\{\alpha_i\}_{i=1}^n$ gives the singular values of ${\rm mat}(\ket{v})$, we have $\sum_{i=1}^{n}\alpha_i^2 = \big\|{\rm mat}(\ket{v})\big\|_{F}^2 = \big\|\ket{v}\big\|^2 = 1$. Since $U$ and $V$ are both unitaries, the sets $\{ \ket{a_i} \}_{i=1}^k$ and $\{ \ket{b_i} \}_{i=1}^k$ are orthonormal, and constructing the vectorization of ${\rm mat}(\ket{v})$ gives
	\begin{align*}
		\ket{v} = \sum_{i=1}^{n}\alpha_i \ket{a_i} \otimes \ket{b_i},
	\end{align*}
	which completes the proof.
\end{proof}

From the above proof, it is clear that the least possible $k$ in Theorem~\ref{thm:schmidt} is equal to the Schmidt rank of $\ket{v}$, which is equal to the rank of the matrix ${\rm mat}(\ket{v})$. Also of interest for us will be the constants $\{\alpha_i\}_{i=1}^k$, which are known as the \emph{Schmidt coefficients} of $\ket{v}$ and are equal to the singular values of ${\rm mat}(\ket{v})$.

The Schmidt rank can roughly be interpreted as the ``amount of entanglement'' contained within a pure state. A pure state is separable if and only if its Schmidt rank equals $1$, and $1 \leq SR(\ket{v}) \leq \min\{m,n\}$ for all $\ket{v} \in \bb{C}^m \otimes \bb{C}^n$. In the case when $SR(\ket{v}) = \min\{m,n\}$ and all of its Schmidt coefficients are equal (and thus equal to $1/\sqrt{\min\{m,n\}}$), we refer to $\ket{v}$ as \emph{maximally entangled}. We have already seen the maximally-entangled pure state $\ket{\psi_{+}} := \frac{1}{\sqrt{m}}\sum_{i=1}^{m} \ket{i} \otimes \ket{i} \in \bb{C}^m \otimes \bb{C}^m$ in Theorem~\ref{thm:choi_cp}, and because of its use in the construction of Choi matrices we will continue to see it throughout this work.
\begin{exam}\label{exam:k_pos}{\rm
 	Let $k,n \in \mathbb{N}$ be such that $k \leq n$ and consider the map $\Phi : M_n \rightarrow M_n$ defined by $\Phi(X) = k\Tr(X)I - X$. Using the Schmidt decomposition theorem, we now show that this map is $k$-positive but (if $k < n$) not $(k+1)$-positive. To see that $\Phi$ is not $(k+1)$-positive when $k < n$, consider its action on the projection onto the state $\ket{\psi} := \tfrac{1}{\sqrt{k+1}}\sum_{i=1}^{k+1}\ket{i} \otimes \ket{i} \in \mathbb{C}^{k+1} \otimes \mathbb{C}^n$:
 	\begin{align*}
 		(id_{k+1} \otimes \Phi)(\ketbra{\psi}{\psi}) & = \frac{1}{k+1}\sum_{i,j=1}^{k+1} \ketbra{i}{j} \otimes \Phi\big(\ketbra{i}{j}\big) \\
 		& = \frac{1}{k+1}\Big(k I \otimes I - \sum_{i,j=1}^{k+1} \ketbra{i}{j} \otimes \ketbra{i}{j}\Big). 
 	\end{align*}
	Because the operator $\sum_{i,j=1}^{k+1} \ketbra{i}{j} \otimes \ketbra{i}{j}$ has $k+1$ as an eigenvalue (corresponding to the eigenvector $\ket{\psi}$), we know that $(id_{k+1} \otimes \Phi)(\ketbra{\psi}{\psi})$ has $(k - (k+1))/(k+1) = -1/(k+1)$ as an eigenvalue. It follows that $(id_{k+1} \otimes \Phi)(\ketbra{\psi}{\psi})$ is not positive semidefinite even though $\ketbra{\psi}{\psi}$ is, so $\Phi$ is not $(k+1)$-positive.
	
	On the other hand, we will now show that $\Phi$ is $k$-positive. First notice that, due to linearity, it is enough to show that $(id_k \otimes \Phi)$ is positive on pure states $\ketbra{v}{v}$. Consider an arbitrary such pure state written in its Schmidt decomposition $\ket{v} = \sum_{i=1}^k \alpha_i \ket{a_i} \otimes \ket{b_i}$. Notice that $I \geq \ketbra{b_i}{b_i}$ implies that $kI - \ketbra{b_i}{b_i} \geq (k-1)\ketbra{b_i}{b_i}$. Because the $\ket{b_i}$'s are orthonormal it follows that
		\begin{align*}
			(id_k \otimes \Phi)(\ketbra{v}{v}) = & \ \sum_{i,j=1}^{k}\alpha_i \alpha_j\ketbra{a_i}{a_j} \otimes (k\braket{b_j}{b_i}I - \ketbra{b_i}{b_j}) \\
			\geq & \ \sum_{i=1}^{k}(k-1)\alpha_i^2\ketbra{a_i}{a_i} \otimes \ketbra{b_i}{b_i} - \sum_{\stackrel{i,j=1}{i \neq j}}^{k}\alpha_i \alpha_j\ketbra{a_i}{a_j} \otimes \ketbra{b_i}{b_j} \\
			= & \ \sum_{\stackrel{i,j=1}{i \neq j}}^{k}\Big(\alpha_i^2\ketbra{a_i}{a_i} \otimes \ketbra{b_i}{b_i} - \alpha_i \alpha_j\ketbra{a_i}{a_j} \otimes \ketbra{b_i}{b_j}\Big) \\
			= & \ \sum_{i=1}^{k}\sum_{j=i+1}^{k}\Big(\alpha_i^2\ketbra{a_i}{a_i} \otimes \ketbra{b_i}{b_i} - \alpha_i \alpha_j\ketbra{a_i}{a_j} \otimes \ketbra{b_i}{b_j} \\
				& \ \quad \quad \quad \quad \quad - \alpha_i \alpha_j\ketbra{a_j}{a_i} \otimes \ketbra{b_j}{b_i} + \alpha_j^2\ketbra{a_j}{a_j} \otimes \ketbra{b_j}{b_j}\Big)
		\end{align*}
This quantity can be factored as
	\begin{align*}
		\sum_{i=1}^{k}\sum_{j=i+1}^{k}\big(\alpha_i\ket{a_i} \otimes \ket{b_i} - \alpha_j\ket{a_j} \otimes \ket{b_j}\big)\big(\alpha_i\bra{a_i} \otimes \bra{b_i} - \alpha_j\bra{a_j} \otimes \bra{b_j}\big) \geq 0,
		\end{align*}
	from which it follows that $\Phi$ is $k$-positive.
	
	The map $\Phi$ is actually rather well-known in operator theory \cite{Tom85} and quantum information theory \cite{TH00} -- it was introduced in the $k = n-1$ case in \cite{Cho72} as the first known example of a map that is $n-1$ positive but not completely positive. We will see in the next section that its positivity properties play an important role in entanglement theory.
}\end{exam}

We close this section with a result that provides a tight bound on the dimension of subspaces consisting entirely of vectors with high Schmidt rank \cite{CW08,Sar08}.
\begin{thm}\label{thm:CMW08}
  The maximum dimension of a subspace $\cl{V} \subseteq \bb{C}^m \otimes \bb{C}^n$ such that $SR(\ket{v}) \geq k$ for all $\ket{v} \in \cl{V}$ is given by $(m - k + 1)(n - k + 1)$.
\end{thm}
Not only is $(m - k + 1)(n - k + 1)$ an upper bound on the dimension of such subspaces, but an explicit method of construction is known that produces such a subspace that attains the bound. This theorem will help us bound a norm based on the Schmidt rank that will be introduced in Section~\ref{sec:sk_operator_norm}.

\subsection{Operator-Schmidt Decomposition}\label{sec:operator_schmidt}

The \emph{operator-Schmidt decomposition} \cite{NDDGMOBHH03,Nie98} does for bipartite operators what the Schmidt decomposition does for bipartite vectors -- it provides a canonical, orthogonal decomposition of the operator into a sum of a minimal number of elementary tensors.

More specifically, if $X \in M_{n,m} \otimes M_{n,m}$ then we can use the vector-operator isomorphism on both copies of $M_{n,m}$ to associate $X$ with a vector $\ket{x} \in (\mathbb{C}^m \otimes \mathbb{C}^n) \otimes (\mathbb{C}^m \otimes \mathbb{C}^n)$. Applying the Schmidt decomposition theorem to $\ket{x}$ then gives $1 \leq k \leq mn$ such that
\begin{align*}
	\ket{x} = \sum_{i=1}^k \alpha_i \ket{a_i} \otimes \ket{b_i} \quad \text{for some orthonormal sets } \big\{\ket{a_i}\big\},\big\{\ket{b_i}\big\} \in \mathbb{C}^m \otimes \mathbb{C}^n
\end{align*}
and real constants $\alpha_i > 0$. Tracing this decomposition back through the vector-operator isomorphism then gives
\begin{align}\label{eq:op_schmidt_decomp}
	X = \sum_{i=1}^k \alpha_i A_i \otimes B_i,
\end{align}
where $A_i = {\rm mat}(\ket{a_i})$ and $B_i = {\rm mat}(\ket{b_i})$ for all $i$. In particular, this implies that the sets of operators $\big\{A_i\big\}$ and $\big\{B_i\big\}$ are orthonormal in the Hilbert--Schmidt inner product.

Indeed, the decomposition~\eqref{eq:op_schmidt_decomp} is the operator-Schmidt decomposition of $X$. Some sources \cite{NDDGMOBHH03} refer to the natural number $k$ as the \emph{Schmidt number} of $X$, but we will introduce another much more common usage of that term in the next section. To avoid confusion, we will instead refer to $k$ as the \emph{operator-Schmidt rank} of $X$. Similarly, we will call the coefficients $\{\alpha_i\}$ the \emph{operator-Schmidt coefficients} of $X$.

\subsection{Schmidt Number and Mixed State Entanglement}\label{sec:schmidt_number}

Although the operator-Schmidt rank provides a natural generalization of the Schmidt rank to the case of operators (i.e., mixed and pure states), it is not particularly informative as a measure of entanglement. Whereas we saw that the Schmidt rank of a pure state equals one if and only if that pure state is separable, recall that a separable mixed state $\rho \in M_m \otimes M_n$ has the form
\begin{align*}
	\rho = \sum_i p_i \sigma_i \otimes \tau_i,
\end{align*}
and so there is no clear relationship between separability of $\rho$ and the operator-Schmidt rank of $\rho$ (although we will see in Section~\ref{sec:realign} that there is a relationship between separability of $\rho$ and the norm of the operators in its operator-Schmidt decomposition).

An extension of Schmidt rank to the case of mixed states that is often much more useful and natural is the \emph{Schmidt number} \cite{TH00}. Given a density matrix $\rho \in M_m \otimes M_n$, the Schmidt number of $\rho$, denoted $SN(\rho)$, is defined to be the least natural number $k$ such that $\rho$ can be written as
\begin{align*}
  \rho = \sum_i p_i \ketbra{v_i}{v_i},
\end{align*}
where $SR(\ket{v_i}) \leq k$ for all $i$ and $\{ p_i \}$ forms a probability distribution. Much like the Schmidt rank (and unlike the operator-Schmidt rank), the Schmidt number of a state can be thought of as a rough measure of how entangled that state is. Some simple special cases include:
\begin{itemize}
	\item The state $\rho$ is separable if and only if $SN(\rho) = 1$.
	\item For a pure state $\ket{v}$ we have $SR(\ket{v}) = SN(\ketbra{v}{v})$.
\end{itemize}

One of the most active areas of research in quantum information theory is the search for operational criteria for determining whether the state $\rho$ is separable or entangled. Much progress has been made on this front over the past two decades. A landmark result of this field of study is that $\rho$ is separable if and only if it remains positive under the application of any positive map to one half of the state \cite{HHH96,P96} -- i.e., if and only if $(id_m \otimes \Phi)(\rho) \geq 0$ whenever $\Phi$ is positive. The ``only if'' direction of this result is trivial, because if we can write $\rho = \sum_i p_i \sigma_i \otimes \tau_i$ with $\sigma_{i}, \tau_{i} \geq 0$ then $\Phi(\tau_{i}) \geq 0$ and so $(id_m \otimes \Phi)(\rho) = \sum_i p_i \sigma_i \otimes \Phi(\tau_{i}) \geq 0$ as well (furthermore, $(id_m \otimes \Phi)(\rho)$ is even separable). The ``if'' direction of the result essentially follows from the separating hyperplane theorem.

An important special case of this separability criterion arises when we choose the positive map $\Phi$ to be the transpose map $T$. In this case, we refer to the operation $id_m \otimes T$ as the \emph{partial transpose}, and we use the shorthand notation $\rho^\Gamma := (id_m \otimes T)(\rho)$. In low-dimensional systems (i.e., when $nm \leq 6$), it turns out that $\rho$ is separable if and only if $\rho^\Gamma \geq 0$ \cite{HHH96,S63,W76}. That is, the only positive map that has to be used to determine separability of a low-dimensional quantum state is the transpose map. Similarly, if ${\rm rank}(\rho) \leq \max\{m,n\}$ then $\rho$ is separable if and only if $\rho^\Gamma \geq 0$ \cite{HLVC00}, but in general the partial transpose only provides a necessary but not sufficient condition for separability. The fact that the transpose map can be used to determine separability in these special cases has led to the study of \emph{positive partial transpose (PPT)} states in arbitrary dimensions, which are density operators $\rho$ such that $\rho^\Gamma \geq 0$.

The following result is a natural generalization of the characterization of separable states in terms of positive maps was implicit in \cite{TH00} and proved in \cite{RA07}.
\begin{thm}\label{thm:sch_kpos_maps}
	Let $\Phi : M_n \rightarrow M_n$ be a linear map and let $\rho \in M_n \otimes M_n$ be a density matrix. Then
	\begin{enumerate}[(a)]
		\item $\Phi$ is $k$-positive if and only if $(id_n \otimes \Phi)(\sigma) \geq 0$ for all $\sigma \in M_n \otimes M_n$ with $SN(\sigma) \leq k$, and
		\item $SN(\rho) \leq k$ if and only if $(id_n \otimes \Psi)(\rho) \geq 0$ for all $k$-positive maps $\Psi : M_n \rightarrow M_n$.
	\end{enumerate}
\end{thm}
Theorem~\ref{thm:sch_kpos_maps} establishes a duality between $k$-positive linear maps and density matrices with Schmidt number at most $k$. This duality will be explored in more generality and depth in Sections~\ref{sec:choi_jamiolkowski_schmidt} and~\ref{sec:dual_cones}.

If we focus on condition (b) of Theorem~\ref{thm:sch_kpos_maps}, we see that choosing any particular $k$-positive map $\Psi$ then gives a necessary criteria for $\rho$ to have $SN(\rho) \leq k$. For example, in the $k = 1$ case if we choose $\Psi = T$ then we see the familiar implication that if $\rho$ is separable then $\rho^\Gamma \geq 0$ (or phrased differently, if $\rho^\Gamma$ has a negative eigenvalue, then $\rho$ is entangled). Another well-known (but weaker \cite{HH99}) separability criteria is the reduction criterion~\cite{CAG99}, which states that if $\rho$ is separable then $\rho \leq \Tr_2(\rho) \otimes I$ and $\rho \leq I \otimes \Tr_1(\rho)$. Much like the partial transpose criterion arises from the transpose map, the reduction criterion arises from the positive map $\Psi(X) = \Tr(X)I - X$. A natural generalization of the reduction criterion for higher Schmidt number is that if $SN(\rho) \leq k$ then $\rho \leq k\Tr_2(\rho) \otimes I$ and $\rho \leq kI \otimes \Tr_1(\rho)$, which follows by using the $k$-positive map of Example~\ref{exam:k_pos} in condition (b) of Theorem~\ref{thm:sch_kpos_maps}.

In spite of Theorem~\ref{thm:sch_kpos_maps}, the structure of the set of separable states is still not well understood, and determining whether or not a given state separable is a difficult problem \cite{G10,G03,Ioa07} and an active area of research. We will see other well-known tests for separability in Sections~\ref{sec:realign} and~\ref{sec:shareable}, and further tests can be found in \cite{Bre06,Hal06,HHH97,QH11,R00}.

\subsection{Block Positive Operators}\label{sec:blockpos}\index{block positive|(}

We say that a Hermitian operator $X = X^\dagger \in M_m \otimes M_n$ is \emph{$k$-block positive} if
\begin{align*}
	\bra{v}X\ket{v} \geq 0 \text{ whenever } SR(\ket{v}) \leq k.
\end{align*}
Observe that if $k = \min\{m,n\}$ then this definition reduces to simply the usual notion of positive semidefiniteness. If $k < \min\{m,n\}$ then this is a strictly weaker notion of positivity in the sense that the resulting set of operators is a strict superset of the set of positive semidefinite operators. Indeed, much like the sets of operators with Schmidt number at most $k$ are nested subsets of the set of positive semidefinite operators, the sets of block positive operators are nested supersets of the set of positive semidefinite operators (see Figure~\ref{fig:schmidt}).

\begin{figure}[ht]
\begin{center}
\includegraphics[width=\textwidth]{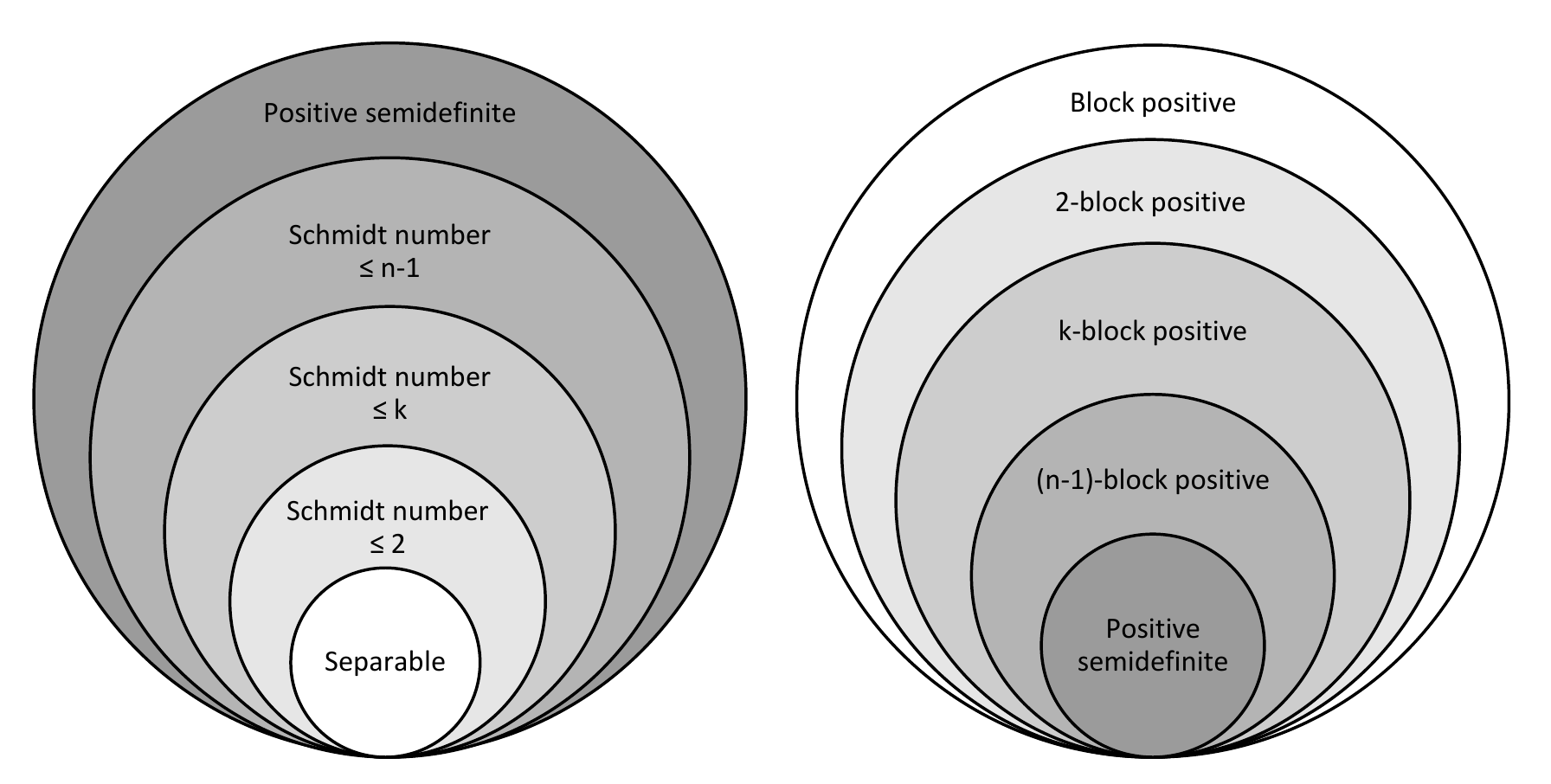}
\end{center}\vspace{-0.3in}
\caption[Block positivity and Schmidt number of operators in $M_n \otimes M_n$]{\hsp A rough depiction of the set of operators with Schmidt number at most $k$ and the sets of $k$-block positive operators in $M_n \otimes M_n$. Sets that are the same shade of gray are dual to each other in the sense of Proposition~\ref{prop:kPos}. The set of positive semidefinite operators is self-dual and equals the set of operators with Schmidt number no greater than $n$, which equals the set of $n$-block positive operators.}\label{fig:schmidt}
\end{figure}

In the $k = 1$ case, we will simply refer to operators such that $\bra{v}X\ket{v} \geq 0$ whenever $\ket{v}$ is separable as \emph{block positive} (rather than $1$-block positive). To see where this terminology comes from, it is instructive to write $X = \sum_{i,j=1}^{m} \ketbra{i}{j} \otimes X_{ij}$ where $X_{ij} \in M_n$ for all $1 \leq i,j \leq m$. Then $X$ is block positive if and only if the following inequality holds for all $\ket{a} \in \mathbb{C}^m$ and $\ket{b} \in \mathbb{C}^n$:
\begin{align*}
	(\bra{a} \otimes \bra{b})X(\ket{a} \otimes \ket{b}) & = \bra{a} \left( \sum_{i,j=1}^{m} (\bra{b}X_{ij}\ket{b}) \ketbra{i}{j} \right) \ket{a} \\
	& = \sspp\bra{a}\begin{bmatrix}\bra{b}X_{11}\ket{b} & \bra{b}X_{12}\ket{b} & \cdots & \bra{b}X_{1m}\ket{b} \\ \bra{b}X_{21}\ket{b} & \bra{b}X_{22}\ket{b} & \cdots & \bra{b}X_{2m}\ket{b} \\ \vdots & \vdots & \ddots & \vdots \\ \bra{b}X_{m1}\ket{b} & \bra{b}X_{m2}\ket{b} & \cdots & \bra{b}X_{mm}\ket{b}\end{bmatrix}\ket{a}\dsp \\
	& \geq 0.
\end{align*}

In other words, if we write $X$ as the block matrix $(X_{ij})$, then $X$ being block positive is equivalent to the matrix $(\bra{b}X_{ij}\ket{b})$ being positive semidefinite for all $\ket{b} \in \mathbb{C}^n$.
\begin{exam}\label{exam:transpose_bp}{\rm
	Let $n \geq 2$ and consider the $n \times n$ transpose map $T : M_n \rightarrow M_n$. We now show that its Choi matrix $C_T$ is block positive, even though we saw in Example~\ref{exam:transpose_not_cp} that it is not positive semidefinite:
\begin{align*}
 (\bra{a} \otimes \bra{b})C_T(\ket{a} \otimes \ket{b}) & = (\bra{a} \otimes \bra{b})\left(\sum_{i,j=1}^{n} \ketbra{i}{j} \otimes \ketbra{j}{i} \right)(\ket{a} \otimes \ket{b}) \\
 & = \sum_{i,j=1}^{n} \braket{a}{i}\braket{j}{a}\braket{b}{j}\braket{i}{b} \\
 & = \sum_{i=1}^{n} \braket{a}{i}\braket{i}{b} \sum_{j=1}^{n}\braket{b}{j}\braket{j}{a} \\
 & = \big| \braket{a}{b} \big|^2 \\
 & \geq 0.
\end{align*}
The fact that the transpose map is positive is directly related to the fact that its Choi matrix is block positive. We will make this connection explicit in Section~\ref{sec:choi_jamiolkowski_schmidt}.
}\end{exam}

We close this section with a well-known result that shows an intricate connection between $k$-block positivity of operators and the Schmidt number of operators. Because the set of operators with Schmidt number no greater than $k$ is a closed and convex subset of the set of positive semidefinite operators, the separating hyperplane theorem says that there must exist operators $\sigma,X$ (with $\sigma \geq 0$) such that $\Tr(X\rho) \geq 0$ for all $\rho$ with $SN(\rho) \leq k$ but $\Tr(X\sigma) < 0$. Indeed, the following theorem says that the separating hyperplanes $X$ are exactly the operators that are $k$-block positive but not positive semidefinite (see Figure~\ref{fig:ent_witness}). Such operators are called \emph{$k$-entanglement witnesses}, or simply \emph{entanglement witnesses} when $k = 1$.
\begin{prop}\label{prop:kPos}
	Let $X, \rho \in M_m \otimes M_n$ be such that $X = X^\dagger$ and $\rho$ is a density matrix. Then
	\begin{enumerate}[(a)]
		\item $X$ is $k$-block positive if and only if $\Tr(X\sigma) \geq 0$ for all $\sigma \in M_m \otimes M_n$ with $SN(\sigma) \leq k$, and
		\item $SN(\rho) \leq k$ if and only if $\Tr(Y\rho) \geq 0$ for all $k$-block positive $Y = Y^\dagger \in M_m \otimes M_n$.
	\end{enumerate}
\end{prop}
Condition (a) of this result follows trivially from the definitions of $k$-block positivity and Schmidt number. Condition (b) is slightly more technical, but follows from the recently-explored dual cone relationship of $k$-block positivity and Schmidt number of \cite{Sko10,SSZ09,S09}. Compare this result to Theorem~\ref{thm:sch_kpos_maps}, which similarly connects $k$-positivity of linear maps and Schmidt number of density matrices. As might be guessed, there is a close connection between $k$-block positivity of operators and $k$-positivity of linear maps, which will be pinned down in Section~\ref{sec:choi_jamiolkowski_isomorphism}.

\begin{figure}[ht]
\begin{center}
\includegraphics[width=0.5816\textwidth]{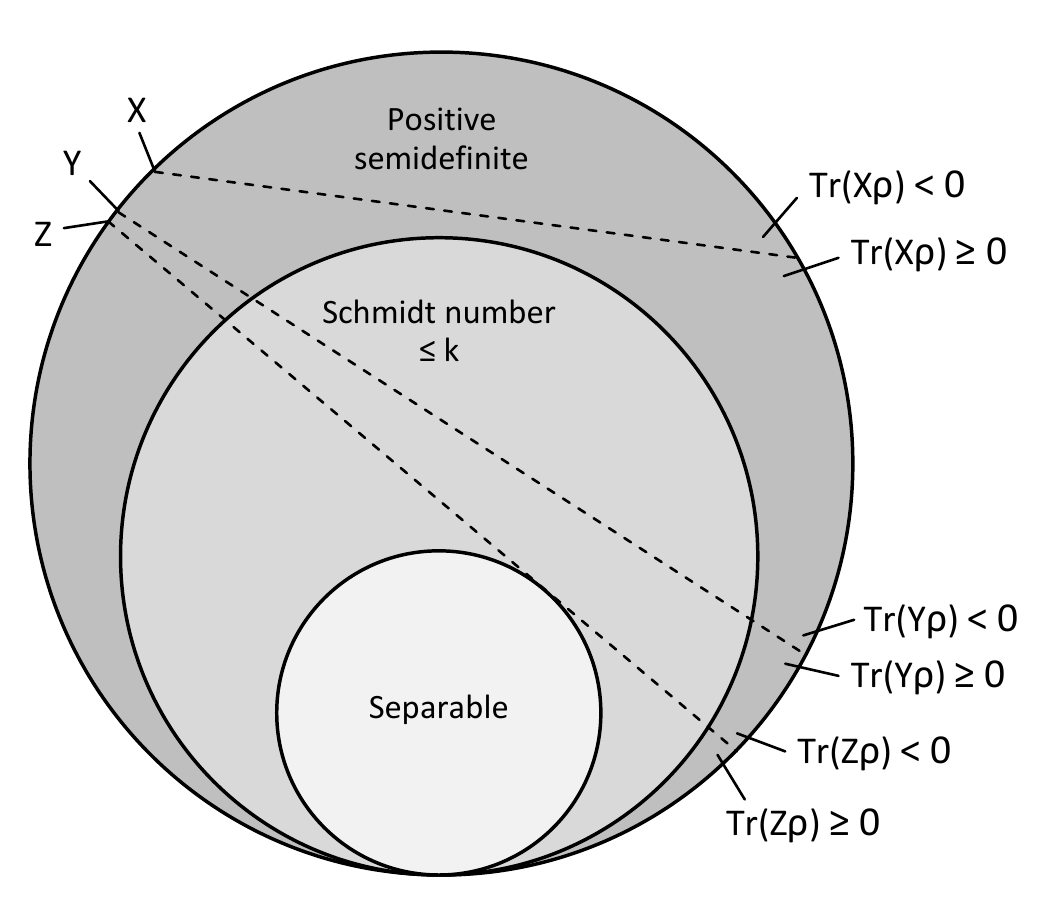}
\end{center}\vspace{-0.3in}
\caption[Entanglement witnesses as separating hyperplanes]{\hsp A representation of entanglement witnesses as separating hyperplanes, as described by Proposition~\ref{prop:kPos}. Any operator that is above one of the separating hyperplanes has entanglement that is detected by the corresponding entanglement witness. The operator $Y$ is a general entanglement witness, $X$ is a $k$-entanglement witness (and hence also an entanglement witness), and $Z$ is an optimal entanglement witness.}\label{fig:ent_witness}
\end{figure}

We close this section by presenting a result of \cite{SWZ08} that provides a simple necessary condition for block positivity.
\begin{prop}[Szarek, Werner, and {\.Z}yczkowski]\label{prop:block_pos_trace}
	Let $X = X^\dagger \in M_m \otimes M_n$. If $X$ is block positive then $\Tr(X^2) \leq \big(\Tr(X)\big)^2$.
\end{prop}\index{block positive|)}
Indeed, the trace inequality of Proposition~\ref{prop:block_pos_trace} is trivially true if $X$ is positive semidefinite. When $X$ is block positive but not positive semidefinite, the inequality provides a restriction on how negative the negative eigenvalues of $X$ can be relative to its positive eigenvalues. We will return to the problem of characterizing the eigenvalues of $k$-block positive operators in Section~\ref{sec:op_sk_norm_block_pos}.

\section{Local Operations and Distillability}\label{sec:local_ops}

In this section we consider the situation in which two parties, traditionally referred to as Alice and Bob, are each in control of a quantum system, but their quantum systems may be entangled with each other. In particular, we will consider what kind of effect Alice and Bob can have on the entanglement between their systems if they are only allowed to perform quantum operations on their own system.

From now on, it will sometimes be useful to let $M_A$ and $M_B$ denote complex matrix spaces that represent the quantum systems controlled by Alice and Bob, respectively. Similarly, we will use $M_{A^\prime}$ and $M_{B^\prime}$ to denote complex matrix spaces that represent the environments of Alice's and Bob's systems. We will use subscripts to indicate which subsystems a state lives in or a map is acting on if there would otherwise be potential for confusion. For example, $id_{A^\prime B^\prime} \otimes \Phi_{AB}$ is the map that acts as the identity on $M_{A^\prime} \otimes M_{B^\prime}$ and as the map $\Phi$ on $M_A \otimes M_B$. We will use the notation $\bb{C}^A$ to denote the complex Euclidean space of dimension corresponding to $M_A$ (i.e., $M_A$ is the space of ${\rm dim}(\bb{C}^A) \times {\rm dim}(\bb{C}^A)$ matrices).

\subsection{LOCC and Separable Channels}\label{sec:locc}

\emph{Local operations and classical communication (LOCC)} \cite{BDFMRSSW99} is the set of channels that can be implemented by Alice applying a quantum channel on her system and communicating classical information to Bob, and then Bob applying a quantum channel on his system and communicating classical information to Alice, and so on. LOCC channels play a particularly important role in entanglement theory, as any meaningful measure of entanglement between two systems intuitively should not increase under the action of an LOCC channel -- a point that we will return to in the next section.

It turns out that LOCC channels are quite messy to represent mathematically, so it is common to work instead with the set of \emph{separable} maps. A completely positive map $\Phi : M_A \otimes M_B \rightarrow M_A \otimes M_B$ is called \emph{separable} \cite{CDKL01,Rai97} if there exist families of operators $\big\{A_\ell\big\} \subset M_A$ and $\big\{B_\ell\big\} \subset M_B$ such that
\begin{align*}
	\Phi(X) = \sum_\ell (A_\ell \otimes B_\ell)X(A_\ell \otimes B_\ell)^\dagger \quad \forall \, X \in M_A \otimes M_B.
\end{align*}

Indeed, every LOCC channel is a separable channel, but the converse is not true. That is, there are separable channels that cannot be implemented via the LOCC paradigm described earlier \cite{BDFMRSSW99}. The distinction between separable and LOCC channels is still not particularly well-understood, but has been explored in \cite{GG08,Ghe10}. Nonetheless, separable maps are useful because the simple form of separable maps generally makes working with them fairly straightforward, and anything that we prove about separable channels is necessarily also true of LOCC channels.

Finally, it is worth pointing out that separable channels are also exactly the channels that preserve separability between Alice and Bob in the case when the original state may be entangled with their individual environments. That is, a channel $\Phi$ is separable if and only if $(id_{A^\prime,B^\prime} \otimes \Phi_{A,B})(\sigma_{A^\prime,A} \otimes \tau_{B^\prime,B})$ is always separable with respect to the $(A^\prime,A)-(B^\prime,B)$ cut (that is, when we treat $M_{A^\prime} \otimes M_A$ as one system and $M_{B^\prime} \otimes M_B$ as the other subsystem). We will prove and expand upon this statement in Section~\ref{sec:choi_jamiolkowski_schmidt}.

\subsection{Distillability and Bound Entanglement}\label{sec:bound_entangle}

Given a bipartite state $\rho \in M_A \otimes M_B$, a natural question to ask is whether or not it can be transformed (with vanishingly small error) via LOCC into the maximally-entangled state $\ket{\psi_+} \in \mathbb{C}_2 \otimes \mathbb{C}_2$. Indeed, this state is the prototypical example of an entangled state that allows for protocols such as quantum teleportation to work \cite{BBCJPW93,Vai94}, so whether or not $\rho$ can be transformed into $\ket{\psi_+}$ can roughly be thought of as an indication of whether or not it contains any ``useful'' entanglement.

It may happen that $\rho$ itself cannot be transformed into $\ket{\psi_+}$ via LOCC operations, but $r$ copies of $\rho$ (i.e., $\rho^{\otimes r}$) can be. Thus we ask whether multiple copies of $\rho$ can be transformed into $\ket{\psi_+}$ via LOCC operations, and we call any state $\rho$ that can be transformed in this way \emph{distillable}.

It should not be surprising that separable states are undistillable -- we should not expect to be able to extract entanglement from a separable state. Conversely, it is known \cite{HHH97} that any entangled state $\rho \in M_2 \otimes M_2$ is distillable. A slightly stronger statement is that any state that violates the reduction criterion is distillable \cite{HH99}. Somewhat surprisingly, however, there are entangled states in $M_m \otimes M_n$ when $mn > 6$ that are undistillable. Indeed, any state $\rho$ with $\rho^\Gamma \geq 0$, where $\Gamma$ refers to the partial transpose, is undistillable \cite{HHH98}, and there are many known entangled states with positive partial transpose when $mn > 6$ \cite{ABLS01,BP00,FLS06,H97,PM07,WW01,YL05}. Entangled states that are undistillable are called \emph{bound entangled}.

Although all PPT states are known to be undistillable, there is still no known simple or useful characterization of undistillable states. In fact, one of the most important open questions in quantum information theory is whether or not there exist any non-positive partial transpose (NPPT) states that are bound entangled \cite{DCL00,DSSTT00,OpenProbGen,OpenProb2}. There is a growing mound of evidence that suggests that NPPT bound entangled states exist \cite{BR03,CS06,JK11,PPHH10,VD06}, but there is still no proof.

One of the more interesting connections between positivity and the NPPT bound entanglement problem says that $\rho$ is undistillable if and only if $(\rho^\Gamma)^{\otimes r}$ is $2$-block positive for all $r \geq 1$ \cite{HHH98}. It is clear that this property is satisfied by any state $\rho$ with $\rho^\Gamma \geq 0$ -- the NPPT bound entanglement problem asks whether or not there exist \emph{other} states satisfying this block positivity property.

In the case when $(\rho^{\otimes r})^\Gamma$ is $2$-block positive for a given value of $r$, we say that $\rho$ is \emph{$r$-copy undistillable}. Determining whether or not an operator is $1$-copy undistillable is already a difficult problem, but determining $r$-copy undistillability for $r \geq 2$ seems to be much more challenging still. For example, we will introduce in Section~\ref{sec:werner_states} a family of states whose $1$-copy undistillability is straightforward to see, but whose $2$-copy undistillability has yet to be proved analytically. One potential reason for this jump in difficulty from the $r = 1$ case to the $r = 2$ case is that the cone generated by the set of $1$-copy undistillable states is easily seen to be convex (see \cite{C05} for implications of this convexity). In the case when $1 < r < \infty$ however, convexity of the set of $r$-copy undistillable states is no longer known, as the tensor copies of $\rho$ interfere. If NPPT bound entangled states do exist, then the set of $r$-copy undistillable states must fail to be convex for at least some $r$ \cite{SST01} (see also \cite{BE08}).

\subsection{Werner States}\label{sec:werner_states}

One especially important class of states in the study of bound entanglement is the family of \emph{Werner states} \cite{W89}, which can be parametrized by a single real variable $\alpha \in [-1,1]$ via
\[
    \rho_\alpha := \frac{1}{n^2 - \alpha n}(I - \alpha S) \in M_n \otimes M_n.
\]
Our interest in Werner states comes from the fact that NPPT bound entangled states exist if and only if there is a Werner state that is NPPT bound entangled \cite{HH99}. That is, to answer the NPPT bound entanglement problem, it is enough to consider only this highly symmetric one-parameter family of states.

The state $\rho_\alpha$ is entangled if and only if $\alpha > 1/n$, and this is also exactly the range of $\alpha$ for which $\rho_\alpha^\Gamma \not\geq 0$. On the other hand, it is known that $\rho_\alpha$ is $1$-copy undistillable whenever $\alpha \leq 1/2$ and $1$-copy distillable otherwise (and we will provide a simple proof of this fact in Section~\ref{sec:bound_entanglement}). Thus, the interval $(1/n, 1/2]$ serves as a ``region of interest'' for values of $\alpha$ -- an NPPT bound entangled state exists if and only if there is some $\alpha \in (1/n, 1/2]$ such that $\rho_\alpha$ is undistillable.

What values of $\alpha$ are associated with even $2$-copy undistillable states is not currently known. Given any fixed value of $r$, it is known that there are states $\rho_\alpha$ that are $r$-copy undistillable \cite{DCL00,DSSTT00}, but in these constructions $\alpha$ depends on $r$ and shrinks to $1/n$ as $r \rightarrow \infty$, and thus does not solve the bound entanglement problem. The two extreme possibilities are that $\rho_\alpha$ is distillable for all $\alpha \in (1/n,1/2]$, or alternatively that $\rho_{1/2}$ is bound entangled (and hence $\rho_\alpha$ is bound entangled for all $\alpha \in (1/n,1/2]$). Many quantum information theorists believe the latter conjecture \cite{DCL00,DSSTT00,PPHH10}, though it is possible that some Werner states in the region of interest are bound entangled, while others are not. In Section~\ref{sec:bound_entanglement}, we will examine the intermediate $\alpha = 2/n$ case extensively.

\begin{figure}[ht]
\begin{center}
\includegraphics[width=\textwidth]{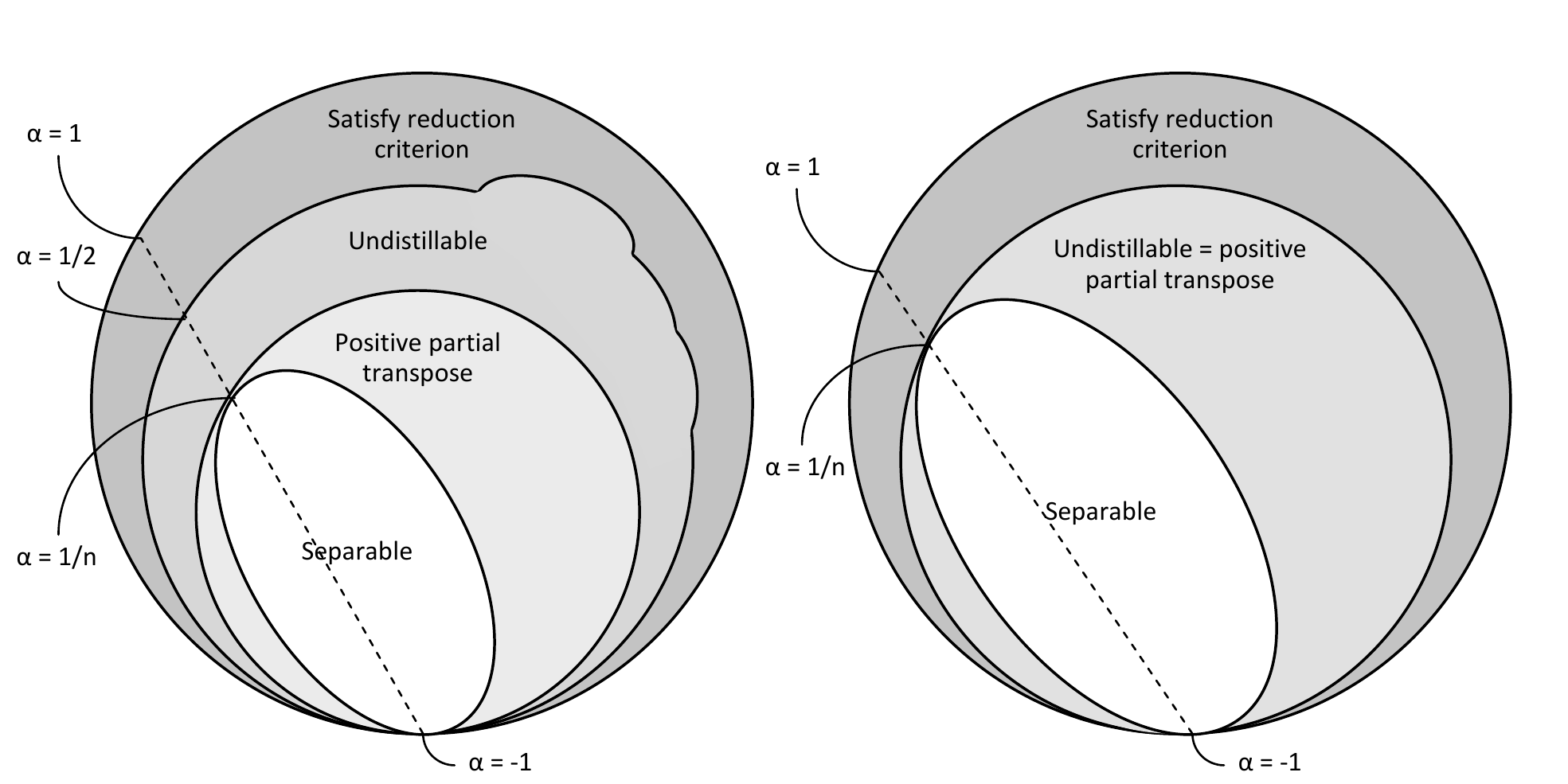}
\end{center}\vspace{-0.3in}
\caption[Werner states and undistillable states relative to the positive partial transpose and reduction criteria]{\hsp Representations of undistillable states relative to states that satisfy the reduction and positive partial transpose criteria. The dotted lines represent the Werner states. The figure on the left assumes that the $\rho_{1/2}$ Werner state is bound entangled, as conjectured, in which case the set of undistillable states is not convex. The figure on the right assumes that NPPT bound entangled states do not exist. The truth may actually be somewhere between these two extreme cases.}\label{fig:undistillable}
\end{figure}

\section{The Symmetric Subspace}\label{sec:symmetric_sub}

One linear operator that will play a particular important role throughout this work is the \emph{swap operator} $S \in M_n \otimes M_n$, which is defined on the standard basis via $S\ket{ij} = \ket{ji}$. We have already seen this operator in Example~\ref{exam:transpose_not_cp}, as $S = n(id_n \otimes T)(\ketbra{\psi_+}{\psi_+})$. The \emph{symmetric subspace} $\cl{S} \subseteq \bb{C}^n \otimes \bb{C}^n$ is the subspace spanned by the states $\ket{v}$ that satisfy $S\ket{v} = \ket{v}$. Equivalently, it is the subspace spanned by the vectors $\ket{ij} + \ket{ji}$ ($1 \leq i,j \leq n$).

It is easily-verified that $S$ corresponds, under the vector-operator isomorphism, to the transpose map. Hence the Takagi factorization \cite{HJ85,Tak24} of complex symmetric matrices (and hence symmetric states) says that $\ket{v} \in \cl{S}$ if and only if $\ket{v}$ has a symmetric Schmidt decomposition: $\ket{v} = \sum_{i=1}^k \alpha_i \ket{a_i} \otimes \ket{a_i}$, where $k = SR(\ket{v})$. We will denote the projection of $\bb{C}^n \otimes \bb{C}^n$ onto $\cl{S}$ by $P_{\cl{S}}$. Notice that $P_{\cl{S}} = \frac{1}{2}(I + S)$ and that the dimension of $\cl{S}$ is $n(n+1)/2$.

In the multipartite setting, things becomes more complicated because there is no longer a unique way to permute subsystems. Instead, there are $p!$ distinct ways to permute the $p$ subsystems of $(\bb{C}^n)^{\otimes p}$, and each such permutation corresponds to a different swap operator. Given a permutation $\sigma : \{1,\ldots,p\} \rightarrow \{1,\ldots,p\}$, we will define the swap operator $S_\sigma : \ket{v_1} \otimes \cdots \otimes \ket{v_p} \mapsto \ket{v_{\sigma(1)}} \otimes \cdots \otimes \ket{v_{\sigma(p)}}$ to be the operator that permutes the $p$ subsystems according to $\sigma$. In this case, the symmetric subspace is the subspace $\cl{S} \subseteq (\bb{C}^n)^{\otimes p}$ spanned by the states $\ket{v}$ that satisfy $S_\sigma \ket{v} = \ket{v}$ for all permutations $\sigma$. As before, the projection onto the symmetric subspace will be denoted by $P_{\cl{S}}$, and we have $P_{\cl{S}} = \frac{1}{p!}\sum_\sigma S_\sigma$, where the sum is taken over all permutations $\sigma : \{1,\ldots,p\} \rightarrow \{1,\ldots,p\}$.

\subsection{Shareable Quantum States and Symmetric Extensions}\label{sec:shareable}

A positive operator $X \in M_m \otimes M_n$ is called \emph{shareable} if there exists $0 \leq \tilde{X} \in M_m \otimes M_n \otimes M_n$ such that $\Tr_2(\tilde{X}) = \Tr_3(\tilde{X}) = X$, where we recall that $\Tr_i$ denotes the partial trace over the $i$-th subsystem. Shareable states are important in quantum information theory, as they are the states such that if one half of the state lives in Alice's system (say $M_m$) and the other half of the state lives in Bob's system $M_n$, there could be a third party that shares the exact same state with Alice. For this reason, shareable states exhibit certain insecurity properties that make them undesirable in quantum key distribution \cite{MCL06}.

More generally, $X \geq 0$ is called \emph{$s$-shareable} if there exists $0 \leq \tilde{X} \in M_m \otimes M_n^{\otimes s}$ such that $\Tr_{\overline{1},\overline{2}}(\tilde{X}) = \Tr_{\overline{1},\overline{3}}(\tilde{X}) = \cdots = \Tr_{\overline{1},\overline{s+1}}(\tilde{X}) = X$, where $\Tr_{\overline{1},\overline{i}}$ denotes the partial trace over all subsystems except the first and $i$-th. Note that all positive operators are $1$-shareable, and $2$-shareable operators are the operators that were simply called shareable in the previous paragraph.

The sets of $s$-shareable operators play a particularly important role in entanglement theory \cite{DPS02,DPS04}, as any separable operator is $s$-shareable for all $s \geq 1$. To see this, write $X = \sum_i c_i \ketbra{v_i}{v_i} \otimes \ketbra{w_i}{w_i}$. Then
\begin{align}\label{eq:sep_extension}
	\tilde{X} = \sum_i c_i \ketbra{v_i}{v_i} \otimes \underbrace{\ketbra{w_i}{w_i} \otimes \cdots \otimes \ketbra{w_i}{w_i}}_{s \text{ copies}}
\end{align}
satisfies the required partial trace conditions. Much more interesting is the fact that the converse of this statement is also true \cite{DPS04,FLV88,RW89,Wer89,Y06}. That is, if $X$ is $s$-shareable for all $s \geq 1$ then it is separable. However, these sets do not collapse in any finite number of steps: for any fixed $s \in \mathbb{N}$ there exist entangled states that are $s$-shareable -- see Figure~\ref{fig:shareable}.

Not only do the sets of $s$-shareable operators approximate the set of separable operators, but they do so in a way that is quite desirable computationally. Whether or not an operator is $s$-shareable is a problem that can be solved via semidefinite programming \cite{DPS04}, which has efficient numerical solution methods. Thus $s$-shareability provides a natural hierarchy of necessary conditions for separability, each of which is not too difficult computationally to test. Much of Chapter~\ref{ch:computation} will focus on semidefinite programs and applications of $s$-shareable operators.

In the definition of $s$-shareable states, note that the requirement that $\Tr_{\overline{1},\overline{2}}(\tilde{X}) = \Tr_{\overline{1},\overline{3}}(\tilde{X}) = \cdots = \Tr_{\overline{1},\overline{s+1}}(\tilde{X}) = X$ could be replaced by the following two properties:
\begin{enumerate}[(a)]
	\item $\Tr_{\overline{1},\overline{2}}(\tilde{X}) = X$; and
	\item $S_{\sigma}\tilde{X}S_{\sigma} = \tilde{X}$ for all permutations $\sigma : \{1,\ldots,s+1\} \rightarrow \{1,\ldots,s+1\}$ with $\sigma(1) = 1$.
\end{enumerate}
It is clear that if there exists $0 \leq \tilde{X} \in M_m \otimes M_n^{\otimes s}$ satisfying these two conditions, then $X$ is $s$-shareable. In the other direction, suppose that there exists $\tilde{X} \geq 0$ such that $\Tr_2(\tilde{X}) = \Tr_3(\tilde{X}) = X$. Then $\frac{1}{2}(\tilde{X} + S_{\{1,3,2\}}\tilde{X}S_{\{1,3,2\}})$ satisfies conditions (a) and (b) (and this same reasoning extends straightforwardly to the $s > 2$ case). It is often useful to use this second (equivalent) definition of $s$-shareability because it places further constraints on the extended operator $\tilde{X}$. An operator $\tilde{X}$ satisfying the two conditions (a) and (b) is called a \emph{$s$-symmetric extension} of $X$.

For the sake of entanglement detection, it is often beneficial to make one additional restriction on $s$-symmetric extensions. Observe that the operator~\eqref{eq:sep_extension} that extends a separable operator is not only symmetric in the sense of condition (b) above, but in fact the symmetric part of the operator is supported on the symmetric subspace. That is, $(I \otimes P_{\cl{S}})\tilde{X}(I \otimes P_{\cl{S}}) = \tilde{X}$. A symmetric extension $\tilde{X}$ that satisfies this stronger condition is called a \emph{$s$-bosonic symmetric extension} ($s$-BSE) of $X$.

In general, having a $s$-symmetric bosonic extension is a strictly stronger property than being $s$-shareable \cite{ML09}. However, the limiting case is still the same: an operator is $s$-shareable for all $s \in \mathbb{N}$ if and only if it has a $s$-symmetric bosonic extension for all $s \in \mathbb{N}$, if and only if it is separable. Because of these relationships, it is often useful to consider bosonic extensions, rather than regular symmetric extensions, when performing tasks related to entanglement detection.

Finally, notice that not only do separable states have $s$-symmetric bosonic extensions for all $s \geq 1$, but they have such an extension that has positive partial transpose (regardless of which subsystems the transpose is applied to). Thus, when considering the existence of $s$-symmetric extensions as necessary conditions for separability, it is often useful to ask that the given state have an $s$-symmetric extension that has the additional property of having positive partial transpose. In this way, we obtain a complete family of necessary criteria for separability, the weakest of which (i.e., the one that arises when $s = 1$) is the standard positive partial transpose criterion. We will see that each of these variants of symmetric extensions is useful in slightly different situations.

\subsection{From Separability to Arbitrary Schmidt Number}\label{sec:shareable_gen}

We saw in the previous section that the sets of $s$-shareable states are useful in that they form a sequence of nested approximations to the set of separable states. It is then natural to ask whether or not there exist (reasonably simple) sets that approximate the set of states $\rho$ with $SN(\rho) \leq k$ when $k > 1$. The answer to this question is ``yes''. To see this, we use the following pair of results, which can by thought of as methods for transforming statements about separability and block positivity into statements about Schmidt number $k$ and $k$-block positivity.
\begin{prop}\label{prop:sep_to_sk}
	Let $\rho \in M_A \otimes M_B$ be a density operator. Then $SN(\rho) \leq k$ if and only if there exists a separable operator $X \in (M_{A^\prime} \otimes M_A) \otimes (M_{B^\prime} \otimes M_B)$ (with ${\rm dim}(\bb{C}^{A^\prime}), {\rm dim}(\bb{C}^{B^\prime}) \leq k$) such that $(\bra{\psi_+}_{A^\prime B^\prime} \otimes I_{AB})X(\ket{\psi_+}_{A^\prime B^\prime} \otimes I_{AB}) = \rho$.
\end{prop}
\begin{proof}
	To see the ``if'' direction, suppose that $X = \sum_\ell p_\ell \ketbra{a_\ell}{a_\ell} \otimes \ketbra{b_\ell}{b_\ell}$, where
	\begin{align*}
		\ket{a_\ell} = \sum_{i=1}^k \alpha_{\ell,i}\ket{i} \otimes \ket{a_{\ell,i}} \in \bb{C}^{A^\prime} \otimes \bb{C}^{A} \quad \text{ and } \quad \ket{b_\ell} = \sum_{i=1}^k \beta_{\ell,i}\ket{i} \otimes \ket{b_{\ell,i}} \in \bb{C}^{B^\prime} \otimes \bb{C}^{B}.
	\end{align*}
	Then
	\begin{align*}
		& (\bra{\psi_+}_{A^\prime B^\prime} \otimes I_{AB})X(\ket{\psi_+}_{A^\prime B^\prime} \otimes I_{AB}) \\
		= \ & \sum_\ell(\bra{\psi_+} \otimes I)\left[\sum_{i,j,r,s=1}^k \alpha_{\ell,i}\alpha_{\ell,j}\beta_{\ell,r}\beta_{\ell,s}\ketbra{ir}{js} \otimes \ketbra{a_{\ell,i} b_{\ell,r}}{a_{\ell,j} b_{\ell,s}}\right](\ket{\psi_+} \otimes I) \\
		= \ & \frac{1}{k}\sum_\ell \sum_{i,j=1}^k \alpha_{\ell,i}\alpha_{\ell,j}\beta_{\ell,i}\beta_{\ell,j} \ketbra{a_{\ell,i} b_{\ell,i}}{a_{\ell,j} b_{\ell,j}} \\
		= \ & \frac{1}{k}\sum_\ell \left(\sum_{i=1}^k \alpha_{\ell,i}\beta_{\ell,i} \ket{a_{\ell,i} b_{\ell,i}}\right)\left(\sum_{j=1}^k \alpha_{\ell,j}\beta_{\ell,j} \bra{a_{\ell,j} b_{\ell,j}}\right),
	\end{align*}
	which clearly has Schmidt number no larger than $k$. To see the converse, simply note that every operator with Schmidt number at most $k$ can be written in the form above.
\end{proof}

\begin{prop}\label{prop:bp_to_kbp}
	Let $X = X^\dagger \in M_A \otimes M_B$. Then $X$ is $k$-block positive if and only if $\ketbra{\psi_+}{\psi_+}_{A^\prime B^\prime} \otimes X_{AB} \in (M_{A^\prime} \otimes M_A) \otimes (M_{B^\prime} \otimes M_B)$ is block positive (where ${\rm dim}(\bb{C}^{A^\prime}) = {\rm dim}(\bb{C}^{B^\prime}) = k$).
\end{prop}
Although we could prove Proposition~\ref{prop:bp_to_kbp} directly, we leave its proof to Section~\ref{sec:choi_jamiolkowski_schmidt}, where we will be able to prove it in a single line.

We can now use Proposition~\ref{prop:sep_to_sk} to produce a hierarchy of necessary tests for whether or not $SN(\rho) \leq k$, much like was done for separable states in the previous section. Let $\rho \in M_A \otimes M_B$. Then $SN(\rho) \leq k$ if and only if there exists $X \in (M_{A^\prime} \otimes M_A) \otimes (M_{B^\prime} \otimes M_B)$, with ${\rm dim}(\bb{C}^{A^\prime}), {\rm dim}(\bb{C}^{B^\prime}) \leq k$, such that $(\bra{\psi_+}_{A^\prime B^\prime} \otimes I_{AB})X(\ket{\psi_+}_{A^\prime B^\prime} \otimes I_{AB}) = \rho$. The operator $X$ is separable if and only if it is $s$-shareable for all $s \geq 1$. By combining these two facts, we see that $SN(\rho) \leq k$ if and only if, for all $s \geq 1$, there exists $0 \leq \tilde{X} \in (M_{A^\prime} \otimes M_A) \otimes (M_{B^\prime} \otimes M_B)^{\otimes s}$ such that
\begin{enumerate}[(a)]
	\item $(\bra{\psi_+}_{A^\prime B^\prime} \otimes I_{AB})\Tr_{\overline{1},\overline{2},\overline{3},\overline{4}}(\tilde{X})(\ket{\psi_+}_{A^\prime B^\prime} \otimes I_{AB}) = \rho$; and
	\item $S_{\sigma}\tilde{X}S_{\sigma} = \tilde{X}$ for all permutations $\sigma : \{1,\ldots,2s+2\} \rightarrow \{1,\ldots,2s+2\}$ with $\sigma(1) = 1$ and $\sigma(2j) = \sigma(2j-1)+1$ for all $1 \leq j \leq s+1$.
\end{enumerate}

For each fixed $s \geq 1$, the above conditions can be checked via semidefinite programming, just like in the case of separability. Furthermore, this method works much more generally -- given any separability criterion, we get a corresponding criterion for Schmidt number of $\rho$ by asking whether or not there exists an extended operator $X$ that satisfies the separability criterion and $(\bra{\psi_+} \otimes I)X(\ket{\psi_+} \otimes I) = \rho$. Similarly, given any $Y = Y^\dagger$, we can apply any test for block positivity to $\ketbra{\psi_+}{\psi_+} \otimes Y$ to get a test for $k$-block positivity of $Y$.

\section{The Choi--Jamio{\l}kowski Isomorphism}\label{sec:choi_jamiolkowski_isomorphism}

Recall from Section~\ref{sec:vector_operator_isomorphism} that the vector-operator isomorphism associated a linear map $\Phi : M_{n,m} \rightarrow M_{n,m}$ with an operator $M_{\Phi} \in M_m \otimes M_n$. While that isomorphism is very useful when dealing with questions related to rank and Schmidt rank, many important properties of the map $\Phi$ are not immediately clear from the operator $M_\Phi$. For example, we know that $\Phi$ is completely positive if and only if we can write $\Phi(X) = \sum_k A_k X A_k^\dagger$, in which case we have $M_\Phi = \sum_k \overline{A_k} \otimes A_k$ -- an operator that does not have any immediately obvious or simple properties that distinguish it.

On the other hand, Theorem~\ref{thm:choi_cp} showed that $\Phi : M_m \rightarrow M_n$ is completely positive if and only if the operator
\begin{align}\label{eq:choi}
	C_\Phi := m(id_m \otimes \Phi)(\ketbra{\psi_{+}}{\psi_{+}}\big) = \sum_{i,j=1}^{m} \ketbra{i}{j} \otimes \Phi(\ketbra{i}{j})
\end{align}
is positive semidefinite, which is an easy property to check. It turns out that many other properties of superoperators are illuminated by looking at the Choi matrix $C_\Phi$ as well. Before proceeding to investigate those properties, we present a simple lemma that illustrates how the Choi matrix of $\Phi$ is related to the Choi matrix of $\Phi^\dagger$.
\begin{lemma}\label{lem:choi_dual}
	Let $\Phi : M_m \rightarrow M_n$ be linear. Then $C_{\Phi^\dagger} = S \overline{C_{\Phi}} S$, where $S$ is the swap operator.
\end{lemma}
\begin{proof}
	Use the singular value decomposition to write $C_{\Phi} = \sum_i \lambda_i \ketbra{v_i}{w_i}$. We will see shortly (in Proposition~\ref{prop:choi_kraus_ops_general}) that $\Phi(X) = \sum_i \lambda_i {\rm mat}(\ket{v_i}) X {\rm mat}(\ket{w_i})^\dagger$. Thus $\Phi^\dagger(X) = \sum_i \lambda_i {\rm mat}(\ket{v_i})^\dagger X {\rm mat}(\ket{w_i})$.
	
	Now recall that $S$ corresponds to the transpose map under the vector-operator isomorphism, so ${\rm mat}(\ket{v})^T = {\rm mat}(S\ket{v})$ for all $\ket{v}$. Thus we see (again using Proposition~\ref{prop:choi_kraus_ops_general}) that $C_{\Phi^\dagger} = \sum_i \lambda_i S\overline{\ketbra{v_i}{w_i}}S$, which is easily seen to be equal to $S \overline{C_{\Phi}} S$.
\end{proof}

The map that sends $\Phi$ to its Choi matrix $C_\Phi$ is a linear isomorphism that is known as the \emph{Choi--Jamio{\l}kowski isomorphism} \cite{C75,J72}. This map, appropriately rescaled by a factor of $m$, is sometimes referred to as \emph{channel-state duality} \cite{AP04,SMR61,ZB04} because it associates quantum channels with density operators, though we will not use this terminology.

It is straightforward to see that the Choi--Jamio{\l}kowski isomorphism is linear. To see that it is bijective, it is perhaps instructive to write the Choi matrix $C_\Phi$ as a block matrix:
\begin{align*}\hsp
	C_\Phi = \begin{bmatrix}\Phi(\ketbra{1}{1}) & \Phi(\ketbra{1}{2}) & \cdots & \Phi(\ketbra{1}{m}) \\ \Phi(\ketbra{2}{1}) & \Phi(\ketbra{2}{2}) & \cdots & \Phi(\ketbra{2}{m}) \\ \vdots & \vdots & \ddots & \vdots \\ \Phi(\ketbra{m}{1}) & \Phi(\ketbra{m}{2}) & \cdots & \Phi(\ketbra{m}{m})\end{bmatrix}.
\dsp\end{align*}
Because the set $\big\{\ketbra{i}{j}\big\}_{i,j=1}^{m}$ is a basis of $M_m$, it follows easily that every map $\Phi$ corresponds to a unique Choi matrix, and vice-versa. This map becomes an isometry when we define an inner product on the space of superoperators by $\langle \Phi | \Psi \rangle := \langle C_\Phi | C_\Psi \rangle = \Tr(C_{\Phi}^\dagger C_{\Psi})$. The following proposition demonstrates some useful properties of this inner product -- these properties are well-known, and an alternative proof can be found in \cite{Sko11}.
\begin{prop}\label{prop:choi_inner_basic}
	Let $\Phi, \Psi : M_m \rightarrow M_n$ and $\Omega : M_n \rightarrow M_n$ be linear. Then
	\begin{enumerate}[(a)]
		\item $\langle \Phi | \Omega \circ \Psi \rangle = \langle \Omega^\dagger \circ \Phi | \Psi \rangle$
		\item $\langle \Phi | \Psi \rangle = \langle \Psi^\dagger | \Phi^\dagger \rangle$.
	\end{enumerate}
\end{prop}
\begin{proof}
	Property (a) follows from simply moving terms around inside the Hilbert--Schmidt inner product:
\begin{align*}
	\langle C_\Phi | C_{\Omega \circ \Psi} \rangle & = \big\langle C_\Phi | (id_m \otimes \Omega)(C_{\Psi}) \big\rangle = \big\langle (id_m \otimes \Omega^\dagger)(C_\Phi) | C_{\Psi} \big\rangle = \langle C_{\Omega^\dagger \circ \Phi} | C_{\Psi} \rangle.
\end{align*}
	For Property (b), we use Lemma~\ref{lem:choi_dual}:
	\begin{align*}
		\langle \Psi^\dagger | \Phi^\dagger \rangle & = \Tr\big(C_{\Psi^\dagger}^\dagger C_{\Phi^\dagger}\big) = \Tr\big((S \overline{C_{\Psi}} S)^\dagger S \overline{C_{\Phi}} S\big) = \Tr\big(C_{\Psi}^T \overline{C_{\Phi}} \big) = \langle \Phi | \Psi \rangle.
	\end{align*}
\end{proof}

Table~\ref{table:choi_jam} gives several examples of equivalences of the Choi--Jamio{\l}kowski isomorphism that will be used repeatedly throughout this work for easy reference. The remainder of this section is devoted to expanding upon, proving, or at least referencing these various equivalences.
\afterpage{\clearpage} 
\begin{table}[ht]\hsp
	\begin{center}
  \begin{tabular}{ c | c }
  	\noalign{\hrule height 0.1em}
    Superoperators $\Phi : M_m \rightarrow M_n$ & Operators $X \in M_m \otimes M_n$ \\
  	\noalign{\hrule height 0.1em}
    all superoperators & all operators \\ \hline
    completely positive maps & positive semidefinite operators\index{completely positive}\index{positive semidefinite} \\ \hline
    Hermiticity-preserving maps & Hermitian operators\index{Hermiticity-preserving}\index{Hermitian} \\ \hline
    trace-preserving maps & operators $X$ with $\Tr_2(X) = I$ \\ \hline
    unital maps & operators $X$ with $\Tr_1(X) = I$\index{unital} \\ \noalign{\hrule height 0.1em}
    positive maps & block positive operators\index{positive map}\index{block positive} \\ \hline
    $k$-positive maps & $k$-block positive operators \\ \hline
    superpositive maps & separable operators\index{superpositive}\index{separable operator} \\ \hline
    $k$-superpositive maps & operators $X$ with $SN(X) \leq k$\index{Schmidt number} \\ \hline
    separable maps & separable operators (via another tensor cut)\index{separable map} \\ \noalign{\hrule height 0.1em}
    completely co-positive maps & positive partial transpose operators\index{completely co-positive} \\ \hline
    binding entanglement maps & bound entangled operators\index{binding entanglement} \\ \hline
    anti-degradable maps & shareable operators\index{anti-degradable} \\ \hline
    $s$-extendible maps & $s$-shareable operators \\ \noalign{\hrule height 0.1em}
  \end{tabular}
	\end{center}
	\caption[Equivalences of the Choi--Jamio{\l}kowski isomorphism]{\hsp The equivalences of several sets of linear operators and linear superoperators via the Choi--Jamio{\l}kowski isomorphism.}\label{table:choi_jam}
\dsp\end{table}

\subsection{Fundamental Correspondences for Quantum Channels}\label{sec:choi_jamiolkowski_basic}

We now derive the most basic and well-known of the associations of the Choi--Jamio{\l}kowski isomorphism -- specifically those that help clarify the structure of the set of quantum channels. These results are all well-known, and proofs of many of these correspondences can be found in \cite{VV03,Wat04}.

\subsubsection*{All superoperators -- All operators}

We already saw that the Choi--Jamio{\l}kowski isomorphism is a bijection between the set of linear maps $\Phi : M_m \rightarrow M_n$ and the set of operators $M_m \otimes M_n$. We now use this isomorphism and a slight modification of the proof of Theorem~\ref{thm:choi_cp} to demonstrate a relationship between the generalized Choi--Kraus operators of $\Phi$ and the Choi matrix $C_\Phi$.
\begin{prop}\label{prop:choi_kraus_ops_general}
	Let $\Phi : M_m \rightarrow M_n$ be a linear map. Then $C_\Phi = \sum_\ell c_\ell \ketbra{v_\ell}{w_\ell}$ if and only if
	\begin{align*}
		\Phi(X) = \sum_\ell c_\ell {\rm mat}(\ket{v_\ell}) X {\rm mat}(\ket{w_\ell})^\dagger \quad \forall \, X \in M_m.
	\end{align*}
\end{prop}
\begin{proof}
	For the ``if'' direction of the proof, we note that
	\begin{align}\label{eq:gen_choi_all1}
		C_\Phi = \sum_\ell c_\ell \sum_{i,j=1}^{m} \ketbra{i}{j} \otimes {\rm mat}(\ket{v_\ell}) \ketbra{i}{j} {\rm mat}(\ket{w_\ell})^\dagger.
	\end{align}
	Now recall that if $\ket{v_\ell} = \sum_{i=1}^{m}c^{(v)}_{\ell,i}\ket{i} \otimes \ket{v_{\ell,i}}$ then ${\rm mat}(\ket{v_\ell}) = \sum_{i=1}^{m}c^{(v)}_{\ell,i}\ketbra{v_{\ell,i}}{i}$. It follows from Equation~\eqref{eq:gen_choi_all1} that
	\begin{align*}
		C_\Phi = \sum_\ell c_\ell \sum_{i,j=1}^{m} c^{(v)}_{\ell,i}c^{(w)}_{\ell,j} \ketbra{i}{j} \otimes \ketbra{v_{\ell,i}}{w_{\ell,j}} = \sum_\ell c_\ell \ketbra{v_\ell}{w_\ell}.
	\end{align*}
	
	For the ``only if'' direction of the proof, we mimic the proof of Theorem~\ref{thm:choi_cp}. Suppose $C_\Phi = \sum_\ell c_\ell \ketbra{v_\ell}{w_\ell}$ and write each $\ket{v_k}$ as a linear combination of elementary tensors: $\ket{v_\ell} = \sum_{i=1}^{m}c^{(v)}_{\ell,i}\ket{i} \otimes \ket{v_{\ell,i}}$ (and decompose $\ket{w_\ell}$ similarly). If we multiply $C_\Phi$ on the left by $\bra{i} \otimes I$ and on the right by $\ket{j} \otimes I$, then from the definition of $C_\Phi$ we have
	\begin{align}\label{eq:choi_reduce_gen1}
		(\bra{i} \otimes I)C_\Phi(\ket{j} \otimes I) & = \Phi\big(\ketbra{i}{j}\big).
	\end{align}
	Similarly, from $C_\Phi = \sum_\ell c_\ell \ketbra{v_\ell}{w_\ell}$ we have
	\begin{align}\begin{split}\label{eq:choi_reduce_gen2}
		(\bra{i} \otimes I)C_\Phi(\ket{j} \otimes I) & = \sum_\ell c_\ell c^{(v)}_{\ell,i}c^{(w)}_{\ell,j} \ketbra{v_{\ell,i}}{w_{\ell,j}} \\
		& = \sum_\ell c_\ell \left(\sum_{k=1}^{m}c^{(v)}_{\ell,k}\ketbra{v_{\ell,k}}{k}\right) \ketbra{i}{j} \left(\sum_{k=1}^{m}c^{(w)}_{\ell,k}\ketbra{k}{w_{\ell,k}}\right) \\
		& = \sum_\ell c_\ell {\rm mat}(\ket{v_\ell}) \ketbra{i}{j} {\rm mat}(\ket{w_\ell})^\dagger.
	\end{split}\end{align}
	It follows by equating Equations~\eqref{eq:choi_reduce_gen1} and~\eqref{eq:choi_reduce_gen2} that
	\begin{align*}
		\Phi\big(\ketbra{i}{j}\big) = \sum_\ell c_\ell {\rm mat}(\ket{v_\ell}) \ketbra{i}{j} {\rm mat}(\ket{w_\ell})^\dagger \quad \forall \, 0 \leq i,j \leq m-1.
		\end{align*}
	Extending by linearity shows that $\Phi$ has the desired form.
\end{proof}

In particular, the singular value decomposition $C_\Phi = \sum_{k=1}^{mn} \sigma_k \ketbra{v_k}{w_k}$ implies via Proposition~\ref{prop:choi_kraus_ops_general} that we can write
\begin{align*}
	\Phi(X) = \sum_{i=1}^{mn} A_i X B_i^\dagger,
\end{align*}
where $\big\{A_i\big\}, \big\{B_i\big\} \subset M_{n,m}$ are sets of operators that are orthogonal in the Hilbert--Schmidt inner product.

Using Proposition~\ref{prop:choi_kraus_ops_general}, we can now present a simple proposition (which was proved in the special cases of quantum channels in \cite{JK11b} and positive maps in \cite{S11}) that allows us to relate the Choi matrices of $\Phi$ and $\Phi^{\dagger}$.
\begin{prop}\label{prop:left_right_choi}
	Let $\Phi : M_m \rightarrow M_n$ be a linear map. Then $C_{\Phi^\dagger} = S \overline{C_{\Phi}} S^\dagger$, where $S \in M_{n,m} \otimes M_{m,n}$ is the swap operator.
\end{prop}
\begin{proof}
	Write $C_{\Phi} = \sum_\ell c_\ell \ketbra{v_\ell}{w_\ell}$ so that $\Phi(X) = \sum_\ell c_\ell {\rm mat}(\ket{v_\ell}) X {\rm mat}(\ket{w_\ell})^\dagger$ by Proposition~\ref{prop:choi_kraus_ops_general}. Then $\Phi^\dagger(X) = \sum_\ell c_\ell {\rm mat}(\ket{v_\ell})^\dagger X {\rm mat}(\ket{w_\ell})$, so using Proposition~\ref{prop:choi_kraus_ops_general} together with the fact that ${\rm mat}(\ket{v_\ell})^\dagger = {\rm mat}(S\overline{\ket{v_\ell}})$ gives
	\begin{align*}
		C_{\Phi^\dagger} = \sum_\ell c_\ell S\overline{\ketbra{v_\ell}{w_\ell}}S^\dagger = S \overline{C_{\Phi}} S^\dagger.
	\end{align*}
\end{proof}

\subsubsection*{Completely positive maps -- Positive semidefinite operators}

We already saw in Theorem~\ref{thm:choi_cp} that the set of completely positive maps corresponds to the set of positive semidefinite operators via the Choi--Jamio{\l}kowski isomorphism. The Kraus representation of a completely positive map then follows immediately from Proposition~\ref{prop:choi_kraus_ops_general} and the spectral decomposition $C_\Phi = \sum_{\ell=1}^{mn} \lambda_\ell \ketbra{v_\ell}{v_\ell}$. As we saw in the proof of Theorem~\ref{thm:choi_cp}, if we define $A_\ell := \sqrt{\lambda_\ell}{\rm mat}(\ket{v_\ell})$ then
\begin{align*}
	\Phi(X) = \sum_{\ell=1}^{mn} A_\ell X A_\ell^\dagger.
\end{align*}
Other proofs of these facts can be found in \cite{LS93,SSF05}.

\subsubsection*{Hermiticity-preserving maps -- Hermitian operators}

A superoperator $\Phi : M_m \rightarrow M_n$ is called \emph{Hermiticity-preserving} if $\Phi(X)^\dagger = \Phi(X)$ whenever $X^\dagger = X$ (or equivalently, if $\Phi(X^\dagger) = \Phi(X)^\dagger$ for all $X$). Completely positive (and even just positive) maps are necessarily Hermiticity-preserving because for any Hermitian matrix $X$ we can write $X = P - Q$ for some positive semidefinite $P$ and $Q$. Then $\Phi(X) = \Phi(P)-\Phi(Q) = \Phi(P)^\dagger - \Phi(Q)^\dagger = \Phi(X)^\dagger$, where the second equality follows from the fact that $\Phi(P)$ and $\Phi(Q)$ are positive semidefinite and hence Hermitian.

Hermiticity-preserving maps were originally characterized in \cite{P67} (see also \cite{Hil73,PH81}) -- they have a structure very similar to that of completely positive maps.
\begin{prop}\label{prop:hermiticity_char}
	Let $\Phi : M_m \rightarrow M_n$ be a linear map. The following are equivalent:
  \begin{enumerate}[(a)]
  	\item $\Phi$ is Hermiticity-preserving;
  	\item $C_\Phi$ is Hermitian; and
  	\item there exist operators $\{A_\ell\}_{\ell=1}^{mn}$ and real numbers $\{\lambda_\ell\}_{\ell=1}^{mn}$ such that
  	\begin{align*}
  		\Phi(X) = \sum_{\ell=1}^{mn} \lambda_\ell A_\ell X A_\ell^\dagger \quad \forall \, X \in M_m.
  	\end{align*}
  \end{enumerate}
\end{prop}
\begin{proof}
	The implication (a) $\implies$ (b) follows from simple algebra:
	\begin{align*}
		C_\Phi^\dagger = \left(\sum_{i,j=1}^{m} \ketbra{i}{j} \otimes \Phi(\ketbra{i}{j})\right)^\dagger = \sum_{i,j=1}^{m} \ketbra{j}{i} \otimes \Phi(\ketbra{j}{i}) = C_\Phi.
	\end{align*}
	
	To see (b) $\implies$ (c), use the spectral decomposition to write $C_\Phi = \sum_{\ell=1}^{mn} \lambda_\ell \ketbra{v_\ell}{v_\ell}$ with each $\lambda_\ell$ real. If we define $A_\ell := {\rm mat}(\ket{v_\ell})$ then Proposition~\ref{prop:choi_kraus_ops_general} gives the desired form of $\Phi$.
	
	Finally, the implication (c) $\implies$ (a) is trivial:
	\begin{align*}
		\Phi(X)^\dagger = \left(\sum_{\ell=1}^{mn} \lambda_\ell A_\ell X A_\ell^\dagger\right)^\dagger = \sum_{\ell=1}^{mn} \lambda_\ell A_\ell X^\dagger A_\ell^\dagger = \Phi(X^\dagger) \quad \forall \, X \in M_m.
	\end{align*}
\end{proof}

In fact, it is clear from the proof of Proposition~\ref{prop:hermiticity_char} that the operators $\big\{A_\ell\big\}$ can be chosen to be orthonormal in the Hilbert--Schmidt inner product. If we relax this condition to orthogonality, then by absorbing constants into the $A_\ell$ operators we can choose $\lambda_\ell \in \{-1,1\}$ for all $\ell$.

A simple corollary of Proposition~\ref{prop:hermiticity_char} is that a linear map is Hermiticity-preserving if and only if it is the difference of two completely positive maps, and it is exactly this property that causes these maps to arise frequently in quantum information theory. For example, if we wish to measure the distance between two quantum channels, this often reduces to the problem of computing a norm (such as the diamond norm of Section~\ref{sec:cb_norm}) on the corresponding Hermiticity-preserving map.

\subsubsection*{Unital or trace-preserving maps -- Operators with identity partial trace}

Recall that a quantum channel $\Phi$ is not only completely positive, but also trace-preserving. It is thus useful to understand how trace-preservation (and the closely-related property of being unital, i.e., $\Phi(I) = I$) is reflected in the Choi--Jamio{\l}kowski isomorphism. Begin by taking the partial traces of the Choi matrix:
\begin{align*}
	\Tr_1(C_\Phi) & = \sum_{i=1}^{m} \Phi(\ketbra{i}{i}) = \Phi(I), \quad \text{ and} \\
	\Tr_2(C_\Phi) & = \sum_{i,j=1}^{m} \ketbra{i}{j}\Tr(\Phi(\ketbra{i}{j})) = \sum_{i,j=1}^{m} \ketbra{i}{j}\big(\Phi^\dagger(I)\big)^\dagger\ketbra{i}{j} = \overline{\Phi^\dagger(I)}.
\end{align*}
It follows that $\Phi$ is unital if and only if $\Tr_1(C_\Phi) = I$ and $\Phi$ is trace-preserving if and only if $\Tr_2(C_\Phi) = I$. In particular, $\Phi$ is a quantum channel if and only if $C_\Phi$ is positive semidefinite with $\Tr_2(C_\Phi) = I$.

When a completely positive map is both trace-preserving and unital, it is called \emph{bistochastic}. Bistochastic quantum channels are special in that they can only add mixedness to the states they act on (see \cite[Lemma 5]{KS06} and \cite{ZB04}), and they are characterized by having the special operator-sum decomposition $\Phi(X) = \sum_\ell \lambda_\ell U_\ell X U_\ell^\dagger$, where each $U_\ell$ is unitary and $\sum_\ell \lambda_\ell = 1$ \cite{Men08,MW09}. Note that in general, even though such channels are completely positive, we can't choose $\lambda_\ell \geq 0$ in this operator-sum representation \cite{LS93}.

Bistochastic quantum channels arise frequently in quantum information theory. For example, the ability to perform error correction for errors represented by these channels is much better-understood than in the general case \cite{CJK09,HKL04,JK11c,KLPL06,Kri03,KS06}, and bistochastic quantum channels arise frequently in capacity and additivity problems \cite{AHW00,Cor04,KR01,VV03}.

\subsection{Correspondences Related to Schmidt Number}\label{sec:choi_jamiolkowski_schmidt}

In this section we present the correspondences of separable states (and in more generality, states with Schmidt number no larger than $k$ for some natural number $k$) through the Choi--Jamio{\l}kowski isomorphism. We also present similar characterizations of block positive and $k$-block positive operators, as Proposition~\ref{prop:kPos} showed that such operators can be used to describe Schmidt number. While most of these correspondences are fairly well-known by now, they are much more recent than the correspondences introduced in the previous section. Although Theorem~\ref{thm:ksep_map_choi} is a known result in the $k = 1$ case \cite{CDKL01}, we believe that our generalization of it for arbitrary $k$ is new (albeit straightforward).

\subsubsection*{Positive maps -- Block positive operators}\index{block positive|(}\index{positive map|(}

Despite how simple Theorem~\ref{thm:choi_cp} makes it to determine whether or not a linear map is completely positive, determining whether or not a linear map is positive is a difficult problem. In fact, a linear map is positive if and only if its Choi matrix is block positive (recall from Section~\ref{sec:blockpos} that $X$ is block positive if $\bra{v}X\ket{v} \geq 0$ for all separable $\ket{v}$); a result that was originally proved in  \cite{Cho75,Jam74} (for another proof, see \cite{And04}).

To see this correspondence, suppose that $\Phi : M_m \rightarrow M_n$ is a positive linear map. Then let's consider what happens when we multiply the Choi matrix of $\Phi$ on the left and right by a separable state $\ket{a} \otimes \ket{b}$:
\begin{align*}
	(\bra{a} \otimes \bra{b})C_{\Phi}(\ket{a} \otimes \ket{b}) & = \frac{1}{m}\sum_{i,j=1}^{m}(\bra{a} \otimes \bra{b})\big(\ketbra{i}{j} \otimes \Phi(\ketbra{i}{j})\big)(\ket{a} \otimes \ket{b}) \\
	& = \frac{1}{m} \bra{b} \Phi\Big(\sum_{i,j=1}^{m} \braket{a}{i}\ket{i} \braket{j}{a}\bra{j}\Big) \ket{b} \\
	& = \frac{1}{m} \bra{b} \Phi\big( \overline{\ketbra{a}{a}} \big) \ket{b} \\
	& \geq 0,
\end{align*}
where the final inequality follows from the facts that $\overline{\ketbra{a}{a}}$ is positive semidefinite and $\Phi$ is a positive linear map. We have thus shown that if $\Phi$ is positive, then its Choi matrix is block positive. To see the converse, note that the string of equalities above shows that if $C_\Phi$ is block positive then $\Phi$ is positive on rank-$1$ positive semidefinite operators. By linearity it then follows that $\Phi$ is a positive map.

We have already seen this correspondence in a few examples. In Example~\ref{exam:transpose_not_cp} it was noted that the transpose map $T : M_2 \rightarrow M_2$ is positive but not completely positive, and in Example~\ref{exam:transpose_bp} it was noted that its Choi matrix is block positive but not positive semidefinite.

\subsubsection*{$k$-positive maps -- $k$-block positive operators}

Given that we have already seen that positive maps correspond to block positive operators and completely positive maps correspond to positive semidefinite operators, it is perhaps not surprising that $k$-positive maps correspond to $k$-block positive operators via the Choi--Jamio{\l}kowski isomorphism. This correspondence was used implicitly in \cite{C05,TH00} and proved explicitly in \cite{RA07,Sko08,SSZ09}, but we provide an elementary proof here for completeness.

We use the fact that $(id_k \otimes \Phi)$ is positive if and only if $(\bra{a} \otimes \bra{b})C_{id_k \otimes \Phi}(\ket{a} \otimes \ket{b}) \geq 0$ for all $\ket{a} \in \bb{C}^{k} \otimes \bb{C}^m, \ket{b} \in \bb{C}^{k} \otimes \bb{C}^n$, which was proved in the previous section.
	
Now write $\ket{a} = \sum_{i=1}^{k}\alpha_i\ket{i}\otimes\ket{a_i}$ and $\ket{b} = \sum_{i=1}^{k}\beta_i\ket{i}\otimes\ket{b_i}$. Then
\begin{align*}
	\ket{a} \otimes \ket{b} = \sum_{i,j=1}^{k}\alpha_i\beta_j\ket{i}\otimes\ket{a_i} \otimes \ket{j}\otimes\ket{b_j}.
\end{align*}
Then after simplification we have
\begin{align*}
	(\bra{a} \otimes \bra{b})C_{id_k \otimes \Phi}(\ket{a} \otimes \ket{b}) & = \left(\sum_{i=1}^{k}\alpha_i\beta_i \bra{a_i} \otimes \bra{b_i}\right) C_\Phi \left(\sum_{i=1}^{k}\alpha_i\beta_i \ket{a_i} \otimes \ket{b_i}\right).
\end{align*}
Since $\sum_{i=1}^{k}\alpha_i\beta_i \ket{a_i} \otimes \ket{b_i}$ is (up to scaling) an arbitrary state with Schmidt rank $\leq k$, it follows that $\Phi$ is $k$-positive if and only if $C_\Phi$ is $k$-block positive.

Note that this correspondence provides an immediate proof of Proposition~\ref{prop:bp_to_kbp}: $C_\Phi$ is $k$-block positive if and only if $\Phi$ is $k$-positive, if and only if $id_k \otimes \Phi$ is positive, if and only if $C_{id_k \otimes \Phi} = k\ketbra{\psi_+}{\psi_+} \otimes C_{\Phi}$ is block positive.

\subsubsection*{Superpositive maps -- Separable operators}\index{separable operator|(}\index{superpositive|(}

A linear map $\Phi : M_m \rightarrow M_n$ is called \emph{superpositive} \cite{And04} if it admits a Kraus representation
\begin{align*}
	\Phi(X) = \sum_{\ell} A_\ell X A_\ell^\dagger
\end{align*}
with ${\rm rank}(A_\ell) = 1$ for all $\ell$. The following characterization of superpositive maps was originally proved in \cite{HSR03,And04}:
\begin{thm}\label{thm:superpos_choi}
	Let $\Phi : M_m \rightarrow M_n$ be a linear map. The following are equivalent:
	\begin{enumerate}[(a)]
		\item $C_\Phi$ is separable;
		\item $(id_m \otimes \Phi)(\rho)$ is separable for all $\rho \in M_m \otimes M_n$; and
		\item $\Phi$ is superpositive.
	\end{enumerate}
\end{thm}
\begin{proof}
	The implication $(b) \implies (a)$ follows trivially by choosing $\rho = \ketbra{\psi_+}{\psi_+}$. To see that $(a) \implies (c)$, write
	\begin{align*}
		C_\Phi = \sum_\ell p_\ell \ketbra{v_\ell}{v_\ell} \otimes \ketbra{w_\ell}{w_\ell}.
	\end{align*}
	Then Proposition~\ref{prop:choi_kraus_ops_general} implies that we can write
	\begin{align*}
		\Phi(X) = \sum_\ell p_i {\rm mat}(\ket{v_\ell} \otimes \ket{w_\ell}) X {\rm mat}(\ket{v_\ell} \otimes \ket{w_\ell})^\dagger = \sum_\ell p_i \ketbra{w_\ell}{v_\ell} X \ketbra{v_\ell}{w_\ell}.
	\end{align*}
	Because ${\rm rank}(\ketbra{w_\ell}{v_\ell}) = 1$ for all $\ell$, $\Phi$ is superpositive.
	
	For the $(c) \implies (b)$ implication, assume without loss of generality that ${\rm rank}(\rho) = 1$ and $\Phi$ can be written as $\Phi(X) = AXA^\dagger$ with ${\rm rank}(A) = 1$ (the general result for arbitrary $\rho$ and arbitrary superpositive $\Phi$ will then follow by convexity of the cone of separable operators). Write $A = c\ketbra{x}{b_1}$ and $\rho = \ketbra{v}{v}$, where $\ket{v} = \sum_{i=1}^{\min\{m,n\}}d_i \ket{a_i} \otimes \ket{b_i}$, where $\{\ket{b_i}\}$ is an orthonormal set in $\mathbb{C}^m$ that extends $\ket{b_1}$. Then
	\begin{align*}
		(id_m \otimes \Phi)(\rho) = (I_m \otimes A)\ketbra{v}{v}(I_m \otimes A^\dagger) = c^2 d_1^2 \ketbra{a_1}{a_1} \otimes \ketbra{x}{x},
	\end{align*}
	which is separable.
\end{proof}

Condition (b) of Theorem~\ref{thm:superpos_choi} shows that superpositive quantum channels are exactly the quantum channels that destroy any entanglement between the system that the channel acts on and its environment. For this reason, superpositive quantum channels are often called \emph{entanglement-breaking channels}. Entanglement-breaking channels were introduced in \cite{Hol98,Sho02} and have been further explored in \cite{Hol08,KABLA08,Kin03,Rus03,Sho02}.

\subsubsection*{$k$-superpositive maps -- Operators with Schmidt number $\leq k$}\index{Schmidt number|(}

A natural generalization of superpositive maps are \emph{$k$-superpositive maps} \cite{SSZ09}, which are linear map $\Phi : M_m \rightarrow M_n$ that have a Kraus representation
\begin{align*}
	\Phi(X) = \sum_{\ell} A_\ell X A_\ell^\dagger
\end{align*}
with ${\rm rank}(A_\ell) \leq k$ for all $\ell$. As might be intuitively expected based on the characterization of superpositive maps, the following three conditions are equivalent:
\begin{enumerate}[(a)]
	\item $SN(C_\Phi) \leq k$;
	\item $SN\big((id_m \otimes \Phi)(\rho)\big) \leq k$ for all $\rho \in M_m \otimes M_n$; and
	\item $\Phi$ is $k$-superpositive.
\end{enumerate}

The above equivalences were originally demonstrated in \cite{CK06} and can be proved by a simple modification of the proof of Theorem~\ref{thm:superpos_choi}. Another proof of the equivalence of conditions (a) and (c) can be found in \cite{SSZ09}. In the case when a quantum channel is $k$-superpositive, it is sometimes called a \emph{$k$-partially entanglement breaking channel} \cite{CK06} due to condition (b) above -- such channels have been further studied in \cite{AKMS05,Hua06,JKPP11,Xthesis,Xha11}.\index{Schmidt number|)}\index{superpositive|)}

\subsubsection*{Separable maps -- Separable operators (via another tensor cut)}\index{separable map|(}

Recall from Section~\ref{sec:locc} that a completely positive map $\Phi : M_A \otimes M_B \rightarrow M_A \otimes M_B$ (where $M_A$ is a complex matrix space representing Alice's quantum system, and $M_B$ represents Bob's quantum system) is called \emph{separable} if it can be written in the following form:
\begin{align*}
	\Phi(X) = \sum_\ell (A_\ell \otimes B_\ell)X(A_\ell \otimes B_\ell)^\dagger \quad \forall \, X \in M_A \otimes M_B.
\end{align*}

The Choi matrix of a separable map $\Phi$ is the following operator in $M_{A^\prime} \otimes M_{B^\prime} \otimes M_A \otimes M_B$:
\begin{align*}
	C_\Phi & = \sum_{i,j,r,s=1}^{m} \ketbra{i}{j} \otimes \ketbra{r}{s} \otimes \Phi(\ketbra{i}{j} \otimes \ketbra{r}{s}) \\
	& = \sum_{i,j,r,s=1}^{m} \sum_\ell \ketbra{i}{j} \otimes \ketbra{r}{s} \otimes A_\ell \ketbra{i}{j} A_\ell^\dagger \otimes B_\ell \ketbra{r}{s} B_\ell^\dagger,
\end{align*}
which is separable across the $(A^\prime,A)-(B^\prime,B)$ cut. This is perhaps made clearer by swapping the order of $M_{B^\prime}$ and $M_A$, so that we interpret $C_\Phi$ as an operator in $M_{A^\prime} \otimes M_A \otimes M_{B^\prime} \otimes M_B$:
\begin{align*}
	C_\Phi & = \sum_\ell \left(\sum_{i,j=1}^{m} \ketbra{i}{j} \otimes A_\ell \ketbra{i}{j} A_\ell^\dagger \right) \otimes \left(\sum_{i,j=1}^{m} \ketbra{i}{j} \otimes B_\ell \ketbra{i}{j} B_\ell^\dagger \right).
\end{align*}
To see that this operator is separable across the $(A^\prime,A)-(B^\prime,B)$ cut, note that $\sum_{i,j=1}^{m} \ketbra{i}{j} \otimes A_\ell \ketbra{i}{j} A_\ell^\dagger$ is positive semidefinite for all $\ell$ because it is the Choi matrix of the completely positive map $X \mapsto A_\ell X A_\ell^\dagger$ (and similarly if we replace $A_\ell$ by $B_\ell$). In fact, it was proved in \cite{CDKL01} (see also \cite{SZG09}) that $C_\Phi$ is separable across the $(A^\prime,A)-(B^\prime,B)$ cut if and only if $\Phi$ is a separable map. Contrast this with the case of superpositive maps, which would have $C_\Phi$ be separable across the $(A^\prime,B^\prime)-(A,B)$ cut (see Figure~\ref{fig:sep_eb_channels}).
\begin{figure}[ht]
\begin{center}
\includegraphics[width=\textwidth]{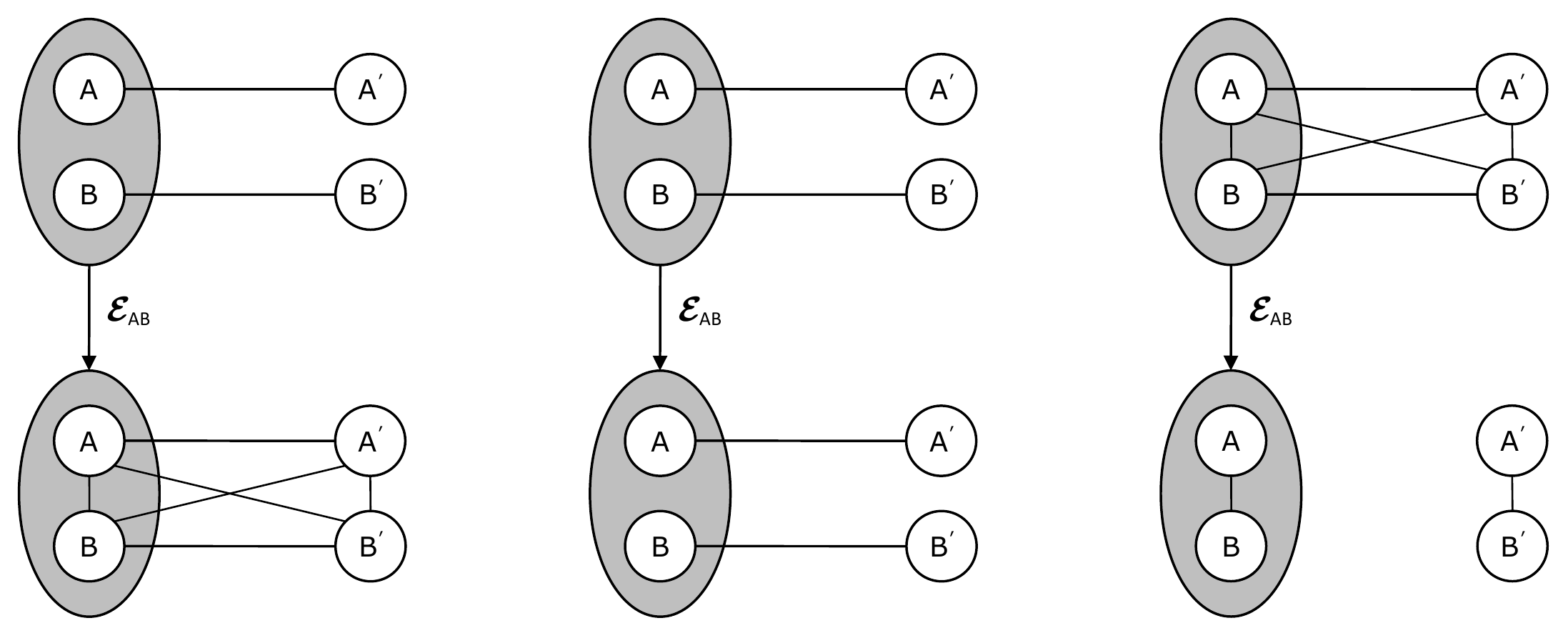}
\end{center}\vspace{-0.3in}
\caption[A comparison of separable and entanglement-breaking channels]{\hsp The action of general quantum channels (left), separable channels (center), and entanglement-breaking channels (right) acting on $M_A \otimes M_B$. Lines between subsystems represent entanglement. General channels can introduce entanglement between $M_A$, $M_B$, and the environment $M_{A^\prime} \otimes M_{B^\prime}$. Separable channels preserve separability between $M_A \otimes M_{A^\prime}$ and $M_B \otimes M_{B^\prime}$. Entanglement-breaking channels destroy entanglement between $M_A \otimes M_B$ and the environment $M_{A^\prime} \otimes M_{B^\prime}$.}\label{fig:sep_eb_channels}
\end{figure}

We now prove a generalization of this result for higher Schmidt number that makes use of the operator-Schmidt rank of the map's Kraus operators. The channels characterized by the following theorem are exactly the channels such that if Alice and Bob each have their own states that are potentially entangled with their own environments, but are not entangled with each other's systems, then the Schmidt number between Alice and Bob after the channel is applied is no greater than $k$.
\begin{thm}\label{thm:ksep_map_choi}
	Let $\Phi : M_A \otimes M_B \rightarrow M_A \otimes M_B$ be a completely positive linear map and let $k \in \mathbb{N}$. The following are equivalent:
	\begin{enumerate}[(a)]
		\item $SN(C_\Phi) \leq k$ across the $(A^\prime,A)-(B^\prime,B)$ cut;
		\item $SN\big((id_{A^\prime,B^\prime} \otimes \Phi_{A,B})(\rho)\big) \leq k$ across the $(A^\prime,A)-(B^\prime,B)$ cut whenever $\rho$ is separable across the $(A^\prime,A)-(B^\prime,B)$ cut; and
		\item there exist Kraus operators $\big\{K_\ell\big\} \subset M_A \otimes M_B$ each with operator-Schmidt rank $\leq k$ such that
		\begin{align*}
			\Phi(X) = \sum_\ell K_\ell X K_\ell^\dagger \quad \forall \, X \in M_A \otimes M_B.
		\end{align*}
	\end{enumerate}
\end{thm}
\begin{proof}
	The implication $(b) \implies (a)$ is trivially true by choosing
	\begin{align*}
		\rho = \ketbra{\psi_{+}}{\psi_{+}}_{A^\prime,A} \otimes \ketbra{\psi_{+}}{\psi_{+}}_{B^\prime,B}.
	\end{align*}
	To see that $(c) \implies (b)$, assume without loss of generality that $\rho$ is a pure state $\ketbra{v}{v}_{A^\prime,A} \otimes \ketbra{w}{w}_{B^\prime,B}$ (the general result will then follow easily from convexity). Write each $A_\ell$ in its operator-Schmidt decomposition (where we absorb the operator-Schmidt coefficients into the operators $A_{\ell,h}$ themselves):
	\begin{align*}
		K_\ell = \sum_{h=1}^k A_{\ell,h} \otimes B_{\ell,h}.
	\end{align*}
	Then writing $(id_{A^\prime,B^\prime} \otimes \Phi_{A,B})(\rho)$ as an operator in $M_{A^\prime} \otimes M_{A} \otimes M_{B^\prime} \otimes M_{B}$ gives
	\begin{align}\label{eq:sep_chan_gen1}
		\sum_\ell \sum_{g,h=1}^{k} \left((I_{A^\prime} \otimes A_{\ell,g})\ketbra{v}{v}(I_{A^\prime} \otimes A_{\ell,h}^\dagger) \right) \otimes \left((I_{B^\prime} \otimes B_{\ell,g})\ketbra{w}{w}(I_{B^\prime} \otimes B_{\ell,h}^\dagger) \right).
	\end{align}
	Now define
	\begin{align*}
		c_{\ell,g}\ket{v_{\ell,g}} := (I_{A^\prime} \otimes A_{\ell,g})\ket{v} \quad \text{ and } \quad d_{\ell,g}\ket{w_{\ell,g}} := (I_{B^\prime} \otimes B_{\ell,g})\ket{w}.	
	\end{align*}
	Plugging the above definitions into Equation~\ref{eq:sep_chan_gen1} gives
	\begin{align*}
		(id_{A^\prime,B^\prime} \otimes \Phi_{A,B})(\rho) & = \sum_\ell \sum_{g,h=1}^{k} c_{\ell,g}c_{\ell,h}d_{\ell,g}d_{\ell,h}\ketbra{v_{\ell,g}}{v_{\ell,h}} \otimes \ketbra{w_{\ell,g}}{w_{\ell,h}} \\
		& = \sum_\ell \left(\sum_{g=1}^k c_{\ell,g}d_{\ell,g} \ket{v_{\ell,g}} \otimes \ket{w_{\ell,g}} \right)\left(\sum_{h=1}^k c_{\ell,h}d_{\ell,h} \bra{v_{\ell,h}} \otimes \bra{w_{\ell,h}} \right),
	\end{align*}
	which has Schmidt number no larger than $k$.
	
	To see that $(a) \implies (c)$, write $C_\Phi$ as a sum of rank-one operators in $(M_{A^\prime} \otimes M_A) \otimes (M_{B^\prime} \otimes M_B)$:
	\begin{align*}
		C_\Phi = \sum_\ell \left(\sum_{i=1}^k c_{\ell,i}\ket{v_{\ell,i}} \otimes \ket{w_{\ell,i}}\right)\left(\sum_{i=1}^k c_{\ell,i}\bra{v_{\ell,i}} \otimes \bra{w_{\ell,i}}\right).
	\end{align*}
	Proposition~\ref{prop:choi_kraus_ops_general} then says that $\Phi(X) = \sum_\ell K_\ell X K_\ell^\dagger$, where
	\begin{align*}
		K_\ell := \sum_{i=1}^k c_{\ell,i}{\rm mat}(\ket{v_{\ell,i}}) \otimes {\rm mat}(\ket{w_{\ell,i}}),
	\end{align*}
	which clearly has operator-Schmidt rank no larger than $k$.
\end{proof}

\subsection{Other Correspondences Related to Entanglement}\label{sec:choi_jamiolkowski_separable}

\subsubsection*{Completely co-positive maps -- Positive partial transpose operators}

A map $\Phi : M_m \rightarrow M_n$ is called \emph{completely co-positive} if $T \circ \Phi$ is completely positive. It follows easily from our results on completely positive maps that $\Phi$ is completely co-positive if and only if $C_{T \circ \Phi} = (id_m \otimes T)(C_\Phi) \geq 0$. In other words, $\Phi$ is completely co-positive if and only if its Choi matrix has positive partial transpose.

\subsubsection*{Binding entanglement maps -- Bound entangled operators}

Recall from Section~\ref{sec:bound_entangle} that a bound entangled state is one that is entangled, yet contains no ``useful'' entanglement. An \emph{entanglement binding map} \cite{HHH00} is a completely positive map $\Phi : M_m \rightarrow M_n$ such that $(id_m \otimes \Phi)(\rho)$ is either bound entangled or separable for any quantum state $\rho \in M_m \otimes M_m$. It turns out that $\Phi$ is a binding entanglement map if and only if its Choi matrix $C_\Phi$ is bound entangled or separable. Via the result of the previous section, we see that a map binds entanglement if it is both completely positive and completely co-positive. The question of whether or not there are other binding entanglement maps is equivalent to the question of whether or not there exist NPPT bound entangled states.

\subsubsection*{Anti-degradable maps -- shareable operators}

A completely positive map $\Phi : M_m \rightarrow M_n$ is called \emph{anti-degradable} \cite{CRS08,WPG07} if there exists a quantum channel $\Psi$ such that $\Psi \circ \Phi^{C} = \Phi$, where we recall that $\Phi^{C}$ is the complementary map of $\Phi$. Anti-degradable maps have gotten attention in quantum information theory recently because they are one of only two families of maps (the other being binding entanglement maps) that are known to have zero quantum capacity \cite{SS12}. For convenience, we will denote the set of anti-degradable maps by $\cl{AD}(M_m,M_n)$, or simply $\cl{AD}$ if the dimensions of the input and output spaces are unimportant or clear from context.

It was shown in \cite{ML09} that a map $\Phi$ is anti-degradable if and only if its Choi matrix $C_\Phi$ is shareable (strictly speaking, this equivalence was only shown for quantum channels $\Phi$, but the same proof applies to this slightly more general case). For clarity, it is worth recalling that we have made the convention that $\Phi$ acts on the second subsystem of the Choi matrix and that it is also the second subsystem of a shareable operator that is shared.

With this correspondence in mind, several properties of anti-degradable maps immediately follow from the corresponding properties of shareable operators. First, $\cl{AD}$ is convex and $\Psi \circ \Phi \in \cl{AD}$ whenever $\Phi \in \cl{AD}$ and $\Psi$ is a quantum channel -- two facts that were proved directly in \cite{CRS08}. Because the shareable operators are a strict superset of the separable operators, it also follows that superpositive maps are anti-degradable, which was also proved directly in \cite{CRS08}.

We similarly know that the set of anti-degradable maps satisfies $\Phi \circ \Psi \in \cl{AD}$ whenever $\Phi \in \cl{AD}$ and $\Psi$ is completely positive -- a property that we call right CP-invariance, which we investigate in Section~\ref{sec:cp_invariant}. To see this, first use Lemma~\ref{lem:comp_adjoint}: if $\Phi = \Psi \circ \Phi^{C}$ then $\Phi \circ {\rm Ad}_B = \Psi \circ \Phi^{C} \circ {\rm Ad}_B = \Psi \circ (\Phi \circ {\rm Ad}_B)^C$, so $\Phi \circ {\rm Ad}_B$ is anti-degradable whenever $\Phi$ is anti-degradable. Convexity of $\cl{AD}$ then gives the result.

\subsubsection*{$s$-extendible maps -- $s$-shareable operators}

A completely positive map $\Phi : M_m \rightarrow M_n$ is called \emph{$s$-extendible} \cite{PBHS11} (or sometimes \emph{local $s$-broadcasting} \cite{LRKKB11}) if there exists a completely positive map $\Psi : M_m \rightarrow M_n^{\otimes s}$ with the property that $\Phi = \Tr_{\overline{i}} \circ \Psi$ for all $1 \leq i \leq s$, where $\Tr_{\overline{i}}$ denotes the partial trace over all subsystems except for the $i$-th. The map $\Psi$ is said to be $s$-broadcasting, and maps of this type generalize the notion of quantum cloning \cite{BCFJS96}. We will denote the cone of $s$-extendible maps by $\cl{B}_s(M_m,M_n)$, or simply $\cl{B}_s$.

It was shown in \cite{LRKKB11} that a map $\Phi$ is $s$-extendible if and only if its Choi matrix $C_\Phi$ is $s$-shareable. Once again, we note that it is the subsystem of $C_\Phi$ that $\Phi$ acts on that can be shared with $s$ parties. Remarkably, this leads immediately to the following result.
\begin{thm}\label{thm:antidegrad_2broadcast}
	Let $\Phi$ be completely positive. Then $\Phi$ is anti-degradable if and only if it is $2$-extendible.
\end{thm}
\begin{proof}
	A simple proof of this statement follows by looking at these sets of maps through the Choi--Jamio{\l}kowski isomorphism: $\Phi$ is anti-degradable if and only if $C_\Phi$ is shareable \cite{ML09} if and only if $\Phi$ is $2$-extendible \cite{LRKKB11}. However, we also present a direct proof of the result for completeness.
	
	We begin with the ``only if'' direction. If $\Phi$ is anti-degradable then write it in its Stinespring form $\Phi = \Tr_1 \circ {\rm Ad}_A$. Let $\Psi$ be a quantum channel such that $\Psi \circ \Phi^C = \Phi$. Then
	\begin{align*}
		\Phi = \Tr_1 \circ {\rm Ad}_A = \Psi \circ \Tr_2 \circ {\rm Ad}_A = \Psi \circ \Phi^C.
	\end{align*}
	Thus if we define the completely positive map $\tilde{\Phi} = (\Psi \otimes id) \circ {\rm Ad}_A$ then we have $\Tr_1 \circ \tilde{\Phi} = \Tr_2 \circ \tilde{\Phi} = \Phi$, so $\Phi$ is $2$-extendible.
	
	For the ``if'' direction of the proof, let $\tilde{\Phi} : M_m \rightarrow M_n \otimes M_n$ be a completely positive map such that $\Tr_1 \circ \tilde{\Phi} = \Tr_2 \circ \tilde{\Phi} = \Phi$. Write $\tilde{\Phi}$ in its Stinespring form $\tilde{\Phi} = \Tr_{3} \circ {\rm Ad}_A$ (where we consider the third subsystem as the environment and the first two subsystems as the output of $\tilde{\Phi}$). Then $\Phi$ has Stinespring representations $\Phi = \Tr_{\overline{1}} \circ {\rm Ad}_A = \Tr_{\overline{2}} \circ {\rm Ad}_A$. We thus have $\Phi^C = \Tr_2 \circ {\rm Ad}_A$, so $\Tr_3 \circ \Phi^C = \Tr_{\overline{1}} \circ {\rm Ad}_A = \Phi$, which shows that $\Phi$ is anti-degradable.
\end{proof}

All of the properties of anti-degradable maps that were discussed in the previous section apply to $s$-extendible maps as well. In particular, $\cl{B}_s$ is a convex right CP-invariant cone. We also have the family of inclusions $\cl{AD} \supseteq \cl{B}_3 \supseteq \cdots \supseteq \cl{B}_k \supseteq \cdots$, and the intersection of all these cones is the cone of superpositive maps -- see Figure~\ref{fig:antidegradable}.
\begin{figure}[ht]
\begin{center}
\includegraphics[width=\textwidth]{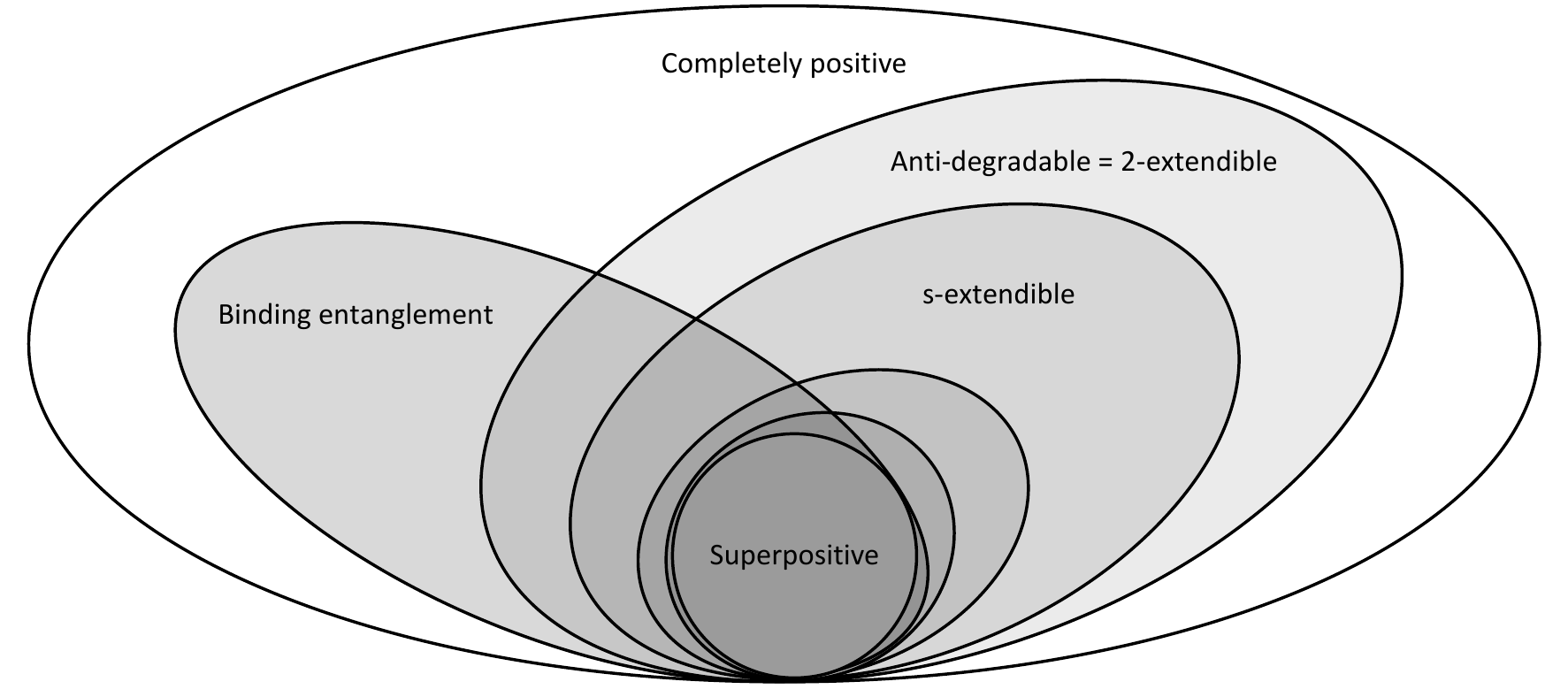}
\end{center}\vspace{-0.25in}
\caption[The cones of binding entanglement and anti-degradable maps]{\hsp The cones of binding entanglement, anti-degradable, and $s$-extendible maps, relative to the cones of completely positive and superpositive maps. Channels in any of the shaded cones have zero quantum capacity.}\label{fig:antidegradable}
\end{figure}

\chapter{Cones, Norms, and Linear Preservers}\label{ch:geometry}

Many of the objects introduced in Chapter~\ref{ch:basics} remain unaffected by scaling. For example, if $\Phi$ is a (completely) positive map, then so is $\lambda \Phi$ for any $\lambda \geq 0$. In other words, the set of completely positive maps is a \emph{cone}. Similarly for the sets of superpositive maps, separable maps, anti-degradable maps, $s$-extendible maps, and completely co-positive maps. This chapter begins with a review of general facts about cones of operators. We then investigate what more can be said when we ask that the cones in question have additional properties, such as being invariant under composition with completely positive maps.

We then shift focus and introduce various norms of operators and linear maps. We present the basics of unitarily-invariant norms, dual norms, and completely bounded norms. We also introduce the realignment criterion, which is our first separability criterion based on norms. Here we see our first glimpse of a theme that is common throughout the rest of this work: that norms can be just as useful as positivity properties for solving questions related to separability.

The final portion of the chapter is devoted to discussing preserver problems. That is, given a specific set (such as the set of separable states) or norm, what operators leave that set or norm unchanged? We answer this question in the case of (even multipartite) separable states, discuss some consequences of our results, and lay the groundwork that allows us to derive the isometry groups for the norms to be introduced in Chapter~\ref{chap:Sk_norms}.

\section{Cones of Linear Maps and Operators}\label{sec:cones}

A \emph{cone} $\cl{C} \subseteq M_n$ is a set of Hermitian operators with the property that if $X \in \cl{C}$ then $\lambda X \in \cl{C}$ for all $0 \leq \lambda \in \bb{R}$. Some cones that we have already seen are the sets of $k$-block positive operators and the sets of $s$-shareable operators. Similarly, a \emph{cone} $\cl{C}$ of superoperators is a set of Hermiticity-preserving maps with the property that if $\Phi \in \cl{C}$ then $\lambda \Phi \in \cl{C}$ for all $0 \leq \lambda \in \bb{R}$.

For convenience, we define some notation for common manipulations of the cone of superoperators $\cl{C}$:
\begin{align*}
	C_\cl{C} := \big\{ C_\Phi : \Phi \in \cl{C} \big\} \quad \text{ and } \quad \cl{C}^{\dagger} := \big\{ \Phi^\dagger : \Phi \in \cl{C} \big\}.
\end{align*}
It follows from the equivalence of conditions (a) and (b) of Proposition~\ref{prop:hermiticity_char} that $\cl{C}$ is a (convex) cone of superoperators if and only if $C_\cl{C}$ is a (convex) cone of operators.

\subsection{Dual Cones}\label{sec:dual_cones}

The \emph{dual cone} $\cl{C}^{\circ}$ of a cone $\cl{C} \subseteq M_n$ is defined via the Hilbert--Schmidt inner product as
\begin{align*}
	\cl{C}^{\circ} := \{ Y \in M_n : \Tr(XY) \geq 0 \ \text{ for all } X \in \cl{C} \}.
\end{align*}
Similarly, the dual cone $\cl{C}^{\circ}$ of a cone $\cl{C}$ of superoperators is defined via the Choi--Jamio{\l}kowski isomorphism as $\cl{C}^{\circ} := \{ \Psi : M_n \rightarrow M_n : \Tr(C_{\Phi}C_{\Psi}) \geq 0 \ \text{ for all } \Phi \in \cl{C} \}$. We note that for all cones $\cl{C} \subseteq M_n$, we have $\cl{C}^{\circ \circ} = \overline{{\rm hull}(\cl{C})}$ -- the closure of the convex hull of $\cl{C}$ (i.e., the smallest closed convex cone containing $\cl{C}$). This fact is well-known in convex analysis and follows from \cite[Theorem 3.4.3]{GY02} or \cite[Theorem 14.1]{R97}. In particular, if $\cl{C}$ is a closed convex cone, then $\cl{C}^{\circ\circ} = \cl{C}$.

Some other useful and easily-verified facts about dual cones are:
\begin{itemize}
	\item The cone of positive semidefinite operators is self-dual. That is, $(M_n^+)^\circ = M_n^+$.
	\item If $\cl{C},\cl{D}$ are two cones such that $\cl{C} \subseteq \cl{D}$ then $\cl{D}^\circ \subseteq \cl{C}^\circ$.
\end{itemize}

Because the cone of $k$-block positive operators contains $(M_m \otimes M_n)^+$, the above properties tell us that its dual cone must be contained within $(M_m \otimes M_n)^+$. Indeed, the dual of the cone of $k$-block positive operators is exactly the set of (unnormalized) states that have Schmidt number no larger than $k$ \cite{SSZ09} (and vice-versa) -- see Table~\ref{table:dual_cones} and refer back to Figure~\ref{fig:schmidt}.

\begin{table}[ht]\hsp
	\begin{center}
  \begin{tabular}{ c | c | c | c }
  	\noalign{\hrule height 0.1em}
  	\multicolumn{2}{c|}{Operators $X \in M_m \otimes M_n$} & \multicolumn{2}{c}{Superoperators $\Phi : M_m \rightarrow M_n$} \\
  	\hline
    Cone & Dual cone & Cone & Dual cone \\
  	\noalign{\hrule height 0.1em}
    block positive & separable & positive & superpositive \\ \hline
    $k$-block positive & Schmidt number $\leq k$ & $k$-positive & $k$-superpositive \\ \hline
    \multicolumn{2}{c|}{positive semidefinite} & \multicolumn{2}{c}{completely positive} \\ \noalign{\hrule height 0.1em}
  \end{tabular}
	\end{center}
\caption[Cones and dual cones of operators and linear maps]{\hsp Some cones of operators and linear maps and their associated dual cones. The cone of positive semidefinite operators (and thus the cone of completely positive maps) is its own dual cone. The chain of inclusions of these cones follows a U-shape in the table, with the block positive operators (positive maps) being the largest cone and the separable operators (superpositive maps) being the smallest cone.}\label{table:dual_cones}
\dsp\end{table}

Similarly, because the cone of operators with an $s$-bosonic symmetric extension is contained within $(M_m \otimes M_n)^+$, its dual cone contains $(M_m \otimes M_n)^+$. A straightforward calculation reveals that if $\rho$ has an $s$-BSE $\tilde{\rho}$ then
\begin{align*}
	\Tr(\rho X) & = \Tr\big(\tilde{\rho} (X \otimes I_n^{\otimes (s-1)})\big) \\
	& = \Tr\big(\tilde{\rho} (I_m \otimes P_{\cl{S}})(X \otimes I_n^{\otimes (s-1)})(I_m \otimes P_{\cl{S}})\big),
\end{align*}
where $P_{\cl{S}}$ is the projection onto the symmetric subspace of $(\bb{C}^n)^{\otimes s}$. Thus the dual of the cone of operators with an $s$-BSE is the cone
\begin{align}\label{eq:s_bse_dual}
	\big\{ X \in M_m \otimes M_n : (I_m \otimes P_{\cl{S}})(X \otimes I_n^{\otimes (s-1)})(I_m \otimes P_{\cl{S}}) \geq 0 \big\}.
\end{align}

Because the set of operators with $s$-BSE approaches the set of separable operators from the outside, the dual cones~\eqref{eq:s_bse_dual} approach the set of block positive operators from the inside (see Figure~\ref{fig:shareable}).

\begin{figure}[ht]
\begin{center}
\includegraphics[width=\textwidth]{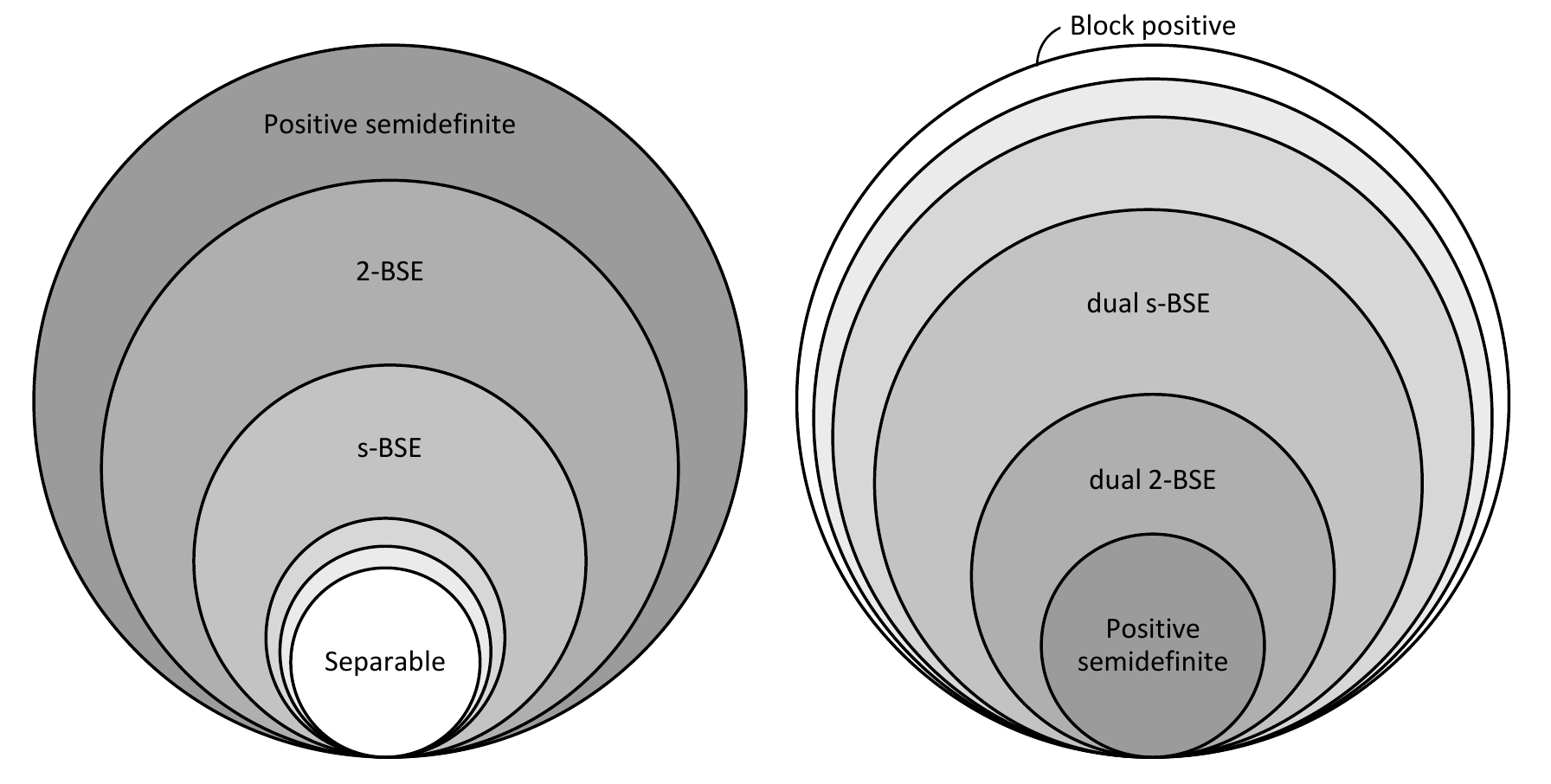}
\end{center}\vspace{-0.3in}
\caption[The cones and dual cones of operators with $s$-symmetric extension]{\hsp A rough depiction of the cones and dual cones of operators with $s$-bosonic symmetric extension relative to the cones of separable, positive semidefinite, and block positive operators. The positive semidefinite cone is self-dual and equals the cone of operators with $1$-BSE. The separable cone is the intersection over all $s \geq 1$ of the cones of operators with $s$-BSE, and the block positive cone is the closure of the union over all $s \geq 1$ of the duals of the cones of $s$-shareable operators.}\label{fig:shareable}
\end{figure}

We close this section with a simple lemma that demonstrates how $\cl{C}^\circ$ behaves with $\cl{C}^\dagger$.
\begin{lemma}\label{lem:sup_cone_duals}
	Let $\cl{C}$ be a cone of superoperators. Then $(\cl{C}^\dagger)^\circ = (\cl{C}^\circ)^\dagger$.
\end{lemma}
\begin{proof}
	Suppose that $\Omega \in (\cl{C}^\dagger)^\circ$. Then
	\begin{align*}
		\langle \Omega | \Phi^\dagger \rangle \geq 0 \quad \text{ for all } \Phi \in \cl{C}.
	\end{align*}
	Condition (b) of Proposition~\ref{prop:choi_inner_basic}, together with the fact that all maps considered here are Hermiticity-preserving, tells us that this is the same as
	\begin{align*}
		\langle \Omega^\dagger | \Phi \rangle \geq 0 \quad \text{ for all } \Phi \in \cl{C}.
	\end{align*}
	This means that $\Omega^\dagger \in \cl{C}^\circ$, so $\Omega \in (\cl{C}^\circ)^\dagger$, so $(\cl{C}^\dagger)^\circ \subseteq (\cl{C}^\circ)^\dagger$. The opposite inclusion is proved by following this same argument backward.
\end{proof}

\subsection{Mapping Cones}\label{sec:mapping_cones}

A \emph{mapping cone} \cite{S86} is a nonzero closed cone $\cl{C}$ of positive maps on $M_n$ with the property that $\Phi \circ \Omega \circ \Psi \in \cl{C}$ whenever $\Omega \in \cl{C}$ and $\Phi,\Psi : M_n \rightarrow M_n$ are completely positive. Many of the cones that we have considered already are mapping cones: the cones of completely positive maps, $k$-positive maps, $k$-superpositive maps, and completely co-positive maps are all examples.

Mapping cones have gained some interest lately due to the fact that many properties of $k$-positive maps and $k$-superpositive maps stem from the fact that they are mapping cones \cite{Sko11,SS10,SSZ09,S11,Sto11}. We do not mention these properties further here because we generalize them in the next section.

\subsection{Right CP-Invariant Cones}\label{sec:cp_invariant}

It is sometimes be useful to consider (not necessarily closed) cones $\cl{C}$ such that $\Omega \circ \Psi \in \cl{C}$ whenever $\Omega \in \cl{C}$ and $\Psi$ is completely positive -- that is, cones that are closed under right-composition, but not necessarily left-composition, by completely positive maps. We call such cones \emph{right CP-invariant}. Left CP-invariant cones can be defined analogously, and it is clear that $\cl{C}$ is right CP-invariant if and only if $\cl{C}^{\dagger}$ is left CP-invariant. By definition, all mapping cones are right CP-invariant, but there are right CP-invariant cones that are not mapping cones. The most familiar example for us of a cone that is right CP-invariant but not a mapping cone is the cone of anti-degradable maps, which is not left CP-invariant.

To help motivate why right CP-invariant cones are interesting for us, we first remark on a pattern that we have seen a few times. Completely positive maps are defined by the fact that $(id_n \otimes \Phi)(\rho) \geq 0$ for all $\rho \geq 0$, yet it is enough to check only that $C_\Phi \geq 0$ to determine complete positivity. Similarly, we saw that a map $\Phi$ is $k$-superpositive if and only if $SN((id_n \otimes \Phi)(\rho)) \leq k$ for all $\rho \geq 0$, which is equivalent to the seemingly simpler condition $SN(C_\Phi) \leq 0$. A similar statement holds for $k$-positive maps if we replace Schmidt number by $k$-block positivity. In other words, we have seen that instead of checking that $(id_n \otimes \Phi)(\rho)$ satisfies a given property for all $\rho \geq 0$, it is often enough to simply check that property is satisfied when $\rho = \ketbra{\psi_+}{\psi_+}$. The following proposition shows that this happens exactly because the cones we have discussed are right CP-invariant.
\begin{prop}\label{tlem:genCharacterize}
	Let $\cl{C} \subseteq \cl{L}(M_m, M_n)$ be a right CP-invariant cone and let $\Phi \in \cl{L}(M_m, M_n)$. Then the following are equivalent:
	\begin{enumerate}[(a)]
		\item $\Phi \in \cl{C}$; and
		\item $(id_m \otimes \Phi)(\rho) \in C_{\cl{C}}$ for all $0 \leq \rho \in M_m \otimes M_m$.
	\end{enumerate}
\end{prop}
\begin{proof}
	The implication (b) $\Rightarrow$ (a) is trivial because we can take $\rho = \ketbra{\psi_+}{\psi_+}$. To see that (a) $\Rightarrow$ (b), note that because $\cl{C}$ is a right CP-invariant cone, we have that $\Phi \circ \Psi \in \cl{C}$ for all $\Psi \in \cl{CP}$. Thus $m(id_m \otimes (\Phi \circ \Psi))(\ketbra{\psi_+}{\psi_+}) = (id_m \otimes \Phi)(C_{\Psi}) \in C_{\cl{C}}$ for all $\Psi \in \cl{CP}$. The result then comes from the fact that $\Psi \in \cl{CP}$ if and only if $C_{\Psi} \geq 0$.
\end{proof}

Some of the most interesting properties of a right CP-invariant cone $\cl{C}$ involve relationships between $\cl{C}$ and $\cl{C}^\circ$. The following proposition is a starting point.
\begin{prop}\label{prop:right_cp_dual}
	If $\cl{C} \subseteq \cl{L}(M_m, M_n)$ is a right (left) CP-invariant cone then so is $\cl{C}^\circ$.
\end{prop}
\begin{proof}
	We prove the result for left CP-invariant cones $\cl{C}$. The corresponding result for right CP-invariance follows from the fact that $\cl{C}^\dagger$ is right CP-invariant and $(\cl{C}^\dagger)^\circ = (\cl{C}^\circ)^\dagger$.
	
	Let $\Psi \in \cl{C}^\circ$. Then $\Tr(C_\Psi C_\Omega) \geq 0$ for all $\Omega \in \cl{C}$. However, left CP-invariance of $\Phi$ tells us that $\Phi \circ \Omega \in \cl{C}$ for all completely positive $\Phi$. Thus $\Tr(C_\Psi C_{\Phi \circ \Omega}) \geq 0$ for all $\Omega \in \cl{C}$ and all completely positive $\Phi$. Then
	\begin{align*}
		\Tr(C_\Psi C_{\Phi \circ \Omega}) & = \Tr(C_{\Phi^\dagger \circ \Psi} C_{\Omega}) \geq 0,
	\end{align*}
	so $\Phi^\dagger \circ \Psi \in \cl{C}^\circ$ for all completely positive $\Phi$. Left CP-invariance of $\cl{C}^\circ$ follows.
\end{proof}

Proposition~\ref{prop:right_cp_dual} generalizes the fact that the dual of a mapping cone is itself a mapping cone -- a fact that was originally proved in the special case of ``symmetric'' mapping cones in \cite{S11} and in general in \cite{Sko11}. We also note that our proof is much simpler than all previous proofs of this fact.

Our next result on right CP-invariant cones shows how the composition of a map from $\cl{C}^\dagger$ and a map from $\cl{C}^\circ$ behaves. This result is also known in the special case of mapping cones \cite{Sko11,S09}.
\begin{prop}\label{prop:right_cp_main}
	Let $\cl{C} \subseteq \cl{L}(M_m, M_n)$ be a cone of superoperators and let $\Phi \in \cl{L}(M_m, M_n)$. If $\cl{C}$ is right CP-invariant then conditions (a) and (b) below are equivalent. If $\cl{C}$ is left CP-invariant then conditions (a) and (c) are equivalent.
	\begin{enumerate}[(a)]
		\item $\Phi \in \cl{C}^\circ$;
		\item $\Omega^\dagger \circ \Phi$ is completely positive for all $\Omega \in \cl{C}$; and
		\item $\Phi \circ \Omega^\dagger$ is completely positive for all $\Omega \in \cl{C}$.
	\end{enumerate}
\end{prop}
\begin{proof}
	To see that (a) $\implies$ (b) when $\cl{C}$ is right CP-invariant, first let $\Phi \in \cl{C}^\circ$. Then $\langle \Phi | \Omega \circ \Psi \rangle \geq 0$ for all $\Omega \in \cl{C}$ and completely positive $\Psi$. Then property~(a) of Proposition~\ref{prop:choi_inner_basic} says that $\langle \Omega^\dagger \circ \Phi | \Psi \rangle \geq 0$. Thus $\Omega^\dagger \circ \Phi$ is in the dual cone of the completely positive maps for all $\Omega$. Since the cone of completely positive maps is self-dual, condition (b) follows. This argument also works in reverse to show that (b) $\implies$ (a).
	
	The proof of equivalence of conditions (a) and (c) when $\cl{C}$ is left CP-invariant is similar. For any $\Phi, \Psi,$ and $\Omega$ we have
	\begin{align*}
		\langle \Phi | \Psi \circ \Omega \rangle = \langle \Psi^\dagger \circ \Phi | \Omega \rangle = \langle \Phi^\dagger \circ \Psi | \Omega^\dagger \rangle = \langle \Psi | \Phi \circ \Omega^\dagger \rangle,
	\end{align*}
	where each equality follows from either property~(a) or~(b) of Proposition~\ref{prop:choi_inner_basic}. The result follows by noting that the inner product on the left is nonnegative exactly when the inequality on the right is nonnegative.
\end{proof}

We now give a more concrete corollary of Proposition~\ref{prop:right_cp_main} by showing what it says about the cones of $k$-positive maps $\cl{P}_k$ and $k$-superpositive maps $\cl{S}_k$.
\begin{cor}\label{cor:right_cp_main_cor2}
	Let $\Omega \in \cl{L}(M_m, M_n)$. The following are equivalent:
	\begin{enumerate}[(a)]
		\item $\Omega \in \cl{S}_k$;
		\item $\Omega \circ \Phi$ is completely positive for all $\Phi \in \cl{P}_k$; and
		\item $\Phi \circ \Omega$ is completely positive for all $\Phi \in \cl{P}_k$.
	\end{enumerate}
\end{cor}

\begin{cor}\label{cor:right_cp_main_cor1}
	Let $\Phi \in \cl{L}(M_m, M_n)$. The following are equivalent:
	\begin{enumerate}[(a)]
		\item $\Phi \in \cl{P}_k$;
		\item $\Omega \circ \Phi$ is completely positive for all $\Omega \in \cl{S}_k$; and
		\item $\Phi \circ \Omega$ is completely positive for all $\Omega \in \cl{S}_k$.
	\end{enumerate}
\end{cor}

Both of these corollaries follow immediately from Proposition~\ref{prop:right_cp_main} and the following simple facts: $\cl{P}_k^\dagger = \cl{P}_k$, $\cl{S}_k^\dagger = \cl{S}_k$, $\cl{P}_k^\circ = \cl{S}_k$, $\cl{S}_k^\circ = \cl{P}_k$, and $\cl{P}_k$ and $\cl{S}_k$ are both left and right CP-invariant. Corollary~\ref{cor:right_cp_main_cor2} was originally proved in the $k = 1$ case in \cite{HSR03}, for arbitrary $k$ in \cite{CK06}, and then re-proved along with Corollary~\ref{cor:right_cp_main_cor1} in \cite{SSZ09}.

We return to right CP-invariant cones in Section~\ref{sec:opSysRight}, where we show that they are the ``natural'' cones of superoperators that arise when dealing with operator systems on matrices.


\subsection{Semigroup Cones}\label{sec:semigroup_cone}

It was shown in \cite[Theorem 3.8]{SSZ09} that Corollaries~\ref{cor:right_cp_main_cor2} and~\ref{cor:right_cp_main_cor1} are in some sense weaker than they should be, since it is actually the case that $\Omega \circ \Phi$ and $\Phi \circ \Omega$ are $k$-superpositive (rather than just completely positive) whenever $\Omega \in \cl{S}_k$ and $\Phi \in \cl{P}_k$. It is natural to ask what property of the cones $\cl{P}_k$ and $\cl{S}_k$ makes this the case, since it is not a consequence of either right or left CP-invariance.

We now show that it is the fact that $\cl{P}_k$ is a \emph{semigroup} (i.e., it satisfies $\Phi \circ \Psi \in \cl{P}_k$ for all $\Phi, \Psi \in \cl{P}_k$) that gives this extra structure.
\begin{prop}\label{prop:semigroup}
	Let $\cl{C} \supseteq \cl{CP}(M_m, M_n)$ be a semigroup cone and let $\Phi \in \cl{L}(M_m, M_n)$. The following are equivalent:
	\begin{enumerate}[(a)]
		\item $\Phi \in \cl{C}^\circ$;
		\item $\Omega^\dagger \circ \Phi \in \cl{C}^\circ$ for all $\Omega \in \cl{C}$; and
		\item $\Phi \circ \Omega^\dagger \in \cl{C}^\circ$ for all $\Omega \in \cl{C}$.
	\end{enumerate}
\end{prop}
\begin{proof}
	The proof is almost identical to the proof of Proposition~\ref{prop:right_cp_main}, except we let $\Psi \in \cl{C}$ rather than $\Psi \in \cl{CP}$ throughout the proof.
\end{proof}

\begin{prop}\label{prop:semigroup2}
	Let $\cl{C} \supseteq \cl{CP}(M_m, M_n)$ be a closed, convex semigroup cone with $\cl{C}^\circ \circ \cl{C}^\circ = \cl{C}^\circ$ and let $\Omega \in \cl{L}(M_m, M_n)$. The following are equivalent:
	\begin{enumerate}[(a)]
		\item $\Omega \in \cl{C}$;
		\item $\Omega^\dagger \circ \Phi \in \cl{C}$ for all $\Phi \in \cl{C}^\circ$; and
		\item $\Phi \circ \Omega^\dagger \in \cl{C}$ for all $\Phi \in \cl{C}^\circ$.
	\end{enumerate}
\end{prop}
\begin{proof}
	The implications (a) $\Rightarrow$ (b) and (a) $\Rightarrow$ (c) both follow immediately from Proposition~\ref{prop:semigroup} and the fact that $\cl{C}^\circ \subseteq \cl{C}$. To see that (b) $\Rightarrow$ (a), first note that $\cl{C}$ being a semigroup implies that $\cl{C}$ is a both left and right CP-invariant, which implies via Proposition~\ref{prop:right_cp_dual} that $\cl{C}^\circ$ is also left and right CP-invariant. Because $\cl{C}^\circ \subseteq \cl{CP}$, it follows that $\Phi \circ \Psi^\dagger \in \cl{C}^\circ$ whenever $\Phi, \Psi \in \cl{C}^\circ$. Thus, if $\Omega^\dagger \circ \Phi \in \cl{C}$ for all $\Phi \in \cl{C}^\circ$ then
	\begin{align*}
		0 & \leq \big\langle \Omega^\dagger \circ \Phi | \Psi \big\rangle = \big\langle \Phi \circ \Psi^\dagger | \Omega \big\rangle \quad \text{ for all } \quad \Phi, \Psi \in \cl{C}^\circ,
	\end{align*}
	where we have used Proposition~\ref{prop:choi_inner_basic} twice. It follows that $\Omega \in \cl{C}^{\circ\circ} = \cl{C}$. The proof of the implication (c) $\Rightarrow$ (a) is similar.
\end{proof}

We return to semigroup cones in Section~\ref{sec:semigroups}, where we show that they also play a natural role in operator system theory. For now we close this section with Table~\ref{table:cone_invariance_prop}, which provides an easy reference for the invariance properties satisfied by the superoperator cones of interest for us.
\begin{table}[ht]\hsp
	\begin{center}
  \begin{tabular}{ l | c c c }
  	\noalign{\hrule height 0.1em}
    Superoperator cone & Right CP-inv. & Left CP-inv. & Semigroup \\
  	\noalign{\hrule height 0.1em}
    completely positive & \checkmark & \checkmark & \checkmark \\ \hline
    $k$-positive & \checkmark & \checkmark & \checkmark \\ \hline
    $k$-superpositive & \checkmark & \checkmark & \checkmark \\ \hline
    separable & & & \checkmark \\ \hline
    $k$-separable & & & \\ \hline
    binding entanglement & \checkmark & \checkmark & \checkmark \\ \hline
    completely co-positive & \checkmark & \checkmark & \\ \hline
    anti-degradable & \checkmark & & \checkmark \\ \hline
    $s$-extendible & \checkmark & & \checkmark \\
    \noalign{\hrule height 0.1em}
  \end{tabular}
	\end{center}
	\caption[Invariance properties of cones of superoperators]{\hsp The invariance properties of several cones of superoperators.}\label{table:cone_invariance_prop}
\dsp\end{table}

\section{Norms}\label{sec:norms}

Much of this work will focus on the interplay between a variety of different norms. Here we will introduce the vector and matrix norms that are of most use in quantum information theory. For a more general introduction to norms on $\mathbb{C}^n$ and $M_n$, the interested reader is directed to \cite{Bha97,HJ85,HJ91,L94,Li00}.

The norm on $\mathbb{C}^n$ that we will use most frequently is the \emph{Euclidean norm}, which we will denote simply by $\|\cdot\|$ and is defined by
\begin{align*}
	\big\| (v_1,\ldots,v_n) \big\| = \sqrt{\sum_{i=1}^n |v_i|^2}.
\end{align*}
If we ever refer to the length or norm of a vector without specifying what norm we are considering, we implicitly mean the Euclidean norm. For example, when we said that a pure quantum state $\ket{v}$ is always a unit vector, we meant that $\big\|\ket{v}\big\| = 1$.

More generally, we define the \emph{$p$-norm} of a vector for $1 \leq p \leq \infty$ by
\begin{align*}
	\big\| (v_1,\ldots,v_n) \big\|_p = \begin{cases}\left(\displaystyle\sum_{i=1}^n |v_i|^p\right)^\frac{1}{p} & \text{ if } 1 \leq p < \infty \\ \displaystyle\lim_{q \rightarrow\infty} \big\| (v_1,\ldots,v_n) \big\|_q & \text{ if } p = \infty\end{cases},
\end{align*}
and we note that $\big\| (v_1,\ldots,v_n) \big\|_\infty = \displaystyle\max_{1 \leq i \leq n} |v_i|$. Indeed, if $p = 2$ then the Euclidean norm itself is obtained as a special case.

One norm on matrices that we have already seen is the \emph{Frobenius norm} $\|\cdot\|_F$, defined for $X = (x_{ij}) \in M_{n,m}$ by $\big\|X\big\|_F := \sqrt{\sum_{i=1}^n\sum_{j=1}^m x_{ij}^2}$. The Frobenius norm is essentially the Euclidean norm on the space of matrices, as we noted when we investigated the vector-operator isomorphism in Section~\ref{sec:vector_operator_isomorphism}. Another simple characterization of the Frobenius norm is $\big\|X\big\|_F = \sqrt{\Tr\big(X^\dagger X\big)} = \sqrt{\sum_{i=1}^{\min\{m,n\}} \sigma_i^2}$, where $\sigma_1 \geq \sigma_2 \geq \cdots \geq \sigma_{\min\{m,n\}} \geq 0$ are the singular values of $X$.

The two norms on $M_{n,m}$ that are most frequently used in quantum information theory are the \emph{operator norm} and \emph{trace norm}, defined by
\begin{align*}
	\big\|X\big\| := \sup_{\ket{v},\ket{w}} \Big\{ \big|\bra{w}X\ket{v}\big| \Big\} \quad \text{ and } \quad \big\|X\big\|_{tr} := \Tr\big(\sqrt{X^\dagger X}\big),
\end{align*}
respectively. The name of the trace norm comes from the fact that if $X \geq 0$ then we have $\big\|X\big\|_{tr} = \Tr(X)$. Both of these norms can be characterized in terms of singular values, much like the Frobenius norm: $\big\|X\big\| = \sigma_1$ and $\big\|X\big\|_{tr} = \sum_{i=1}^{\min\{m,n\}} \sigma_i$.

The \emph{Schatten $p$-norms} \cite{Sch60}, defined by $\big\|X\big\|_p := \left(\sum_{i=1}^{\min\{m,n\}} \sigma_i^p \right)^{1/p}$ for $1 \leq p < \infty$ and $\big\|X\big\|_{\infty} := \displaystyle\lim_{p\rightarrow\infty}\big\|X\big\|_p = \sigma_1$, generalize the operator and trace norms. In particular, the Schatten $1$-norm equals the trace norm and the Schatten $\infty$-norm equals the operator norm. The Frobenius norm also arises in the $p = 2$ case. The operator and trace norms are also generalized by the \emph{Ky Fan $k$-norms} \cite{F51}, defined by $\big\|X\big\|_{(k)} := \sum_{i=1}^k \sigma_i$ for $1 \leq k \leq \min\{m,n\}$. In this case we have the Ky Fan $1$-norm equal to the operator norm and the Ky Fan $\min\{m,n\}$-norm equal to the trace norm. 

A natural generalization of the Ky Fan and Schatten norms are the \emph{$(k,p)$-norms} \cite{MF85}, defined by $\big\|X\big\|_{(k,p)} := \left(\sum_{i=1}^{k} \sigma_i^p \right)^{1/p}$ for $1 \leq k \leq \min\{m,n\}$, $1 \leq p < \infty$ and $\big\|X\big\|_{(k,\infty)} := \displaystyle\lim_{p\rightarrow\infty}\big\|X\big\|_{(k,p)} = \sigma_1$ for $1 \leq k \leq \min\{m,n\}$.

Observe that every matrix norm $\mnorm{\cdot}$ introduced so far has the property that if $U \in M_n$ and $V \in M_m$ are unitary matrices, then $\bmnorm{UXV} = \bmnorm{X}$. Norms with this property are called \emph{unitarily-invariant}, and they are special in that the norm $\mnorm{\cdot}$ is unitarily invariant if and only if there is a function $f : \mathbb{R}^{\min\{m,n\}} \rightarrow \mathbb{R}$ such that $\bmnorm{X} = f(\sigma_1, \sigma_2, \ldots, \sigma_{\min\{m,n\}})$ for all $X \in M_{n,m}$, where $\sigma_1 \geq \sigma_2 \geq \cdots \geq 0$ are the singular values of $X$. They are also special in the sense of the following proposition, which was proved in \cite[Proposition~IV.2.4]{Bha97}.
\begin{prop}\label{prop:unitary_invar_symmetric}
	Let $\mnorm{\cdot}$ be a norm on $M_n$. Then $\mnorm{\cdot}$ is unitarily-invariant if and only if
	\begin{align*}
		\bmnorm{ABC} \leq \big\|A\big\| \bmnorm{B} \big\|C\big\|
	\end{align*}
	for all $A,B,C \in M_n$, where $\|\cdot\|$ refers to the operator norm.
\end{prop}

Given a superoperator $\Phi : M_m \rightarrow M_n$ and real numbers $1 \leq p,q \leq \infty$, the \emph{induced Schatten superoperator norm} \cite{KR04,KNR05,Wat05} of $\Phi$ is defined by
\begin{align*}
	\big\|\Phi\big\|_{q \rightarrow p} := \sup_{X}\left\{ \big\|\Phi(X)\big\|_{p} : \big\|X\big\|_q = 1 \right\},
\end{align*}
where $\|\cdot\|_p$ and $\|\cdot\|_q$ are the Schatten $p$- and $q$-norms. Some special cases of these norms that are particularly important include:
\begin{itemize}
	\item $p = q = 1$ gives the \emph{induced trace norm} $\big\|\Phi\big\|_{tr}$, which is the key building block of diamond norm that will be introduced in Section~\ref{sec:cb_norm}. This norm can be useful for measuring the distance between quantum channels in the restricted case when environment subsystems are not permitted. Furthermore, we show in Section~\ref{sec:contrac_sep_crit} that this norm can be used as a tool for detecting entanglement.
	\item $p = q = \infty$ gives the \emph{induced operator norm} $\big\|\Phi\big\|$, which appears frequently in operator theory \cite{P03} and is the key building block of the completely bounded norm that will be introduced in Section~\ref{sec:cb_norm}.
	\item $p = q = 2$ gives the \emph{induced Frobenius norm} of $\Phi$, which is easily seen to be equal to the standard operator norm of $M_\Phi$ -- the matrix associated with $\Phi$ via the vector-operator isomorphism of Section~\ref{sec:vector_operator_isomorphism}.
	\item In the case when $\Phi$ is a quantum channel, $p = \infty, q = 1$ gives the \emph{maximal output purity}, which we will explore in Section~\ref{sec:max_output_purity}.
\end{itemize}

\subsection{Dual Norms}\label{sec:dual_norms}

Given a norm $\mnorm{\cdot}$ on $\mathbb{C}^n$, its \emph{dual norm} $\mnorm{\cdot}^\circ$ is defined as follows:
\begin{align}\label{eq:dual_norm_vec}
	\bmnorm{\ket{v}}^\circ := \sup_{\ket{w}} \left\{ \frac{\big| \braket{w}{v} \big|}{\mnorm{\ket{w}}} \right\}.
\end{align}
Even though Equation~\eqref{eq:dual_norm_vec} only involves unit vectors $\ket{v}$, it extends in the natural way to all of $\mathbb{C}^n$. A direct consequence of this definition is that $|\braket{w}{v}| \leq \mnorm{\ket{v}}^\circ\mnorm{\ket{w}}$ for any $\ket{v}, \ket{w} \in \mathbb{C}^n$ and any norm $\mnorm{\cdot}$. Similarly, if $\mnorm{\cdot}_{(a)}$ and $\mnorm{\cdot}_{(b)}$ are any two norms such that $\mnorm{\cdot}_{(a)} \leq \mnorm{\cdot}_{(b)}$, then $\mnorm{\cdot}_{(a)}^\circ \geq \mnorm{\cdot}_{(b)}^\circ$. Furthermore, the dual of the dual norm is the original norm itself (i.e., $\mnorm{\cdot}^{\circ \circ} = \mnorm{\cdot}$).

The Euclidean norm is self-dual in that $\|\cdot\| = \|\cdot\|^\circ$. More generally, the $p$-norm $\|\cdot\|_p$ is dual to $\|\cdot\|_q$, where $\frac{1}{p} + \frac{1}{q} = 1$, and the convention is made that $p = 1$ implies $q = \infty$ (and vice-versa).

There is nothing particularly special about the space $\mathbb{C}^n$ in the above discussion -- all we need to define a dual norm is a norm and an inner product. Thus we similarly say that if $\mnorm{\cdot}$ is a norm on the space of matrices $M_n$, its dual norm $\mnorm{\cdot}^\circ$ is defined via the Hilbert--Schmidt inner product as follows:
\begin{align*}
	\bmnorm{X}^\circ := \sup_{Y \in M_n}\Big\{ \big|\Tr(X^\dagger Y)\big| : \mnorm{Y} \leq 1 \Big\}.
\end{align*}
Similar to before, it follows that $|\Tr(X^\dagger Y)| \leq \mnorm{X}^\circ\mnorm{Y}$ for any $X, Y \in M_n$ and any norm $\mnorm{\cdot}$, and the other basic properties of dual norms carry over to this setting naturally as well.

The trace norm is the dual of the operator norm (and hence the operator norm is the dual of the trace norm), and the Frobenius norm is its own dual. A well-known generalization of these relationships is that the dual of the Schatten $p$-norm is the Schatten $q$-norm, where $\frac{1}{p} + \frac{1}{q} = 1$ (with the usual convention that $1/\infty = 0$). The duals of the Ky Fan $k$-norms \cite[Exercise IV.2.12(iii)]{Bha97} are slightly less well-known:
\begin{align*}
	\big\|X\big\|_{(k)}^\circ = \max\Big\{ \big\|X\big\|, \frac{1}{k}\big\|X\big\|_{tr} \Big\}.
\end{align*}
The following result \cite[Theorem~3.3]{MF85} goes one step further and characterizes the duals of the $(k,p)$-norms (which we recall contain the operator, trace, Frobenius, Schatten, and Ky Fan norms as special cases).
\begin{thm}\label{thm:pk_norm_dual}
	Let $X \in M_{n,m}$ have singular values $\sigma_1 \geq \sigma_2 \geq \cdots \geq \sigma_{\min\{m,n\}} \geq 0$. Let $r$ be the largest index $1 \leq r < k$ such that $\sigma_r > \sum_{i=r+1}^{\min\{m,n\}}\sigma_{i}/(k-r)$ (or take $r = 0$ if no such index exists). Also define $\tilde{\sigma} := \sum_{i=r+1}^{\min\{m,n\}}\sigma_i/(k-r)$. Then	
	\begin{align*}
		\big\|X\big\|_{(k,p)}^\circ = \begin{cases}\left(\displaystyle\sum_{i=1}^r \sigma_i^q + (k - r)\tilde{\sigma}^q\right)^\frac{1}{q} & \text{ if } p > 1 \\
		\max\{\sigma_1,\tilde{\sigma}\} & \text{ if } p = 1\end{cases},
	\end{align*}
	where $q$ is such that $\frac{1}{p} + \frac{1}{q} = 1$ (with the usual convention that $p = \infty \implies q = 1$).
\end{thm}

For superoperators, we don't consider dual norms themselves. However, we frequently make use of the fact that $1/p + 1/p^\prime = 1/q + 1/q^\prime = 1$ then $\big\|\Phi\big\|_{q \rightarrow p} = \big\|\Phi^\dagger\big\|_{p^\prime \rightarrow q^\prime}$, which follows easily from the duality result for Schatten $p$-norms.

\subsection{Completely Bounded Norms}\label{sec:cb_norm}

Recall the induced Schatten superoperator norms defined by
\begin{align*}
	\big\|\Phi\big\|_{q \rightarrow p} := \sup_{X}\left\{ \big\|\Phi(X)\big\|_{p} : \big\|X\big\|_q = 1 \right\}.
\end{align*}
Based on these norms, we define $\big\|\Phi\big\|_{k,q\rightarrow p} := \big\|id_k \otimes \Phi\big\|_{q \rightarrow p}$ and note that $\big\|\Phi\big\|_{1,q\rightarrow p} = \big\|\Phi\big\|_{q\rightarrow p}$. In the $p = q = \infty$ and $p = q = 1$ cases, we also define the \emph{completely bounded} and \emph{diamond} versions of these norms:
\begin{align*}
	\big\|\Phi\big\|_{cb} := \sup_{k \geq 1} \Big\{ \big\| \Phi \big\|_{k,\infty \rightarrow \infty} \Big\} \ \ text{ and } \ \ \big\|\Phi\big\|_{\diamond} := \sup_{k \geq 1} \Big\{ \big\| \Phi \big\|_{k,1 \rightarrow 1} \Big\}.
\end{align*}
Note that, because $\big\|\Phi\big\| = \big\|\Phi^\dagger\big\|_{tr}$ for all $\Phi$, we also have $\big\|\Phi\big\|_{cb} = \big\|\Phi^\dagger\big\|_{\diamond}$.

It was shown by Smith \cite{S83} (and independently later by Kitaev \cite{Kit97} from the dual perspective) that if $\Phi : M_m \rightarrow M_n$ then it suffices to fix $k = n$ so that $\big\|\Phi\big\|_{cb} = \big\| id_n \otimes \Phi \big\|$. We also have the following well-known result \cite{JKP09}, which follows from a generalization of Stinespring's theorem to completely bounded maps \cite[Theorem~8.4]{P03}.
\begin{thm}\label{thm:cb_norm_gen_kraus}
	Let $\Phi : M_m \rightarrow M_n$ be a linear map. Then
	\begin{align*}
		\big\|\Phi\big\|_{cb}^2 = \inf\left\{ \Big\| \sum_i A_i A_i^\dagger \Big\| \Big\| \sum_i B_i B_i^\dagger \Big\| \right\},
	\end{align*}
	where the infimum is taken over all generalized Choi--Kraus representations $\Phi = \sum_i A_i X B_i^\dagger$. Furthermore, the infimum is attained.
\end{thm}

\subsection{Fidelity}\label{sec:fidelity}

There are many different tools that can be used to measure how similar two quantum states $\rho$ and $\sigma$ are \cite{Fuc96}. The measure that will be most useful for us is the \emph{fidelity} $\cl{F}$ \cite{Uhl76,Joz94}, which is defined by
\begin{align*}
	\cl{F}(\rho,\sigma) := \big\|\sqrt{\rho}\sqrt{\sigma}\big\|_{tr}^2.
\end{align*}
In the case when $\sigma$ is a pure state, the fidelity reduces to simply
\begin{align}\label{eq:fidelity_pure}
	\cl{F}(\rho,\ketbra{v}{v}) = \big\|\sqrt{\rho}\ketbra{v}{v}\big\|_{tr}^2 = \Tr \left( \sqrt{ \ketbra{v}{v} \rho \ketbra{v}{v}} \right)^2 = \bra{v}\rho\ket{v}.
\end{align}
The fidelity can be thought of as a measure of the overlap between $\rho$ and $\sigma$, and it satisfies $0 \leq \cl{F}(\rho,\sigma) \leq 1$ with $\cl{F}(\rho,\sigma) = 1$ if and only if $\rho = \sigma$.

\subsection{Realignment Criterion}\label{sec:realign}

We now briefly present a well-known separability criterion based on the trace norm. The \emph{realignment criterion} \cite{CW03} (sometimes called the \emph{computable cross-norm criterion} \cite{R03}) says that if a state $\rho \in M_m \otimes M_n$ is separable, then the following (equivalent) conditions must hold:
\begin{enumerate}[(1)]
	\item $\big\|R(\rho)\big\|_{tr} \leq 1$, where $R : M_m \otimes M_n \rightarrow M_{m,n} \otimes M_{m,n}$ is the linear map defined by $R(\ketbra{i}{j} \otimes \ketbra{k}{\ell}) \rightarrow \ketbra{i}{k} \otimes \ketbra{j}{\ell}$; and

	\item if we write $\rho$ in its operator-Schmidt decomposition $\rho = \sum_{i} \alpha_i A_i \otimes B_i$, then $\sum_i \alpha_i \leq 1$.
\end{enumerate}

To see condition (1), observe that if $\rho = \sum_i p_i \ketbra{v_i}{v_i} \otimes \ketbra{w_i}{w_i}$ is separable, then $R(\rho) = \sum_i p_i \ket{v_i}\overline{\bra{w_i}} \otimes \overline{\ket{v_i}}\bra{w_i}$. If we let $\mnorm{\cdot}$ be a unitarily-invariant matrix norm scaled so that $\bmnorm{\ketbra{1}{1} \otimes \ketbra{1}{1}} = 1$, then
\begin{align*}
	\bmnorm{R(\rho)} = \mnorm{\sum_i p_i \ket{v_i}\overline{\bra{w_i}} \otimes \overline{\ket{v_i}}\bra{w_i}} \leq \sum_i p_i \mnorm{\ket{v_i}\overline{\bra{w_i}} \otimes \overline{\ket{v_i}}\bra{w_i}} = \sum_i p_i = 1.
\end{align*}

The reason for choosing $\mnorm{\cdot} = \|\cdot\|_{tr}$ is that the trace norm is the largest unitarily-invariant matrix norm satisfying the given scaling condition, so the trace norm provides the strongest of these conditions for separability.

To see that condition (2) is equivalent to condition (1), simply note that
\begin{align}\label{eq:sing_val_Rrho}
	R(\rho) = R\left(\sum_{i} \alpha_i A_i \otimes B_i\right) = \sum_i \alpha_i S{\rm vec}(A_i){\rm vec}(B_i)^T,
\end{align}
where we recall that $S$ is the swap operator. Since the families of operators $\big\{A_i\big\}$ and $\big\{B_i\big\}$ are orthonormal in the Hilbert--Schmidt inner product, the families $\big\{S{\rm vec}(A_i)\big\}$ and $\big\{{\rm vec}(B_i)\big\}$ are also orthonormal, so the sum~\eqref{eq:sing_val_Rrho} is the singular value decomposition of $R(\rho)$. Thus the coefficients $\{\alpha_i\}$ are exactly the singular values of $R(\rho)$, so $\sum_i \alpha_i = \big\|R(\rho)\big\|_{tr}$.

It is worth looking at the map $R$ in terms of block matrices. If we write $\rho \in M_m \otimes M_n$ as a block matrix
\begin{align*}\sspp
	\rho = \begin{bmatrix}\rho_{11} & \rho_{12} & \cdots & \rho_{1m} \\ \rho_{21} & \rho_{22} & \cdots & \rho_{2m} \\ \vdots & \vdots & \ddots & \vdots \\ \rho_{m1} & \rho_{m2} & \cdots & \rho_{mm}\end{bmatrix},
\dsp\end{align*}
where each $\rho_{ij} \in M_n$, then
\begin{align*}\sspp
	R(\rho) = \begin{bmatrix} {\rm vec}(\rho_{11})^T \\ \vdots \\ {\rm vec}(\rho_{m1})^T \\ \vdots \\ {\rm vec}(\rho_{1m})^T \\ \vdots \\ {\rm vec}(\rho_{mm})^T \dsp\end{bmatrix}.
\end{align*}

Yet another way of looking at the realignment map is simply as a composition of a swap operator and the partial transpose map. More specifically, we have $R = \Phi \circ (id \otimes T) \circ \Phi$, where $\Phi(X) = XS$ is the map that multiplies on the right by the swap operator.

The realignment map, much like the partial transpose map, simply moves the matrix elements of $\rho$ around in the standard basis. Much like the partial transpose map is the prototypical example of a separability criterion based on the cone of positive maps, the realignment criterion is the prototypical example of a separability criterion based on norms. The connection between separability criteria based on positive cones and those based on norms was explored in \cite{HHH06}, and is the focus of much of Chapters~\ref{chap:Sk_norms} and~\ref{ch:optheory}.

The remainder of this section is devoted to proving the following generalization of the realignment criterion to arbitrary Schmidt number. Recall that the norm $\|\cdot\|_{(k^2,2)}^{\circ}$ of the following theorem was characterized by Theorem~\ref{thm:pk_norm_dual}.
\begin{thm}\label{thm:gen_realign}
	If $\rho \in M_m \otimes M_n$ has $SN(\rho) \leq k$, then $\big\|R(\rho)\big\|_{(k^2,2)}^{\circ} \leq 1$.
\end{thm}
\begin{proof}
	Begin by writing $\rho$ as a convex combination of projections onto states with Schmidt rank no greater than $k$:
	\begin{align*}
		\rho = \sum_i p_i \sum_{j,\ell=1}^k \alpha_{ij}\alpha_{i\ell}\ketbra{v_{ij}}{v_{i\ell}} \otimes \ketbra{w_{ij}}{w_{i\ell}}
	\end{align*}
	Then
	\begin{align*}
		R(\rho) = \sum_i p_i \left(\sum_{j=1}^k \alpha_{ij}\ket{v_{ij}}\overline{\bra{w_{ij}}}\right) \otimes \left(\sum_{\ell=1}^k\alpha_{i\ell}\overline{\ket{v_{i\ell}}}\bra{w_{i\ell}}\right).
	\end{align*}
	If we define $A_i := \sum_{j=1}^k \alpha_{ij}\ket{v_{ij}}\overline{\bra{w_{ij}}}$ then we have $R(\rho) = \sum_i p_i A_i \otimes \overline{A_i}$, where ${\rm rank}(A_i) \leq k$ and $\big\|A_i\big\|_F = 1$ for all $i$. In particular then, we have $R(\rho) = \sum_i p_i B_i$, where ${\rm rank}(B_i) \leq k^2$ and $\big\|B_i\big\|_F = 1$ for all $i$. Let $\mnorm{\cdot}$ be a unitarily-invariant matrix norm with the property that $\bmnorm{X} = \big\|X\big\|_F$ for all $X$ with ${\rm rank}(X) \leq k^2$. Then
\begin{align*}
	\bmnorm{R(\rho)} = \mnorm{\sum_i p_i B_i} \leq \sum_i p_i \mnorm{B_i} = \sum_i p_i \big\|B_i\big\|_F = \sum_i p_i = 1.
\end{align*}
All that remains is to make a suitable choice for $\mnorm{\cdot}$, so that this test for Schmidt number is as strong as possible. To this end, notice that $\|\cdot\|_{(k^2,2)}$ is clearly the smallest unitarily-invariant matrix norm with the required rank property. Also notice that, because the Frobenius norm is self-dual, $\|\cdot\|_{(k^2,2)}^{\circ}$ must satisfy the same rank property, and in particular must be the largest such matrix norm. We thus choose $\mnorm{\cdot} = \|\cdot\|_{(k^2,2)}^{\circ}$, which completes the proof.
\end{proof}

Notice that when $k = 1$, $\|\cdot\|_{(k^2,2)}^{\circ} = \|\cdot\|_{tr}$, so Theorem~\ref{thm:gen_realign} gives the standard realignment criterion in this case. On the other extreme, if $k = \min\{m,n\}$ then $\|\cdot\|_{(k^2,2)}^{\circ} = \|\cdot\|_{F}$. Because $R$ preserves the Frobenius norm, Theorem~\ref{thm:gen_realign} then simply says that $\|\rho\|_{F} \leq 1$ for all quantum states $\rho$, which is trivially true because $\|\rho\|_{F} \leq \|\rho\|_{tr} = 1$. The conditions given for the remaining values of $k$ are all non-trivial, yet easy to compute.

\section{Linear Preserver Problems}\label{sec:linear_preservers}

\subsection{Classical Results}\label{sec:linear_preservers_classical}

Some of our results will stem from classical results about linear maps on complex matrices that preserve some of their properties (such as their rank, singular values, or operator norm). The problem of characterizing such maps is known as a linear preserver problem, and the interested reader can find an overview of the subject in \cite{GLS00,LP01,LT92}.

One linear preserver problem that we have already seen is the problem of characterizing maps $\Phi : M_m \rightarrow M_n$ such that $(id_m \otimes \Phi)(X)$ is positive semidefinite whenever $X \in M_m \otimes M_m$ is positive semidefinite -- these are completely positive maps, which were characterized by Theorem~\ref{thm:choi_cp}. 

Another classical linear preserver problem that will be of great use to us concerns maps that are rank-$k$-non-increasing (or equivalently, rank-$k$-preserving).

\begin{prop}\label{prop:rankKpreserver}
	Let $k,m,n$ be positive integers such that $1 \leq k < \min\{m, n\}$ and let $\Phi : M_{n,m} \rightarrow M_{n,m}$ be an invertible linear map. Then ${\rm rank}(\Phi(X)) \leq k$ whenever ${\rm rank}(X) \leq k$ if and only if there exist nonsingular $P \in M_n$ and $Q \in M_m$ such that $\Phi$ is of one of the following two forms:
	\begin{align}\label{eq:rankKform}
		\Phi(X) = PXQ \quad \text{ or } \quad n = m \text{ and } \Phi(X) = PX^TQ.
	\end{align}
\end{prop}

Proposition~\ref{prop:rankKpreserver} is more often stated for maps that send operators with rank $k$ to operators with exactly rank $k$ \cite{B88,BL90,L89}. The above stronger version involving operators of rank at most $k$ can be found in \cite{GLS00}.

\subsection{Norm Isometries}\label{sec:linear_preservers_isometries}

The problem of characterizing linear maps that preserve a certain norm (i.e., the problem of characterizing the \emph{isometries} of that norm) can be thought of as a specific type of linear preserver problem. The set of isometries of the Euclidean norm on $\mathbb{C}^n$ is exactly the unitary group $U(n) \in M_n$, and in fact the set of isometries of any norm is always a group. Slightly less obvious is the fact that if an operator preserves the Euclidean norm of separable pure states then it preserves the norm of \emph{all} pure states (and hence is unitary). To prove this statement, we first need the following lemma.
\begin{lemma}\label{lem:sep_prod_zero}
	Let $X \in M_m \otimes M_n$. Then $\bra{v}X\ket{v} = 0$ for all separable $\ket{v} \in \bb{C}^m \otimes \bb{C}^n$ if and only if $X = 0$.
\end{lemma}
\begin{proof}
The ``if'' implication is trivial. To see the ``only if'' implication, write $X = \sum_{ij=1}^m \ketbra{i}{j} \otimes X_{ij}$, where $\big\{ X_{ij} \big\} \subset M_n$. If we write $\ket{v} = \ket{v_1}\otimes \ket{v_2}$ then
	\begin{align*}
		\bra{v_1}\Big( \sum_{ij=1}^m \bra{v_2}X_{ij}\ket{v_2} \ketbra{i}{j} \Big)\ket{v_1} = 0 \quad \forall \, \ket{v_1} \in \bb{C}^m,\ket{v_2} \in \bb{C}^n.
	\end{align*}
	It follows that $\sum_{ij=1}^m \bra{v_2}X_{ij}\ket{v_2} \ketbra{i}{j} = 0$ for all $\ket{v_2} \in \bb{C}^n$. However, because the set of operators $\big\{ \ketbra{i}{j} \big\}_{ij=1}^m$ is linearly independent, this implies that
	\begin{align*}
		\bra{v_2}X_{ij}\ket{v_2} = 0 \quad \forall \, i,j \text{ and } \forall \, \ket{v_2} \in \bb{C}^n.
	\end{align*}
	It follows that $X_{ij} = 0$ for all $i,j$ and so $X = 0$.
\end{proof}

\begin{prop}\label{prop:sep_unitary}
	Let $U \in M_m \otimes M_n$. Then $\big\| U\ket{v} \big\| = \big\| \ket{v} \big\|$ for all separable $\ket{v} \in \bb{C}^m \otimes \bb{C}^n$ if and only if $U$ is unitary.
\end{prop}
\begin{proof}
The ``if'' implication is trivial. To see the ``only if'' implication, note that $\bra{v}U^{\dagger}U\ket{v} = 1$ for all separable $\ket{v}$, so
\begin{align}\label{eq:UUI1}
	\bra{v}(U^{\dagger}U - I)\ket{v} = 0 \quad \forall \, \ket{v} \in \bb{C}^m \otimes \bb{C}^n \text{ with } SR(\ket{v}) = 1.
\end{align}
This immediately implies that $U^{\dagger}U = I$ via Lemma~\ref{lem:sep_prod_zero}, so $U$ is unitary.
\end{proof}

One family of matrix norm isometries follows easily from the vector-operator isomorphism of Section~\ref{sec:vector_operator_isomorphism}. Because the Euclidean norm on $\mathbb{C}^m \otimes \mathbb{C}^n$ corresponds to the Frobenius norm on $M_{n,m}$ via the vector-operator isomorphism, we know that a linear map $\Phi : M_{n,m} \rightarrow M_{n,m}$ preserves the Frobenius norm of all operators $X \in M_{n,m}$ (i.e., $\Phi$ is an isometry for the Frobenius norm) if and only if the associated operator $M_\Phi \in M_m \otimes M_n$ is unitary. We already saw that if $\Phi = \sum_k A_k X B_k^\dagger$ then $M_\Phi = \sum_k \overline{B_k} \otimes A_k$, which provides a simple concrete characterization of the isometries of the Frobenius norm. Because of this association between unitary matrices and isometries of the Frobenius norm, we will refer to any linear map on complex matrices that preserves the Frobenius norm as \emph{unitary}.

Another well-known result states that if $\Phi$ is a linear map that preserves the operator norm, then there exist unitary operators $U \in M_n$ and $V \in M_m$ such that either $\Phi(X) = UXV$ or $\Phi(X) = UX^TV$ \cite{K51,M59,S25}. This result was strengthened in \cite{GM77}, where it was shown that any isometry of a Ky Fan norm must have the same form. The problem of characterizing isometries of unitarily-invariant complex matrix norms was completely solved in \cite{LT90,S81}:
\begin{thm}\label{thm:unitary_norm_isometries}
	Let $\Phi : M_{n,m} \rightarrow M_{n,m}$ be a linear map and let $\|\cdot\|_{ui}$ be a unitarily-invariant norm that is not a multiple of the Frobenius norm. Then $\Phi$ is an isometry of $\|\cdot\|_{ui}$ if and only if there exist unitary matrices $U \in M_n$ and $V \in M_m$ such that either
	\begin{align*}
		\Phi(X) \equiv U X V \quad \text{ or } \quad n = m \text{ and } \Phi(X) \equiv U X^T V.
	\end{align*}
\end{thm}

We will prove a similar statement for a family of norms related to the Schmidt decomposition of pure states in Section~\ref{sec:sk_matrix_isometries}, though we will see that the form of $U$ and $V$ is slightly restricted in our setting. Our proof will rely on the following result of \cite{LT90,L94}, which is rephrased here slightly to suit our purposes. Note that a group $\cl{G} \subseteq M_{n}$ is said to be \emph{bounded} if there exists $K > 0$ such that $\big\|X\big\| \leq K$ for all $X \in \cl{G}$, and it is said to be \emph{irreducible} if ${\rm span}(\cl{G}) = M_n$.
\begin{prop}\label{prop:subgroup_lift}
	Let $\cl{G} \subseteq M_n$ be a bounded group that contains an irreducible subgroup of the unitary group $U(n) \subset M_n$. Then $\cl{G} \subseteq U(n)$.
\end{prop}
Additionally, Proposition~\ref{prop:subgroup_lift} will help us characterize the operators that preserve the geometric measure of entanglement in Theorem~\ref{thm:geometric_measure_entanglement}.

\subsection{Bipartite Separability Preservers}\label{sec:bip_sep_preservers}

In the design of quantum algorithms, a particularly important role is played by \emph{entangling gates} -- i.e., unitary operators that are capable of mapping a separable pure state into an entangled pure state \cite{DBE95,NC00}. Thus it is desirable to have a characterization of entangling gates that allows us to easily recognize whether or not a given unitary is entangling. In this section, we will phrase this problem slightly differently as characterizing the operators that map separable pure states to separable pure states. It is clear that any unitary $U \in M_m \otimes M_n$ of the form
\begin{align}\label{eq:localU1}
	U = U_1 \otimes U_2	\quad \text{ or } \quad n = m \text{ and } U = S(U_1 \otimes U_2),
\end{align}
where $U_1 \in M_m$ and $U_2 \in M_n$ are unitary, and $S \in M_n \otimes M_n$ is the swap operator introduced in Section~\ref{sec:symmetric_sub}, is an example of one such operator. In fact, it was shown in \cite{MM59,HPSSSL06} that \emph{all} bipartite operators that preserve the set of separable pure states are of this form, though both proofs are quite long and involved. The fact that all such operators have the form~\eqref{eq:localU1} also follows from several related results that have been proved more recently in \cite{AS10,FLPS11}.

 We now prove a stronger result -- that if an operator maps the set of pure states with Schmidt rank at most $k$ into itself for some $1 \leq k < \min\{m,n\}$ then it must be a unitary of the form~\eqref{eq:localU1}. Moreover, our proof is quick and elementary thanks to Proposition~\ref{prop:rankKpreserver}. In fact, it will be useful for us to consider the slightly more general nonsingular operators $L \in M_m \otimes M_n$ of the form
\begin{align}\label{eq:localL}
	L = P \otimes Q	\quad \text{ or } \quad n = m \text{ and } L = S(P \otimes Q),
\end{align}
where $P \in M_m$ and $Q \in M_n$ are nonsingular.

\begin{thm}\label{thm:bipartite_schmidt_preserver}
	Let $L \in M_m \otimes M_n$ be an invertible linear operator and define $\cl{V}_k$ to be the set of scalar multiples of pure states with Schmidt rank no larger than $k$:
	\begin{align*}
		\cl{V}_k := \big\{ c\ket{w} \in \bb{C}^m \otimes \bb{C}^n : c \in \mathbb{R}, SR(\ket{w}) \leq k \big\}.
	\end{align*}
	Then the following are equivalent:
	\begin{enumerate}[(a)]
		\item there exists some $1 \leq k < \min\{m,n\}$ such that $L\cl{V}_k \subseteq \cl{V}_k$;
		\item $L\cl{V}_k = \cl{V}_k$ for all $1 \leq k \leq \min\{m,n\}$; and
		\item $L$ is an operator of the form \eqref{eq:localL}.
	\end{enumerate}
	Furthermore, if $L$ sends pure states in $\cl{V}_k$ to pure states (i.e., it does not alter their norm) then $L$ is a unitary of the form~\eqref{eq:localU1}.
\end{thm}

\begin{proof}
	It is straightforward to see that $(c) \Rightarrow (b) \Rightarrow (a)$, so we only prove the implication $(a) \Rightarrow (c)$.

	To this end, recall that the vector-operator isomorphism associates a pure state $\ket{v} \in \bb{C}^m \otimes \bb{C}^n$ with an operator $A_v \in M_{n,m}$. The same isomorphism associates the operator $L \in M_m \otimes M_n$ with a superoperator $\Phi_L : M_{n,m} \rightarrow M_{n,m}$. Then condition $(a)$ is equivalent to the statement that there exists some $1 \leq k < \min\{m,n\}$ such that ${\rm rank}(\Phi_L(A_v)) \leq k$ whenever ${\rm rank}(A_v) \leq k$. Proposition~\ref{prop:rankKpreserver} then says that there exist nonsingular $P \in M_n$ and $Q \in M_m$ such that $\Phi_L$ is of the form~\eqref{eq:rankKform}.
	
	The given form of $\Phi_L$ says, again via the vector-operator isomorphism, that either
	\begin{align*}
		L = Q^T \otimes P \quad \text{ or } \quad n = m \text{ and } L = S (P \otimes Q^T),
	\end{align*}
	completing the $(a) \Rightarrow (c)$ implication. The final claim is trivial -- if $L$ preserves the length of separable pure states then $P$ and $Q$ must each be unitary, so $L$ must be a unitary of the form~\eqref{eq:localU1}.
\end{proof}

In order to demonstrate that the invertibility hypothesis of the above result is indeed required, consider the operator $L \in M_2 \otimes M_2$ defined by $L := \ketbra{1}{1} \otimes \ketbra{1}{1} + \ketbra{1}{2} \otimes \ketbra{1}{2}$. It is clear that the range of $L$ is ${\rm span}(\ket{1}\otimes\ket{1})$ so $L\ket{v}$ is always a multiple of a separable state. However, neither $L$ nor $SL$ can be written as an elementary tensor $P \otimes Q$, even if $P$ and $Q$ are allowed to be singular. However, if $L$ preserves the length of states with Schmidt rank no greater than $k$, then we can ignore the invertibility hypothesis of Theorem~\ref{thm:bipartite_schmidt_preserver} because Proposition~\ref{prop:sep_unitary} tells us that $L$ is necessarily unitary (and thus invertible).

\subsection{Multipartite Separability Preservers}\label{sec:mul_sep_preservers}

Thus far we have only considered operators that preserve separability and entanglement in bipartite quantum systems -- i.e., in $\bb{C}^m \otimes \bb{C}^n$, which is the tensor product of only two Hilbert spaces. We now consider the separability preserver problem in the case of multipartite quantum systems $\bb{C}^{n_1} \otimes \cdots \otimes \bb{C}^{n_p}$ where $p \geq 3$. In particular, we will show that exactly what might naively be expected to happen in this more general setting does indeed happen -- an operator $U$ sends separable pure states to separable pure states if and only if it is a composition of local unitaries and a swap operator. The difference is that in the bipartite case there were only two subsystems to swap so there were only two possible swap operators -- the identity operator and the operator $S$, which we simply referred to as \emph{the} swap operator in the bipartite setting. In the multipartite case, there are $p!$ different swap operators, each corresponding to a different permutation of the $p$ subsystems.

To make this result rigorous, we must first clarify our terminology in this setting. A pure state $\ket{v} \in \bb{C}^{n_1} \otimes \cdots \otimes \bb{C}^{n_p}$ is said to be separable if it can be written in the form $\ket{v} = \ket{v_1} \otimes \cdots \otimes \ket{v_p}$, where $\ket{v_i} \in \bb{C}^{n_i}$ for all $i$. Given a permutation $\sigma : \{1,\ldots,p\} \rightarrow \{1,\ldots,p\}$, recall that the swap operator $S_\sigma : \ket{v_1} \otimes \cdots \otimes \ket{v_p} \mapsto \ket{v_{\sigma(1)}} \otimes \cdots \otimes \ket{v_{\sigma(p)}}$ that permutes the $p$ subsystems according to $\sigma$. Additionally, we will define $\cl{V}$ to be the set of scalar multiples of \emph{multipartite} separable pure states:
	\begin{align*}
		\cl{V} := \big\{ c\ket{w} \in \bb{C}^{n_1} \otimes \cdots \otimes \bb{C}^{n_p} : c \in \mathbb{R}, \text{ $\ket{w}$ is separable } \big\}.
	\end{align*}
	We are now in a position to present the main result of this section. However, before proceeding we note that this result was also derived in \cite{W67}. We nonetheless present a full proof here for two reasons: our proof is significantly different than that provided in \cite{W67}, and several of our other results (Theorem~\ref{thm:geometric_measure_entanglement}, Proposition~\ref{prop:sepLPP}, and Theorem~\ref{thm:opNormIso}) rely on this result, so we prove it for completeness.

\begin{thm}\label{thm:multipartite}
	Let $L \in M_{n_1} \otimes \cdots \otimes M_{n_p}$ be an invertible linear operator. Then $L\cl{V} \subseteq \cl{V}$ if and only if there exist invertible operators $P_i \in M_{n_i}$ ($1 \leq i \leq p$) and a permutation $\sigma : \{1,\ldots,p\} \rightarrow \{1,\ldots,p\}$ such that $L = S_\sigma(P_1 \otimes \cdots \otimes P_p)$. Furthermore, if $L$ sends pure states in $\cl{V}$ to pure states (i.e., it does not alter their norm) then each $P_i$ is unitary.
\end{thm}

Before proving the theorem, we will need a few lemmas. Our first lemma says that the sum of two separable pure states is separable if and only if the states are either identical or differ on only one subsystem. When we say that two states $\ket{a} := \ket{a_1} \otimes \cdots \otimes \ket{a_p}$ and $\ket{b} := \ket{b_1} \otimes \cdots \otimes \ket{b_p}$ differ on the $i$-th subsystem, we mean that $\ket{a_i} \nparallel \ket{b_i}$ -- in other words, $\ket{a_i}$ and $\ket{b_i}$ are not related by a complex number with modulus one. Conversely, if we write $\ket{a_i} \parallel \ket{b_i}$ then we mean that $\ket{a_i}$ and $\ket{b_i}$ are linearly dependent, so they differ by a complex number with modulus one, and we say that $\ket{a}$ and $\ket{b}$ agree on the $i$-th subsystem.
\begin{lemma}\label{lem:multiHelper}
	Let $\ket{a} := \ket{a_1} \otimes \cdots \otimes \ket{a_p}, \ket{b} := \ket{b_1} \otimes \cdots \otimes \ket{b_p} \in \bb{C}^{n_1} \otimes \cdots \otimes \bb{C}^{n_p}$. Then $\ket{a} + \ket{b}$ is separable if and only if $\ket{a_i} \parallel \ket{b_i}$ for all indices $1 \leq i \leq p$ with the exception of at most one.
\end{lemma}
\begin{proof}
	The ``if'' implication of the lemma is trivial. We prove the ``only if'' implication by induction. The $p = 2$ case can be seen by using the (bipartite) vector-operator isomorphism of Section~\ref{sec:vector_operator_isomorphism} to associate $\ket{a_1} \otimes \ket{a_2}$ and $\ket{b_1} \otimes \ket{b_2}$ with $\ket{a_2}\overline{\bra{a_1}}$ and $\ket{b_2}\overline{\bra{b_1}}$, respectively. Then $\ket{a_1} \otimes \ket{a_2} + \ket{b_1} \otimes \ket{b_2}$ is separable if and only if $\ket{a_2}\overline{\bra{a_1}} + \ket{b_2}\overline{\bra{b_1}}$ is a rank one operator. This operator is rank one if and only if $\ket{a_1} \parallel \ket{b_1}$ or $\ket{a_2} \parallel \ket{b_2}$ (or both), which establishes the base case of $p = 2$.
	
	Now suppose that the claim is true for some particular $p \geq 2$ and consider linear functions $f_{i,\ket{v}} : \bb{C}^{n_1} \otimes \cdots \otimes \bb{C}^{n_{p+1}} \rightarrow \bb{C}^{n_1} \otimes \cdots \otimes \bb{C}^{n_{i-1}} \otimes \bb{C}^{n_{i+1}} \otimes \cdots \otimes \bb{C}^{n_{p+1}}$ defined on elementary tensors by
	\begin{align*}
		f_{i,\ket{v}}(\ket{a_1} \otimes \cdots \otimes \ket{a_{p+1}}) = \braket{v}{a_i} \ket{a_1} \otimes \cdots \otimes \ket{a_{i-1}} \otimes \ket{a_{i+1}} \otimes \cdots \otimes \ket{a_{p+1}}.
	\end{align*}
	Clearly $f_{i,\ket{v}}(\ket{a})$ is always a multiple of a separable state whenever $\ket{a}$ is separable. Now pick two vectors $\ket{a} := \ket{a_1} \otimes \cdots \otimes \ket{a_{p+1}}$, $\ket{b} := \ket{b_1} \otimes \cdots \otimes \ket{b_{p+1}}$ such that $\ket{a_i} \nparallel \ket{b_i}$ and $\ket{a_j} \nparallel \ket{b_j}$ for some $i \neq j$. Pick another index $k \neq i,j$ and choose a state $\ket{v} \in \bb{C}^{n_k}$ such that $\braket{v}{a_k},\braket{v}{b_k} \neq 0$. Then $f_{k,\ket{v}}(\ket{a})$ and $f_{k,\ket{v}}(\ket{b})$ are nonzero multiples of separable pure states living on a tensor product of $p$ Hilbert spaces that differ on the $i$-th and $j$-th subsystems. By the inductive hypothesis, $f_{k,\ket{v}}(\ket{a} + \ket{b})$ is not separable, so neither is $\ket{a} + \ket{b}$.
\end{proof}

\begin{lemma}\label{lem:multiHelperCor}
	Let $L \in M_{n_1} \otimes \cdots \otimes M_{n_p}$ be an invertible operator such that $L\cl{V} \subseteq \cl{V}$. Let $1 \leq r \leq p$ and $\ket{a} := \ket{a_1} \otimes \cdots \otimes \ket{a_p}, \ket{b} := \ket{b_1} \otimes \cdots \otimes \ket{b_p} \in \bb{C}^{n_1} \otimes \cdots \otimes \bb{C}^{n_p}$ be such that there are exactly $r$ indices $h$ with $\ket{a_h} \nparallel \ket{b_h}$. If we write
	\begin{align*}
		L\ket{a} = c_a\ket{a_1^\prime} \otimes \cdots \otimes \ket{a_p^\prime} \ \ \ \ \text{ and } \ \ \ \ L\ket{b} & = c_b\ket{b_1^\prime} \otimes \cdots \otimes \ket{b_p^\prime},
	\end{align*}
	then there are at most $r$ indices $h^\prime$ with $\ket{a_{h^\prime}^\prime} \nparallel \ket{b_{h^\prime}^\prime}$.
\end{lemma}
\begin{proof}
	Suppose without loss of generality that $\ket{a}$ and $\ket{b}$ differ on the first $r$ subsystems. Then define
	\begin{align*}
		\ket{a^{(i)}} := \ket{b_1} \otimes \cdots \otimes \ket{b_{i}} \otimes \ket{a_{i+1}} \otimes \cdots \otimes \ket{a_p} \quad \text{for $0 \leq i \leq r$},
	\end{align*}
	with the understanding that $\ket{a^{(0)}} = \ket{a}$ and $\ket{a^{(r)}} = \ket{b}$. It is clear that, for any $1 \leq i \leq r$, $\ket{a^{(i-1)}} + \ket{a^{(i)}}$ is separable, so $L(\ket{a^{(i-1)}} + \ket{a^{(i)}})$ is separable as well, which implies via Lemma~\ref{lem:multiHelper} that $L\ket{a^{(i-1)}}$ and $L\ket{a^{(i)}}$ differ on at most one subsystem. It follows that $L\ket{a}$ and $L\ket{b}$ differ on at most $r$ subsystems.
\end{proof}

Our final lemma shows that if $L\cl{V} \subseteq \cl{V}$ and three separable states $\ket{v},\ket{x},\ket{y}$ are such that $\ket{x}$ and $\ket{y}$ each differ from $\ket{v}$ on a single subsystem, then $L\ket{x}$ and $L\ket{y}$ differ from $L\ket{v}$ on a single subsystem as well. Furthermore, $\ket{x}$ and $\ket{y}$ differ from $\ket{v}$ on the same subsystem if and only if $L\ket{x}$ and $L\ket{y}$ differ from $L\ket{v}$ on the same subsystem.
\begin{lemma}\label{lem:multiHelperB}
	Let $L \in M_{n_1} \otimes \cdots \otimes M_{n_p}$ be an invertible operator such that $L\cl{V} \subseteq \cl{V}$. Let $1 \leq i,j \leq p$ and
	\begin{align*}
		\ket{v} & := \ket{v_1} \otimes \cdots \otimes \ket{v_p}, \\
		\ket{x} & := \ket{v_1} \otimes \cdots \otimes \ket{v_{i-1}} \otimes \ket{\tilde{x}} \otimes \ket{v_{i+1}} \otimes \cdots \otimes \ket{v_p} \quad (\ket{\tilde{x}} \nparallel \ket{v_i}), \\
		\ket{y} & := \ket{v_1} \otimes \cdots \otimes \ket{v_{j-1}} \otimes \ket{\tilde{y}} \otimes \ket{v_{j+1}} \otimes \cdots \otimes \ket{v_p} \quad (\ket{\tilde{y}} \nparallel \ket{v_j}).
	\end{align*}
	Write
	\begin{align*}
		L\ket{v} & = c_v\ket{v_1^\prime} \otimes \cdots \otimes \ket{v_p^\prime} \text{ for some } c_v \in \bb{C} \text{ and } \ket{v_h^\prime} \in \bb{C}^{n_h} \ (1 \leq h \leq p).
	\end{align*}
	Then there exist $1 \leq k,\ell \leq p$, $c_x,c_y \in \bb{C}$, and $\ket{\tilde{x}^\prime} \in \bb{C}^{n_k}, \ket{\tilde{y}^\prime} \in \bb{C}^{n_\ell}$ such that
	\begin{align*}
		L\ket{x} & = c_x\ket{v_1^\prime} \otimes \cdots \otimes \ket{v_{k-1}^\prime} \otimes \ket{\tilde{x}^\prime} \otimes \ket{v_{k+1}^\prime} \otimes \cdots \otimes \ket{v_p^\prime}, \text{ and}\\
		L\ket{y} & = c_y\ket{v_1^\prime} \otimes \cdots \otimes \ket{v_{\ell-1}^\prime} \otimes \ket{\tilde{y}^\prime} \otimes \ket{v_{\ell+1}^\prime} \otimes \cdots \otimes \ket{v_p^\prime}.
	\end{align*}
	Furthermore, $k = \ell$ if and only if $i = j$.
\end{lemma}
\begin{proof}
	It is clear that $\ket{v} + \ket{x}$ is separable, so $L(\ket{v} + \ket{x})$ is separable as well. It follows from Lemma~\ref{lem:multiHelper} that $L\ket{v}$ and $L\ket{x}$ differ on a single subsystem, which allows us to write $L\ket{x}$ in the form described by the lemma. The fact that $L\ket{y}$ can be written in the desired form is proved analogously. All that remains to be proved is the final claim that $k = \ell$ if and only if $i = j$.
	
	First suppose that $i = j$. Notice that $\ket{x} + \ket{y}$ is separable in this case, so $L(\ket{x} + \ket{y})$ must be separable as well. The fact that $k = \ell$ then follows from Lemma~\ref{lem:multiHelper}. Now suppose that $i \neq j$ (without loss of generality, suppose that $i < j$) -- we will prove that $k \neq \ell$ by contradiction. To this end, assume that $k = \ell$. Consider an arbitrary separable state $\ket{w}$ and three related states $\ket{w^{(i)}}, \ket{w^{(j)}}$ and $\ket{w^{(i,j)}}$, defined as follows:
	\begin{align*}
		\ket{w} & := \ket{w_1} \otimes \cdots \otimes \ket{w_p},\\
		\ket{w^{(i)}} & := \ket{v_1} \otimes \cdots \otimes \ket{v_{i-1}} \otimes \ket{w_i} \otimes \ket{v_{i+1}} \otimes \cdots \otimes \ket{v_p},\\
		\ket{w^{(j)}} & := \ket{v_1} \otimes \cdots \otimes \ket{v_{j-1}} \otimes \ket{w_j} \otimes \ket{v_{j+1}} \otimes \cdots \otimes \ket{v_p},\\
		\ket{w^{(i,j)}} & := \ket{v_1} \otimes \cdots \otimes \ket{v_{i-1}} \otimes \ket{w_i} \otimes \ket{v_{i+1}} \otimes \cdots \otimes \ket{v_{j-1}} \otimes \ket{w_j} \otimes \ket{v_{j+1}} \otimes \cdots \otimes \ket{v_p}.
	\end{align*}
	Our goal is to show that $L\ket{w}$ is contained within a fixed nontrivial subspace of $\bb{C}^{n_1} \otimes \cdots \otimes \bb{C}^{n_p}$. Because $\ket{w}$ is an arbitrary separable state, and separable states span all of $\bb{C}^{n_1} \otimes \cdots \otimes \bb{C}^{n_p}$, this contradicts the fact that $L$ is invertible and will establish the lemma.
	
	Because $\ket{v} + \ket{w^{(i)}}$ and $\ket{x} + \ket{w^{(i)}}$ are separable, $L(\ket{v} + \ket{w^{(i)}})$ and $L(\ket{x} + \ket{w^{(i)}})$ are separable as well, and so by Lemma~\ref{lem:multiHelper} we have that $L\ket{v}$ and $L\ket{w^{(i)}}$ differ on a single subsystem, and similarly that $L\ket{x}$ and $L\ket{w^{(i)}}$ differ on a single subsystem. By invertibility of $L$ we know $\ket{v_k^\prime} \nparallel \ket{\tilde{x}^\prime}$ so it must be the case that $L\ket{v}$ and $L\ket{w^{(i)}}$ differ on the $k$-th subsystem. A similar argument shows that $L\ket{v}$ and $L\ket{w^{(j)}}$ differ on the $k$-th subsystem as well. Thus we can write
	\begin{align*}
		L\ket{w^{(i)}} & = c_{w^{(i)}}\ket{v_1^\prime} \otimes \cdots \otimes \ket{v_{k-1}^\prime} \otimes \ket{w_i^\prime} \otimes \ket{v_{k+1}^\prime} \otimes \cdots \otimes \ket{v_p^\prime}, \text{ and}\\
		L\ket{w^{(j)}} & = c_{w^{(j)}}\ket{v_1^\prime} \otimes \cdots \otimes \ket{v_{k-1}^\prime} \otimes \ket{w_j^\prime} \otimes \ket{v_{k+1}^\prime} \otimes \cdots \otimes \ket{v_p^\prime}.
	\end{align*}
	
	Similarly, $\ket{w^{(i)}} + \ket{w^{(i,j)}}$ and $\ket{w^{(j)}} + \ket{w^{(i,j)}}$ are separable so Lemma~\ref{lem:multiHelper} tells us that $L\ket{w^{(i)}}$ and $L\ket{w^{(i,j)}}$ differ on a single subsystem and that $L\ket{w^{(j)}}$ and $L\ket{w^{(i,j)}}$ differ on a single subsystem. Once again, this is only possible if $L\ket{w^{(i)}}$ and $L\ket{w^{(i,j)}}$ differ on the $k$-th subsystem. Thus, there exists some $\ket{\tilde{w}^\prime} \in \bb{C}^{n_k}$ such that we can write
	\begin{align*}
		L\ket{w^{(i,j)}} = c_{w^{(i,j)}}\ket{v_1^\prime} \otimes \cdots \otimes \ket{v_{k-1}^\prime} \otimes \ket{\tilde{w}^\prime} \otimes \ket{v_{k+1}^\prime} \otimes \cdots \otimes \ket{v_p^\prime}.
	\end{align*}
	
	Now observe that $\ket{w}$ and $\ket{w^{(i,j)}}$ differ on at most $p-2$ subsystems, so Lemma~\ref{lem:multiHelperCor} tells us that $L\ket{w}$ and $L\ket{w^{(i,j)}}$ differ on at most $p-2$ subsystems as well. It follows that $L\ket{w}$ is contained within the set
	\begin{align*}
		\cl{T} := \big\{ c\ket{z_1} \otimes \cdots \otimes \ket{z_p} \in \bb{C}^{n_1} \otimes \cdots \otimes \bb{C}^{n_p} \ | \ \exists \, h \text{ with } \ket{z_h} = \ket{v_h^\prime} \big\}.
	\end{align*}
	Because $\ket{w}$ is an arbitrary separable state and separable states span all of $\bb{C}^{n_1} \otimes \cdots \otimes \bb{C}^{n_p}$, it follows that the range of $L$ is contained in the span of $\cl{T}$. Now let $\ket{z_h^\prime} \in \bb{C}^{n_h}$ for $1 \leq h \leq p$ be such that $\braket{z_h^\prime}{v_h^\prime} = 0$. Then clearly $(\bra{z_1^\prime} \otimes \cdots \otimes \bra{z_p^\prime})\ket{z} = 0$ for all $\ket{z} \in \cl{T}$, so $\cl{T}$ spans a strict subspace of $\bb{C}^{n_1} \otimes \cdots \otimes \bb{C}^{n_p}$. This contradicts invertibility of $L$ and completes the proof.
\end{proof}

\begin{proof}[Proof of Theorem~\ref{thm:multipartite}]
	As in the bipartite case, the ``if'' implication is trivial. For the ``only if'' implication, we prove the following claim:
	\begin{claim}
		Let $L \in M_{n_1} \otimes \cdots \otimes M_{n_p}$ be an invertible operator such that $L\cl{V} \subseteq \cl{V}$. Fix $1 \leq i \leq p$ and vectors $\ket{v_h} \in \bb{C}^{n_h}$ ($i < h \leq p$). Then there exist a permutation $S_{\sigma}$, operators $P_h$ ($1 \leq h \leq i$), and vectors $\ket{w_h} \in \bb{C}^{n_h}$ ($i < h \leq p$) such that
		\begin{align*}
			S_{\sigma}L(\ket{v_1} \otimes \cdots \otimes \ket{v_p}) = P_1\ket{v_1} \otimes \cdots \otimes P_i\ket{v_i} \otimes \ket{w_{i+1}} \otimes \cdots \otimes \ket{w_{p}} \ & \\
			\forall \, \ket{v_h} \in \bb{C}^{n_h} (1 \leq h \leq i). &
		\end{align*}
	\end{claim}
	If we can prove the above claim then we are done, because if $i = p$ then we can use the fact that there exists a separable basis of $\bb{C}^{n_1} \otimes \cdots \otimes \bb{C}^{n_p}$ to conclude that $S_{\sigma}L = P_1 \otimes \cdots \otimes P_p$. Invertibility of each $P_j$ then follows from invertibility of $L$, and Theorem~\ref{thm:multipartite} is proved.
	
	To prove the claim, we proceed by induction on $i$. For the base case, assume $i = 1$ and fix vectors $\ket{v_h} \in \bb{C}^{n_h}$ ($2 \leq h \leq p$). Consider the $n_1$ vectors $\ket{v^{(j)}}$ ($1 \leq j \leq n_1$) defined by
	\begin{align*}
		\ket{v^{(j)}} := \ket{j} \otimes \ket{v_2} \otimes \cdots \otimes \ket{v_p}.
	\end{align*}
	By Lemma~\ref{lem:multiHelperB} we know that there exists $1 \leq k \leq p$, $\ket{w_h} \in \bb{C}^{n_h}$ ($h \neq k$), and $c_j \in \bb{R}$, $\ket{\tilde{v}^{(j)}} \in \bb{C}^{n_k}$ ($1 \leq j \leq n_i$) such that
	\begin{align*}
		L\ket{v^{(j)}} = c_j\ket{w_1} \otimes \cdots \otimes \ket{w_{k-1}} \otimes \ket{\tilde{v}^{(j)}} \otimes \ket{w_{k+1}} \otimes \cdots \otimes \ket{w_p} \quad \forall \, 1 \leq j \leq n_1.
	\end{align*}
	If $\sigma : \{1,2,\ldots,p\} \rightarrow \{1,2,\ldots,p\}$ is the permutation that swaps $1$ and $k$ then it follows by linearity of $L$ that there exists an operator $P_1 \in M_{n_1}$ such that
	\begin{align*}
		& S_{\sigma}L(\ket{v_1} \otimes \cdots \otimes \ket{v_p}) \\
		& \quad \quad \quad = P_1\ket{v_1} \otimes \ket{w_2} \otimes \cdots \otimes \ket{w_{k-1}} \otimes \ket{w_1} \otimes \ket{w_{k+1}} \otimes \cdots \otimes \ket{w_p} \quad \forall \, \ket{v_1} \in \bb{C}^{n_1}.
	\end{align*}
	We have thus proved the base case $i = 1$ of the claim.
	
	We now proceed to the inductive step. Assume that the claim holds for some specific value of $i$ and fix vectors $\ket{v_h} \in \bb{C}^{n_h}$ ($i+2 \leq h \leq p$). By the inductive hypothesis, there exist a permutation $S_{\sigma}$, operators $P_{h}$ ($1 \leq h \leq i$), and vectors $\ket{w_{h}} \in \bb{C}^{n_h}$ ($i+1 \leq h \leq p$) such that
	\begin{align}\begin{split}\label{eq:inducHyp}
		& S_{\sigma}L(\ket{z_1} \otimes \cdots \otimes \ket{z_i} \otimes \ket{1} \otimes \ket{v_{i+2}} \otimes \cdots \otimes \ket{v_p}) \\
		& \quad \quad \quad = P_{1}\ket{z_1} \otimes \cdots \otimes P_{i}\ket{z_i} \otimes \ket{w_{i+1}} \otimes \cdots \otimes \ket{w_{p}} \ \ \forall \, \ket{z_h} \in \bb{C}^{n_h} (1 \leq h \leq i).
	\end{split}\end{align}
	Fix $\ket{v_h},\ket{x_h} \in \bb{C}^{n_h}$ ($1 \leq h \leq i$) and consider the $2n_{i+1}$ vectors $\ket{v^{(j)}}$ and $\ket{x^{(j)}}$ ($1 \leq j \leq n_{i+1}$) defined by
	\begin{align*}
		\ket{v^{(j)}} & := \ket{v_1} \otimes \cdots \otimes \ket{v_{i}} \otimes \ket{j} \otimes \ket{v_{i+2}} \otimes \cdots \otimes \ket{v_p} \ \text{ and}\\
		\ket{x^{(j)}} & := \ket{x_1} \otimes \cdots \otimes \ket{x_{i}} \otimes \ket{j} \otimes \ket{v_{i+2}} \otimes \cdots \otimes \ket{v_p}.
	\end{align*}
	By Lemma~\ref{lem:multiHelperB} we know that there exists $1 \leq k \leq p$, independent of $j$, such that $S_{\sigma}L\ket{v^{(1)}}$ and $S_{\sigma}L\ket{v^{(j)}}$ differ only on the $k$-th subsystem. If $k \leq i$ then we can create a vector $\ket{v^{(1)\prime}}$ that differs from $\ket{v^{(1)}}$ only on the $k$-th subsystem and observe that $S_{\sigma}L\ket{v^{(1)}}$ and $S_{\sigma}L\ket{v^{(1)\prime}}$ differ on the $k$-th subsystem as well. This contradicts Lemma~\ref{lem:multiHelperB}, so we see that $k \geq i + 1$.
	
	Similarly, there exists $i + 1 \leq \ell \leq p$, independent of $j$, such that $S_{\sigma}L\ket{x^{(1)}}$ and $S_{\sigma}L\ket{x^{(j)}}$ differ only on the $\ell$-th subsystem. Now suppose there are $r$ indices $h$ such that $\ket{v_h} \neq \ket{x_h}$ ($1 \leq h \leq i$). Then for any $j$, $\ket{v^{(j)}}$ and $\ket{x^{(j)}}$ differ on $r$ of the first $i$ subsystems, so in particular $S_{\sigma}L\ket{v^{(1)}}$ and $S_{\sigma}L\ket{x^{(1)}}$ differ on $r$ of the first $i$ subsystems as well, by Equation~\eqref{eq:inducHyp}. However, $S_{\sigma}L\ket{v^{(1)}}$ and $S_{\sigma}L\ket{v^{(j)}}$ differ on the $k$-th subsystem and $S_{\sigma}L\ket{x^{(1)}}$ and $S_{\sigma}L\ket{x^{(j)}}$ differ on the $\ell$-th subsystem, $S_{\sigma}L\ket{v^{(j)}}$ and $S_{\sigma}L\ket{x^{(j)}}$ differ on $r+2$ subsystems if $k \neq \ell$, which contradicts Lemma~\ref{lem:multiHelperCor}. It follows that $k = \ell$ and furthermore that $S_{\sigma}L\ket{v^{(j)}}$ and $S_{\sigma}L\ket{x^{(j)}}$ agree on the $k$-th subsystem for all $j$.
	
	Now let $\tau : \{1,2,\ldots,p\} \rightarrow \{1,2,\ldots,p\}$ be the permutation that swaps $i+1$ and $k$. By using the fact that $\ket{v_h},\ket{x_h} \in \bb{C}^{n_h}$ ($1 \leq h \leq i$) were chosen arbitrarily, it follows that there exists $P_{i+1}$ (independent of $j$) such that
	\begin{align*}
		& S_{\tau}S_{\sigma}L(\ket{z_1} \otimes \cdots \otimes \ket{z_i} \otimes \ket{j} \otimes \ket{v_{i+2}} \otimes \cdots \otimes \ket{v_p}) \\
		= \ & P_{1}\ket{z_1} \otimes \cdots \otimes P_{i}\ket{z_i} \otimes P_{i+1}\ket{j} \otimes \ket{w_{i+2}} \otimes \cdots \otimes \ket{w_{k-1}} \otimes \ket{w_{i+1}} \otimes \ket{w_{k+1}} \otimes \cdots \otimes \ket{w_{p}} \\
		& \quad \quad \quad \quad \quad \quad \quad \quad \quad \quad \quad \quad \quad \quad \quad \quad \quad \quad \ \ \forall \, \ket{z_h} \in \bb{C}^{n_h} (1 \leq h \leq i), \ \forall \, 1 \leq j \leq n_{i+1}.
	\end{align*}
	The inductive step is completed by noting that $\{\ket{j}\}$ forms a basis for $\bb{C}^{n_{i+1}}$ and using linearity. The desired form of $L$ follows. The claim about each $P_i$ being unitary if $L$ does not alter the norm of pure states is trivial.
\end{proof}

\subsection{Geometric Measure of Entanglement}\label{sec:geom_meas_ent}

One useful application of Theorem~\ref{thm:multipartite} is we can now characterize all operators that preserve the \emph{geometric measure of entanglement} \cite{BL01,Shi95,WG03}, which is defined for pure states $\ket{v} \in \mathbb{C}^{n_1} \otimes \cdots \otimes \mathbb{C}^{n_p}$ in terms of the maximum overlap of $\ket{v}$ with a separable state:
\begin{align*}
	E(\ket{v}) := 1 - \sup_{\ket{w_i} \in \mathbb{C}^{n_i}}\Big\{ \big|(\bra{w_1} \otimes \cdots \otimes \bra{w_p})\ket{v}\big|^2 \Big\}.
\end{align*}
The geometric measure of entanglement is an entanglement monotone \cite{Vid00} and thus can be thought of as a measurement of ``how entangled'' a given pure state is. It has been shown to be related to the relative entropy of entanglement \cite{VP98,WEGM04}, the generalized robustness of entanglement \cite{Cav06,HN03,Ste03}, and quantum state estimation \cite{CZW11,PR04}. It thus plays a very important role in entanglement theory, especially in the multipartite setting where separability and entanglement monotones are still not well understood.

We will now use Theorem~\ref{thm:multipartite} to describe the operators that preserve the geometric measure of entanglement. Instead of dealing with $E$ itself, it will be useful to consider the quantity $G(\ket{v}) := \sqrt{1 - E(\ket{v})}$ and simply note that $E(L\ket{v}) = E(\ket{v})$ for all $\ket{v}$ if and only if $G(L\ket{v}) = G(\ket{v})$ for all $\ket{v}$. It is not difficult to see that $G$ is a norm so the group of operators
\begin{align*}
	\cl{G} := \big\{ L \in M_{n_1} \otimes \cdots \otimes M_{n_p} : G(L\ket{v}) = G(\ket{v}) \text{ for all } \ket{v} \in \bb{C}^{n_1} \otimes \cdots \otimes \bb{C}^{n_p} \big\}
\end{align*}
is bounded. Furthermore, it is clear that the group
\begin{align*}
	\cl{G}_S := \big\{ U_1 \otimes \cdots \otimes U_p \in M_{n_1} \otimes \cdots \otimes M_{n_p} : U_i \in M_{n_i} \text{ is unitary for all } i \big\}
\end{align*}
is a subgroup of both $\cl{G}$ and the unitary group. To see that $\cl{G}_S$ is irreducible, recall that the unitary group $U(n)$ spans all of $M_n$. Thus the span of $\cl{G}_S$ contains all operators of the form
\begin{align*}
	X_1 \otimes \cdots \otimes X_p \in M_{n_1} \otimes \cdots \otimes M_{n_p} : X_i \in M_{n_i} \text{ is not necessarily unitary.}
\end{align*}
Because operators of this form span all of $M_{n_1} \otimes \cdots \otimes M_{n_p}$, it follows that $\cl{G}_S$ spans all of $M_{n_1} \otimes \cdots \otimes M_{n_p}$ and is thus irreducible. By Proposition~\ref{prop:subgroup_lift} it follows that $\cl{G}$ is contained in the unitary group, so if $E(L\ket{v}) = E(\ket{v})$ for all $\ket{v}$ then $L$ must be unitary.

Now using the fact that $E(\ket{v}) = 0$ if and only if $\ket{v}$ is separable, we see that $L\ket{v}$ must be separable whenever $\ket{v}$ is separable for any $L \in \cl{G}$. By invoking Theorem~\ref{thm:multipartite} we have proved the following:
\begin{thm}\label{thm:geometric_measure_entanglement}
	Let $U \in M_{n_1} \otimes \cdots \otimes M_{n_p}$ be a linear operator. Then $E(U\ket{v}) = E(\ket{v})$ for all $\ket{v}$ if and only if there exist unitaries $U_i \in M_{n_i}$ ($1 \leq i \leq p$) and a swap operator $S_{\sigma} : \ket{v_1} \otimes \cdots \otimes \ket{v_p} \mapsto \ket{v_{\sigma(1)}} \otimes \cdots \otimes \ket{v_{\sigma(p)}}$ such that $U = S_\sigma(U_1 \otimes \cdots \otimes U_p)$.
\end{thm}

\chapter{Norms Arising from Schmidt Rank}\label{chap:Sk_norms}

In this chapter we develop families of vector and operator norms that arise naturally from the Schmidt rank of pure states. We characterize these norms in a variety of different ways, and derive various basic results such as characterizations of their isometry groups. We characterize the dual of the vector norms, but we defer a characterization of the dual of the operator norms until Section~\ref{sec:dual_op_norm}, since we do not yet have the necessary mathematical tools to derive this result.

We spend a significant amount of time producing various bounds on the norms introduced in this chapter. Our reason for deriving these bounds is that Chapter~\ref{ch:computation} deals primarily with applications of these norms, so the inequalities derived here apply immediately to these various problems. We return to more theoretical aspects of these norms in Chapter~\ref{ch:optheory}, where we show that they arise from a natural family of operator spaces.

\section{The s(k)-Vector Norm}\label{sec:sk_vector_norm}

Given a pure quantum state $\ket{v} \in \bb{C}^m \otimes \bb{C}^n$, it is natural to ask for a measure of how close $\ket{v}$ is to being separable. The notion of Schmidt rank plays a role in answering that question, but in some ways seems insufficient because it misses (for example) the fact that if $\varepsilon > 0$ is small then a Schmidt rank-$3$ state with Schmidt coefficients $\varepsilon$, $\varepsilon$, and $\sqrt{1 - 2\varepsilon^2}$ is in some sense ``closer'' to being separable than a Schmidt rank-$2$ state with Schmidt coefficients $1/\sqrt{2}$ and $1/\sqrt{2}$. The norms introduced in this section, which we refer to as ``$s(k)$-vector norms'', can be seen as filling in this gap and providing a measure of how close a pure state is to having Schmidt rank $k$.
\begin{defn}\label{defn:sk_norm_vec}
  Let $\ket{v} \in \bb{C}^m \otimes \bb{C}^n$ and let $1 \leq k \leq \min\{m,n\}$. Then we define the {\em $s(k)$-vector norm} of $\ket{v}$, denoted $\big\| \ket{v} \big\|_{s(k)}$, by
  \begin{align*}
    \big\| \ket{v} \big\|_{s(k)} & := \sup_{\ket{w}} \Big\{ \big| \braket{w}{v} \big| : SR(\ket{w}) \leq k \Big\}.
  \end{align*}
\end{defn}

Note that even though Definition~\ref{defn:sk_norm_vec} is only stated for unit vectors $\ket{v}$, it extends in the obvious way to a norm on all of $\bb{C}^m \otimes \bb{C}^n$. Also note that these are indeed norms; positive homogeneity and the triangle inequality follow immediately from the corresponding properties of the complex modulus and supremum. The fact that $\big\|\ket{v}\big\|_{s(k)} = 0$ if and only if $\ket{v} = 0$ can be seen by noting that the set of separable pure states (and thus the set of states with Schmidt rank no larger than $k$) spans all of $\bb{C}^m \otimes \bb{C}^n$.

As a brief note on notation, we use the subscript $s(k)$ to differentiate these norms from the usual vector $k$-norm, which is traditionally denoted simply by the subscript $k$. The lowercase ``s'' simply refers to Schmidt (we use an uppercase ``S'' for an operator version of these norms later in this chapter).

The $s(k)$-vector norms have been considered in \cite{CK09,CKS09,JK10} as a tool for detecting $k$-block positivity of operators. We return to this topic in Section~\ref{sec:op_sk_norm_block_pos}, but for now we focus on more fundamental mathematical properties of these norms. Also note that in the $k = 1$ case, the $s(1)$-norm is the well-known \emph{smallest reasonable crossnorm} or \emph{injective crossnorm} \cite{Gro53} (see also \cite[Chapter~1]{DGFS08}). We discuss the similarities between the $s(k)$-norms and reasonable crossnorms in Section~\ref{sec:sk_vector_dual_norm}.

\subsection{Basic Properties}\label{sec:sk_vector_norm_prop}

The case of $k = \min\{m,n\}$ of the $s(k)$-norms is very familiar -- $\big\| \ket{v} \big\|_{s(\min\{m,n\})}$ is just the standard Euclidean norm of $\ket{v}$. Similarly, the norm $\big\|\ket{v}\big\|_{s(1)}$ is exactly the quantity $G(\ket{v})$ that was defined in terms of the geometric measure of entanglement in Section~\ref{sec:geom_meas_ent} in the bipartite setting. It is clear from the definition that $\big\| \ket{v} \big\|_{s(k)} \leq \big\| \ket{v} \big\|$ for all $k$, and moreover that we have an increasing family of norms leading up to the Euclidean norm:
\begin{align*}
  \big\| \ket{v} \big\|_{s(1)} \leq \big\| \ket{v} \big\|_{s(2)} \leq \cdots \leq \big\| \ket{v} \big\|_{s(\min\{m,n\}-1)} \leq \big\| \ket{v} \big\|.
\end{align*}

The first result of this section shows that the $s(k)$-vector norm is easy to calculate.

\begin{thm}\label{thm:sk_vector_norm}
  Suppose $\ket{v} \in \bb{C}^m \otimes \bb{C}^n$ has Schmidt coefficients $\alpha_1 \geq \alpha_2 \geq \cdots \geq 0$. Then
    \[
      \big\| \ket{v} \big\|_{s(k)} = \sqrt{\sum_{i=1}^{k}\alpha_i^2}.
    \]
\end{thm}
\begin{proof}
  Assume without loss of generality that $n \geq m$. To see that $\big\| \ket{v} \big\|_{s(k)} \geq \sqrt{\sum_{i=1}^{k}\alpha_i^2}$, use the Schmidt Decomposition to write $\ket{v} = \sum_{i=1}^m \alpha_i \ket{u_i} \otimes \ket{v_i}$. Now let
  \[
    \ket{w} = \frac{\sum_{i=1}^k \alpha_i \ket{u_i} \otimes \ket{v_i}}{\sqrt{\sum_{i=1}^k \alpha_i^2}}.
  \]

  \noindent Observe that $SR(\ket{w}) = k$. Some algebra then reveals that
  \begin{align*}
    \braket{w}{v} & = \frac{1}{\sqrt{\sum_{i=1}^k \alpha_i^2}} \Big( \sum_{i=1}^{m}\alpha_i\bra{u_i} \otimes \bra{v_i} \Big) \Big(\sum_{i=1}^k \alpha_i \ket{u_i} \otimes \ket{v_i}\Big) \\
    & = \frac{1}{\sqrt{\sum_{i=1}^k \alpha_i^2}} \sum_{i=1}^{m}\sum_{j=1}^k \alpha_i \alpha_j \braket{u_i}{u_j} \otimes \braket{v_i}{v_j} \\
    & = \frac{1}{\sqrt{\sum_{i=1}^k \alpha_i^2}} \sum_{j=1}^k \alpha_j^2 \\
    & = \sqrt{\sum_{i=1}^k \alpha_i^2}.
  \end{align*}

  \noindent To see the opposite inequality, let $\ket{w} \in \bb{C}^m \otimes \bb{C}^n$ have $SR(\ket{w}) \leq k$ and Schmidt Decomposition $\ket{w} = \sum_{i=1}^k \beta_i \ket{w_i} \otimes \ket{x_i}$ (where we take some of the $\beta_i$'s to be zero if $SR(\ket{w}) < k$). Then
  \begin{align*}
    \big| \braket{w}{v} \big| & = \left| \Big(\sum_{i=1}^k \beta_i \bra{w_i} \otimes \bra{x_i}\Big) \Big( \sum_{i=1}^{m}\alpha_i\ket{u_i} \otimes \ket{v_i} \Big) \right| \\
    & \leq \sum_{i=1}^{m}\sum_{j=1}^k\alpha_i\beta_j\big|\braket{w_j}{u_i}\braket{x_j}{v_i}\big| \\
    & = {\bf \alpha}^{T}D{\bf \beta}
  \end{align*}

\noindent where ${\bf \alpha}^{T} = (\alpha_1,\ldots,\alpha_m)$ and ${\bf \beta}^{T} = (\beta_1,\ldots,\beta_k,0,\ldots,0)$ are vectors of Schmidt coefficients, and $D$ is the $m \times m$ matrix given by $D_{ij} := \big|\braket{w_j}{u_i}\braket{x_j}{v_i}\big|$, where we have extended $\{\ket{w_j}\}$ and $\{\ket{x_j}\}$ to orthonormal bases of their respective spaces. Notice that, by the Cauchy--Schwarz inequality,
\begin{align*}
	\sum_{j=1}^m D_{ij} & = \sum_{j=1}^m \big|\braket{w_j}{u_i}\braket{x_j}{v_i}\big| \\
	& \leq \sqrt{\left(\sum_{j=1}^m \big|\braket{w_j}{u_i}\big|^2\right)\left(\sum_{j=1}^m \big|\braket{x_j}{v_i}\big|^2\right)} \\
	& \leq \sqrt{\big\|\ket{u_i}\big\|^2 \big\|\ket{v_i}\big\|^2} \\
	& = 1.
\end{align*}
In other words, the row sums of $D$ are no greater than $1$. A similar argument shows that the column sums of $D$ are no greater than $1$, so $D$ is doubly-sub-stochastic. The Hardy-Littlewood-Polya Theorem then tells us that the vector ${\bf \gamma}^{T} := (D{\bf \beta})^{T} = (\gamma_1,\ldots,\gamma_m)$ satisfies
\begin{align*}
    \sum_{i=1}^j \gamma_i \leq \sum_{i=1}^j \beta_i \quad \forall \, 1 \leq j \leq m.
\end{align*}

\noindent Because $\alpha_1 \geq \alpha_2 \geq \cdots \geq 0$, this tells us that ${\bf \alpha}^{T}D{\bf \beta} \leq {\bf \alpha}^{T}{\bf \beta}$. The Cauchy--Schwarz inequality then implies that
    \begin{align*}
    \big| \braket{w}{v} \big| = {\bf \alpha}^{T}D{\bf \beta} \leq {\bf \alpha}^{T}{\bf \beta} \leq \sqrt{\sum_{i=1}^k \alpha_i^2}\sqrt{\sum_{i=1}^k \beta_i^2} = \sqrt{\sum_{i=1}^k \alpha_i^2},
  \end{align*}
and the result follows.
\end{proof}

One useful way of looking at Theorem~\ref{thm:sk_vector_norm} is through the vector-operator isomorphism. Because the Schmidt coefficients of $\ket{v}$ are the singular values of ${\rm mat}(\ket{v})$, we have $\big\| \ket{v} \big\|_{s(k)} = \big\| {\rm mat}(\ket{v}) \big\|_{(k,2)}$, where the operator $(k,2)$-norm was defined in Section~\ref{sec:norms}. We will make use of this equivalence in Sections~\ref{sec:sk_vector_dual_norm} and~\ref{sec:sk_vector_isometries} in order to investigate the dual norm and isometries of the $s(k)$-vector norm.

Because $\bb{C}^m \otimes \bb{C}^n$ is finite-dimensional, the $s(k)$-norms on it must be equivalent. The following corollary of Theorem~\ref{thm:sk_vector_norm} quantifies the equivalence of these norms.
\begin{cor}\label{cor:vecEquiv}
    Let $\ket{v} \in \bb{C}^m \otimes \bb{C}^n$ and suppose $1 \leq h \leq k \leq \min\{m,n\}$. Then
    \begin{align*}
      \big\| \ket{v} \big\|_{s(h)} & \leq \big\| \ket{v} \big\|_{s(k)} \leq \sqrt{\frac{k}{h}} \big\| \ket{v} \big\|_{s(h)}.
    \end{align*}
    Furthermore, equality is attained on the left if and only if $\big\| \ket{v} \big\|_{s(h)} = 1$ if and only if $SR(\ket{v}) \leq h$. Equality is attained on the right if and only if the $k$ largest Schmidt coefficients of $\ket{v}$ are equal.
\end{cor}
\begin{proof}
	The left inequality follows trivially from the definition of the $s(k)$-vector norm. To see the right inequality, use Theorem~\ref{thm:sk_vector_norm} to write
	\begin{align*}
		k\big\| \ket{v} \big\|_{s(h)}^2 & = k\sum_{i=1}^{h}\alpha_i^2 \\
		& = h\sum_{i=1}^{h}\alpha_i^2 + (k-h)\sum_{i=1}^h\alpha_i^2 \\
		& \geq h\sum_{i=1}^{h}\alpha_i^2 + (k-h)h\alpha_h^2 \\
		& \geq h\sum_{i=1}^{h}\alpha_i^2 + h\sum_{i=h+1}^k\alpha_i^2 \\
		& = h\big\| \ket{v} \big\|_{s(k)}^2,
	\end{align*}
	where we used the fact that $\alpha_i \geq \alpha_{i+1}$ for all $i$ twice. Dividing through by $h$ and taking the square root of both sides gives $\big\| \ket{v} \big\|_{s(k)} \leq \sqrt{\frac{k}{h}} \big\| \ket{v} \big\|_{s(h)}$. The remaining claims follow easily from the facts that $\alpha_i \geq 0$ and $\alpha_i \geq \alpha_{i+1}$ for all $i$.
\end{proof}

Corollary~\ref{cor:vecEquiv} supports the interpretation of the $s(k)$-vector norms as a measure of how close a pure state is to having Schmidt Rank $k$, as it shows explicitly that $\big\| \ket{v} \big\|_{s(k)} = 1$ (the largest possible value that norm can take on pure states) if and only if $SR(\ket{v}) \leq k$. On the other hand, consider the maximally-entangled state $\ket{\psi_+} := \frac{1}{\sqrt{\min\{m,n\}}}\sum_{i=1}^{\min\{m,n\}}\ket{i}\otimes\ket{i} \in \bb{C}^m \otimes \bb{C}^n$ -- Corollary~\ref{cor:vecEquiv} also implies that $\big\| \ket{\psi_+} \big\|_{s(k)} = \sqrt{\frac{k}{n}}$, which is the smallest the norm can ever be on pure states. We can make this interpretation of the vector norms more precise by using the fidelity. It is not difficult to show via Equation~\eqref{eq:fidelity_pure} that
\begin{align}\label{eq:sk_vec_fid}
  \big\| \ket{v} \big\|_{s(k)}^2 & = \sup_{\rho} \Big\{ F(\rho,\ketbra{v}{v}) : SN(\rho) \leq k \Big\}.
\end{align}

The final result of this section shows that the $s(k)$-vector norm of a pure state is equal to the Ky Fan $k$-norm of its reduced density matrix. In particular, this implies that the $s(k)$-vector norms are closely related to the \emph{relative entropy of entanglement} \cite{BBPS96,PV07}, which is defined in terms of the eigenvalues of a state's reduced density matrix.
\begin{cor}\label{cor:sk_norm_ky_fan}
  Let $\ket{v} \in \bb{C}^m \otimes \bb{C}^n$ and suppose $1 \leq k \leq \min\{m,n\}$. Then
  \begin{align*}
    \big\| \ket{v} \big\|_{s(k)}^2 = \big\| \Tr_1(\ketbra{v}{v}) \big\|_{(k)} = \big\| \Tr_2(\ketbra{v}{v}) \big\|_{(k)}.
  \end{align*}
\end{cor}
\begin{proof}
	To see the first equality, write $\ket{v}$ in its Schmidt decomposition:
	\begin{align*}
		\ket{v} = \sum_{i=1}^{\min\{m,n\}} \alpha_i \ket{a_i} \otimes \ket{b_i}.
	\end{align*}
	Then
	\begin{align*}
		\Tr_1(\ketbra{v}{v}) = \sum_{i=1}^{\min\{m,n\}} \alpha_i^2 \ketbra{b_i}{b_i},
	\end{align*}
	from which it follows that $\big\| \Tr_1(\ketbra{v}{v}) \big\|_{(k)} = \sum_{i=1}^k \alpha_i^2 = \big\| \ket{v} \big\|_{s(k)}^2$. The second equality is proved analogously.
\end{proof}

\subsection{Dual Norms}\label{sec:sk_vector_dual_norm}

We now investigate the duals of the $s(k)$-vector norms, which we write as $\|\cdot\|_{s(k)}^\circ$. That is, we investigate the norm on $\bb{C}^m \otimes \bb{C}^n$ defined by
\begin{align*}
	\big\|\ket{v}\big\|_{s(k)}^\circ := \sup_{c,\ket{w}} \Big\{ c\big| \braket{w}{v} \big| : c\big\|\ket{w}\big\|_{s(k)} \leq 1 \Big\} = \sup_{\ket{w}} \left\{ \frac{\big| \braket{w}{v} \big|}{\big\|\ket{w}\big\|_{s(k)}} \right\}.
\end{align*}
Recall that the Euclidean norm $\|\cdot\|$ is self-dual; $\|\cdot\| = \|\cdot\|^{\circ}$. It follows that, much like the $s(k)$-norms form an increasing family of norms leading up to the Euclidean norm, the duals of the $s(k)$-norms form a decreasing family of norms leading down to the Euclidean norm. That is, we have the following chain of inequalities:
\begin{align*}
  \big\| \ket{v} \big\|_{s(1)}^\circ \geq \big\| \ket{v} \big\|_{s(2)}^\circ \geq \cdots \geq \big\| \ket{v} \big\|_{s(\min\{m,n\}-1)}^\circ \geq \big\| \ket{v} \big\|.
\end{align*}

The dual of the $s(k)$-vector norm has a similar interpretation to that of the $s(k)$-norm itself -- it measures how close a pure state is to having Schmidt rank $k$, with the primary difference being that while a larger value of $\big\|\ket{v}\big\|_{s(k)}$ corresponds to $\ket{v}$ being closer to the set of states with Schmidt rank $k$, a \emph{smaller} value of $\big\|\ket{v}\big\|_{s(k)}^\circ$ has the same interpretation. Similarly, the extreme cases of the dual norms behave very similarly to those of the $s(k)$-norms themselves. Notice that $\big\|\ket{v}\big\|_{s(k)}^\circ = 1$ if and only if $\big| \braket{w}{v} \big| \leq \big\|\ket{w}\big\|_{s(k)}$ for all $\ket{w} \in \bb{C}^m \otimes \bb{C}^n$. It follows that $\big\|\ket{v}\big\|_{s(k)}^\circ = 1$ if and only if $SR(\ket{v}) \leq k$; if $SR(\ket{v}) \leq k$ then $\big| \braket{w}{v} \big| \leq \big\|\ket{w}\big\|_{s(k)}$ by definition of $\big\|\ket{w}\big\|_{s(k)}$, and if $SR(\ket{v}) > k$ then choosing $\ket{w} = \ket{v}$ violates that inequality.

With the above properties in mind, we now able to prove our first characterization of $\|\cdot\|_{s(k)}^\circ$.
\begin{thm}\label{thm:sk_dual_groth}
	Let $\ket{v} \in \bb{C}^m \otimes \bb{C}^n$. Then
	\begin{align}
		\big\|\ket{v}\big\|_{s(k)}^\circ = \inf\Big\{ \sum_i |c_i| : \ket{v} = \sum_i c_i\ket{v_i} \text{ with } SR(\ket{v_i}) \leq k \ \forall \, i \Big\},
	\end{align}
	where the infimum is taken over all decompositions of $\ket{v}$ of the given form.
\end{thm}
\begin{proof}
	The $k = 1$ version of this result is well-known: $\|\cdot\|_{s(1)}$ is the ``injective crossnorm'' and $\|\cdot\|_{s(1)}^\circ$ is the ``projective crossnorm'', which are well-known to be duals of each other. We now prove that the result for arbitrary $k$, using similar ideas to those in the proof of the $k = 1$ case in \cite[Chapter~1]{DGFS08}.
	
	We call a norm $\mnorm{\cdot}$ on $\bb{C}^m \otimes \bb{C}^n$ with the property that $\bmnorm{\ket{v}} = \bmnorm{\ket{v}}^\circ = 1$ whenever $SR(\ket{v}) \leq k$ a \emph{$k$-crossnorm}. We showed earlier in this section that $\|\cdot\|_{s(k)}$ is a $k$-crossnorm, and we note that $\mnorm{\cdot}^\circ$ is a $k$-crossnorm whenever $\mnorm{\cdot}$ is a $k$-crossnorm.
	
	Our first step is to prove that $\|\cdot\|_{s(k)}$ is the smallest $k$-crossnorm. To this end, let $\mnorm{\cdot}$ be any $k$-crossnorm. Then
	\begin{align*}
		\big\| \ket{v} \big\|_{s(k)} & = \sup_{\ket{w}} \Big\{ \big| \braket{w}{v} \big| : SR(\ket{w}) \leq k \Big\} \leq \sup_{\ket{w}} \Big\{ \big| \braket{w}{v} \big| : \bmnorm{\ket{w}}^\circ \leq 1 \Big\} = \mnorm{\cdot}.
	\end{align*}
	It follows that $\| \cdot \|_{s(k)}$ is the smallest $k$-crossnorm, so $\| \cdot \|_{s(k)}^\circ$ is the largest $k$-crossnorm. The remainder of the proof is devoted to showing that the infimum given in the statement of the theorem is also the largest $k$-crossnorm. From now on we denote this infimum by $\|\cdot\|_{k,\textup{inf}}$ for convenience.
	
	To see that $\|\cdot\|_{k,\textup{inf}}$ is a norm, we only show the triangle inequality, as the remaining properties are trivial to check. We first fix $\varepsilon > 0$. If $\ket{v} = \sum_i c_i \ket{v_i}$ and $\ket{w} = \sum_i d_i \ket{w_i}$ are decompositions of $\ket{v}$ and $\ket{w}$ so that $SR(\ket{v_i}),SR(\ket{w_i}) \leq k$ for all $i$, $\sum_i |c_i| \leq \big\|\ket{v}\big\|_{k,\textup{inf}} + \varepsilon$, and $\sum_i |d_i| \leq \big\|\ket{w}\big\|_{k,\textup{inf}} + \varepsilon$, then
	\begin{align*}
		\big\|\ket{v} + \ket{w}\big\|_{k,\textup{inf}} \leq \sum_i |c_i| + \sum_i |d_i| \leq \big\|\ket{v}\big\|_{k,\textup{inf}} + \big\|\ket{w}\big\|_{k,\textup{inf}} + 2\varepsilon.
	\end{align*}
	Since $\varepsilon > 0$ is arbitrary, the triangle inequality follows, so $\|\cdot\|_{k,\textup{inf}}$ is a norm.
	
	We now show that $\|\cdot\|_{k,\textup{inf}}$ is a $k$-crossnorm. If $SR(\ket{v}) \leq k$ then $\big\|\ket{v}\big\|_{k,\textup{inf}} \leq 1$ by definition. The fact that $\big\|\ket{v}\big\|_{k,\textup{inf}} \geq 1$ follows from $\big\|\ket{v}\big\| \leq \big\|\ket{v}\big\|_{k,\textup{inf}}$ (which in turn follows from the triangle inequality for the Euclidean norm). Similarly, $\big\|\ket{v}\big\|_{k,\textup{inf}}^\circ \leq \big\|\ket{v}\big\| \leq 1$. Finally, since $\big\|\ket{v}\big\|_{k,\textup{inf}} = 1$ we have $\big\|\ket{v}\big\|_{k,\textup{inf}}^\circ \geq \big|\braket{v}{v}\big| = 1$, which shows that $\|\cdot\|_{k,\textup{inf}}$ is a $k$-crossnorm.
	
	To complete the proof, we show that if $\mnorm{\cdot}$ is any $k$-crossnorm, then $\mnorm{\cdot} \leq \|\cdot\|_{k,\textup{inf}}$. Indeed, this fact relies simply on the triangle inequality for $\mnorm{\cdot}$. If we write $\ket{v} = \sum_i c_i \ket{v_i}$ with $SR(\ket{v_i}) \leq k$ for all $i$ then
	\begin{align*}
		\bmnorm{\ket{v}} = \mnorm{\sum_i c_i \ket{v_i}} \leq \sum_i |c_i| \bmnorm{\ket{v_i}} = \sum_i |c_i|.
	\end{align*}
	Taking the infimum over all such decompositions of $\ket{v}$ gives $\bmnorm{\ket{v}} \leq \big\|\ket{v}\big\|_{k,\textup{inf}}$, as desired.
\end{proof}

While Theorem~\ref{thm:sk_dual_groth} is of theoretical interest, it says nothing about how to compute $\|\cdot\|_{s(k)}^\circ$. A basic property of dual norms implies that $\big\|\ket{v}\big\|_{s(k)}^\circ \geq \big\|\ket{v}\big\|_{s(k)}^{-1}$ for all $\ket{v} \in \bb{C}^m \otimes \bb{C}^n$, but that inequality is not always attained. In order to compute $\|\cdot\|_{s(k)}^\circ$, we will use our earlier observation that $\big\| \ket{v} \big\|_{s(k)} = \big\| {\rm mat}(\ket{v}) \big\|_{(k,2)}$. Then by applying Theorem~\ref{thm:pk_norm_dual}, which characterizes the duals of the $(k,p)$-operator norms, we immediately have a characterization of $\big\| \ket{v} \big\|_{s(k)}^\circ$ in terms of the Schmidt coefficients of $\ket{v}$.
\begin{thm}\label{thm:sk_vector_norm_dual}
	Let $\ket{v} \in \bb{C}^m \otimes \bb{C}^n$ have Schmidt coefficients $\alpha_1 \geq \alpha_2 \geq \cdots \geq 0$. Let $r$ be the largest index $1 \leq r < k$ such that $\alpha_r > \sum_{i=r+1}^{\min\{m,n\}}\alpha_{i}/(k-r)$ (or take $r = 0$ if no such index exists). Also define $\tilde{\alpha} := \sum_{i=r+1}^{\min\{m,n\}}\alpha_i/(k-r)$. Then	
	\begin{align*}
		\big\|\ket{v}\big\|_{s(k)}^\circ = \sqrt{\sum_{i=1}^r \alpha_i^2 + (k - r)\tilde{\alpha}^2}.
	\end{align*}
\end{thm}
\begin{proof}
	We simply associate $\ket{v}$ with the operator ${\rm mat}(\ket{v})$ via the vector-operator isomorphism and note that $\big\| \ket{v} \big\|_{s(k)} = \big\| {\rm mat}(\ket{v}) \big\|_{(k,2)}$ (and similarly, because the vector-operator isomorphism preserves the inner product, $\big\| \ket{v} \big\|_{s(k)}^\circ = \big\| {\rm mat}(\ket{v}) \big\|_{(k,2)}^\circ$). Applying Theorem~\ref{thm:pk_norm_dual} with $p = 2$ then gives
	\begin{align*}
		\big\|{\rm mat}(\ket{v})\big\|_{(k,2)}^\circ = \sqrt{\sum_{i=1}^r \alpha_i^2 + (k - r)\tilde{\alpha}^2}.
	\end{align*}
	The result follows.
\end{proof}

Recall from the proof of Theorem~\ref{thm:sk_vector_norm} that if $\ket{v} = \sum_{i=1}^{\min\{m,n\}}\alpha_i\ket{u_i}\otimes\ket{v_i}$ then a vector $\ket{w}$ with $SR(\ket{w}) \leq k$ such that $\big|\braket{w}{v}\big| = \big\|\ket{v}\big\|_{s(k)}$ is the normalization of $\sum_{i=1}^{k}\alpha_i\ket{u_i}\otimes\ket{v_i}$. For the dual norm, a similar role is played by the normalization $\ket{w}$ of $c\ket{w} := \sum_{i=1}^r\alpha_i\ket{u_i}\otimes\ket{v_i} + \sum_{i=r+1}^{\min\{m,n\}}\tilde{\alpha}\ket{u_i}\otimes\ket{v_i}$, where $r$ and $\tilde{\alpha}$ are as defined in Theorem~\ref{thm:sk_vector_norm_dual}. Then
\begin{align*}
	c\big|\braket{w}{v}\big| & = \Big( \sum_{i=1}^r\alpha_i\bra{u_i}\otimes\bra{v_i} + \sum_{i=r+1}^{\min\{m,n\}}\tilde{\alpha}\bra{u_i}\otimes\bra{v_i} \Big)\Big( \sum_{i=1}^{\min\{m,n\}}\alpha_i\ket{u_i}\otimes\ket{v_i} \Big) \\
	& = \sum_{i=1}^r\alpha_i^2 + \sum_{i=r+1}^{\min\{m,n\}}\tilde{\alpha}\alpha_i \\
	& = \sum_{i=1}^r\alpha_i^2 + (k-r)\tilde{\alpha}^2.
\end{align*}
Similarly,
\begin{align*}
	c\big\|\ket{w}\big\|_{s(k)} = \sqrt{\sum_{i=1}^r\alpha_i^2 + \sum_{i=r+1}^k\tilde{\alpha}^2} = \sqrt{\sum_{i=1}^r\alpha_i^2 + (k-r)\tilde{\alpha}^2}.
\end{align*}
It follows that $\big|\braket{w}{v}\big|/\big\|\ket{w}\big\|_{s(k)} = \sqrt{\sum_{i=1}^r\alpha_i^2 + (k-r)\tilde{\alpha}^2} = \big\|\ket{v}\big\|_{s(k)}^\circ$, so $\ket{w}$ attains the supremum that defines the norm $\big\|\ket{v}\big\|_{s(k)}^\circ$.

\subsection{Isometries}\label{sec:sk_vector_isometries}

We now consider the problem of characterizing the isometries of the $s(k)$-vector norms -- that is, the operators $U \in M_m \otimes M_n$ such that $\big\|U\ket{v}\big\|_{s(k)} = \big\|\ket{v}\big\|_{s(k)}$ for all $\ket{v} \in \bb{C}^m \otimes \bb{C}^n$. In the $k = \min\{m,n\}$ case, the $s(k)$-norm is simply the Euclidean norm, so the isometries are exactly the unitary operators. However, when $k < \min\{m,n\}$ it is clear that the isometry group does not contain all unitary operators -- for example, any unitary that sends a separable state $\ket{v}$ to the maximally-entangled state $\ket{\psi_+} = \frac{1}{\sqrt{\min\{m,n\}}}\sum_{i=1}^{\min\{m,n\}}\ket{ii}$ has $\big\|\ket{v}\big\|_{s(k)} = 1$ but $\big\|U\ket{v}\big\|_{s(k)} = \sqrt{\frac{k}{\min\{m,n\}}} < 1$.

It is also clear that the isometry group contains all operators $U$ of the form
\begin{align}\label{eq:localU}
	U = U_1 \otimes U_2	\quad \text{ or } \quad n = m \text{ and } U = S(U_1 \otimes U_2),
\end{align}
where $U_1 \in M_m$ and $U_2 \in M_n$ are unitary, and $S \in M_n \otimes M_n$ is the swap operator introduced in Section~\ref{sec:symmetric_sub}. In fact, it follows from Theorem~\ref{thm:geometric_measure_entanglement} and the fact that $\big\|\ket{v}\big\|_{s(1)} = \sqrt{1 - E(\ket{v})}$ in the bipartite case, that the isometries of the $s(1)$-norm are exactly the unitaries of the form~\ref{eq:localU}. We now prove that these operators actually form the isometry group of the $s(k)$-norm for all $k < \min\{m,n\}$.
\begin{thm}\label{thm:sk_vector_norm_isometries}
	Let $1 \leq k < \min\{m,n\}$ and $U \in M_m \otimes M_n$. Then $\big\|U\ket{v}\big\|_{s(k)} = \big\|\ket{v}\big\|_{s(k)}$ for all $\ket{v} \in \bb{C}^m \otimes \bb{C}^n$ if and only if $U$ is a unitary of the form~\eqref{eq:localU}.
\end{thm}
\begin{proof}
	The ``if'' implication is trivial. To see the ``only if'' implication, use the vector-operator isomorphism and recall that $\big\|\ket{v}\big\|_{s(k)} = \big\|{\rm mat}(\ket{v})\big\|_{(k,2)}$. Since the operator $(k,2)$-norm is unitarily-invariant and not a multiple of the Frobenius norm when $k < \min\{m,n\}$, it follows from Theorem~\ref{thm:unitary_norm_isometries} that the map $\Phi_U$ associated to $U$ through the vector-operator isomorphism is of the form $\Phi_U(X) = U_1 XU_2$, or $n = m$ and $\Phi_U(X) = U_1 X^T U_2$, for some unitaries $U_1 \in M_m$ and $U_2 \in M_n$. Thus we can write either $U = U_1 \otimes U_2^T$ or $U = S(U_2^T \otimes U_1)$ (if $n = m$).
\end{proof}

Another method of proving Theorem~\ref{thm:sk_vector_norm_isometries} would be to mimic the proof of Theorem~\ref{thm:geometric_measure_entanglement} and first argue that $U$ must be unitary. Next, one could then show that it must map the set of states with Schmidt rank at most $k$ back into itself, and finally invoke Theorem~\ref{thm:bipartite_schmidt_preserver}. We use this approach to investigate the maps that preserve the $S(k)$-norm, which is introduced in the next section.

\section{The S(k)-Operator Norm}\label{sec:sk_operator_norm}

In this section we define and investigate a family of operator norms that arise from the Schmidt rank of pure states in a manner similar to the $s(k)$-vector norms of the previous section. The vector norms are recovered in the special case of rank-one operators, and will be used to derive an upper bound for the operator norms.
\begin{defn}\label{defn:operator_sk_norm}
  Let $X \in M_m \otimes M_n$ and let $1 \leq k \leq \min\{m,n\}$. Then we define the {\em $S(k)$-operator norm} of $X$, denoted $\big\| X \big\|_{S(k)}$, by
  \begin{align*}
    \big\|X\big\|_{S(k)} & := \sup_{\ket{v},\ket{w}} \Big\{ \big|\bra{w}X\ket{v}\big| : SR(\ket{v}),SR(\ket{w}) \leq k \Big\}.
  \end{align*}
\end{defn}

To see that these quantities are indeed norms, notice that (as with the vector norms) positive homogeneity and the triangle inequality follow from the corresponding properties of the complex modulus and supremum. The fact that $\big\|X\big\|_{S(k)} = 0$ if and only if $X = 0$ follows from Lemma~\ref{lem:sep_prod_zero}.

Before continuing, let us comment briefly on the definition of the $S(k)$-operator norms. We could just as well have defined another generalization of the $s(k)$-vector norms to the case of operators by using the bipartite version of the vector-operator isomorphism (i.e., the isomorphism of Section~\ref{sec:operator_schmidt}). Then the $s(k)$-vector norm on $(\bb{C}^{m} \otimes \bb{C}^{m}) \otimes (\bb{C}^{n} \otimes \bb{C}^{n})$ (with $1 \leq k \leq \min\{m^2,n^2\}$) corresponds to a norm on $M_m \otimes M_n$. However, one motivation for investigating the norm given by Definition~\ref{defn:operator_sk_norm} instead is that the $s(k)$-vector norms are in a sense trivial since they can be computed efficiently, as shown in Theorem~\ref{thm:sk_vector_norm}. Because the quantum separability problem is known to be NP-hard, as is the problem of determining block positivity of an operator \cite{G03}, it seems unlikely that an easily-computable operator norm could tell us much about block positivity or Schmidt number.

We will see in Section~\ref{sec:op_sk_norm_block_pos} that the $S(k)$-operator norm is a very powerful tool for detecting $k$-block positivity. Additionally, these norms build on the general principle that properties of pure states are easier to determine than properties of mixed states. We will see in Proposition~\ref{prop:rankOneNorm} that the $S(k)$-operator norm of a pure state reduces simply to the square of the $s(k)$-vector norm of the corresponding pure vector state. Thus, the operator norms can efficiently be computed for pure states, but we will see that computing them for general mixed states is quite difficult.

We will now investigate various aspects of these norms, much as was done for the $s(k)$-vector norms in the previous section. One key difference with our presentation of these norms is that we will not consider the dual of the $S(k)$-operator norm until Section~\ref{sec:dual_op_norm}, as its characterization relies on techniques from the theory of operator spaces that have not yet been presented.

\subsection{Basic Properties}\label{sec:sk_operator_norm_prop}

In analogy with the $s(k)$-vector norms, notice that $\big\| X \big\|_{S(\min\{m,n\})} = \big\| X \big\|$ and $\big\| X \big\|_{S(k)} \leq \big\| X \big\|$ for all $k$. Furthermore, the $S(k)$-operator norms form an increasing family of norms that lead up to the standard operator norm:
\begin{align*}
  \big\| X \big\|_{S(1)} \leq \big\| X \big\|_{S(2)} \leq \cdots \leq \big\| X \big\|_{S(\min\{m,n\}-1)} \leq \big\| X \big\|.
\end{align*}
Moreover, although $\big\| X^\dagger \big\|_{S(k)} = \big\| X \big\|_{S(k)}$, it is {\em not} the case in general that $\big\| X^{\dagger}X \big\|_{S(k)} = \big\| X \big\|_{S(k)}^2$. They also do not satisfy any natural submultiplicativity relationships.
\begin{prop}\label{prop:rankOneNorm}
  Let $\ketbra{x}{y} \in M_m \otimes M_n$ be a rank-$1$ operator. Then
  \begin{align*}
    \big\| \ketbra{x}{y} \big\|_{S(k)} & = \big\| \ket{x} \big\|_{s(k)}\big\| \ket{y} \big\|_{s(k)}.
  \end{align*}
\end{prop}
\begin{proof}
	The proof follows easily from the relevant definitions:
	\begin{align*}
    \big\| \ketbra{x}{y} \big\|_{S(k)} & = \sup_{\ket{v},\ket{w}}\Big\{ \big| \braket{w}{x}\braket{y}{v} \big| : SR(\ket{v}),SR(\ket{w}) \leq k \Big\} \\
     & = \sup_{\ket{w}}\Big\{ \big| \braket{w}{x} \big| : SR(\ket{w}) \leq k \Big\}\sup_{\ket{v}}\Big\{ \big| \braket{y}{v} \big| : SR(\ket{v}) \leq k \Big\} \\
     & = \big\| \ket{x} \big\|_{s(k)} \big\| \ket{y} \big\|_{s(k)}.
  \end{align*}
\end{proof}
We now present an important example to make use of Proposition~\ref{prop:rankOneNorm}.
\begin{exam}\label{exam:MatrixNormChoi}{\rm
  Recall the rank-$1$ projection operator $\ketbra{\psi_+}{\psi_+} := \frac{1}{n}\sum_{i,j=1}^n \ketbra{i}{j} \otimes \ketbra{i}{j} \in M_n \otimes M_n$. By Proposition~\ref{prop:rankOneNorm} we have that
  \begin{align*}
    \big\| \ketbra{\psi_+}{\psi_+} \big\|_{S(k)} & = \big\| \sum_{i=1}^n \frac{1}{\sqrt{n}}\ket{i} \otimes \ket{i} \big\|_{s(k)}^2 = \sum_{i=1}^{k}\left(\frac{1}{\sqrt{n}}\right)^2 = \frac{k}{n}.
\end{align*}

  We will see that this simple example can be used to show that some inequalities that we derive in the next section are tight. It will also have applications to bound entanglement in Section~\ref{sec:bound_entanglement}.}
\end{exam}

The following proposition shows if $X$ is positive semidefinite then it is enough to take the supremum only over $\ket{v}$ in the definition of the $S(k)$-operator norms.
\begin{prop}\label{prop:MatMultDiffVectors}
  Let $X \in M_m \otimes M_n$ be positive semidefinite. Then
  \begin{align}\label{eq:MatMultDiffVectors}
    \big\| X \big\|_{S(k)} & = \sup_{\ket{v}} \big\{ \bra{v}X\ket{v} : SR(\ket{v}) \leq k \big\} \\ \label{eq:MatMultSN}
    & = \sup_{\rho} \big\{ \Tr(X \rho) : SN(\rho) \leq k \big\}.
  \end{align}
\end{prop}
\begin{proof}
  To show the first equality, write $X$ in its Spectral Decomposition as $X = \sum_i \lambda_i \ketbra{v_i}{v_i}$. For any $\ket{v}$ and $\ket{w}$ with $SR(\ket{v}),SR(\ket{w}) \leq k$, we have $\bra{v}X\ket{v} = \sum_i \lambda_i \big|\braket{v_i}{v}\big|^2$ and $\bra{w}X\ket{w} = \sum_i \lambda_i \big|\braket{v_i}{w}\big|^2$. Now define the $i^{th}$ component of two vectors ${\bf v}^\prime$ and ${\bf w}^\prime$ by $v_i^\prime := \sqrt{\lambda_i}\big|\braket{v_i}{v}\big|$ and $w_i^\prime := \sqrt{\lambda_i}\big|\braket{w}{v_i}\big|$. Applying the Cauchy--Schwarz inequality to ${\bf v}^\prime$ and ${\bf w}^\prime$ gives $\big|\bra{w}X\ket{v}\big| \leq \sqrt{\bra{v}X\ket{v}}\sqrt{\bra{w}X\ket{w}} \leq \max\left\{\bra{v}X\ket{v},\bra{w}X\ket{w}\right\}$. It follows that $\big\| X \big\|_{S(k)} \leq \sup_{\ket{v}} \big\{ \bra{v}X\ket{v} : SR(\ket{v}) \leq k \big\}$, and the other inequality is trivial.\\
  To see the second equality, simply write
  \[
    \sup_{\ket{v}} \big\{ \bra{v}X\ket{v} : SR(\ket{v}) \leq k \big\} = \sup_{\ket{v}} \big\{ \Tr(X\ketbra{v}{v}) : SR(\ket{v}) \leq k \big\},
  \]
	and note that the supremum on the right cannot become larger when taking the supremum over mixed states since mixed states can be written as convex combinations of pure states.
\end{proof}

Equation~\eqref{eq:MatMultDiffVectors} captures a well-known property of the operator norm of positive operators in the $k = \min\{m,n\}$ case. We also note that Proposition~\ref{prop:MatMultDiffVectors} says that the $S(1)$-norm, $\big\|\cdot\big\|_{S(1)}$, when acting on positive operators, coincides with the \emph{product numerical radius} $r^{\otimes}$ \cite{GPMSZ10,PGMSCZ11} (see also \cite{JRC10}). That is, if $X$ is positive then $\big\|X\big\|_{S(1)} = r^{\otimes}(X)$. In the more general case of arbitrary $k$, this quantity has been referred to as the \emph{maximal SN-$k$ expectation value} \cite{SV11}. Equation~\eqref{eq:MatMultSN} is perhaps a more natural way of looking at $\big\| X \big\|_{S(k)}$ from a quantum information perspective.

Because Equation~\eqref{eq:MatMultDiffVectors} holds (up to absolute value) for the operator norm not just for positive semidefinite operators, but more generally for normal operators, one might initially expect that Proposition~\ref{prop:MatMultDiffVectors} can be extended to the case of normal operators as well. We now present an example to show that this generalization is actually not true, even just for Hermitian operators.
\begin{exam}\label{ex:sk_norm_hermitian}
  {\rm Define $X := \ketbra{11}{22} + \ketbra{22}{11}$ and observe that $X$ is Hermitian with $\big\|X\big\| = 1$. Also, $\big\|X\big\|_{S(1)} = 1$ because $\bra{11}X\ket{22} = 1$. However, if we restrict the supremum that defines $\big\|X\big\|_{S(1)}$ to the $\ket{v} = \ket{w}$ case then we have
	  \begin{align*}\sspp
      \sup_{\ket{v}} \Big\{ \big|\bra{v}X\ket{v}\big| : SR(\ket{v}) = 1 \Big\} & = \sup_{\ket{a},\ket{b}} \Big\{ \big|\bra{ab}X\ket{ab}\big| \Big\} \\
      & \leq 2\sup_{\ket{a},\ket{b}} \Big\{ \big|\braket{a}{1}\braket{b}{1}\braket{2}{a}\braket{2}{b}\big| \Big\} \\
      & = 2\sup_{\ket{a}} \Big\{ \big|\braket{a}{1}\braket{2}{a}\big|^2 \Big\} \\
      & \leq \frac{1}{2}.\dsp
    \end{align*}
    The final inequality above can be seen by writing $\ket{a} = (a_1, a_2)^T$ and then observing that $\big|\braket{a}{0}\braket{1}{a}\big| = |\overline{a_1} a_2|$, which is bounded above by $1/2$ because $|a_1|^2 + |a_2|^2 = 1$. It is worth noting that the upper bound of $1/2$ is in fact attained by the vector $\ket{v} = \frac{1}{2}(\ket{1} + \ket{2}) \otimes (\ket{1} + \ket{2})$.}
\end{exam}

For a general mixed state $\rho$, one might want to think of $\big\| \rho \big\|_{S(k)}$ as measuring how close $\rho$ is to having Schmidt number of $k$ or less, but this interpretation is not quite right. Consider the following example, which shows that it is not true that $SN(\rho) \leq k$ implies $\|\rho\|_{S(k)} = \|\rho\|$ (contrast this with the corresponding true statement for the $s(k)$-vector norms that $SR(\ket{v}) \leq k$ implies $\big\|\ket{v}\big\|_{s(k)} = \big\|\ket{v}\big\|$).
\begin{exam}\label{ex:BadSNorm}
  {\rm Let $\rho \in M_2 \otimes M_2$ have the following matrix representation in the standard basis $\{\ket{11},\ket{12},\ket{21},\ket{22}\}$:
    \[\sspp
        \rho = \frac{1}{8}\begin{bmatrix}5 & 1 & 1 & 1 \\ 1 & 1 & 1 & 1 \\ 1 & 1 & 1 & 1 \\ 1 & 1 & 1 & 1\end{bmatrix} = \frac{1}{2}\begin{bmatrix}1 & 0 \\ 0 & 0\end{bmatrix} \otimes \begin{bmatrix}1 & 0 \\ 0 & 0\end{bmatrix} + \frac{1}{8}\begin{bmatrix}1 & 1 \\ 1 & 1\end{bmatrix} \otimes \begin{bmatrix}1 & 1 \\ 1 & 1\end{bmatrix}.\dsp
    \]
	It is clear that $SN(\rho) = 1$. However, the eigenvector corresponding to the (distinct) maximal eigenvalue $3/4$ is $\ket{v} := \frac{1}{2\sqrt{3}}(3, 1, 1, 1)^T$. It is easily verified that $SR(\ket{v}) = 2$, so $\|\rho\|_{S(1)} < \|\rho\|$ (in fact, we will see in Example~\ref{ex:BadSNorm_return} that $\|\rho\|_{S(1)} = \frac{1}{8}(3 + 2\sqrt{2}) \approx 0.7286$).}
\end{exam}

Nonetheless, if the eigenspace corresponding to the maximal eigenvalue of $\rho$ contains a state $\ket{v}$ with $SR(\ket{v}) \leq k$ then $\|\rho\|_{S(k)} = \|\rho\|$. More importantly though, we can see via fidelity that the correct interpretation of $\|\rho\|_{S(k)}$ is as a measure of how close $\rho$ is to a \emph{pure} state $\ket{v}$ with $SR(\ket{v}) \leq k$. More precisely, it is not difficult to show that
\begin{align}\label{eq:sk_op_fid}
  \| \rho \|_{S(k)} & = \sup_{\ket{v}} \Big\{ F(\rho,\ketbra{v}{v}) : SR(\ket{v}) \leq k \Big\}.
\end{align}
This shows that the $S(k)$-operator norms are, in a sense, dual to the $s(k)$-vector norms -- compare Equation~\eqref{eq:sk_op_fid} with Equation~\eqref{eq:sk_vec_fid}.

The following corollary of Proposition~\ref{prop:MatMultDiffVectors} shows that the $S(k)$-operator norms are non-increasing under local quantum operations.
\begin{cor}\label{cor:LocalOps}
  Let $X \in M_m \otimes M_n$ be positive and let $\Phi : M_n \rightarrow M_n$ be a quantum channel (i.e. completely positive and trace-preserving). Then
  \[
    \big\| (id_m \otimes \Phi^\dagger)(X) \big\|_{S(k)} \leq \big\| X \big\|_{S(k)}.
  \]
\end{cor}
\begin{proof}
  By Proposition~\ref{prop:MatMultDiffVectors} we know that
  \begin{align*}
    \big\| (id_m \otimes \Phi^\dagger)(X) \big\|_{S(k)} & = \sup_{\rho} \Big\{ \Tr((id_m \otimes \Phi^\dagger)(X) \rho) : SN(\rho) \leq k \Big\} \\
    & = \sup_{\rho} \Big\{ \Tr(X (id_m \otimes \Phi)(\rho)) : SN(\rho) \leq k \Big\}.
  \end{align*}
	The result follows from the fact that Schmidt number is non-increasing under the action of local quantum channels \cite{TH00}, so $SN((id_m \otimes \Phi^\dagger)(\rho)) \leq k$.
\end{proof}
In fact, it follows from Proposition~\ref{prop:semigroup} that the map $\Phi$ of Corollary~\ref{cor:LocalOps} need not be a quantum channel, but can be chosen to be just $k$-positive and trace-preserving.

We will see in the upcoming sections that of particular importance is the problem of computing the $S(k)$-norm of orthogonal projections. The following proposition will thus be of much use, as it allows us to describe the $S(k)$-norm in terms of the $s(k)$-vector norm.

\begin{prop}\label{prop:orth_proj_norm}
	Let $P = P^\dagger = P^2 \in M_m \otimes M_n$ be an orthogonal projection. Then
	\begin{align*}
		\big\|P\big\|_{S(k)} = \sup_{\ket{v}}\Big\{ \big\|\ket{v}\big\|_{s(k)}^2 : \ket{v} \in {\rm Range}(P) \Big\}.
	\end{align*}
\end{prop}
\begin{proof}
	To see the ``$\geq$'' inequality, choose an arbitrary $\ket{v} \in {\rm Range}(P)$ and use Proposition~\ref{prop:MatMultDiffVectors} to write
	\begin{align*}
		\big\|P\big\|_{S(k)} & = \sup_{\ket{w}}\big\{ \braket{w}{v} \braket{v}{w} + \bra{w}(P - \ketbra{v}{v})\ket{w} : SR(\ket{w}) \leq k \big\} \\
		& \geq \sup_{\ket{w}}\big\{ \braket{w}{v} \braket{v}{w} : SR(\ket{w}) \leq k \big\} \\
		& = \big\|\ket{v}\big\|_{s(k)}^2.
	\end{align*}
	To see the ``$\leq$'' inequality, observe that the set of states with Schmidt rank at most $k$ is compact, so there is a particular $\ket{w}$ with $SR(\ket{w}) \leq k$ such that $\bra{w}P\ket{w} = \big\|P\big\|_{S(k)}$. Define $\ket{v_1} := P\ket{w}/\big\|P\ket{w}\big\|$ and extend $\ket{v_1}$ to an orthonormal basis $\big\{\ket{v_i}\big\}_{i=1}^{{\rm rank}(P)}$ of the range of $P$. Then $\braket{v_i}{w} = 0$ for all $i \geq 2$, so if we write $P = \sum_{i=1}^{{\rm rank}(P)}\ketbra{v_i}{v_i}$ then we see that $\big\|P\big\|_{S(k)} = \bra{w}P\ket{w} = \big|\braket{v_1}{w}\big|^2 \leq \big\|\ket{v_1}\big\|_{s(k)}^2$. Because $\ket{v_1} \in {\rm Range}(P)$, the proof is complete.
\end{proof}

Finally, the last result of this section makes a crucial connection between the $S(k)$-norm and $k$-block positivity of an operator.
\begin{cor}\label{cor:kPosInf1}
  Let $0 \leq X \in M_m \otimes M_n$ be positive semidefinite and let $c \in \bb{R}$. Then $cI - X$ is $k$-block positive if and only if $c \geq \big\|X\big\|_{S(k)}$.
\end{cor}
\begin{proof}
  By Proposition~\ref{prop:kPos} we know that $cI - X$ is $k$-block positive if and only if
  \[
    \Tr((cI - X)\rho) = c - \Tr(X\rho) \geq 0 \quad \forall \, \rho \in (M_m \otimes M_n)^{+} \text{ with } SN(\rho) \leq k.
  \]
	Proposition~\ref{prop:MatMultDiffVectors} tells us that this is true precisely when $c \geq \big\|X\big\|_{S(k)}$.
\end{proof}

In particular, Corollary~\ref{cor:kPosInf1} shows that the problem of computing the operator norms is equivalent to the problem of determining $k$-block positivity of a Hermitian operator. Since the $k$-positivity problem seems to be very difficult in general, computing these norms even just for positive operators is likely a very difficult problem as well. Nevertheless, we shall see in the following sections that this connection leads to a new perspective for a number of different problems in quantum information.

\subsection{Inequalities}\label{sec:MatrixNormInequalities}

Since computing the $S(k)$-norms in general seems to be difficult, it will be useful to have explicitly calculable bounds for them. The following upper bound is thus of interest because it is easily computable in light of Theorem~\ref{thm:sk_vector_norm}.
  \begin{prop}\label{prop:matrixVectorNorms}
    Let $X \in M_m \otimes M_n$ be normal with eigenvalues $\{\lambda_i\}$ and corresponding eigenvectors $\{ \ket{v_i} \}$. Then
    \[
        \big\|X\big\|_{S(k)} \leq \sum_i |\lambda_i|\big\|\ket{v_i}\big\|^2_{s(k)}.
    \]
  \end{prop}
    \begin{proof}
        Let $\ket{v}, \ket{w} \in \bb{C}^m \otimes \bb{C}^n$ have $SR(\ket{v}),SR(\ket{w}) \leq k$. Then
        \[
            \big| \bra{w}X\ket{v} \big| = \Big| \sum_i \lambda_i \braket{w}{v_i}\braket{v_i}{v} \Big| \leq \sum_i |\lambda_i| |\braket{w}{v_i} ||\braket{v_i}{v} | \leq \sum_i |\lambda_i | \big\| \ket{v_i} \big\|^2_{s(k)}.
        \]
    \end{proof}
  The following upper bound makes use of the $(k,p)$-norm introduced in Section~\ref{sec:norms} and the realignment map $R$ of Section~\ref{sec:realign}.
  \begin{prop}\label{prop:sk_realign}
  	Let $0 \leq X \in M_m \otimes M_n$. Then $\big\|X\big\|_{S(k)} \leq \big\| R(X) \big\|_{(k^2,2)}$.
  \end{prop}
	\begin{proof}
		The proof is mostly by simple algebra:
		\begin{align*}
			\big\|X\big\|_{S(k)} & = \sup_{\rho} \left\{ \Tr(X\rho) : SN(\rho) \leq k \right\} \\
			& \leq \sup_{\rho} \left\{ \Tr(X\rho) : \|R(\rho)\|_{(k^2,2)}^\circ \leq 1 \right\} \\
			& = \sup_{Y} \left\{ \Tr(X R(Y)) : \|Y\|_{(k^2,2)}^\circ \leq 1, \Tr(R(Y)) = 1, R(Y) \geq 0 \right\} \\
			& \leq \sup_{Y} \left\{ \Tr(R(X) Y) : \|Y\|_{(k^2,2)}^\circ \leq 1 \right\} \\
			& = \big\|R(X)\big\|_{(k^2,2)}.
		\end{align*}
		The first equality above comes from Proposition~\ref{prop:MatMultDiffVectors}, the first inequality comes from Theorem~\ref{thm:gen_realign}, the next equality comes from the fact that $R = R^{-1}$ (i.e., we are setting $Y = R(\rho)$), and the second inequality uses the fact that $R = R^\dagger$.
	\end{proof}
	
	Some special cases of Property~\ref{prop:sk_realign} are worth pointing out. In the $k = 1$ case, it simply says that $\big\|X\big\|_{S(1)} \leq \big\| R(X) \big\|$ (much like $\big\| X \big\|_{S(1)} \leq \big\|X\big\|$, which we already noted). At the other extreme, if $k = \min\{m,n\}$ then the result says that $\big\|X\big\| \leq \big\| R(X) \big\|_F$, which is trivially true since $\big\| R(X) \big\|_F = \big\| X \big\|_F$.
  
Because $M_m \otimes M_n$ is finite-dimensional, we know that the $S(k)$-operator norms are equivalent. In order to quantify this fact, we will first need the following simple lemma.
\begin{lemma}\label{lem:Schmidt01}
    Let $h \leq k$ and suppose $\ket{v} \in \bb{C}^m \otimes \bb{C}^n$ is a unit vector with $SR(\ket{v}) \leq k$. Then there exist nonnegative real constants $\{ d_j \}$ and (not necessarily distinct) unit vectors $\{ \ket{v_j} \} \subseteq \bb{C}^m \otimes \bb{C}^n$ for $1 \leq j \leq k$ such that $\sum_{j=1}^k d_j^2 = h$, $SR(\ket{v_j}) \leq h$, and
    \[
      h\ket{v} = \sum_{j=1}^k d_j\ket{v_j}.
  \]
\end{lemma}
\begin{proof}
  We can write $\ket{v}$ via the Schmidt Decomposition as $\ket{v} = \sum_{j=1}^k c_j \ket{a_j} \otimes \ket{b_j}$ with $\sum_{j=1}^k|c_j|^2 = 1$ and $\{ \ket{a_j} \}$, $\{ \ket{b_j} \}$ orthonormal sets. Thus
  \[
    h\ket{v} = \sum_{i=1}^h\sum_{j=1}^k c_j \ket{a_j} \otimes \ket{b_j}.
  \]
  Because $h \leq k$, we can rearrange the summations in such a way that we sum over $k$ sets of orthonormal vectors, with $h$ vectors in each set. We thus have $h\ket{v} = \sum_{j=1}^k d_j\ket{v_j}$ for some unit vectors $\ket{v_j}$ with $SR(\ket{v_j}) \leq h$ and constants $d_j$ satisfying $\sum_{j=1}^k d_j^2 = h$.
\end{proof}

\begin{thm}\label{thm:MatrixEquiv01}
    Let $X \in M_m \otimes M_n$ and suppose $h \leq k$. Then
    \begin{align*}
      \big\| X \big\|_{S(h)} & \leq \big\| X \big\|_{S(k)} \leq \frac{k}{h} \big\| X \big\|_{S(h)}.
    \end{align*}
\end{thm}
\begin{proof}
  The left inequality is trivial by the definition of the operator norms. To see the right inequality, suppose $\ket{v}$ and $\ket{w}$ have $SR(\ket{v}),SR(\ket{w}) \leq k$. Use Lemma~\ref{lem:Schmidt01} to write $h\ket{v} = \sum_{j=1}^k d_j\ket{v_j}$ and $h\ket{w} = \sum_{j=1}^k f_j\ket{w_j}$ so that
  \begin{align*}
    h^2\big| \bra{w}X\ket{v} \big| = \big| \sum_{i,j=1}^k{f_i d_j \bra{w_i} X \ket{v_j}} \big| \leq \Big(\sum_{i=1}^k{f_i}\Big)\Big(\sum_{i=1}^k{d_i}\Big) \big\| X \big\|_{S(h)} \leq kh \big\| X \big\|_{S(h)},
  \end{align*}
	where the rightmost inequality follows from applying the Cauchy--Schwarz inequality to the vectors $(d_1, \ldots, d_k)^T$ and $(1, \ldots, 1)^T$, and to the vectors $(f_1, \ldots, f_k)^T$ and $(1, \ldots, 1)^T$. The result follows by dividing through by $h^2$.
\end{proof}

To see that the inequalities of Theorem~\ref{thm:MatrixEquiv01} are tight, simply recall Example~\ref{exam:MatrixNormChoi}. Also observe that a straightforward consequence of this result is the inequality $\big\| X \big\|_{S(k)} \geq \frac{k}{m} \big\| X \big\|$ for all $k \leq m$. We now derive lower bounds that are much better in many situations.

\begin{prop}\label{prop:lowerBoundEig}
    Let $X = X^\dagger \in M_m \otimes M_n$ have eigenvalues $\lambda_1 \leq \lambda_2 \leq \cdots \leq \lambda_{mn}$. Then for any $r \geq k$,
    \[
      \big\| X \big\|_{S(k)} \geq \frac{k\lambda_{mn - (n-r)(m-r)}}{r}.
    \]
		Furthermore, there exists $X = X^\dagger$ such that $\big\| X \big\|_{S(k)} < \lambda_{nm - (n - k)(m - k) + 1}$.
\end{prop}
\begin{proof}
  Let $\cl{V}$ be the span of the eigenvectors $\big\{\ket{v_{i}}\big\}_{i=nm - (n - r)(m - r)}^{mn}$ corresponding to the eigenvalues $\{\lambda_{i}\}_{i=nm - (n - r)(m - r)}^{mn}$. Then because ${\rm dim}(\cl{V}) = (n - r)(m - r) + 1$, by Theorem~\ref{thm:CMW08}, we know that there exists a vector $\ket{v} \in \cl{V}$ with $SR(\ket{v}) \leq r$. It follows that
    \[
        \big\|X \big\|_{S(r)} \geq \big| \bra{v}X\ket{v} \big| \geq \sum_{i=nm - (n - r)(m - r)}^{mn} \lambda_i |\braket{v_i}{v}|^2 \geq \lambda_{nm - (n - r)(m - r)}.
    \]

    \noindent Using Theorem~\ref{thm:MatrixEquiv01} then shows that if $k \leq r$,
  \begin{align*}
    \big\| X \big\|_{S(k)} \geq \frac{k}{r}\big\| X \big\|_{S(r)} \geq \frac{k\lambda_{mn - (n-r)(m-r)}}{r}.
  \end{align*}

  To see the final claim, note that the dimension given by Theorem~\ref{thm:CMW08} is tight, so we can construct a positive operator $X$ with distinct eigenvalues such that the span of the eigenvectors corresponding to its $(n - k)(m - k)$ largest eigenvalues does not contain any states $\ket{w}$ with $SR(\ket{w}) \leq k$. It follows that $\bra{v}X\ket{v} < \lambda_{nm - (n - k)(m - k) + 1}$ for all $\ket{v}$ with $SR(\ket{v}) \leq k$.
\end{proof}

\begin{thm}\label{thm:s1_norm_trace}
	Let $X = X^\dagger \in M_m \otimes M_n$. Then
	\begin{align*}
		\big\|X\big\|_{S(1)} \geq \frac{1}{mn}\left(\Tr(X) + \sqrt{\frac{mn\Tr\big( X^2\big) - \Tr(X)^2}{mn-1}}\right).
	\end{align*}
\end{thm}
\begin{proof}
	We begin by demonstrating the weaker inequality $\big\|X\big\|_{S(1)} \geq \frac{\Tr(X)}{mn}$, which we prove for two reasons. First, its proof is elementary, and it is instructive to see how the statement of the theorem compares to this simpler result. Second, this weaker inequality will be needed to overcome a slight technicality in the proof of the more general inequality.
	
	Begin by defining $p := {\rm rank}(X)$. Write $X$ in its spectral decomposition $X = \sum_{i=1}^p \lambda_i\ketbra{v_i}{v_i}$. Now write the vectors $\ket{v_i}$ in the form
    \[
        \ket{v_i} = \sum_{j=1}^{m}\sum_{\ell=1}^{n} c_{ij\ell}\ket{j}\otimes\ket{\ell},
    \]
		where $\{c_{ij\ell}\} \in \bb{C}$ is a family of constants such that
    \begin{align}\label{eq:sumeq}
        \sum_{j=1}^{m}\sum_{\ell=1}^{n} |c_{ij\ell}|^2 = 1 \quad \forall \, i = 1, 2, \ldots, p.
    \end{align}
		It follows that there exists some fixed $j$ and $\ell$ such that
    \[
        \sum_{i=1}^p\lambda_i |c_{ij\ell}|^2 \geq \frac{\Tr(X)}{mn},
    \]
		since otherwise Equation~\eqref{eq:sumeq} would be violated. Then for this specific $j$ and $\ell$,
    \begin{align*}
      \big\|X\big\|_{S(1)} \geq (\bra{j} \otimes \bra{\ell})X(\ket{j} \otimes \ket{\ell}) = \sum_{i=1}^p\lambda_i\big|\bra{v_i}(\ket{j} \otimes \ket{\ell})\big|^2 = \sum_{i=1}^p \lambda_i |c_{ij\ell}|^2 \geq \frac{\Tr(X)}{mn},
    \end{align*}
    as desired.
    
	We now prove the inequality described in the statement of the theorem. If $X \geq 0$ then Corollary~\ref{cor:kPosInf1} tells us that $\big\|X\big\|_{S(1)}I - X$ is block positive. Allowing $X$ to be Hermitian instead of positive semidefinite can only make this operator more positive, so $\big\|X\big\|_{S(1)}I - X$ is block positive in this case as well. Using Proposition~\ref{prop:block_pos_trace} then shows that
	\begin{align*}
		\Tr\left(\big(\big\|X\big\|_{S(1)}I - X\big)^2\right) \leq \left(\Tr\big(\big\|X\big\|_{S(1)}I - X\big)\right)^2.
	\end{align*}
	Expanding and rearranging terms gives
	\begin{align}\label{eq:s1_quad}
		mn\big\|X\big\|_{S(1)}^2 - 2\Tr(X)\big\|X\big\|_{S(1)} + \frac{\big(\Tr(X)\big)^2 - \Tr\big( X^2\big)}{mn-1} \geq 0.
	\end{align}
	We can treat the left-hand side of Inequality~\eqref{eq:s1_quad} as a quadratic in $\big\|X\big\|_{S(1)}$. Using the quadratic equation gives its roots as
	\begin{align*}
		\frac{\Tr(X) \pm \sqrt{(mn\Tr\big( X^2\big) - \Tr(X)^2)/(mn-1)}}{mn}.
	\end{align*}
	Inequality~\eqref{eq:s1_quad} is satisfied exactly when $\big\|X\big\|_{S(1)}$ is not strictly between these two roots. However, we already saw that $\big\|X\big\|_{S(1)} \geq \Tr(X)/mn$, which rules out the solutions smaller than the lesser of the two roots. The result follows.
\end{proof}

We now turn our attention to orthogonal projections $P = P^\dagger = P^2 \in M_m \otimes M_n$. In this case, we can improve the left inequality of Theorem~\ref{thm:MatrixEquiv01}.
\begin{prop}\label{prop:proj_norm_ineq}
    Let $P = P^\dagger = P^2 \in M_m \otimes M_n$ be an orthogonal projection and let $h \leq k$. Then
    \[
        \big\|P\big\|_{S(k)} \geq \big\|P\big\|_{S(h)} + \frac{k-h}{\min\{m,n\}-h}\big(1 - \big\|P\big\|_{S(h)}\big).
    \]
\end{prop}
\begin{proof}
  Assume without loss of generality that $m \leq n$. Use Proposition~\ref{prop:orth_proj_norm} and Theorem~\ref{thm:sk_vector_norm} to write
    \begin{align*}
        \big\|P\big\|_{S(h)} = \sup_{\ket{w} \in {\rm Range}(P)}\big\{ \sum_{i=1}^h \alpha_i^2 : \alpha_1 \geq \alpha_2 \geq \cdots \geq 0 \text{ are Schmidt coefficients of } \ket{w} \big\}.
    \end{align*}
	Now let $\ket{w} \in {\rm Range}(P)$ have Schmidt coefficients $\{ \alpha_i \}$ such that $\sum_{i=1}^h \alpha_i^2 = \big\| P \big\|_{S(h)}$. Then using the facts that $\sum_{i=1}^m\alpha_i^2 = 1$ and $\alpha_i \geq \alpha_j$ for $i \leq j$, it follows that $\sum_{i=h+1}^m \alpha_i^2 = 1 - \big\| P \big\|_{S(h)}$ and so $\sum_{i=h+1}^k \alpha_i^2 \geq \frac{k-h}{m-h}(1 - \big\| P \big\|_{S(h)})$. Thus
    \[
        \big\|P\big\|_{S(k)} \geq \sum_{i=1}^k \alpha_i^2 = \big\|P\big\|_{S(h)} + \sum_{i=h+1}^k \alpha_i^2 \geq \big\|P\big\|_{S(h)} + \frac{k-h}{m-h}\big(1 - \big\|P\big\|_{S(h)}\big).
    \]
\end{proof}

If $P = P^\dagger = P^2 \in M_m \otimes M_n$ is an orthogonal projection, then by Theorem~\ref{thm:MatrixEquiv01} we have that $\frac{k}{m} \leq \big\| P \big\|_{S(k)} \leq 1$. The left inequality was seen to be tight by a rank-$1$ projection in Example~\ref{exam:MatrixNormChoi}, and it is not difficult to construct projection operators of any rank that have $\big\| P \big\|_{S(k)} = 1$. However, the following result shows that we can improve the lower bound if we take the rank of the projection into account. Note that Inequality~\eqref{eq:projIneq2} as stated here is stronger than the corresponding inequality given in \cite{JK10}.
\begin{thm}\label{thm:mainProjRes}
    Let $P = P^\dagger = P^2 \in M_m \otimes M_n$ be an orthogonal projection and define $r := {\rm rank}(P)$. Then
    \begin{align}\label{eq:projIneq1}
      \big\| P \big\|_{S(k)} & \geq \min\Big\{1,\frac{k}{\big\lceil \frac{1}{2}\big( n + m - \sqrt{(n-m)^2 + 4r - 4} \big) \big\rceil}\Big\} \text{ and } \\ \label{eq:projIneq2}
      \big\|P\big\|_{S(k)} & \geq \frac{\min\{m,n\} - k}{mn(\min\{m,n\}-1)}\left( r + \sqrt{\frac{mnr - r^2}{mn-1}} \right) + \frac{k-1}{\min\{m,n\}-1}.
    \end{align}
\end{thm}
\begin{proof}
  To prove Inequality~\eqref{eq:projIneq1}, let $k \leq p \leq \min\{m,n\}$ and notice that Theorem~\ref{thm:CMW08} implies that $\big\| P \big\|_{S(p)} = 1$ whenever ${\rm rank}(P) \geq (n-p)(m-p) + 1$. Solving this inequality for $p$ gives
  \[
    p \geq \frac{1}{2}\Big( n + m - \sqrt{(n-m)^2 + 4{\rm rank}(P) - 4} \Big).
  \]
Thus, choose $p = \max\Big\{ k, \Big\lceil \frac{1}{2}\big( n + m - \sqrt{(n-m)^2 + 4{\rm rank}(P) - 4} \big) \Big\rceil \Big\}$. Then using Theorem~\ref{thm:MatrixEquiv01} shows
  \[
    \big\| P \big\|_{S(k)} \geq \frac{k}{\big\lceil \frac{1}{2}\big( n + m - \sqrt{(n-m)^2 + 4{\rm rank}(P) - 4} \big) \big\rceil}.
  \]

  To show Inequality~\eqref{eq:projIneq2} holds, we first note that the $k = 1$ case follows immediately from Theorem~\ref{thm:s1_norm_trace}. For the $k > 1$ case, use Proposition~\ref{prop:proj_norm_ineq} with $h = 1$.
\end{proof}

Proposition~\ref{prop:proj_norm_ineq} and Theorem~\ref{thm:mainProjRes} are particularly important because we will see that several important problems in quantum information theory could be answered if we were able to compute, or bound tightly, the $S(k)$-norms of projections. Inequality~\eqref{eq:projIneq1} provides the best lower bound we have when ${\rm rank}(P)$ is small or large (e.g., ${\rm rank}(P) \leq m$ or ${\rm rank}(P) \geq (n-1)(m-1)$), but Inequality~\eqref{eq:projIneq2} is much tighter for moderate-rank projections (e.g., when ${\rm rank}(P) \approx \frac{mn}{2}$).

The two special cases of $k = \min\{m,n\}$ and $k = 1$ of Inequality~\eqref{eq:projIneq2} give lower bounds of $1$ and $\frac{1}{mn}\left( r + \sqrt{\frac{mnr - r^2}{mn-1}} \right)$, respectively -- the remaining lower bounds are just the linear interpolation of these two extremal cases. The bounds provided by Inequality~\eqref{eq:projIneq1} and Inequality~\eqref{eq:projIneq2} will be used in Sections~\ref{sec:op_sk_norm_block_pos}. See Figure~\ref{fig:proj_norm_compare} for a more detailed comparison of these inequalities.
\begin{figure}[ht]
\begin{center}
\includegraphics[width=\textwidth]{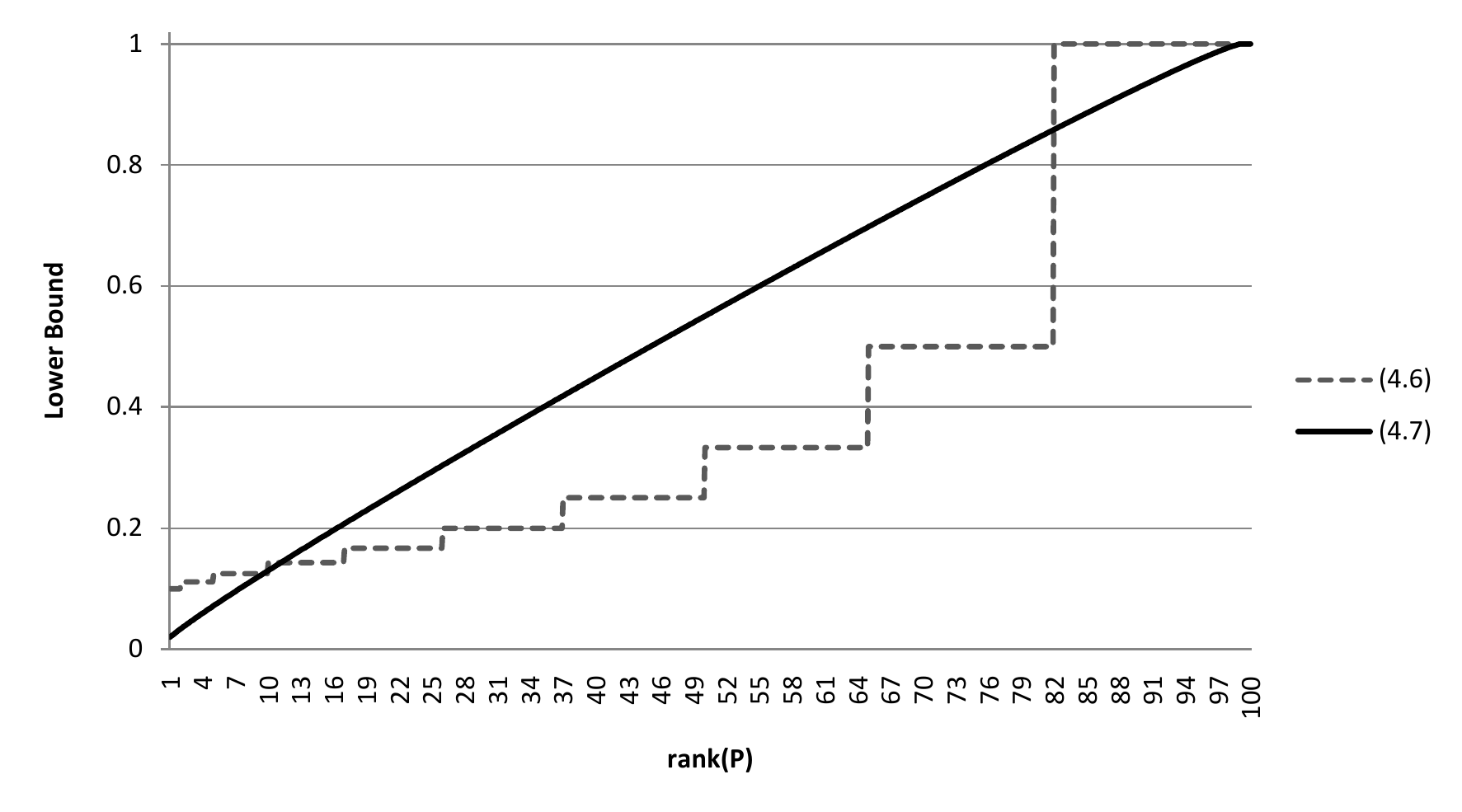}
\end{center}\vspace{-0.25in}
\caption[Comparison of lower bounds for the $S(1)$-norm of projections]{\hsp A comparison of the lower bounds for $\big\|P\big\|_{S(1)}$, where $P$ is an orthogonal projection, provided by Inequalities~\eqref{eq:projIneq1} and~\eqref{eq:projIneq2} in the $m = n = 10$ case. Inequality~\eqref{eq:projIneq1} provides a better lower bound when ${\rm rank}(P)$ is high or low, but Inequality~\eqref{eq:projIneq2} provides a better bound for moderate-rank projections.}\label{fig:proj_norm_compare}
\end{figure}

	We close this section with an inequality that demonstrates how the $S(k)$-norms behave on ``typical'' projections. The proof relies on some methods of convex geometry, and in particular we will make use of Dvoretzky's theorem \cite{M71} and techniques presented in \cite{ASW10}.

\begin{thm}[Dvoretzky]\label{thm:dvoretzky}Let $\mnorm{\cdot}$ be a norm on $\bb{C}^n$ and suppose $b > 0$ is such that $\mnorm{\cdot} \leq b\|\cdot\|$, where $\|\cdot\|$ is the Euclidean norm. Denote $M := \bb{E}\mnorm{X}$, the expectation of $\mnorm{X}$, where $X$ is a random variable uniformly distributed on the unit sphere. Let $\epsilon > 0$ and let $m \leq c\epsilon^2(M/b)^2n$, where $c > 0$ is an appropriate universal constant. Then, for most $m$-dimensional subspaces $E$ (in the sense of the invariant measure on the corresponding Grassmannian) we have
\begin{align*}
	(1 - \epsilon)M\|x\| \leq \mnorm{x} \leq (1 + \epsilon)M\|x\| \quad \forall \, x \in E.
\end{align*}
\end{thm}

Using the above version of Dvoretzky's theorem we can prove the following result, which says that the $S(k)$-norm of orthogonal projections of low rank in $M_n \otimes M_n$ is typically near $k/n$. This result was used in \cite{JK11} to resolve a conjecture of Brand{\~ a}o \cite{B09} in the negative, via an argument of Stanislaw Szarek.

\begin{thm}\label{thm:genProjections}
	There exists a universal constant $C$, independent of $n$ and $k$, such that for most projections $P \in M_n \otimes M_n$ with rank$(P) \leq kn$, we have
	\begin{align*}
		\frac{k}{n} \leq \|P\|_{S(k)} \leq C\frac{k}{n}.
	\end{align*}
\end{thm}
\begin{proof}
	The left inequality is true for all projections simply by Theorem~\ref{thm:MatrixEquiv01}. We will prove the right inequality by making use of Theorem~\ref{thm:dvoretzky}. Let $P \in M_n \otimes M_n$ be an orthogonal projection and use Proposition~\ref{prop:orth_proj_norm} and Theorem~\ref{thm:sk_vector_norm} to write
	\begin{align}\label{eq:Pnorm}
		\sqrt{\|P\|_{S(k)}} = \sup_{\ket{v} \in Range(P)}\Big\{ \sqrt{\sum_{i=1}^k \alpha_i^2} : \{ \alpha_i \} \text{ are the Schmidt coefficients of $\ket{v}$}\Big\}.
	\end{align}
	But now by associating $\bb{C}^n \otimes \bb{C}^n$ with $M_n$, the quantity~\eqref{eq:Pnorm} equals
	\begin{align}\label{eq:Anorm}
		\sup_{A \in R}\Big\{ \sqrt{\sum_{i=1}^k s_i^2(A)} : \|A\|_F = 1, s_1(A) \geq \cdots \geq s_n(A) \geq 0 \text{ are singular values of $A$} \Big\},
	\end{align}
	
	\noindent where $R$ is the subspace of $M_n$ associated with the range of $P$ through the standard bipartite vector to operator isomorphism. So now the goal is to show that there exists a constant $C$ such that $\sqrt{\sum_{i=1}^k s_i^2(A)} \leq C\sqrt{k/n}\|A\|_F$ for $A$ in general subspaces $R$ of dimension $kn$. To this end, we need to bound the constants $b$ and $M$ of Dvoretzky's theorem. It is trivial to see that $\sqrt{\sum_{i=1}^k s_i^2(A)} \leq \|A\|_F$ and that equality is attained for some operators $A$, so $b = 1$.
	
	To upper-bound $M$, recall from \cite{ASW10} that the expectation of the operator norm, $\bb{E}\|A\|$, is upper-bounded by $\frac{C_0}{\sqrt{n}}$ for some absolute constant $C_0$. Thus
	\begin{align*}
		M := \bb{E}\left(\sqrt{\sum_{i=1}^k s_i^2(A)}\right) \leq \bb{E}(\sqrt{k}s_1(A)) = \sqrt{k}\bb{E}\|A\| \leq C_0 \sqrt{\frac{k}{n}}.
	\end{align*}
	It follows via Dvoretzky's theorem that there is a constant $c$ such that if we choose $\epsilon = 1/(C_0 \sqrt{c})$, then for general subspaces $R$ with $\dim(R) \leq c \epsilon^2 C_0^2kn = kn$, we have
	\begin{align*}
		\sqrt{\sum_{i=1}^k s_i^2(A)} \leq (1 + \epsilon)M\|A\|_F \leq (1 + \frac{1}{C_0 \sqrt{c}})C_0\sqrt{\frac{k}{n}}\|A\|_F.
	\end{align*}
\end{proof}

\subsection{Isometries}\label{sec:sk_matrix_isometries}

We will now characterize the isometries of the $S(k)$-operator norm. Recall that in the $k = \min\{m,n\}$ case, the $S(k)$-norm on $M_m \otimes M_n$ is simply the operator norm, and we recall from Section~\ref{sec:linear_preservers_isometries} that the isometries of the operator norms are the maps of the form $\Phi(X) = UXV$ or $\Phi(X) = UX^T V$, where $U$ and $V$ are unitary matrices. We will see shortly that the isometries of the other $S(k)$-norms are similar, but the unitaries $U$ and $V$ must both be of the local form~\eqref{eq:localU}.

Before proceeding, we will need two intermediate results.
\begin{lemma}\label{lem:sepLPPhelp}
	Let $p \geq 3$ and suppose $X_1, X_2, \ldots, X_p \in M_{n,m}$ are rank one. If $X_i + X_j$ is rank one for all $i \neq j$ then $X_1 + X_2 + \cdots + X_p$ is rank one.
\end{lemma}
\begin{proof}
	Begin by writing
	\begin{align*}
		X_i = c_i \ketbra{a_i}{b_i} \quad \text{for some } c_i \in \bb{R}, \ket{a_i} \in \bb{C}^n, \ket{b_i} \in \bb{C}^m.
	\end{align*}
	If $X_1 + X_2$ is rank one then it follows that either $\ket{a_1} \parallel \ket{a_2}$ or $\ket{b_1} \parallel \ket{b_2}$, where we recall from Section~\ref{sec:mul_sep_preservers} that $\ket{a} \parallel \ket{b}$ means that $\ket{a}$ and $\ket{b}$ are linearly dependent. Assume without loss of generality that $\ket{a_1} \parallel \ket{a_2}$, and adjust $\ket{b_2}$ appropriately so that we can write $X_2 = c_2\ketbra{a_1}{b_2}$. Because $X_1 + X_3$ and $X_2 + X_3$ are also rank one, the following two statements are also true:
	\begin{itemize}
		\item $\ket{a_1} \parallel \ket{a_3}$ or $\ket{b_1} \parallel \ket{b_3}$,\\ \hspace*{0.52in} -- and --
		\item $\ket{a_1} \parallel \ket{a_3}$ or $\ket{b_2} \parallel \ket{b_3}$.
	\end{itemize}
	It follows that we have two possibilities: either $\ket{a_1} \parallel \ket{a_3}$ or $\ket{b_1} \parallel \ket{b_2} \parallel \ket{b_3}$. In either case, $X_1 + X_2 + X_3$ is rank one. The proof extends straightforwardly to more than three matrices.
\end{proof}

The following proposition is of independent interest as it can be thought of as a bipartite version of Proposition~\ref{prop:rankKpreserver} in the case of rank-$1$ operators.
\begin{prop}\label{prop:sepLPP}
	Let $k,m,n$ be positive integers such that $1 \leq k < \min\{m, n\}$ and let $\Phi : M_m \otimes M_n \rightarrow M_m \otimes M_n$ be an invertible linear map. Define $\cl{V} \subseteq M_m \otimes M_n$ to be the set of rank-$1$ operators whose row and column space both have Schmidt rank no greater than $k$:
	\begin{align*}
		\cl{V} := \big\{ c\ketbra{v}{w} \in M_m \otimes M_n : c \in \mathbb{R}, SR(\ket{v}),SR(\ket{w}) \leq k \big\}.
	\end{align*}
	Then $\Phi(\cl{V}) \subseteq \cl{V}$ if and only if $\Phi$ can be written as a composition of one or more of the following maps:
	\begin{enumerate}[(a)]
		\item $X \mapsto LXM$, where $L, M \in M_m \otimes M_n$ are invertible operators of the local form~\eqref{eq:localL},
		\item the transpose map $T$, and
		\item if $k = 1$, the partial transpose map $(id_m \otimes T)$.
	\end{enumerate}
\end{prop}
\begin{proof}
	The ``if'' implication of the proposition is trivial, so we focus on the ``only if'' implication. Notice that there is an isomorphism between pure separable states $\ket{x_1} \otimes \ket{x_2} \otimes \ket{y_1} \otimes \ket{y_2} \in \bb{C}^m \otimes \bb{C}^m \otimes \bb{C}^n \otimes \bb{C}^n$ and rank one separable (not necessarily positive) operators $\ket{x_2}\overline{\bra{x_1}} \otimes \ket{y_2}\overline{\bra{y_1}} \in M_m \otimes M_n$. The $k = 1$ case of the result then follows by applying Theorem~\ref{thm:multipartite} and using this isomorphism -- the various swap operators $S_{\sigma}$ on $\bb{C}^m \otimes \bb{C}^m \otimes \bb{C}^n \otimes \bb{C}^n$ correspond on $M_m \otimes M_n$ to the transpose map, partial transpose map, and multiplication on the left and/or right by the swap operator $S$.

	For the case when $k \geq 2$, suppose $\Phi(\cl{V}) \subseteq \cl{V}$ where $\cl{V}$ is as defined in the statement of the theorem. We prove the following two claims:
	\begin{enumerate}[a)]
		\item $\Phi(\ketbra{v}{w})$ is rank one for all $\ket{v},\ket{w}$ with $SR(\ket{w}) \leq k$; and
		\item $\Phi(\ketbra{v}{w})$ is rank one for all $\ket{v},\ket{w}$.
	\end{enumerate}
	Once (b) is established we will know that $\Phi$ must map the set of rank one matrices into itself and so the result follows by Proposition~\ref{prop:rankKpreserver} and Theorem~\ref{thm:bipartite_schmidt_preserver}.
	
	We first prove (a). Take any arbitrary states $\ket{v},\ket{w}$ with $SR(\ket{w}) \leq k$. Write
	\begin{align*}
		\ket{v} = \sum_{i=1}^n \alpha_i \ket{v_i} \ \ \text{ with } \alpha_i \in \bb{C}, SR(\ket{v_i}) = 1 \ \ \forall \, i.
	\end{align*}
	For any $i \neq j$, $\alpha_i \ketbra{v_i}{w} + \alpha_j \ketbra{v_j}{w} \in \cl{V}$ and so $\Phi(\alpha_i \ketbra{v_i}{w}) + \Phi(\alpha_j \ketbra{v_j}{w}) \in \cl{V}$ as well and hence it must be rank one. It follows from Lemma~\ref{lem:sepLPPhelp} that
	\begin{align*}
		\Phi(\ketbra{v}{w}) = \Phi( \alpha_1 \ketbra{v_1}{w}) + \cdots + \Phi( \alpha_n \ketbra{v_n}{w})
	\end{align*}
	is rank one as well, which establishes (a). Now take any arbitrary states $\ket{v},\ket{w}$ (not necessarily with Schmidt rank at most $k$) and write
	\begin{align*}
		\ket{w} = \sum_{i=1}^n \beta_i \ket{w_i} \ \ \text{ with } \beta_i \in \bb{C}, SR(\ket{w_i}) = 1 \ \ \forall \, i.
	\end{align*}
	For any $i \neq j$, $\Phi\big(\ket{v}(\overline{\beta_i}\bra{w_i} + \overline{\beta_j}\bra{w_j})\big) = \Phi(\overline{\beta_i}\ketbra{v}{w_i}) + \Phi(\overline{\beta_j}\ketbra{v}{w_j})$ is rank one by (a). It follows from Lemma~\ref{lem:sepLPPhelp} that
	\begin{align*}
		\Phi(\ketbra{v}{w}) = \Phi( \overline{\beta_1} \ketbra{v}{w_1}) + \cdots + \Phi( \overline{\beta_n} \ketbra{v}{w_n})
	\end{align*}
	is rank one as well. Claim~(b) follows and the proof is complete.
\end{proof}

We are now in a position to prove the main result of this section.
\begin{thm}\label{thm:opNormIso}
	Let $1 \leq k < \min\{m,n\}$ and $\Phi : M_m \otimes M_n \rightarrow M_m \otimes M_n$. Then $\big\|\Phi(X)\big\|_{S(k)} = \big\|X\big\|_{S(k)}$ for all $X \in M_m \otimes M_n$ if and only if $\Phi$ can be written as a composition of one or more of the following maps:
	\begin{enumerate}[(a)]
		\item $X \mapsto UXV$, where $U$ and $V$ are unitaries of the local form~\eqref{eq:localU},
		\item the transpose map $T$, and
		\item if $k = 1$, the partial transpose map $(id_m \otimes T)$.
	\end{enumerate}
\end{thm}
\begin{proof}
	Again, the ``if'' implication is trivial. For the ``only if'' implication, we first use Proposition~\ref{prop:subgroup_lift} along with the vector-operator isomorphism to show that any map that preserves the $S(k)$-norm also preserves the Frobenius norm $\|\cdot\|_{F}$. We then show that these maps must send rank $1$ operators to rank $1$ operators, and finally we use Theorem~\ref{thm:bipartite_schmidt_preserver} to pin down the result.
	
	We begin in much the same way as in the proof of Theorem~\ref{thm:geometric_measure_entanglement} by defining
	\begin{align*}
		\cl{G} := \big\{ \Phi : M_m \otimes M_n \rightarrow M_m \otimes M_n : \big\|\Phi(X)\big\|_{S(k)} = \big\|X\big\|_{S(k)} \text{ for all } X \big\}.
	\end{align*}
	Clearly $\cl{G}$ is bounded because it is the set of isometries under the norm $\|\cdot\|_{S(k)}$ and all norms on a finite-dimensional space are equivalent. Additionally, $\cl{G}$ contains the subgroup of unitary maps
	\begin{align*}
		\cl{G}_S := \big\{ \Phi \in \cl{G} : \Phi(X) = (U_1 \otimes U_2)X(V_1 \otimes V_2) \text{ for some unitaries } U_1,U_2,V_1,V_2 \big\}.
	\end{align*}
	To see that $\cl{G}_S$ is irreducible, recall that the recall that the unitary group $U(n)$ spans all of $M_n$, so if we fix $U_2,V_1,V_2$ then we can find maps in $\cl{G}_S$ that span the space of operators of the form
	\begin{align*}
		\Phi(X) = (A \otimes U_2)X(V_1 \otimes V_2) \text{ for some } A \in M_m \text{ and unitaries } U_2,V_1,V_2.
	\end{align*}
	Similarly, we can obtain any map of the form $\Phi(X) = (A \otimes B)X(C \otimes D)$ in the span of $\cl{G}_S$, where $A,B,C,D$ are arbitrary. Operators of the form $A \otimes B$ span all of $M_m \otimes M_n$, so the span of $\cl{G}_S$ actually contains all maps of the form $\Phi(X) = EXF$ and hence all maps of the form $\Phi(X) = \sum_iE_iXF_i$. Since all linear maps can be written in this form, it follows that $\cl{G}_S$ spans the entire space of linear maps and hence is irreducible. By Proposition~\ref{prop:subgroup_lift} and the vector-operator isomorphism it follows that $\cl{G}$ is contained in the unitary group and so if $\big\|\Phi(X)\big\|_{S(k)} = \big\|X\big\|_{S(k)}$ for all $X$, then $\big\|\Phi(X)\big\|_{F} = \big\|X\big\|_{F}$ for all $X$ as well.
	
	We will now consider how an isometry $\Phi$ of the $S(k)$-norm acts on rank-$1$ operators. In particular, let $\ket{v},\ket{w} \in \bb{C}^m \otimes \bb{C}^n$ with $SR(\ket{v}),SR(\ket{w}) \leq k$. Then
	\begin{align*}
		1 = \big\| \ketbra{v}{w} \big\|_{F} = \big\| \ketbra{v}{w} \big\|_{S(k)} = \big\| \Phi(\ketbra{v}{w}) \big\|_{F} = \big\| \Phi(\ketbra{v}{w}) \big\|_{S(k)}.
	\end{align*}
	However, $\big\|X\big\|_{S(k)} \leq \big\|X\big\| \leq \big\|X\big\|_{F}$ for all $X$ and $\big\|X\big\| = \big\|X\big\|_{F}$ if and only if $X$ has rank $1$. In this case, $\big\|X\big\|_{S(k)} = \big\|X\big\|$ if and only if there exist $\ket{x},\ket{y}$ with $SR(\ket{x}),SR(\ket{y}) \leq k$ such that $X = \ketbra{x}{y}$. Thus $\Phi(\ketbra{v}{w}) = \ketbra{x}{y}$ for some $\ket{x},\ket{y}$ with $SR(\ket{x}),SR(\ket{y}) \leq k$. Proposition~\ref{prop:sepLPP} then applies to $\Phi$ (the fact that $\Phi$ is invertible follows from it being an isometry). To finish the proof, simply note that $\Phi$ being unitary implies that the operators $L$ and $M$ of Proposition~\ref{prop:sepLPP} must be unitary.
\end{proof}

\subsection{Spectral Tests for Block Positivity}\label{sec:op_sk_norm_block_pos}

	In this section we derive a set of conditions for testing when a Hermitian operator is and is not $k$-block positive based on its eigenvalues and eigenvectors. Equivalently, we derive conditions for testing when a superoperator is $k$-positive based on its generalized Choi--Kraus operators. The tests derived here generalize several known tests for $k$-positivity.

	Throughout this section, if $X = X^\dagger$ then we will denote the positive eigenvalues of $X$ by $\{\lambda^{+}_i\}$ and the corresponding eigenvectors by $\{\ket{v^{+}_i}\}$. We will similarly denote the negative eigenvalues by $\{\lambda^{-}_i\}$ and the corresponding eigenvectors by $\{\ket{v^{-}_i}\}$, and the eigenvectors corresponding to the zero eigenvalues by $\{\ket{v^{0}_i}\}$. We also define $X^{+} := \sum_i \lambda^{+}_i \ketbra{v^{+}_i}{v^{+}_i} \geq 0$ and $X^{-} := \sum_i \lambda^{-}_i \ketbra{v^{-}_i}{v^{-}_i} \leq 0$ to be the positive and negative parts of $X$, respectively. Similarly, $P_X^{0} := \sum_i \ketbra{v^{0}_i}{v^{0}_i}$ and $P_X^{-} := \sum_i \ketbra{v^{-}_i}{v^{-}_i}$ denote the projections onto the nullspace and negative part of $X$, respectively.
\begin{thm}\label{thm:kposSpectral}
  Let $X = X^\dagger \in M_m \otimes M_n$.
  \begin{enumerate}[(a)]
    \item If $\big\| P_X^{-} \big\|_{S(k)} = 1$ then $X$ is not $k$-block positive.
    \item If $\big\| P_X^{0} + P_X^{-} \big\|_{S(k)} < 1$ and $\lambda_i^+ \geq \frac{\| X^{-} \|_{S(k)}}{1 - \| P_X^0 + P_X^{-} \|_{S(k)}}$ for all $i$, then $X$ is $k$-block positive.
    \item If $\big\| P_X^{-} \big\|_{S(k)} < 1$, all of the negative eigenvalues are equal, $X$ is nonsingular, and $\lambda_i^+ < \frac{\| X^{-} \|_{S(k)}}{1 - \| P_X^{-} \|_{S(k)}}$ for all $i$, then $X$ is not $k$-block positive.
    \end{enumerate}
\end{thm}
\begin{proof}
  To see statement (a), observe that there must be a vector $\ket{v} \in {\rm Range}(P_X^{-})$ such that $SR(\ket{v}) \leq k$. It follows that $\bra{v}X\ket{v} = \bra{v}X^{-}\ket{v} < 0$ and so $X$ is not $k$-block positive.

	To see statement (b), let $\ket{v}$ be such that $SR(\ket{v}) \leq k$ and define $\mu := \frac{\| X^{-} \|_{S(k)}}{1 - \| P_X^0 + P_X^{-} \|_{S(k)}}$. Then, using the spectral decomposition for $X^+$, the definition of the $S(k)$-operator norm, and the hypotheses of (b), we have
  \begin{align*}
    \bra{v}X\ket{v} & = \bra{v}X^{+}\ket{v} - \big|\bra{v}X^{-}\ket{v}\big| \\
    & \geq \sum_i \lambda_i^{+} |\braket{v}{v_i^{+}}|^2 - \big\|X^{-}\big\|_{S(k)} \\
    & \geq \mu \sum_i |\braket{v}{v_i^{+}}|^2 - \big\|X^{-}\big\|_{S(k)} \\
    & \geq \mu (1 - \| P_X^0 + P_X^{-} \|_{S(k)}) - \big\|X^{-}\big\|_{S(k)} \\
    & = 0,
  \end{align*}
  so $X$ is $k$-block positive.

  To see statement (c), observe that the set of unit vectors $\ket{v}$ with $SR(\ket{v}) \leq k$ is compact and so there exists a particular $\ket{v}$ with $SR(\ket{v}) \leq k$ such that $\big|\bra{v}X^{-}\ket{v}\big| = \big\|X^{-}\big\|_{S(k)}$. Define $\mu := \frac{\| X^{-} \|_{S(k)}}{1 - \| P_X^{-} \|_{S(k)}}$. Then similarly we have
  \begin{align*}
    \bra{v}X\ket{v} & = \bra{v}X^{+}\ket{v} - \big|\bra{v}X^{-}\ket{v}\big| \\
    & = \sum_i \lambda_i^{+} |\braket{v}{v_i^{+}}|^2 - \big\|X^{-}\big\|_{S(k)} \\
    & < \mu \sum_i |\braket{v}{v_i^{+}}|^2 - \big\|X^{-}\big\|_{S(k)} \\
    & = \mu (1 - \| P_X^{-} \|_{S(k)}) - \big\|X^{-}\big\|_{S(k)} \\
    & = 0,
  \end{align*}
	so $X$ is not $k$-block positive.
\end{proof}

On its face, Theorem~\ref{thm:kposSpectral} appears to be a very technical result that may not be of much use due to the difficulty of computing the $S(k)$-operator norms. However, it is not difficult to derive computable corollaries from it. In fact, it implies a wide array of previously-known and new tests for $k$-block positivity and $k$-positivity of linear maps. These consequences are presented below.

We first show that Theorem~\ref{thm:kposSpectral} implies the $k$-positivity results of \cite{CK09,CK11}:
\begin{cor}\label{cor:kPosSpectral_cor1}
	Let $\Phi : M_m \rightarrow M_n$ be a Hermiticity-preserving linear map represented via the canonical generalized Choi--Kraus representation $\Phi(X) = \sum_{i=1}^a \lambda_i^{+} A_i X A_i^\dagger + \sum_{i=1}^b \lambda_i^{-} B_i X B_i^\dagger$, with the set $\big\{A_1, \ldots, A_a, B_1, \ldots, B_b\big\}$ forming an orthonormal set in the Hilbert--Schmidt inner product, and $\lambda_i^{+} > 0$ and $\lambda_i^{-} < 0$ for all $i$. Furthermore, let $\{C_i\}$ be a set of operators that make $\big\{A_1, \ldots, A_a, B_1, \ldots, B_b, C_1, \ldots, C_{mn-a-b}\big\}$ a full orthonormal basis.
  \begin{enumerate}[(a)]
    \item Suppose that $\sum_i \big\|B_i\big\|_{(k,2)}^2 + \sum_i \big\|C_i\big\|_{(k,2)}^2 < 1$ and
    \begin{align*}
    	\lambda_j^+ \geq \frac{\sum_i \lambda_i^{-}\big\|B_i\big\|_{(k,2)}^2}{1 - \sum_i \big\|B_i\big\|_{(k,2)}^2 - \sum_i \big\|C_i\big\|_{(k,2)}^2} \text{ for all $j$}.
    \end{align*}
    Then $\Phi$ is $k$-positive.
    \item Suppose that $a = mn-1, b = 1, \big\|B_1\big\|_{(k,2)}^2 < 1$, and
    \begin{align*}
    	\lambda_j^+ < \frac{\lambda_1^{-}\big\|B_1\big\|_{(k,2)}^2}{1 - \big\|B_1\big\|_{(k,2)}^2} \text{ for all $j$}.
    \end{align*}
    Then $\Phi$ is not $k$-positive.
  \end{enumerate}
\end{cor}
\begin{proof}
	To see condition (a), let $X$ be the Choi matrix of $\Phi$ and use condition (b) of Theorem~\ref{thm:kposSpectral} together with Proposition~\ref{prop:matrixVectorNorms} to see that if $\sum_i \big\|\ket{v_i^{-}}\big\|_{s(k)}^2 + \sum_i \big\|\ket{v_i^{0}}\big\|_{s(k)}^2 < 1$ and
	\begin{align}\label{eq:kpos_eq_1}
		\lambda_j^+ \geq \frac{\sum_i \lambda_i^{-}\big\|\ket{v_i^{-}}\big\|_{s(k)}^2}{1 - \sum_i \big\|\ket{v_i^{-}}\big\|_{s(k)}^2 - \sum_i \big\|\ket{v_i^{0}}\big\|_{s(k)}^2} \text{ for all $j$,}
	\end{align}
	then $X$ is $k$-block positive (i.e., $\Phi$ is $k$-positive). A slight modification of the proof of Theorem~\ref{thm:choi_cp} shows that the Choi--Kraus operators of $\Phi$ satisfy ${\rm vec}(B_i) = \ket{v_i^{-}}$ and ${\rm vec}(C_i) = \ket{v_i^{0}}$. Thus $\big\|B_i\big\|_{(k,2)} = \big\|\ket{v_i^{-}}\big\|_{s(k)}$ and $\big\|C_i\big\|_{(k,2)} = \big\|\ket{v_i^{0}}\big\|_{s(k)}$, which allows us to rewrite Equation~\eqref{eq:kpos_eq_1} as saying that $\sum_i \big\|B_i\big\|_{(k,2)}^2 + \sum_i \big\|C_i\big\|_{(k,2)}^2 < 1$ and
	\begin{align*}
		\lambda_j^+ \geq \frac{\sum_i \lambda_i^{-}\big\|B_i\big\|_{(k,2)}^2}{1 - \sum_i \big\|B_i\big\|_{(k,2)}^2 - \sum_i \big\|C_i\big\|_{(k,2)}^2} \text{ for all $j$}.
	\end{align*}
	Statement (a) follows immediately. The proof of statement (b) follows similarly by using condition (c) of Theorem~\ref{thm:kposSpectral}.
\end{proof}

Theorem~\ref{thm:kposSpectral} also implies the $k$-positivity test of \cite{TT83} and the positivity test of \cite{BFP04}, as it was noted in \cite{CK09} that condition (a) of Corollary~\ref{cor:kPosSpectral_cor1} implies those results.

We note that in \cite{CKS09} it was shown that the $k$-positivity tests of Corollary~\ref{cor:kPosSpectral_cor1} cannot be used to find entanglement
witnesses that detect non-positive partial transpose states, and thus are not useful for trying to determine whether NPPT bound entangled states exist. We will see in Section~\ref{sec:bound_entanglement} that the more general Theorem~\ref{thm:kposSpectral} likely \emph{is} strong enough to detect bound entanglement.

Corollary~\ref{cor:kPosSpectral_cor1} was proved by seeing what conditions (b) and (c) of Theorem~\ref{thm:kposSpectral} say about the generalized Choi--Kraus operators of a linear map. The following corollary, originally proved in \cite{KS05}, shows what condition (a) of Theorem~\ref{thm:kposSpectral} says in the same situation.
\begin{cor}
	Let $\Phi : M_m \rightarrow M_n$ be a Hermiticity-preserving linear map represented via the canonical generalized Choi--Kraus representation $\Phi(X) = \sum_{i=1}^a \lambda_i^{+} A_i X A_i^\dagger + \sum_{i=1}^b \lambda_i^{-} B_i X B_i^\dagger$, with the set $\big\{A_1, \ldots, A_a, B_1, \ldots, B_b\big\}$ forming an orthonormal set in the Hilbert--Schmidt inner product, and $\lambda_i^{+} > 0$ and $\lambda_i^{-} < 0$ for all $i$. If ${\rm rank}(B_i) \leq k$ for some $i$, then $\Phi$ is not $k$-positive. 
\end{cor}
\begin{proof}
  Just like in the proof of Corollary~\ref{cor:kPosSpectral_cor1}, recall that the generalized Choi--Kraus operators $B_i$ can be scaled so that ${\rm vec}(B_i) = \ket{v_i^{-}}$, where $\big\{\ket{v_i^{-}}\big\}$ is the set of eigenvectors corresponding to the negative eigenspace of $C_\Phi$. Thus ${\rm rank}(B_i) = SR(\ket{v_i^{-}})$. If $SR(\ket{v_i^{-}}) \leq k$ for some $i$ then $\big|\bra{v_i^{-}}P_X^{-}\ket{v_i^{-}}\big| = 1$, so $\big\|P_X^{-}\big\|_{S(k)} = 1$. Condition (a) of Theorem~\ref{thm:kposSpectral} then gives the result.
\end{proof}

The next corollary, which provides a tight bound on the maximum number of negative eigenvalues that a $k$-block positive operator can have, appeared in \cite{Sar08,SS10}, though we expect that this bound was fairly well-known even earlier.
\begin{cor}\label{cor:maxKboxNegEval}
  If $X = X^\dagger \in M_m \otimes M_n$ is $k$-block positive then it has at most $(n - k)(m - k)$ negative eigenvalues.
\end{cor}
\begin{proof}
  Suppose $X$ has more than $(n - k)(m - k)$ negative eigenvalues. Then by Theorem~\ref{thm:CMW08} it follows that there exists $\ket{v} \in {\rm Range}(P_X^{-})$ with $SR(\ket{v}) \leq k$. Hence we have $\big\| P_X^{-} \big\|_{S(k)} = 1$ and so condition (a) of Theorem~\ref{thm:kposSpectral} tells us that $X$ is not $k$-block positive.
\end{proof}
The fact that there exist $k$-block positive operators in $M_m \otimes M_n$ with $(n-k)(m-k)$ negative eigenvalues is easily seen from the tightness of the bound provided by Theorem~\ref{thm:CMW08}. If $P$ is the projection onto a subspace of dimension $(n-k)(m-k)$ that consists entirely of vectors with Schmidt rank higher than $k$, then the operator $(1 - \varepsilon)I - P$ is $k$-block positive when $\varepsilon > 0$ is sufficiently small.

The following corollary shows just how negative the negative eigenvalues of a $k$-block positive operator can be.
\begin{cor}\label{cor:mostNegEval}
  Suppose $X = X^\dagger \in M_m \otimes M_n$ is $k$-block positive. Denote the maximal and minimal eigenvalues of $X$ by $\lambda_{\textup{max}}$ and $\lambda_{\textup{min}}$, respectively, and let $\ket{v_{\textup{min}}}$ be an eigenvector corresponding to the eigenvalue $\lambda_{\textup{min}}$. Then
  \begin{align*}
    \frac{\lambda_{\textup{min}}}{\lambda_{\textup{max}}} \geq 1 - \frac{1}{\big\| \ket{v_{\textup{min}}} \big\|_{s(k)}^2} \geq 1 - \frac{\min\{m,n\}}{k}.
  \end{align*}
\end{cor}
\begin{proof}
  If $\lambda_{\textup{min}} \geq 0$ then the result is trivial. We thus assume that $\lambda_{\textup{min}} < 0$. Suppose without loss of generality that $X$ has only one negative eigenvalue and is nonsingular (certainly the case when $\lambda_{\textup{min}}$ is minimal occurs when there is just one negative eigenvalue and all other eigenvalues equal $\lambda_{\textup{max}}$). If $X$ is $k$-block positive then condition (a) of Theorem~\ref{thm:kposSpectral} says that $\big\|P_X^{-}\big\|_{S(k)} < 1$. Condition (c) then says that
  \begin{align*}
    \lambda_{\textup{max}} \geq \frac{\big\| X^{-} \big\|_{S(k)}}{1 - \big\| P_X^{-} \big\|_{S(k)}} = -\lambda_{\textup{min}}\frac{\big\| P_X^{-} \big\|_{S(k)}}{1 - \big\| P_X^{-} \big\|_{S(k)}}.
  \end{align*}
  Then
  \begin{align*}
    \frac{\lambda_{\textup{min}}}{\lambda_{\textup{max}}} \geq \frac{\big\| P_X^{-} \big\|_{S(k)} - 1}{\big\| P_X^{-} \big\|_{S(k)}} = 1 - \frac{1}{\big\| \ket{v_{\textup{min}}} \big\|_{s(k)}^2} \geq 1 - \frac{\min\{m,n\}}{k},
  \end{align*}
  with the final inequality following from Corollary~\ref{cor:vecEquiv}.
\end{proof}
The fact that there exist $k$-block positive operators in $M_m \otimes M_n$ with $\frac{\lambda_{\textup{min}}}{\lambda_{\textup{max}}} = 1 - \frac{\min\{m,n\}}{k}$ can be seen by letting $\ket{\psi_{+}} = \frac{1}{\sqrt{m}}\sum_{i=1}^{m} \ket{i} \otimes \ket{i}$ be the standard maximally-entangled state. Then $I - \frac{\min\{m,n\}}{k}\ketbra{\psi_{+}}{\psi_{+}}$ is $k$-block positive by Corollary~\ref{cor:kPosInf1}.

By using Theorem~\ref{thm:mainProjRes} in the proof of Corollary~\ref{cor:mostNegEval}, we can derive the following bounds that in some sense interpolate between Corollary~\ref{cor:maxKboxNegEval} and Corollary~\ref{cor:mostNegEval}, giving lower bounds on $\lambda_{\textup{min}}$ that depend on the number of negative eigenvalues of $X$.
\begin{cor}\label{cor:kpos_eig_strong}
  Suppose $X = X^\dagger \in M_m \otimes M_n$ is $k$-block positive with $r$ negative eigenvalues. Denote the maximal and minimal eigenvalues of $X$ by $\lambda_{\textup{max}}$ and $\lambda_{\textup{min}}$, respectively. Then
  \begin{align*}
    \frac{\lambda_{\textup{min}}}{\lambda_{\textup{max}}} & \geq 1 - \frac{\big\lceil \frac{1}{2}\big( n + m - \sqrt{(n-m)^2 + 4{r - 4}} \big) \big\rceil}{k} \ \ \text{ and} \\
    \frac{\lambda_{\textup{min}}}{\lambda_{\textup{max}}} & \geq 1 - \frac{mn(\min\{m,n\}-1)}{mn(k-1) + (\min\{m,n\} - k)\left( r + \sqrt{\frac{mnr - r^2}{mn-1}} \right)}
    \end{align*}
\end{cor}
In general, the first inequality of Corollary~\ref{cor:kpos_eig_strong} is stronger when $r$ is small or large (i.e., close to $1$ or $mn$), while the second inequality is stronger when $r$ is intermediate (i.e., close to $mn/2$). The following example makes use of the second inequality when $m = n$ and $r = n(n-1)/2$.
\begin{exam}\label{ex:trans_not_2pos}{\rm
	Recall from Example~\ref{exam:transpose_not_cp} that the transpose map $T : M_n \rightarrow M_n$ is positive but not $2$-positive. Here we use Corollary~\ref{cor:kpos_eig_strong} to provide another proof that $T$ is not $2$-positive based solely on the eigenvalues of its Choi matrix.
	
	Recall from Section~\ref{sec:symmetric_sub} that the Choi matrix of $T$ is the swap operator $S$, which is a unitary with $n(n+1)/2$ eigenvalues equal to $1$ and $n(n-1)/2$ eigenvalues equal to $-1$. Let's assume that $T$ is $2$-positive, so $S$ is $2$-block positive. If we apply the second inequality of Corollary~\ref{cor:kpos_eig_strong} with $r = n(n-1)/2$ and $k = 2$, we have
	\begin{align*}
		-1 = \frac{\lambda_{\textup{min}}}{\lambda_{\textup{max}}} & \geq 1 - \frac{n^2(n-1)}{n^2 + (n-2)\left( \frac{n(n-1)}{2} + \sqrt{\frac{n^3(n-1)/2 - n^2(n-1)^2/4}{n^2-1}} \right)} = -1 + \frac{2}{n} > -1.
	\end{align*}
	This gives a contradiction and shows that $T$ is not $2$-positive for any $n$.
}\end{exam}

One final corollary shows that we now have a complete spectral characterization of the $k$-block positivity of Hermitian operators with exactly two distinct eigenvalues.
\begin{cor}\label{cor:kposTwoEvals}
  Let $X = X^\dagger \in M_m \otimes M_n$ have two distinct eigenvalues $\lambda_1 > \lambda_2$. Then $X$ is $k$-block positive if and only if
  \begin{align}\label{eq:kPosIneq}
    \big\| P_X^{-} \big\|_{S(k)} \leq \frac{\lambda_1}{\lambda_1 - \lambda_2}.
  \end{align}
\end{cor}
\begin{proof}
  If $\lambda_1$ and $\lambda_2$ have the same sign then the result is trivial. We thus assume that $\lambda_1 > 0$ and $\lambda_2 < 0$.

  If $X$ is $k$-block positive, then by condition (a) of Theorem~\ref{thm:kposSpectral} we know that $\| P_X^{-} \|_{S(k)} < 1$. Then condition (c) says that
  \begin{align*}
    \lambda_1 \geq \frac{\big\| X^{-} \big\|_{S(k)}}{1 - \big\| P_X^{-} \big\|_{S(k)}} = -\lambda_2\frac{\big\| P_X^{-} \big\|_{S(k)}}{1 - \big\| P_X^{-} \big\|_{S(k)}}.
  \end{align*}
	The desired inequality follows easily. To see the other direction of the proof, suppose inequality~\eqref{eq:kPosIneq} is satisfied. Then because $\lambda_2 < 0$ it follows that $\big\| P_X^{-} \big\|_{S(k)} < 1$. Simple algebra now shows that condition (b) of Theorem~\ref{thm:kposSpectral} is satisfied.
\end{proof}

Note that Corollaries~\ref{cor:maxKboxNegEval}, \ref{cor:mostNegEval}, and~\ref{cor:kpos_eig_strong} all provide necessary conditions for $k$-block positivity of $X$ that depend only on its eigenvalues. A natural question that can be asked at this point is for a complete characterization of the possible eigenvalues of $k$-block positive matrices. That is, what is the structure of the set $\Lambda \subset \mathbb{R}^{mn}$ with the property that $\mathbf{\lambda} \in \Lambda$ if and only if there is a $k$-block positive matrix $X = X^\dagger \in M_m \otimes M_n$ with eigenvalues that are the entries of $\mathbf{\lambda}$? Compare this problem with the related open problem that asks for the structure of the set $\Lambda \subset \mathbb{R}^{mn}$ with the property that $\mathbf{\lambda} \in \Lambda$ if and only if every matrix $X = X^\dagger \in M_m \otimes M_n$ with eigenvalues that are the entries of $\mathbf{\lambda}$ is necessarily separable \cite{OpenProb15}, which was partially solved in \cite{GB02,VAD01}.

A weaker version of our question would be to ask for a tight version of Corollary~\ref{cor:kpos_eig_strong}. That is, what is the minimal value of $\lambda_{\textup{min}} / \lambda_{\textup{max}}$ for a $k$-block positive matrix? We don't have an answer for either of these questions, but we close by noting that Proposition~\ref{prop:block_pos_trace} provides a spectral test for block positivity that, unlike the other results of this section, is independent of Theorem~\ref{thm:kposSpectral}.

\chapter{Computational Problems and Applications}\label{ch:computation}

In this chapter we investigate various methods of computing the $S(k)$-operator norms and we use these methods to tackle several problems in quantum information theory. Although all of this work applies immediately to the problems of characterizing separable states, states with a given Schmidt number, and block positive operators, this chapter focuses on slightly more exotic and unexpected applications in quantum information theory. In particular, we use the $S(k)$-norms to investigate whether or not there exist bound entangled states that have non-positive partial transpose and to compute the minimum gate fidelity of a quantum channel. We also connect the $S(k)$-norms to the maximum output purity of a quantum channel and the tripartite and quadripartite geometric measure of entanglement, which allows all of our results to apply immediately in these varied settings.

\section{Semidefinite and Conic Programming}\label{sec:SP}

We begin by introducing the reader to the theory of semidefinite programs (SDPs) and conic programs, which are the main tools used throughout this chapter for bounding the $S(k)$-operator norms. Our introduction is brief and suited to our particular purposes -- for a more general and in-depth introduction and discussion, the reader is encouraged to read any of a number of other sources including \cite{Al95,Kl02,Lo03,VB94,WSV00}. Importantly, there are explicit methods that are able to approximately solve semidefinite programs to any desired accuracy in polynomial time \cite{GLS93}. We provide several examples of semidefinite programs throughout this section, some of which we solve analytically and some of which we solve numerically. For numerical solutions, we use the YALMIP modelling language \cite{YALMIP} and the SeDuMi solver \cite{SeDuMi} in MATLAB to carry out the computations.

For our purposes, assume we have a Hermiticity-preserving linear map $\Phi : M_m \rightarrow M_n$, two operators $A \in M_m$ and $B \in M_n$, and a closed convex cone $\cl{C} \subseteq M_n^+$. Then the \emph{conic program} associated with $\Phi$, $A$, $B$, and $\cl{C}$ is defined by the following pair of optimization problems:
\begin{align}\label{sp:form}\hsp
\begin{matrix}
\begin{tabular}{r l c r l}
\multicolumn{2}{c}{{\bf Primal problem}} & \quad \quad \quad & \multicolumn{2}{c}{{\bf Dual problem}} \\
\text{maximize:} & $\Tr(AX)$ & \quad & \text{minimize:} & $\Tr(BY)$ \\
\text{subject to:} & $B - \Phi(X) \in \cl{C}$ & \quad & \text{subject to:} & $\Phi^\dagger(Y) \geq A$ \\
\ & $X \geq 0$ & \ & \ & $Y \in \cl{C}^{\circ}$ \\
\end{tabular}
\end{matrix}
\dsp\end{align}

In the case when $\cl{C} = M_n^+$, the program \eqref{sp:form} is called a \emph{semidefinite program}, which is the case that will be of most use for us. Recall that the cone of positive semidefinite operators is its own dual cone, so in this case we have $\cl{C}^\circ = M_n^+$ as well. The form~\eqref{sp:form} differs from the standard form of semidefinite programs, but it is equivalent and better suited to our particular needs, and has been used very recently to solve other problems in quantum information \cite{Wa09,JJUW10}. The more general conic form will be used in Section~\ref{sec:coneNorms}, and the interested reader is directed to \cite{BV04} for a more thorough introduction to conic programming.

The \emph{primal feasible} set $\cl{A}$ and \emph{dual feasible} set $\cl{B}$ are defined by
\begin{align*}
	\cl{A} := \big\{ X \geq 0 : B - \Phi(X) \in \cl{C} \big\} \quad \quad \text{ and } \quad \quad \cl{B} := \big\{ Y \in \cl{C}^{\circ} : \Phi^\dagger(Y) \geq A \big\}.
\end{align*}
The \emph{optimal values} associated with the primal and dual problems are defined to be
\begin{align*}
	\alpha := \sup_{X \in \cl{A}} \big\{ \Tr(AX) \big\} \quad \quad \text{ and } \quad \quad \beta := \inf_{Y \in \cl{B}} \big\{ \Tr(BY) \big\},
\end{align*}
and if $\cl{A}$ or $\cl{B}$ is empty then we set $\alpha = -\infty$ or $\beta = \infty$, respectively. The functions being optimized ($\Tr(AX)$ and $\Tr(BY)$) are called the \emph{objective functions}.

Semidefinite and conic programming have a strong theory of duality. Weak duality tells us it is always the case that $\alpha \leq \beta$. Equality is actually attained for many conic programs of interest though, as the following theorem shows.
\begin{thm}[Strong duality]\label{thm:spSlater}
  The following two implications hold for every conic program of the form~\eqref{sp:form}.
  \begin{enumerate}[1.]
		\item Strict primal feasibility: If $\beta$ is finite and there exists an operator $X > 0$ such that $B - \Phi(X)$ is in the interior of $\cl{C}$, then $\alpha = \beta$ and there exists $Y \in \cl{B}$ such that $\Tr(BY) = \beta$.
    \item Strict dual feasibility: If $\alpha$ is finite and there exists an operator $Y$ in the interior of $\cl{C}^\circ$ such that $\Phi^\dagger(Y) > A$, then $\alpha = \beta$ and there exists $X \in \cl{A}$ such that $\Tr(AX) = \alpha$.
  \end{enumerate}
\end{thm}
There are other conditions that imply strong duality, but the conditions of Theorem~\ref{thm:spSlater} (which are known as \emph{Slater-type conditions}) will be sufficient for our needs.
\begin{exam}\label{exam:basic_sdp}{\rm
	Consider the semidefinite program associated with the matrices $\sspp A = \begin{bmatrix}1 & 1 \\ 1 & 1\end{bmatrix}\dsp$ and $\sspp B = \begin{bmatrix}2 & -1 \\ -1 & 2\end{bmatrix}\dsp$ and the linear map $\Phi : M_2 \rightarrow M_2$ defined by
	\begin{align*}\sspp
		\Phi\left(\begin{bmatrix}a & b \\ c & d\end{bmatrix}\right) = \begin{bmatrix}a+d & 0 \\ 0 & -a\end{bmatrix}.
	\dsp\end{align*}
	To see that strict primal feasibility holds, note that $B$ and $Y$ are both positive semidefinite, so $\beta \geq 0$. Furthermore, if we take $\sspp X = \displaystyle\frac{1}{2}\begin{bmatrix}1 & 0 \\ 0 & 1\end{bmatrix} > 0\dsp$ then $\sspp B - \Phi(X) = \displaystyle\frac{1}{2}\begin{bmatrix}2 & -2 \\ -2 & 5\end{bmatrix} > 0\dsp$. To see that strict dual feasibility holds, note that $\alpha$ is finite because the first primal constraint guarantees $x_{11} + x_{22} \leq 2$, which (by positive semidefiniteness of $X$) says that $x_{12} + x_{21} \leq 2$ as well, so $\alpha \leq 4$. Furthermore, if we take $\sspp Y = \begin{bmatrix}3 & 0 \\ 0 & 1\end{bmatrix} > 0\dsp$ then $\sspp \Phi^\dagger(Y) - A = \begin{bmatrix}1 & -1 \\ -1 & 2\end{bmatrix} > 0\dsp$. It follows that the primal and dual problems have the same optimal values (i.e., $\alpha = \beta$), and there are specific $X \in \cl{A}$ and $Y \in \cl{B}$ such that $\Tr(AX) = \Tr(BY) = \alpha = \beta$. We begin by computing $\alpha, \beta$, and $X$ analytically. We then compute $Y$ via numerically via MATLAB.

	If we write $X = (x_{ij})$ and $Y = (y_{ij})$ then the primal and dual forms~\eqref{sp:form} of this semidefinite program simplify as follows:
\begin{align*}\hsp
\begin{matrix}
\begin{tabular}{r l c r l}
\multicolumn{2}{c}{{\bf Primal problem}} & \quad \quad & \multicolumn{2}{c}{{\bf Dual problem}} \\
\text{max.:} & $x_{11} + x_{12} + x_{21} + x_{22}$ & \ \ & \text{min.:} & $2y_{11} - y_{12} - y_{21} + 2y_{22}$ \\
\text{s.t.:} & $\sspp\begin{bmatrix}2 - x_{11} - x_{22} & -1 \\ -1 & 2 + x_{11}\end{bmatrix} \geq 0\hsp$ & \ \ & \text{s.t.:} & $\sspp\begin{bmatrix}y_{11} - y_{22} - 1 & -1 \\ -1 & y_{11} - 1\end{bmatrix} \geq 0\hsp$ \\
\ & $X \geq 0$ & \ & \ & $Y \geq 0$ \\
\end{tabular}
\end{matrix}
\dsp\end{align*}
	If we now use the fact that a $2 \times 2$ matrix $\sspp\begin{bmatrix}a & b \\ \overline{b} & d\end{bmatrix}\dsp$ is positive semidefinite if and only if $ad \geq |b|^2$, we can reduce this semidefinite program to the following:
\begin{align}\label{eq:sdp_exam_2}\hsp
\begin{matrix}
\begin{tabular}{r l c r l}
\multicolumn{2}{c}{{\bf Primal problem}} & \quad \quad & \multicolumn{2}{c}{{\bf Dual problem}} \\
\text{max.:} & $(\sqrt{x_{11}} + \sqrt{x_{22}})^2$ & \ \ & \text{min.:} & $2(y_{11} - \sqrt{y_{11}y_{22}} + y_{22})$ \\
\text{s.t.:} & $x_{11}^2 + 2x_{22} + x_{11}x_{22} \leq 3$ & \ \ & \text{s.t.:} & $(y_{11} - 1)(y_{11} - y_{22}) \geq y_{11}$
\end{tabular}
\end{matrix}
\dsp\end{align}
	Note that we transformed the constraint $X \geq 0$ into $x_{12} + x_{21} \leq 2\sqrt{x_{11}x_{22}}$. Because $x_{12}$ and $x_{21}$ do not appear in the other constraint, we were free to replace $x_{12} + x_{21}$ by $2\sqrt{x_{11}x_{22}}$ in the primal objective function (and similarly for $Y$).
	
	We now use the facts that the constraints and objective functions of~\eqref{eq:sdp_exam_2} are continuous, the variables $x_{11}$ and $x_{22}$ are non-negative, and increasing $x_{11}$ or $x_{22}$ increases the value of the objective function, to see that we can take equality in the constraint of the primal problem. Using Lagrange multipliers on the primal problem then gives the following system of equations:
	\begin{align*}
		1 + \sqrt{\frac{x_{22}}{x_{11}}} + \lambda(2x_{11} + x_{22}) & = 0 \\
		1 + \sqrt{\frac{x_{11}}{x_{22}}} + \lambda(2 + x_{11}) & = 0 \\
		x_{11}^2 + 2x_{22} + x_{11}x_{22} & = 3.
	\end{align*}
	Taking $x_{11}$ times the first equation minus $x_{22}$ times the second equation gives $\lambda = \frac{x_{22} - x_{11}}{2(x_{11}^2 - x_{22})}$. Plugging this into the second equation and using the third equation gives $x_{22} = \frac{3 - x_{11}^2}{2 + x_{11}}$, and plugging that into the first equation finally gives $2x_{11}^5+14x_{11}^4+31x_{11}^3+14x_{11}^2-27x_{11}-24 = 0$. This polynomial has a unique real root at $x_{11} \approx 0.9315$. This in turn gives $x_{22} \approx 0.7274$, so the optimal value $\alpha$ of the semidefinite program is approximately $3.3051$, which is attained by the matrix $\sspp X \approx \begin{bmatrix}0.9315 & 0.8231 \\ 0.8231 & 0.7274\end{bmatrix}\dsp$.

In order to verify our answer, we can solve this semidefinite program in MATLAB and see that the optimal solutions to the primal and dual problems are indeed both approximately $3.3051$. Furthermore, we find the following matrix $Y$ that attains the optimal value in the dual problem:
\begin{align*}
	Y & = \begin{bmatrix}2.1317 & 0.7272 \\ 0.7272 & 0.2480\end{bmatrix}
\end{align*}
Observe that $y_{12} = y_{21} = \sqrt{y_{11}y_{22}}$, as we noted earlier, and plugging these values into the dual objective function of~\eqref{eq:sdp_exam_2} gives
\begin{align*}
	2(y_{11} - \sqrt{y_{11}y_{22}} + y_{22}) \approx 2(2.1317 - \sqrt{2.1317 \times 0.2480} + 0.2480) \approx 3.3051,
\end{align*}
as desired.}
\end{exam}

\section{Computation of the S(k)-Norms}\label{sec:semidefProgramMNorm}

Although we have seen many properties of the $S(k)$-norms, we have not yet discussed the problem of their computation. Proposition~\ref{prop:rankOneNorm} and Theorem~\ref{thm:sk_vector_norm} show that we can compute the $S(k)$-norm of rank-$1$ operators efficiently, since the Schmidt coefficients of a bipartite pure state are easy to compute. However, Corollary~\ref{cor:kPosInf1} shows that the problem of computing the $S(k)$-operator norm of an arbitrary positive operator is equivalent to the problem of determining $k$-block positivity of an arbitrary Hermitian operator and is thus likely very difficult. Furthermore, it has been shown \cite{G03,G10} that computing the $S(1)$-norm is NP-hard, so we don't expect that there is an efficient method for its computation. Nonetheless, we will see some techniques that can be used to compute the $S(k)$-norms in certain cases and in small dimensions, and at least provide nontrivial bounds in general.

One na\"{i}ve method we could use to estimate $\big\|X\big\|_{S(k)}$ would be to simply compute $\big|\bra{w}X\ket{v}\big|$ for several states $\ket{v}, \ket{w} \in \mathbb{C}^m \otimes \mathbb{C}^n$ with $SR(\ket{v}), SR(\ket{w}) \leq k$ and take the largest resulting value. This procedure can be made rigorous via $\varepsilon$-nets \cite{HLW06}, which are finite approximations of the set of states with small Schmidt rank. The downside of this approach is that the number of states in the $\varepsilon$-net grows exponentially in $m+n$ and thus becomes infeasible even for moderately large values of $m$ and $n$. Instead, most of our methods will stem from semidefinite and conic programming.

\subsection{Semidefinite Programs Based on k-Positive Maps}\label{sec:SP_pos_map}

Here we develop a family of semidefinite programs that can be used to provide upper bounds on the $S(k)$-operator norm in general and compute it exactly in low-dimensional cases. Additionally, some simple theoretical results that further establish the link between the $S(k)$-norm and $k$-block positive operators will follow from the duality theory of semidefinite programming.

Given a positive semidefinite operator $X \in (M_m \otimes M_n)^+$ and a natural number $k$, we now present a family of semidefinite programs with the following properties:
\begin{itemize}
	\item Strong duality holds for each semidefinite program.
	\item The optimal value $\alpha$ of each SDP is an upper bound of $\big\|X\big\|_{S(k)}$.
	\item There is a semidefinite program in the family with optimal value $\alpha = \big\|X\big\|_{S(k)}$.
\end{itemize}

Let $X \in (M_m \otimes M_n)^+$ be a positive semidefinite operator for which we wish to compute $\big\|X\big\|_{S(k)}$. Let $\Phi : M_n \rightarrow M_n$ be a fixed $k$-positive linear map and consider the following semidefinite program, where we optimize over density operators $\rho$ in the primal problem and over constants $\lambda \in \bb{R}$ and operators $Y \in (M_m \otimes M_n)^+$ in the dual problem:
\begin{align}\label{sp:matrixNorm}\hsp
\begin{matrix}
\begin{tabular}{r l c r l}
\multicolumn{2}{c}{{\bf Primal problem}} & \quad \quad & \multicolumn{2}{c}{{\bf Dual problem}} \\
\text{max.:} & $\Tr(X\rho)$ & \ \ & \text{min.:} & $\lambda$ \\
\text{s.t.:} & $(id_m \otimes \Phi)(\rho) \geq 0$ & \ \ & \text{s.t.:} & $\lambda I_m \otimes I_n \geq (id_m \otimes \Phi^\dagger)(Y) + X$ \\
\ & $\Tr(\rho) \leq 1$ & \ & \ & $Y \geq 0$ \\
\ & $\rho \geq 0$ & \ & \ & \ \\
\end{tabular}
\end{matrix}
\dsp\end{align}

It may not be immediately obvious that this optimization problem is actually of the form~\eqref{sp:form}, so we first check that these problems are indeed duals of each other and form a valid semidefinite program. To this end, consider the linear map $\Psi : M_m \otimes M_n \rightarrow (M_m \otimes M_n) \oplus M_1$ defined by
\begin{align*}\sspp
  \Psi(\rho) = \begin{bmatrix}-(id_m \otimes \Phi)(\rho) & 0 \\ 0 & \Tr(\rho) \end{bmatrix}.
\dsp\end{align*}
Then the dual map $\Psi^\dagger : (M_m \otimes M_n) \oplus M_1 \rightarrow M_m \otimes M_n$ is given by
\begin{align*}\sspp
  \Psi^\dagger\left( \begin{bmatrix} Y & * \\ * & \lambda \end{bmatrix} \right) = \lambda I_m \otimes I_n - (id_m \otimes \Phi^\dagger)(Y).
\dsp\end{align*}
Finally, setting
\begin{align*}\sspp
  A = X \quad \text{and} \quad B = \begin{bmatrix} 0 & 0 \\ 0 & 1 \end{bmatrix}
\dsp\end{align*}
gives the semidefinite program~\eqref{sp:matrixNorm} in the form~\eqref{sp:form}.

We now show that this program satisfies the Slater-type conditions for strong duality given by Theorem~\ref{thm:spSlater}. It is clear that both $\alpha$ and $\beta$ are finite, as $\Tr(X\rho) \leq \big\| X \big\|$ and $\lambda \geq 0$. Both feasible sets are also non-empty (for example, one could take $\rho$ to be any separable state, $Y = 0$, and $\lambda \geq \big\| X \big\|$). Strong dual feasibility then follows by choosing any $Y > 0$ and a sufficiently large $\lambda$. Strong primal feasibility is not necessarily satisfied, however, as there is no guarantee that $\Phi$ does not introduce singularities in $\rho$ (for example, consider the zero map, which is $k$-positive). We could restrict the family of $k$-positive maps that we are interested in if we really desired strong primal feasibility, but strict dual feasibility is enough for our purposes.

It follows from condition (b) of Theorem~\ref{thm:sch_kpos_maps} that, for any $k$-positive map $\Phi$, the optimal value of the semidefinite program~\eqref{sp:matrixNorm} is an upper bound of $\big\| X \big\|_{S(k)}$ -- the supremum in the primal problem is just being taken over a set that is larger than the set of operators $\rho$ with $SN(\rho) \leq k$. This leads to the following theorem.
\begin{thm}\label{thm:kPosInf2}
	Let $X \in (M_m \otimes M_n)^{+}$. Then
	\begin{align*}
	  \big\| X \big\|_{S(k)} = \inf_{Y} \Big\{ \lambda_{\textup{max}}(X + Y) : Y\text{ is }k\text{-block positive} \Big\}.
	\end{align*}
\end{thm}
\begin{proof}
	Because $\Phi$ is $k$-positive if and only if $\Phi^\dagger$ is $k$-positive, the dual problem~\eqref{sp:matrixNorm} can be rephrased as asking for the infimum of $\lambda_{\textup{max}}(X + Y)$, where the infimum is taken over a subset of the $k$-block positive operators $Y \in M_m \otimes M_n$. The preceding paragraph then showed us that
	\begin{align*}
		\big\| X \big\|_{S(k)} \leq \inf_{Y} \Big\{ \lambda_{\textup{max}}(X + Y) : Y\text{ is }k\text{-block positive} \Big\}.
	\end{align*}
	To see that equality is attained, choose $Y = \big\| X \big\|_{S(k)}I - X$, which we know from Corollary~\ref{cor:kPosInf1} is $k$-block positive. Then
	\begin{align*}
		\lambda_{\textup{max}}(X + Y) = \lambda_{\textup{max}}(X + \big\| X \big\|_{S(k)}I - X) = \big\| X \big\|_{S(k)}.
	\end{align*}
\end{proof}

In fact, it is not difficult to see that there is a particular $k$-positive map $\Phi$ such that $\big\| X \big\|_{S(k)}$ is attained as the optimal value of the semidefinite program~\eqref{sp:matrixNorm} corresponding to $\Phi$ -- simply let $\Phi$ be the map associated with the operator $\big\| X \big\|_{S(k)}I - X$ via the Choi--Jamio{\l}kowski isomorphism.

One additional implication of Theorem~\ref{thm:kPosInf2} is that $\big\| X \big\|_{S(k)} \leq \lambda_{\textup{max}}(X + Y)$ for all $X \in (M_m \otimes M_n)^+$ and all $k$-block positive $Y \in M_m \otimes M_n$. The following corollary shows that this can be strengthened into another characterization of $k$-positivity.
\begin{cor}\label{cor:SPcor}
	Let $Y = Y^\dagger \in M_m \otimes M_n$. Then $Y$ is $k$-block positive if and only if
	\begin{align*}
	  \big\| X \big\|_{S(k)} \leq \lambda_{\textup{max}}(X + Y) \quad \forall \, X \in (M_m \otimes M_n)^{+}.
	\end{align*}
\end{cor}
\begin{proof}
	The ``only if'' direction of the proof follows immediately from Theorem~\ref{thm:kPosInf2}. To see the ``if'' direction, assume that $Y$ is not $k$-block positive and choose $X = cI - Y$, where $c \in \bb{R}$ is large enough that $cI - Y \geq 0$. Then, because $Y$ is not $k$-block positive, there exists a vector $\ket{v}$ with $SR(\ket{v}) \leq k$ such that $\bra{v}Y\ket{v} < 0$. Thus
	\begin{align*}
		\big\| X \big\|_{S(k)} \geq \bra{v}(cI - Y)\ket{v} = c - \bra{v}Y\ket{v} > c = \lambda_{\textup{max}}(X + Y).
	\end{align*}
\end{proof}

Recall that if $m = 2$ and $n \in \{2,3\}$ then the transpose map $T$ alone is enough to determine whether or not $\rho$ is separable (i.e., $SN(\rho) = 1$ if and only if $(id_m \otimes T)(\rho) \geq 0$) \cite{HHH96}. It follows that the semidefinite program~\eqref{sp:matrixNorm} with $\Phi = T$ can be used to compute $\big\| X \big\|_{S(1)}$ for positive operators $X \in (M_2 \otimes M_3)^+$. That is, the infinite family of semidefinite programs reduces to just a single semidefinite program in this situation.

\subsection{Operator Norms Arising from Other Convex Cones}\label{sec:coneNorms}

We now show that many of the results for the $S(k)$-norms actually hold in the much more general setting of arbitrary closed convex cones of operators.
\begin{defn}\label{defn:coneNorm}
  Let $X \in M_n$ and let $\cl{C} \subseteq M_n^+$ be a closed convex cone such that ${\rm span}(\cl{C}) = M_n$. Then we define the \emph{$\cl{C}$-operator norm} of $X$, denoted $\big\| X \big\|_{\cl{C}}$, by
  \begin{align*}
  	\big\|X\big\|_{\cl{C}} & := \sup_{\rho \in \cl{C}} \Big\{ \big|\Tr(X\rho)\big| : \Tr(\rho) = 1 \Big\}.
  \end{align*}
\end{defn}

It is easy to see that the $\cl{C}$-operator norm is indeed a valid norm. The only nontrivial condition is that $\big\|X\big\|_{\cl{C}} = 0$ if and only if $X = 0$, which follows from the requirement that ${\rm span}(\cl{C}) = M_n$. Observe also that if $\cl{C} = \cl{S}_k$ is the cone of (unnormalized) states with Schmidt number no larger than $k$ and $X \geq 0$, then $\big\|X\big\|_{\cl{C}} = \big\|X\big\|_{S(k)}$. The $\cl{C}$-operator norm was studied independently in the case of bipartite systems (i.e., when $\cl{C} \subset M_m \otimes M_n$) in \cite{SS10}.

It is trivial to see that if $\cl{C} \subseteq \cl{D}$, where $\cl{D}$ is another closed convex cone, then $\big\|X\big\|_{\cl{C}} \leq \big\|X\big\|_{\cl{D}}$ (which gives the familiar inequality $\big\|X\big\|_{S(k)} \leq \big\|X\big\|$ when $\cl{C} = \cl{S}_k$). Additionally, several of the characterizations of the $S(k)$-norms carry over in an obvious way to this more general setting, as we now demonstrate.
\begin{prop}\label{prop:kPosInf1_gen}
  Let $X \in M_n^+$. Then $cI - X \in \cl{C}^{\circ}$ if and only if $c \geq \big\|X\big\|_{\cl{C}}$.
\end{prop}
\begin{proof}
	By definition, $cI - X \in \cl{C}^{\circ}$ if and only if
	\begin{align*}
		\Tr\big((cI - X)\rho\big) = c - \Tr(X\rho) \geq 0 \quad \forall \, \rho \in \cl{C},
	\end{align*}
	which is true if and only if $c \geq \big\|X\big\|_{\cl{C}}$.
\end{proof}

Now let $X \in M_n^{+}$ and consider the following conic program:
\begin{align}\label{sp:coneNorm}\hsp
\begin{matrix}
\begin{tabular}{r l c r l}
\multicolumn{2}{c}{{\bf Primal problem}} & \quad \quad \quad & \multicolumn{2}{c}{{\bf Dual problem}} \\
\text{maximize:} & $\Tr(X\rho)$ & \quad & \text{minimize:} & $\lambda$ \\
\text{subject to:} & $\Tr(\rho) \leq 1$ & \quad & \text{subject to:} & $\lambda I \geq Y + X$ \\
\ & $\rho \in \cl{C}$ & \ & \ & $Y \in \cl{C}^\circ$ \\
\end{tabular}
\end{matrix}
\dsp\end{align}
It is easy to see that these problems are indeed duals of each other and form a valid conic program, using the same method as was used in Section~\ref{sec:SP_pos_map} to show that the semidefinite program~\eqref{sp:matrixNorm} is valid -- we have just not made the restriction that $\cl{C} = (M_m \otimes M_n)^{+}$ and we have replaced the map $id_m \otimes \Phi$ by the identity map. Strong dual duality also holds in this setting. The main difference here is that we have $\rho \in \cl{C}$ and $Y \in \cl{C}^\circ$ rather than $\rho, Y \geq 0$ -- we could have stated the semidefinite program~\eqref{sp:matrixNorm} as a conic program in terms of the cone $\cl{S}_k$, but then it would become less clear how to actually implement the semidefinite programs and compute upper bounds of $\big\|X\big\|_{S(k)}$ using $k$-positive maps.

Just as is the case for the $S(k)$-norms, the theory of semidefinite programming leads to the following two results. We state them without proof, as their proofs are almost identical to the proofs of Theorem~\ref{thm:kPosInf2} and Corollary~\ref{cor:SPcor}, respectively.
\begin{thm}\label{thm:kPosInf2_gen}
	Let $X \in M_n^{+}$. Then
	\begin{align*}
	  \big\| X \big\|_{\cl{C}} = \inf_{Y} \Big\{ \lambda_{\textup{max}}(X + Y) : Y \in \cl{C}^\circ \Big\}.
	\end{align*}
\end{thm}

\begin{cor}
	Let $Y = Y^\dagger \in M_n$. Then $Y \in \cl{C}^\circ$ if and only if
	\begin{align*}
	  \big\| X \big\|_{\cl{C}} \leq \lambda_{\textup{max}}(X + Y) \quad \forall \, X \in M_n^{+}.
	\end{align*}
\end{cor}

We now shift focus to one particularly important family of cones and their related norms. Given any positive linear map $\Phi : M_n \rightarrow M_n$, there exists a natural closed convex cone $\cl{C}_\Phi \subseteq (M_m \otimes M_n)^{+}$ associated with $\Phi$:
\begin{align*}
	\cl{C}_\Phi := \Big\{ X \in (M_m \otimes M_n)^{+} : (id_n \otimes \Phi)(X) \geq 0 \Big\}.
\end{align*}
Any such cone satisfies the hypotheses of Definition~\ref{defn:coneNorm} and hence $\|\cdot\|_{\cl{C}_\Phi}$ is indeed a norm. Furthermore, if $X \geq 0$ then we can compute $\big\|X\big\|_{\cl{C}_\Phi}$ to any desired accuracy via semidefinite programming: $\big\|X\big\|_{\cl{C}_{\Phi}}$ is exactly what is computed by the semidefinite program~\eqref{sp:matrixNorm}. It follows that $\big\|X\big\|_{S(k)} = \inf_{\Phi}\big\{ \big\|X\big\|_{\cl{C}_{\Phi}} : \Phi \text{ is } \text{$k$-positive} \big\}$.

In the case of the transpose map $T : M_n \rightarrow M_n$, $\cl{C}_T$ is the cone of unnormalized PPT states, so the norm $\| \cdot \|_{\cl{C}_T}$ can roughly be thought of as a measure of how close a given operator is to having positive partial transpose. It is known \cite{S09} that the dual cone of the set of PPT states is given by
\begin{align*}
	\cl{C}_{T}^\circ = \Big\{ X = X^\dagger \in M_m \otimes M_n : X = Y + Z \text{ for some } Y \geq 0, Z^\Gamma \geq 0 \Big\},
\end{align*}
where we recall the shorthand notation $Z^\Gamma := (id_m \otimes T)(Z)$. This leads immediately to the following characterization of $\| \rho \|_{\cl{C}_T}$ via Theorem~\ref{thm:kPosInf2_gen}.
\begin{prop}
	Let $\rho \in (M_m \otimes M_n)^+$ be a density operator. Then
	\begin{align*}
	  \| \rho \|_{\cl{C}_T} = \inf_{Y} \Big\{ \lambda_{\textup{max}}(\rho + Y) : Y^\Gamma \geq 0 \Big\}.
	\end{align*}
\end{prop}

\subsection{Computation Based on States with Symmetric Extensions}\label{sec:SP_shareable}

One of the disadvantages of the semidefinite programs of Section~\ref{sec:SP_pos_map} was that they required you to have a good selection of $k$-positive maps at your disposal to get good upper bounds. Furthermore, it generally is not clear how close the optimal value of one of the semidefinite programs is to the true value of $\big\|X\big\|_{S(k)}$. We tackle both of these problems in this section by presenting a different family of semidefinite programs that can be used to compute $\big\|X\big\|_{S(k)}$, using the ideas of \cite{DPS04}.

Much like before, strong duality holds for each semidefinite programs to be presented, and each semidefinite program returns an upper bound on $\big\|X\big\|_{S(k)}$. However, the semidefinite programs of this section also have the following properties:
\begin{itemize}
	\item the family of semidefinite programs is indexed by the nonnegative integers. There is no need for $k$-positive maps in the construction of the semidefinite programs;
	\item if $\alpha_s$ is the optimal value of the $s$-th SDP, then $\alpha_1 \geq \alpha_2 \geq \dots \geq \big\|X\big\|_{S(k)}$; and
	\item $\displaystyle\lim_{s\rightarrow\infty} \alpha_s = \big\|X\big\|_{S(k)}$ and we can bound the difference of $\alpha_s$ and $\big\|X\big\|_{S(k)}$.
\end{itemize}

The $s$-th semidefinite program in our infinite family is based on states with an $s$-bosonic symmetric extension, so many of the properties of such states presented in Section~\ref{sec:shareable} will have natural analogues here. For example, it is generally the case that $\alpha_s \gneq \big\|X\big\|_{S(k)}$ for all $s$, but in fact we have $\alpha_1 = \big\|X\big\|_{S(k)}$ in the case when $0 \leq X \in M_2 \otimes M_2$.

\subsubsection*{The Family of SDPs in the k = 1 Case}\label{sec:SP_shareable_S1}

We begin by presenting the family of semidefinite programs in the $k = 1$ case, since their construction is significantly simpler in this case. Let $0 \leq X \in M_n \otimes M_n$ be an operator whose $S(1)$-norm we wish to calculate. Let $s \geq 1$ and consider the following semidefinite program where we optimize in the primal problem over $\rho \in M_n^{\otimes (s+1)}$ and we optimize in the dual problem over operators $W \in M_n^{\otimes (s+1)}$. We use $\Tr_{[s-1]}(\cdot)$ to denote the partial trace over the first $s-1$ copies of $M_n$. Furthermore, we define $P_{\cl{S}_{s}}$ to be the symmetric projection on $(\bb{C}^n)^{\otimes s}$ and, for brevity, we define $P := (P_{\cl{S}_{s}} \otimes I) \in M_n^{\otimes (s+1)}$.
\begin{align}\label{sp:S1_share_sdp}
\begin{matrix}
\begin{tabular}{r l c r l}
\multicolumn{2}{c}{{\bf Primal problem}} & \quad \quad & \multicolumn{2}{c}{{\bf Dual problem}} \\
\text{max.:} & $\Tr(X \Tr_{[s-1]}(\rho))$ & \ \ & \text{min.:} & $\big\|P((I_n^{\otimes (s-1)} \otimes X) + W^\Gamma)P\big\|$ \\
\text{s.t.:} & $P \rho P = \rho$ & \ \ & \text{s.t.:} & $W \geq 0$ \\
\ & $\Tr(\rho) \leq 1$ & \ & \ & \ \\
\ & $\rho, \rho^\Gamma \geq 0$ & \ & \ & \ \\
\end{tabular}
\end{matrix}
\end{align}

Note that we may choose the partial transpotion in the condition $\rho^\Gamma \geq 0$ to be with respect to any subsystems of our choosing, but we will choose the transposition to take place on the last $\lfloor s/2 \rfloor$ subsystems, as this will allow us to use existing results to compute error bounds for this semidefinite program. Also observe that the semidefinite program~\eqref{sp:S1_share_sdp} has an equality constraint that is not present in the general form of semidefinite (or conic) programs~\eqref{sp:form} presented earlier. This is not a problem, however, as any semidefinite program together with equality constraints can be transformed into a semidefinite program with only inequality constraints \cite{Wat04}.

The operator $W$ in the semidefinite program~\eqref{sp:S1_share_sdp} acts as a ``witness'' that proves an upper bound on the $S(1)$-norm, much like entanglement witnesses prove that a state is entangled. The witness $W$ can also be thought of as playing a role that is dual to the states $\ket{v}$ in the supremum that defines $\big\|X\big\|_{S(1)}$: while any given separable state $\ket{v}$ proves a lower bound $\bra{v}X\ket{v}$ on $\big\|X\big\|_{S(1)}$, any given $W \geq 0$ proves the \emph{upper} bound $\big\|P((I_n^{\otimes (s-1)} \otimes X) + W^\Gamma)P\big\|$.

To see that the optimization problems~\eqref{sp:S1_share_sdp} are indeed duals of each other and can be put in form~\eqref{sp:form}, let $A = P(I_n^{\otimes (s-1)} \otimes X)P \in M_n^{\otimes (s+1)}$ and define $\Phi : M_n^{\otimes (s+1)} \rightarrow M_1 \oplus (M_n^{\otimes (s+1)})^{\oplus 3}$ and $B \in M_1 \oplus (M_n^{\otimes (s+1)})^{\oplus 3}$ by
\begin{align*}\sspp
	\Phi(\rho) = \begin{bmatrix}\Tr(P\rho P) & 0 & 0 & 0 \\ 0 & -(P\rho P)^{\Gamma} & 0 & 0 \\ 0 & 0 & P\rho P - \rho & 0 \\ 0 & 0 & 0 & \rho - P\rho P \end{bmatrix} \ \ \text{ and } \ \ B = \begin{bmatrix}1 & 0 & 0 & 0 \\ 0 & 0 & 0 & 0 \\ 0 & 0 & 0 & 0 \\ 0 & 0 & 0 & 0 \end{bmatrix}.
\dsp\end{align*}

Then
\begin{align*}
	\Tr(A\rho) = \Tr\big(P(I_n^{\otimes (s-1)} \otimes X)P \rho\big) = \Tr\big((I_n^{\otimes (s-1)} \otimes X) \rho\big) = \Tr(X \Tr_{[s-1]}(\rho)),
\end{align*}
so the primal problem associated with this choice of $\Phi$, $A$, and $B$ is indeed the primal problem of~\eqref{sp:S1_share_sdp}. To see that the dual problem is as claimed, note that
\begin{align*}\sspp
	\Phi^\dagger\left(\begin{bmatrix}\lambda & * & * & * \\ * & W & * & * \\ * & * & Z_1 & * \\ * & * & * & Z_2 \end{bmatrix}\right) = P (\lambda I - W^\Gamma + Z_1 - Z_2) P - Z_1 + Z_2,
\dsp\end{align*}
where we have used $*$ to denote entries in the nullspace of $\Phi^\dagger$. Thus the dual problem becomes:
\begin{align*}
\begin{matrix}
\begin{tabular}{r l}
\text{min.:} & $\lambda$ \\
\text{s.t.:} & $\lambda P \geq P((I_n^{\otimes (s-1)} \otimes X) + W^\Gamma - Z_1 + Z_2)P + Z_1 - Z_2$ \\
\ & $W,Z_1,Z_2 \geq 0$ \\
\end{tabular}
\end{matrix}
\end{align*}

To simplify the above problem, define $Z = Z_1 - Z_2$ to be a general Hermitian operator. It is straightforward to see that the $Z - P Z P$ portion of the above constraint cannot serve to decrease $\lambda$, so we can choose $Z = 0$ without loss of generality. The dual problem thus simply asks to minimize the maximal eigenvalue of $P((I_n^{\otimes (s-1)} \otimes X) + W^\Gamma)P$. In other words, it asks to minimize $\big\|P((I_n^{\otimes (s-1)} \otimes X) + W^\Gamma)P\big\|$, as claimed.

To see the strong duality holds for each of the given semidefinite programs, we note that we could write the program as a conic program over the cone $\cl{C} := \{ \rho : P \rho P = \rho, \rho \geq 0 \}$. Then the state $\rho = P$ (appropriately normalized) satisfies all of the equality constraints of the primal problem, is in the relative interior of the cone $\cl{C}$, and satisfies the remaining primal inequality ($\rho^\Gamma \geq 0$) strictly. Slater's condition then tells us that strong duality holds, so the primal and dual pair of semidefinite programs~\eqref{sp:S1_share_sdp} have the same optimal value.

Note that in the semidefinite program~\eqref{sp:S1_share_sdp} we included just a single partial transposition constraint, $\rho^\Gamma \geq 0$ (where we recall that this partial transpose is with respect to the $\lceil s/2 \rceil - \lfloor s/2 \rfloor$ cut). We could have included partial transpose constraints with respect to other cuts as well. However, we will see that some information about how quickly $\alpha_s$ approaches $\big\|X\big\|_{S(1)}$ is known in this case of just one partial transpose. The optimal values in the case of multiple partial transposes certainly approach $\big\|X\big\|_{S(1)}$ at least as quickly, but it is not known if they approach strictly faster (i.e., if the additional computational overhead is really worth it). Furthermore, we will see that the above semidefinite program is only solvable on current hardware up to about $s = 3$ anyway, at which point there are only two independent partial transposition conditions.

Based on Proposition~\ref{prop:MatMultDiffVectors} and the fact that the set of states with an $s$-symmetric extension approaches the set of separable states as $s \rightarrow \infty$, it is clear that $\alpha_1 \geq \alpha_2 \geq \cdots \geq \big\|X\big\|_{S(1)}$ and $\displaystyle\lim_{s\rightarrow\infty}\alpha_s = \big\|X\big\|_{S(1)}$. In fact, for operators $0 \leq X \in M_2 \otimes M_2$ we even have $\alpha_1 = \big\|X\big\|_{S(1)}$ because of the fact that $\rho^\Gamma \geq 0$ if and only if $\rho$ is separable in this case.

One variant of the semidefinite programs~\eqref{sp:S1_share_sdp} that is particularly useful is the one that arises by removing the partial transposition requirement $\rho^\Gamma \geq 0$. In this case, the optimal values of the semidefinite programs still approach $\big\|X\big\|_{S(1)}$ from above (albeit more slowly in general), but the dual problem simplifies to simply asking for the value of $\big\|P(I_n^{\otimes {s-1}} \otimes X)P\big\|$, and thus we don't even need to use semidefinite programming techniques to find the optimal values. Indeed, we simply have
\begin{align}\label{eq:s1_share_no_PT}
	\big\|X\big\|_{S(1)} = \lim_{s\rightarrow\infty} \big\|P(I_n^{\otimes {s-1}} \otimes X)P\big\|.
\end{align}

\subsubsection*{Error Bounds in the k = 1 Case}\label{sec:SP_error_bounds}

One of the biggest advantages of the family of semidefinite programs based on symmetric extensions over the semidefinite programs based on positive maps is that there are explicit bounds on how far away the optimal value of the $s$-th semidefinite program is from $\big\|X\big\|_{S(1)}$ in this setting. In the statement of the following theorem, $\alpha_s$ is the optimal value of the semidefinite program~\eqref{sp:S1_share_sdp} and $\beta_s := \big\|P(I_n^{\otimes {s-1}} \otimes X)P\big\|$ is the optimal value of the same semidefinite program without the partial transposition constraint. Also, the quantity $g_s$ is defined as in \cite{NOP09} by
\begin{align*}
	g_s := \begin{cases}\min\big\{ 1 - x : P^{(n-2,0)}_{s/2+1}(x) = 0 \big\} & \text{if $s$ is even} \\
	\min\big\{ 1 - x : P^{(n-2,1)}_{(s+1)/2}(x) = 0 \big\} & \text{if $s$ is odd} \end{cases},
\end{align*}
where $P^{(\alpha,\beta)}_n(x)$ are the Jacobi polynomials \cite{AS72}.
\begin{thm}\label{thm:sdp_error}
	Let $0 \leq X \in M_n \otimes M_n$. Then
	\begin{align*}
		\alpha_s & \geq \big\|X\big\|_{S(1)} \geq \left(1 - \frac{n g_s}{2(n-1)}\right) \alpha_s + \frac{g_s}{2(n-1)}\lambda_{\textup{min}}(X) \ \ \text{ and } \\
		\beta_s & \geq \big\|X\big\|_{S(1)} \geq \frac{s}{n+s} \beta_s + \frac{1}{n+s}\lambda_{\textup{min}}(X),
	\end{align*}
	where $\lambda_{\textup{min}}(X)$ is the minimal eigenvalue of $X$.
\end{thm}
\begin{proof}
	We have already seen why the left inequalities hold, so we only need to show the right inequalities. We begin with the second inequality. It is known \cite[Theorem~2]{NOP09} that if $\rho$ has an $s$-bosonic symmetric extension, then
	\begin{align}\label{eq:s1_separable_op}
		\frac{s}{n+s}\rho + \frac{1}{n+s}\Tr_2(\rho) \otimes I_n
	\end{align}
	is separable. Then
	\begin{align*}
		\big\|X\big\|_{S(1)} & = \sup_{\sigma} \big\{ \Tr(X\sigma) : \sigma \text{ is separable} \big\} \\
		& \geq \sup_{\rho} \Big\{ \Tr\big(X(\tfrac{s}{n+s}\rho + \tfrac{1}{n+s}\Tr_2(\rho) \otimes I_n)\big) : \rho \text{ has $s$-BSE} \Big\} \\
		& = \sup_{\rho} \Big\{ \frac{s}{n+s} \Tr(X\rho) + \frac{1}{n+s}\Tr\big(X(\Tr_2(\rho) \otimes I_n)\big) : \rho \text{ has $s$-BSE} \Big\} \\
		& \geq \frac{s}{n+s} \beta_s + \frac{1}{n+s}\lambda_{\textup{min}}(X).
	\end{align*}

The corresponding inequality for $\alpha_s$ follows similarly from using \cite[Theorem~3]{NOP09}, which says that if $\rho$ has as $s$-bosonic symmetric extension with positive partial transpose with respect to the $\lceil s/2 \rceil - \lfloor s/2 \rfloor$ cut, then
\begin{align*}
	\left(1 - \frac{n g_s}{2(n-1)}\right)\rho + \frac{g_s}{2(n-1)}\Tr_2(\rho) \otimes I_n
\end{align*}
is separable. Then
	\begin{align*}
		\big\|X\big\|_{S(1)} & = \sup_{\sigma} \big\{ \Tr(X\sigma) : \sigma \text{ is separable} \big\} \\
		& \geq \sup_{\rho} \Big\{ \Tr\big(X((1 - \tfrac{n g_s}{2(n-1)})\rho + \tfrac{g_s}{2(n-1)}\Tr_2(\rho) \otimes I_n)\big) : \rho \text{ has $s$-PPT BSE} \Big\} \\
		& = \sup_{\rho} \Big\{ \big(1 - \tfrac{n g_s}{2(n-1)}\big) \Tr(X\rho) + \tfrac{g_s}{2(n-1)}\Tr\big(X(\Tr_2(\rho) \otimes I_n)\big) : \rho \text{ has $s$-PPT BSE} \Big\} \\
		& \geq \left(1 - \frac{n g_s}{2(n-1)}\right) \alpha_s + \frac{g_s}{2(n-1)}\lambda_{\textup{min}}(X).
	\end{align*}
\end{proof}

Since $\displaystyle \lim_{s\rightarrow\infty} g_s = 0$, it is clear that the bounds of Theorem~\ref{thm:sdp_error} all approach $\big\|X\big\|_{S(1)}$ as $s \rightarrow \infty$.

\subsubsection*{The Family of SDPs in the k > 1 Case}\label{sec:SP_shareable_S2}

The family of semidefinite programs used to compute $\big\|X\big\|_{S(1)}$ can be modified using the techniques of Section~\ref{sec:shareable_gen} to compute $\big\|X\big\|_{S(k)}$ for arbitrary $k$. As before, fix $s \geq 1$. We now consider the following semidefinite program where we optimize in the primal problem over $\tilde{\rho} \in (M_k \otimes M_n)^{\otimes (s+1)}$ and we optimize in the dual problem over operators $W \in (M_k \otimes M_n)^{\otimes (1+2)}$. We use $\Tr_{[s-1]}(\cdot)$ to denote the partial trace over the first $s-1$ copies of $(M_k \otimes M_n)$. Furthermore, we define $P_{\cl{S}_{s}}$ to be the symmetric projection on $(\bb{C}^k \otimes \bb{C}^n)^{\otimes s}$ and, for brevity, we define $P := (P_{\cl{S}_{s}} \otimes I) \in (M_k \otimes M_n)^{\otimes (s+1)}$.
\begin{align}\label{sp:Sk_share_sdp}
\begin{matrix}
\begin{tabular}{r l}
\multicolumn{2}{c}{{\bf Primal problem}} \\
\text{max.:} & $\Tr((\ketbra{\psi_+}{\psi_+} \otimes X) \Tr_{[s-1]}(\tilde{\rho}))$ \\
\text{s.t.:} & $P \tilde{\rho} P = \tilde{\rho}$ \\
\ & $\Tr\big( (\bra{\psi_+} \otimes I) \Tr_{[s-1]}(\tilde{\rho}) (\ket{\psi_+} \otimes I) \big) = 1$ \\
\ & $\tilde{\rho}, \tilde{\rho}^\Gamma \geq 0$ \\
\end{tabular}
\end{matrix}
\end{align}

We do not give explicit details to show that the above optimization problem is indeed a semidefinite program for two reason. Firstly, the details are almost exactly the same as in the $k = 1$ case. Secondly, we will now show that these semidefinite programs in the $k > 1$ case are not of much practical computational use anyway.

\subsubsection*{Error Bounds in the k > 1 Case}\label{sec:SP_error_boundsk2}

We now compute bounds on how far away the optimal value of the semidefinite program~\eqref{sp:Sk_share_sdp} can be from $\big\|X\big\|_{S(k)}$. As in the $k = 1$ case, $\alpha_s$ is the optimal value of the semidefinite program~\eqref{sp:Sk_share_sdp} and $\beta_s$ is the optimal value of the same semidefinite program without the partial transposition constraint. The quantity $g_s$ is also the same as it was before.
\begin{thm}\label{thm:sdp_error_k2}
	Let $0 \leq X \in M_n \otimes M_n$. Then
	\begin{align*}
		\alpha_s & \geq \big\|X\big\|_{S(k)} \geq \left(1 - \frac{n^2 g_s}{(2 + g_s n)(n-1)}\right) \alpha_s + \frac{g_s}{(2 + g_s n)(n-1)}\Tr(X) \ \ \text{ and } \\
		\beta_s & \geq \big\|X\big\|_{S(k)} \geq \frac{s}{n^2 + s} \beta_s + \frac{1}{n^2 + s} \Tr(X).
	\end{align*}
\end{thm}
\begin{proof}
	As in the $k = 1$ case, the left inequalities are trivial, so we only need to show the right inequalities. We begin with the second inequality. Note that if $\tilde{\rho}$ has an $s$-BSE and is normalized so that $\Tr\big( (\bra{\psi_+} \otimes I) \tilde{\rho} (\ket{\psi_+} \otimes I) \big) \leq 1$ then we can add $\tfrac{1}{n+s}((I - \Tr_2(\rho)) \otimes I)$ to the separable operator~\eqref{eq:s1_separable_op} to see that the following operator is separable and satisfies the same normalization condition as $\tilde{\rho}$:
	\begin{align}\label{eq:k_s_share_sep}
		\frac{s}{n^2 + s}\tilde{\rho} + \frac{1}{n^2 + s}I \otimes I.
	\end{align}
	The remainder of the proof mimics the proof of Theorem~\ref{thm:sdp_error}, and note that $\tilde{\sigma}$ is assumed to satisfy the same normalization condition as $\tilde{\rho}$:
	\begin{align*}
		\big\|X\big\|_{S(k)} & = \sup_{\sigma} \big\{ \Tr(X\sigma) : SN(\sigma) \leq k \big\} \\
		& = \sup_{\tilde{\sigma}} \Big\{ \Tr\big((\ketbra{\psi_+}{\psi_+} \otimes X)\tilde{\sigma}\big) : \tilde{\sigma} \text{ is separable} \Big\} \\
		& \geq \sup_{\tilde{\rho}} \Big\{ \Tr\big((\ketbra{\psi_+}{\psi_+} \otimes X)(\tfrac{s}{n^2 + s}\tilde{\rho} + \tfrac{1}{n^2 + s}I \otimes I)\big) : \tilde{\rho} \text{ has $s$-BSE} \Big\} \\
		& = \sup_{\tilde{\rho}} \Big\{ \tfrac{s}{n^2 + s} \Tr\big((\ketbra{\psi_+}{\psi_+} \otimes X)\tilde{\rho}\big) + \tfrac{1}{n^2 + s} \Tr\big(\ketbra{\psi_+}{\psi_+} \otimes X\big) : \tilde{\rho} \text{ has $s$-BSE} \Big\} \\
		& \geq \frac{s}{n^2 + s} \beta_s + \frac{1}{n^2 + s} \Tr(X).
	\end{align*}
	The proof of the corresponding inequality involving $\beta_s$ is extremely similar.
\end{proof}

Note that the bounds provided by Theorem~\ref{thm:sdp_error_k2} are significantly worse than the bounds in the $k = 1$ case provided by Theorem~\ref{thm:sdp_error}. For one thing, the lower bounds do not depend on $k$ at all (other than than inherent dependence of $\alpha_s$ and $\beta_s$ on $k$), so we expect that these lower bounds are quite poor when $k$ is small relative to $n$.

On the other hand, these semidefinite programs in the $k > 1$ case also seem to perform quite a bit worse than their $k = 1$ counterparts when $s$ is small. Even in the extremely simple case of $X = \ketbra{\psi_+}{\psi_+} \in M_3 \otimes M_3$ and $k = 2$, we found that $\beta_s > 0.9999$ for $1 \leq s \leq 7$, even though we saw in Example~\ref{exam:MatrixNormChoi} that $\big\|X\big\|_{S(2)} = 2/3$. Based on Theorem~\ref{thm:sdp_error_k2}, we know that $\big\|X\big\|_{S(2)} \geq \frac{s + 1}{s + 9} \beta_s$ for all $s \geq 1$, so we may not see a value of $\beta_s$ that is significantly different from $1$ until $\frac{s + 1}{s + 9} > \frac{2}{3}$ (i.e., $s > 15$).

\subsection{Examples}\label{sec:examples}

	The methods of computing the $S(k)$-operator norms introduced in Sections~\ref{sec:SP_pos_map} and~\ref{sec:SP_shareable} have been implemented in MATLAB. In order to test the semidefinite programs, we will need a theoretical result to compare the computed results to. To test the semidefinite programs of Section~\ref{sec:SP_pos_map}, we analytically compute the $S(k)$-norms of the family of Werner states \cite{W89} and compare the exact answers to the computational results. We also look at the operator norms of randomly generated states from the Bures measure. To test the semidefinite programs of Section~\ref{sec:SP_shareable}, we return to the matrix of Example~\ref{ex:BadSNorm}.

We begin by deriving the $S(k)$-norm of Werner states. Recall that, given a real number $\alpha \in [-1,1]$, the \emph{Werner state} $\rho_\alpha \in M_n \otimes M_n$ is defined by
\begin{align*}
	\rho_\alpha := \frac{1}{n(n - \alpha)}(I - \alpha S),
\end{align*}
where $S$ is the swap operator. The following result shows that if $\alpha \leq 0$ then $\big\|\rho_\alpha\big\|_{S(k)} = \big\|\rho_\alpha\big\|$ for all $k$. If $\alpha > 0$ then $\big\|\rho_\alpha\big\|_{S(1)}$ is smaller, but the rest of the $S(k)$-norms are all equal to $\big\|\rho_\alpha\big\|$.
\begin{prop}\label{prop:WernerSchmidt}
	Let $\rho_\alpha \in M_n \otimes M_n$ be a Werner state. Then
	\begin{align*}
		\|\rho_\alpha\|_{S(1)} = \frac{1 + | \min\{\alpha,0\}|}{n(n - \alpha)} \ \text{ and } \ \big\|\rho_\alpha\big\|_{S(k)} = \frac{1 + |\alpha|}{n(n - \alpha)} \quad \text{ for } \, 2 \leq k \leq n.
	\end{align*}
\end{prop}
\begin{proof}
	Throughout the proof, we will work with the operator $X_\alpha := n(n - \alpha)\rho_\alpha = I - \alpha S$ to simplify the algebra. To see the result when $\alpha \leq 0$, note that for any $k$,
	\begin{align*}
		\big\|X_\alpha\big\|_{S(k)} = \big\|I - \alpha S\big\|_{S(k)} \leq \big\|I\big\|_{S(k)} - \alpha \big\|S\big\|_{S(k)} = 1 - \alpha,
	\end{align*}
	where the inequality comes from the triangle inequality and the rightmost equality comes from the fact that $\big\|S\big\|_{S(k)} = 1$, which is easily verified. To see the other inequality, choose $\ket{v} := \ket{1} \otimes \ket{1}$ and observe that
	\begin{align*}
		\bra{v}X\ket{v} = (\bra{1}\otimes\bra{1})(I - \alpha S)(\ket{1}\otimes\ket{1}) = 1 - \alpha\sum_{i,j=1}^{n}\braket{1}{i}\braket{j}{1}\braket{1}{j}\braket{i}{1} = 1 - \alpha.
	\end{align*}
		
	On the other hand, if $\alpha \geq 0$, then for any vector $\ket{v} = \ket{a} \otimes \ket{b}$, it follows that
	\begin{align*}
		\bra{v}X_\alpha\ket{v} = (\bra{a}\otimes\bra{b})(I - \alpha S)(\ket{a}\otimes\ket{b}) = 1 - \alpha (\bra{a}\otimes\bra{b})(\ket{b}\otimes\ket{a}) = 1 - \alpha |\braket{a}{b}|^2 \leq 1.
	\end{align*}
	Furthermore, equality can easily be seen to be attained when $\ket{v} = \ket{1} \otimes \ket{2}$, which shows that $\big\| X_\alpha \big\|_{S(1)} = 1$. To see the result for $k \geq 2$ and $\alpha \geq 0$, use the triangle inequality again to see that $\big\| X_\alpha \big\|_{S(k)} \leq 1 + \alpha$. To show that equality is attained, let $\ket{v} = \frac{1}{\sqrt{2}}(\ket{1} \otimes \ket{2} - \ket{2} \otimes \ket{1})$ and observe that $\bra{v}X_\alpha \ket{v} = 1 + \alpha$. Since $\ket{v}$ has $SR(\ket{v}) = 2$, the result follows.
\end{proof}
	
	The performance of the semidefinite programs of Section~\ref{sec:SP_pos_map} for the $S(1)$-norm is analyzed in Table~\ref{table:werner}. If the transpose map is used, then we know that the semidefinite program must give exactly $\|\rho_\alpha\|_{S(1)}$ when $n = 2$, which it does. In fact, the positive map $\Phi$ defined by $\Phi(X) = \Tr(X)I - X$ (see Example~\ref{exam:k_pos}) that is used as the basis of the reduction criterion also gives the correct answer in this case. For $n = 3$, the transpose map still happens to give the correct answer, though the reduction map gives a strict upper bound when $\alpha > 0$.
\begin{table}[ht]\hsp
	\begin{center}
  \begin{tabular}{ c | c | c | c | c }
  	\noalign{\hrule height 0.1em}
  	 & & & \multicolumn{2}{c}{Upper bound computed using...} \\
  	\hline
    $n$ & $\alpha$ & Exact $\|\rho_\alpha\|_{S(1)}$ & $\Phi(X) = X^T$ & $\Phi(X) = \Tr(X)I - X$ \\
  	\noalign{\hrule height 0.1em}
    $2$ & $1/2$ & $1/3$ & $0.3333$ & $0.3333$ \\ \hline
    $2$ & $-1/2$ & $3/10$ & $0.3000$ & $0.3000$ \\ \hline
    $3$ & $1/2$ & $2/15$ & $0.1333$ & $0.2000$ \\ \hline
    $3$ & $-1/2$ & $1/7$ & $0.1429$ & $0.1429$ \\
  	\noalign{\hrule height 0.1em}
  \end{tabular}
	\end{center}
\caption[The $S(1)$-operator norm of various Werner states]{\hsp The exact $S(1)$-operator norm of various Werner states as well as the computed upper bounds obtained by using the semidefinite program defined by one of two different positive linear maps.}\label{table:werner}
\dsp\end{table}
	
	As another example, we consider random density operators distributed according to the Bures measure \cite{B69,Uhl76}, which can be generated via the method of~\cite{OSZ10}. We now investigate the general behaviour of the $S(k)$-norms of a density operator in $M_2 \otimes M_2$ and $M_3 \otimes M_3$ relative to its eigenvalues.
	
	In particular, Figure~\ref{fig:4dim} shows how the $S(1)$-norm is distributed compared to the two largest eigenvalues $\lambda_3 \leq \lambda_4$ in $M_2 \otimes M_2$, based on $2 \times 10^6$ randomly-generated density operators. It is not surprising that the $S(1)$-norm lies between $\lambda_3$ and $\lambda_4$, since $\lambda_4$ is equal to the $S(2)$-norm and Theorem~\ref{prop:lowerBoundEig} says that the $S(n-1)$-norm in $M_n \otimes M_n$ is always at least as big as the second-largest eigenvalue. We see in this case that the $S(1)$-norm typically is much closer to $\lambda_4$ than $\lambda_3$.
\begin{figure}[ht]
\begin{center}
\includegraphics[width=0.9\textwidth]{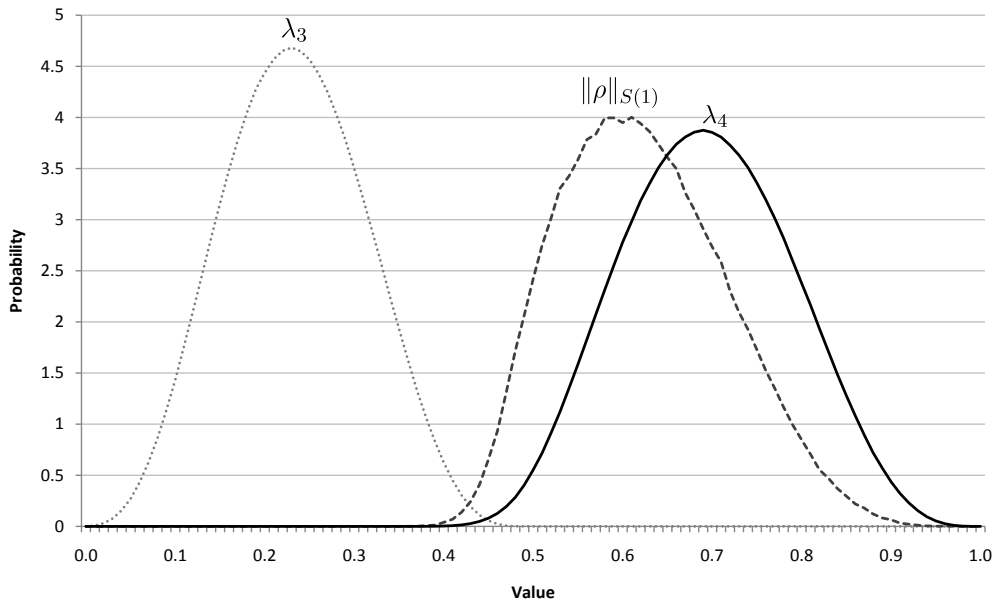}
\end{center}\vspace{-0.3in}
\caption[Approximate distribution of the $S(1)$-norm on $M_2 \otimes M_2$]{\hsp Approximate distributions of the $S(1)$-operator norm and the two largest eigenvalues of random Bures density operators in $M_2 \otimes M_2$.}\label{fig:4dim}
\end{figure}
	
	The $S(1)$-norm in this case was computed using the semidefinite programming method of Section~\ref{sec:semidefProgramMNorm}. A similar plot was presented in \cite{GPMSZ10} for what was called the maximal expectation value among product states, which coincides with the $S(1)$-norm for positive semidefinite operators. There it was similarly observed that this value typically lies closer to $\lambda_4$ than $\lambda_3$ under the Hilbert--Schmidt measure.
	
	Figure~\ref{fig:9dim} shows how the $S(1)$- and $S(2)$-norms typically compare to the two largest eigenvalues $\lambda_8 \leq \lambda_9$ in $M_3 \otimes M_3$, based on $10^5$ randomly-generated density operators. As before, it is not surprising that the $S(2)$-norm lies between $\lambda_8$ and $\lambda_9$. However, Theorem~\ref{prop:lowerBoundEig} also showed that there exist density operators $\rho \in M_3 \otimes M_3$ for which $\lambda_5 \leq \|\rho\|_{S(1)} < \lambda_6$. This situation seems to be extremely rare, as $\|\rho\|_{S(1)}$ generally lies between $\lambda_8$ and $\lambda_9$.
\begin{figure}[ht]
\begin{center}
\includegraphics[width=0.9\textwidth]{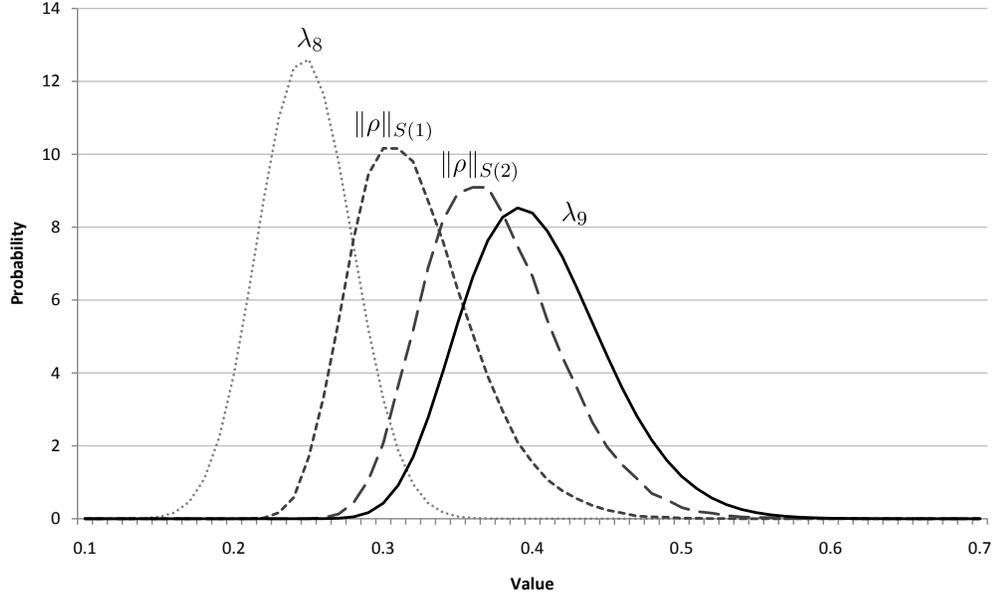}
\end{center}\vspace{-0.3in}
\caption[Approximate distributions of the $S(1)$- and $S(2)$-norms on $M_3 \otimes M_3$]{\hsp Approximate distributions of the $S(1)$- and $S(2)$-operator norms, as well as the two largest eigenvalues of random Bures density operators in $M_3 \otimes M_3$.}\label{fig:9dim}
\end{figure}

Because the semidefinite programming method of Section~\ref{sec:semidefProgramMNorm} does not produce the exact value for the $S(1)$- and $S(2)$-norms in $M_3 \otimes M_3$, the values of the norms used for Figure~\ref{fig:9dim} are estimates that were derived from a simple genetic algorithm.

We now present an example to make use of the semidefinite programs of Section~\ref{sec:SP_shareable}.
\begin{exam}\label{ex:BadSNorm_return}
  {\rm Recall the density matrix
    \begin{align}\label{eq:rho_ex}\sspp
        \rho = \frac{1}{8}\begin{bmatrix}5 & 1 & 1 & 1 \\ 1 & 1 & 1 & 1 \\ 1 & 1 & 1 & 1 \\ 1 & 1 & 1 & 1\end{bmatrix}\dsp
    \end{align}
	that was introduced in Example~\ref{ex:BadSNorm}. We begin by showing that $\|\rho\|_{S(1)} = \frac{1}{8}(3 + 2\sqrt{2})$.
	
	First, let $\ket{v} = \ket{aa}$, where $\ket{a} = \frac{\sqrt{2 + \sqrt{2}}}{2}\ket{1} + \frac{\sqrt{2 - \sqrt{2}}}{2}\ket{2}$. Straightforward computation reveals that
	\begin{align*}\sspp
		\bra{aa}\rho\ket{aa} & = \sspp\frac{1}{128}\begin{bmatrix} 2 + \sqrt{2}, & \sqrt{2}, & \sqrt{2}, & 2 - \sqrt{2}\end{bmatrix} \begin{bmatrix}5 & 1 & 1 & 1 \\ 1 & 1 & 1 & 1 \\ 1 & 1 & 1 & 1 \\ 1 & 1 & 1 & 1\end{bmatrix} \begin{bmatrix} 2 + \sqrt{2} \\ \sqrt{2} \\ \sqrt{2} \\ 2 - \sqrt{2}\end{bmatrix} \\
		& = \sspp\frac{1}{128}\begin{bmatrix} 2 + \sqrt{2}, & \sqrt{2}, & \sqrt{2}, & 2 - \sqrt{2}\end{bmatrix} \begin{bmatrix} 12 + 6\sqrt{2} \\ 4 + 2\sqrt{2} \\ 4 + 2\sqrt{2} \\ 4 + 2\sqrt{2}\end{bmatrix} \\
		& = \frac{1}{128}\big( (36 + 24\sqrt{2}) + (4 + 4\sqrt{2}) + (4 + 4\sqrt{2}) + 4\big) \\
		& = \frac{1}{8}(3 + 2\sqrt{2}).\dsp
	\end{align*}
	
	Thus $\|\rho\|_{S(1)} \geq \frac{1}{8}(3 + 2\sqrt{2})$. To see the opposite inequality, we use the $s = 1$ version of the semidefinite program~\eqref{sp:S1_share_sdp}. In particular, define
	\begin{align*}\sspp
		W = \frac{1}{16}\begin{bmatrix}2\sqrt{2} - 2 & -1 & -1 & 0 \\ -1 & 2\sqrt{2}+2 & -2 & -1 \\ -1 & -2 & 2\sqrt{2}+2 & -1 \\ 0 & -1 & -1 & 2\sqrt{2}+2\end{bmatrix}.
	\dsp\end{align*}
	It is easy to verify that $W \geq 0$ (its eigenvalues are $0, \tfrac{\sqrt{2}}{8}, \tfrac{\sqrt{2}}{16}$, and $\tfrac{2 + \sqrt{2}}{16}$), so the semidefinite program~\eqref{sp:S1_share_sdp} (or equivalently, Theorem~\ref{thm:kPosInf2}) says that $\|\rho\| \leq \big\|\rho + W^\Gamma\big\|$. A simple calculation reveals that $\big\|\rho + W^\Gamma\big\| = \frac{1}{8}(3 + 2\sqrt{2})$, so the desired inequality follows.
	
	Notice that the semidefinite program~\eqref{sp:S1_share_sdp} gives the correct value of $\|\rho\|_{S(1)}$ already in the $s = 1$ case here, as we knew it would, since $\rho \in M_2 \otimes M_2$. In larger dimensions, we cannot expect the upper bounds constructed in this way to be tight for $s = 1$. We also cannot expect the bounds to be tight for any fixed $s$ if we ignore the partial transposition constraint, as we now illustrate.
	
	Define $\beta_s := \big\|P(I_n^{\otimes s} \otimes X)P\big\|$, where $P$ is as it was in Equation~\eqref{eq:s1_share_no_PT}. Table~\ref{table:s1_share_approx} shows the value of $\beta_s$ for $s \leq 20$, as computed by MATLAB. As expected, the values of $\beta_s$ start at $\|\rho\| = 0.75$ when $s = 1$ and then decrease as $s$ increases. Furthermore, the values of $\beta_s$ seem to be decreasing to $\|\rho\|_{S(1)} = \frac{1}{8}(3 + 2\sqrt{2}) \approx 0.7286$, as they should. The lower bounds provided by Theorem~\ref{thm:sdp_error} similarly are increasing to $\|\rho\|_{S(1)}$ -- the lower bound when $s = 1$ is $0.2500$, while the lower bound computed when $s = 20$ is $0.6635$.
\begin{table}[ht]\hsp
	\begin{center}
  \begin{tabular}{ c c | c c | c c | c c }
  	\noalign{\hrule height 0.1em}
    $s$ & $\beta_s$ & $s$ & $\beta_s$ & $s$ & $\beta_s$ & $s$ & $\beta_s$ \\
  	\noalign{\hrule height 0.1em}
    $1$ & $0.7500$ & $6$ & $0.7329$ & $11$ & $0.7310$ & $16$ & $0.7302$ \\ \hline
    $2$ & $0.7405$ & $7$ & $0.7323$ & $12$ & $0.7308$ & $17$ & $0.7301$ \\ \hline
    $3$ & $0.7368$ & $8$ & $0.7318$ & $13$ & $0.7306$ & $18$ & $0.7300$ \\ \hline
    $4$ & $0.7349$ & $9$ & $0.7315$ & $14$ & $0.7304$ & $19$ & $0.7300$ \\ \hline
    $5$ & $0.7337$ & $10$ & $0.7312$ & $15$ & $0.7303$ & $20$ & $0.7299$ \\
  	\noalign{\hrule height 0.1em}
  \end{tabular}
	\end{center}
\caption[Upper bounds for the $S(1)$-operator norm]{\hsp Upper bounds for the $S(1)$-norm of the density matrix $\rho$ given by Equation~\eqref{eq:rho_ex}. Observe that $\beta_1 = \|\rho\|$ and $\beta_s$ seems to be decreasing to $\|\rho\|_{S(1)} \approx 0.7286$, as expected.}\label{table:s1_share_approx}
\dsp\end{table}}
\end{exam}

\section{Bound Entanglement}\label{sec:bound_entanglement}

Recall from Section~\ref{sec:bound_entangle} the NPPT bound entanglement problem, which asked whether or not there exists a state $\rho$ such that $(\rho^\Gamma)^{\otimes r}$ is $2$-block positive for all $r \geq 1$. Also recall that it is enough to consider the NPPT bound entanglement problem on the Werner states \cite{HH99,W89}
\[
    \rho_\alpha := \frac{1}{n(n - \alpha)}(I - \alpha S) \in M_n \otimes M_n,
\]
where $\alpha \in [-1,1]$ and $S$ is the swap operator.

Because the partial transpose of Werner states have only two distinct eigenvalues (as noted in the proof of the following proposition), Corollary~\ref{cor:kposTwoEvals} applies to this situation and the $S(k)$-operator norms are a natural tool for approaching this problem. The following result is a starting point.
\begin{prop}\label{prop:werner01}
    Let $\rho_\alpha \in M_n \otimes M_n$ be a Werner state. Then $\rho_\alpha^\Gamma$ is $k$-block positive if and only if $\alpha \leq \frac{1}{k}$.
\end{prop}
\begin{proof}
  Simply note that $(n^2 - \alpha n)\rho_\alpha^\Gamma = I - \alpha n\ketbra{\psi_+}{\psi_+}$ has only two distinct eigenvalues: $1$ and $1 - \alpha n$. Corollary~\ref{cor:kposTwoEvals} then implies that $\rho_\alpha^\Gamma$ is $k$-block positive if and only if $\big\|\ketbra{\psi_+}{\psi_+}\big\|_{S(k)} \leq \frac{1}{\alpha n}$. We saw in Example~\ref{exam:MatrixNormChoi} that $\big\|\ketbra{\psi_+}{\psi_+}\big\|_{S(k)} = \frac{k}{n}$, so the result follows.
\end{proof}

The special case $k = n$ of the above proposition is very well-known and states that $\rho_\alpha$ is PPT if and only if
$\alpha \leq \frac{1}{n}$. Moreover, Proposition~\ref{prop:werner01}
shows that Werner states cannot be bound entangled for $\alpha > \frac{1}{2}$, which is also well-known. It has been conjectured that Werner states are bound entangled for all $\alpha \leq \frac{1}{2}$; this is exactly the set of values for which $\rho_\alpha^\Gamma$ is $2$-positive.

Although we now have determined $k$-block positivity of $\rho_\alpha^\Gamma$, determining $k$-block positivity (or even $2$-block positivity) of $(\rho_\alpha^\Gamma)^{\otimes r}$ for $r > 1$ is not so simple in general because the projection onto the negative eigenspaces is no longer rank-$1$, so we cannot exactly compute its $S(k)$-norm. Additionally, $(\rho_\alpha^\Gamma)^{\otimes r}$ has more than two distinct eigenvalues in general so we can no longer use Corollary~\ref{cor:kposTwoEvals}. To simplify the problem somewhat, consider the $\alpha = \frac{2}{n}$ case. Notice that this value of $\alpha$ is in the ``region of interest'' $(1/n,1/2]$ if and only if $n \geq 4$ -- an assumption that we make for the remainder of this section.

In this case, the operator $X := (n^2 - 2)\rho_{2/n} = I - 2\ketbra{\psi_+}{\psi_+}^\Gamma$ has eigenvalues $1$ and $-1$, so $(X^\Gamma)^{\otimes r}$ has only those two distinct eigenvalues regardless of $r$. Corollary~\ref{cor:kposTwoEvals} then says that $\rho_{2/n}$ is bound entangled if and only if $\big\| P_{n,r}^{-} \big\|_{S(2)} \leq \frac{1}{2}$ for all $r \geq 1$, where $P_{n,r}^{-}$ is the projection onto the $-1$ eigenspace of $(\rho_{2/n}^\Gamma)^{\otimes r}$. This procedure mirrors the approach attempted in \cite{PPHH10} to find a bound entangled NPPT Werner state, though that paper considers the $n = 4$ case exclusively.

Before using the computational techniques introduced in this chapter to estimate $\big\| P_{n,r}^{-} \big\|_{S(2)}$, we show that in the limit as $r$ tends to infinity, it is not possible to do any better than $\big\| P_{n,r}^{-} \big\|_{S(2)} \leq \frac{1}{2}$. More precisely, we show that
\[
    \lim_{r \to \infty} \big\| P_{n,r}^{-} \big\|_{S(2)} \geq \frac{1}{2}.
\]

To prove this claim, observe that
\begin{align}\label{eq:proj_def}\begin{split}
	P_{n,1}^{-} & = \ketbra{\psi_+}{\psi_+} \in M_n \otimes M_n, \\
	P_{n,r}^{-} & = (I - P_{n,1}^{-}) \otimes P_{n,r-1}^{-} + P_{n,1}^{-} \otimes (I - P_{n,r-1}^{-}) \quad \forall \, r \geq 2.
\end{split}\end{align}
In particular, this means that ${\rm rank}(P_{n,1}^{-}) = 1$ and ${\rm rank}(P_{n,r}^{-}) = {\rm rank}(P_{n,r-1}^{+}) + (n^2 - 1){\rm rank}(P_{n,r-1}^{-})$ for all $r \geq 2$. Standard techniques for solving recurrence relations then show that ${\rm rank}(P_{n,r}^-) = \frac{1}{2}(n^{2r} - (n^2 - 2)^r)$ for all $r \geq 1$. Plugging this into the lower bound
\begin{align*}
	\big\|P_{n,r}^-\big\|_{S(2)} \geq \frac{(n^r-2){\rm rank}(P_{n,r-1}^{-})}{n^{2r}(n^r-1)} + \frac{1}{n^r-1},
\end{align*}
which follows from throwing away the square root term in Inequality~\eqref{eq:projIneq2}, reveals that
\begin{align*}
    \big\|P_{n,r}^-\big\|_{S(2)} & \geq \frac{(n^r-2)(n^{2r} - (n^2 - 2)^r)}{2n^{2r}(n^r-1)} + \frac{1}{n^r-1} \\
    & = \frac{n^r-2}{2(n^r-1)} - \frac{(n^r-2)(n^2 - 2)^r - 2n^{2r}}{2n^{2r}(n^r-1)}.
\end{align*}
It is not difficult to verify that the lower bound on the right is always, for $n \geq 4$, strictly less than $\frac{1}{2}$. Furthermore, as $r \rightarrow \infty$, the rightmost fraction tends to zero and the left fraction tends to $\frac{1}{2}$. This shows that, asymptotically, $\frac{1}{2}$ is the smallest that we could ever hope $\big\|P_{n,r}^-\big\|_{S(2)}$ to be. We have thus proved the following.

\begin{thm}\label{thm:bound_ent_s2}
The Werner state $\rho_{2/n}$ is bound entangled if and only if
\[
    \lim_{r \to \infty} \big\| P_{n,r}^{-} \big\|_{S(2)} = \frac{1}{2}.
\]
\end{thm}

In order to make progress on the NPPT bound entanglement problem via Theorem~\ref{thm:bound_ent_s2}, we now present the best bounds that we have on $\big\| P_{n,r}^{-} \big\|_{S(2)}$. Although these $S(2)$-norms are still unknown, we can analytically compute the $S(1)$-norm of each of these projections using the semidefinite programming method of the previous sections.
\begin{prop}\label{prop:S1}
	Let $P_{n,r}^{-}$ be the projection defined by Equations~\eqref{eq:proj_def}. Then
	\begin{align*}
		\big\|P_{n,r}^{-}\big\|_{S(1)} = \frac{1}{2} - \frac{1}{2}\left(1 - \frac{2}{n}\right)^r.
	\end{align*}
\end{prop}
\begin{proof}
	To see the ``$\geq$'' inequality, consider the separable vector $\ket{v} := \ket{11} \otimes \ket{11\cdots 1} \in (\bb{C}^n)^{\otimes 2} \otimes (\bb{C}^n)^{\otimes {2r-2}}$. Then define the quantity
	\begin{align*}
		c_{n,r} := \bra{v} P_{n,r}^{-} \ket{v}.
	\end{align*}
It follows that
	\begin{align*}
		c_{n,r} & = \bra{11}(I - P_{n,1}^{-})\ket{11} \bra{11\cdots 1}P_{n,r-1}^{-}\ket{11\cdots 1} \\
		& \quad \ + \bra{11}P_{n,1}^{-}\ket{11} \bra{11\cdots 1}(I - P_{n,r-1}^{-})\ket{11\cdots 1} \\
		 & = \frac{n-1}{n}c_{n,r-1} + \frac{1}{n}(1 - c_{n,r-1}) \\
		 & = \Big(1 - \frac{2}{n}\Big)c_{n,r-1} + \frac{1}{n}.
	\end{align*}
Standard methods for solving recurrence relations yield $c_{n,r} = \frac{1}{2} - \frac{1}{2}\big(1 - \frac{2}{n}\big)^r$. Noting that $\big\|P_{n,r}^{-}\big\|_{S(1)} \geq c_{n,r}$ gives the desired inequality.
	
	To see the ``$\leq$'' inequality, we will use the dual form of the semidefinite program~\eqref{sp:matrixNorm} with the transpose map $\Phi_1(X) := X^T$. To this end, notice that if $\lambda_{n,r}^{\textup{max}}$ is the maximal eigenvalue of $(P_{n,r}^{-})^\Gamma$, then $\lambda_{n,r}^{\textup{max}}I - (P_{n,r}^{-})^\Gamma$ is positive semidefinite and so Theorem~\ref{thm:kPosInf2} says that
	\begin{align*}
		\big\|P_{n,r}^{-}\big\|_{S(1)} \leq \big\| P_{n,r}^{-} + (\lambda_{n,r}^{\textup{max}}I - (P_{n,r}^{-})^\Gamma)^\Gamma \big\| = \big\| \lambda_{n,r}^{\textup{max}}I \big\| = \lambda_{n,r}^{\textup{max}}.
	\end{align*}
In order to compute $\lambda_{n,r}^{\textup{max}}$, let us consider the partial transpose of the family of projections~\eqref{eq:proj_def}:
\begin{align*}
	(P_{n,1}^{-})^\Gamma & = \frac{1}{n}S \in M_n \otimes M_n, \\
	(P_{n,r}^{-})^\Gamma & = \frac{1}{n}S \otimes (I - (P_{n,r-1}^{-})^\Gamma) + (I - \frac{1}{n} S) \otimes (P_{n,r-1}^{-})^\Gamma \quad \forall \, r \geq 2.
\end{align*}
It is clear that the eigenvectors of $(P_{n,r}^{-})^\Gamma$ are each of the form $\ket{x} \otimes \ket{y}$ for some eigenvector $\ket{x}$ of $S$ and some eigenvector $\ket{y}$ of $(P_{n,r-1}^{-})^\Gamma$. If we recall that the eigenvalues of $S$ are $\pm 1$, it follows that
	\begin{align*}
		\lambda_{n,r}^{\textup{max}} = \max\Big\{ (1 - \frac{2}{n})\lambda_{n,r-1}^{\textup{max}} + \frac{1}{n} , (1 + \frac{2}{n})\lambda_{n,r-1}^{\textup{max}} - \frac{1}{n} \Big\}.
	\end{align*}
If $\lambda_{n,r-1}^{\textup{max}} \leq \frac{1}{2}$ then $(1 + \frac{2}{n})\lambda_{n,r-1}^{\textup{max}} - \frac{1}{n} \leq (1 - \frac{2}{n})\lambda_{n,r-1}^{\textup{max}} + \frac{1}{n} \leq \frac{1}{2}$, so it follows via induction (and the fact that $\lambda_{n,1}^{\textup{max}} = \frac{1}{n} \leq \frac{1}{2}$) that $\lambda_{n,r}^{\textup{max}} = (1 - \frac{2}{n})\lambda_{n,r-1}^{\textup{max}} + \frac{1}{n}$. We already saw that this recurrence relation has the closed form $\lambda_{n,r}^{\textup{max}} = \frac{1}{2} - \frac{1}{2}\big(1 - \frac{2}{n}\big)^r$, which finishes the proof.
\end{proof}

Proposition~\ref{prop:S1} shows that not only does $\big\|P_{n,r}^{-}\big\|_{S(2)}$ approach $1/2$ from below as $r \rightarrow \infty$, but even $\big\|P_{n,r}^{-}\big\|_{S(1)}$ does, and it does so exponentially quickly. One way to tackle the problem of computing $\big\|P_{n,r}^{-}\big\|_{S(2)}$ would be to hope that $\big\|P_{n,r}^{-}\big\|_{S(1)} = \big\|P_{n,r}^{-}\big\|_{S(2)}$ -- we now show that this is not the case. It is worth pointing out that the following proposition shows the best lower and upper bounds on $\big\|P_{n,r}^{-}\big\|_{S(2)}$ that we have.
\begin{prop}\label{prop:S2}
	Let $n \geq 3$ and let $P_{n,r}^{-}$ be the projection defined by Equations~\eqref{eq:proj_def}. Then
	\begin{align}\label{eq:S2_lowbound}
		\big\|P_{n,r}^{-}\big\|_{S(2)} & \geq \frac{1}{2} - \left(\frac{1}{2} - \frac{1}{n-2}\right)\left(1 - \frac{2}{n}\right)^r \quad \text{and} \\ \label{eq:S2_upbound}
		\big\|P_{n,r}^{-}\big\|_{S(2)} & \leq 1 - \left(1 - \frac{2}{n}\right)^r.
	\end{align}
\end{prop}
\begin{proof}
	Inequality~\eqref{eq:S2_upbound} simply follows from Proposition~\ref{prop:S1} and Theorem~\ref{thm:MatrixEquiv01}. For Inequality~\eqref{eq:S2_lowbound}, we construct a specific vector $\ket{v} \in (\bb{C}^n)^{\otimes 2r}$ with $SR(\ket{v}) = 2$ such that $\bra{v}P_{n,r}^{-}\ket{v}$ is the given quantity.
	
	To this end, let $\ket{v} = \frac{1}{\sqrt{2}}(\ket{11} \otimes \ket{11\cdots 1} + \ket{22} \otimes \ket{11\cdots 1}) \in (\bb{C}^n)^{\otimes 2} \otimes (\bb{C}^n)^{\otimes {2r-2}}$ and define the following quantity:
	\begin{align*}
		c_{n,r} & := \bra{11\cdots 1} P_{n,r}^{-} \ket{11\cdots 1}.
	\end{align*}
	We already saw in the proof of Proposition~\ref{prop:S1} that $c_{n,r} = \frac{1}{2} - \frac{1}{2}\big(1 - \frac{2}{n}\big)^r$. Thus we get (omitting some messy algebra at the start of the calculation):
	\begin{align*}
		\bra{v}P_{n,r}^{-}\ket{v} & = \frac{n-1}{n}c_{n,r-1} + \frac{1}{n}(1 - c_{n,r-1}) - \frac{1}{n}c_{n,r-1} + \frac{1}{n}(1 - c_{n,r-1}) \\
		& = \frac{n-4}{n} c_{n,r-1} + \frac{2}{n} \\
		& = \frac{n-4}{2n} - \frac{n-4}{2n}\left(1 - \frac{2}{n}\right)^{r-1} + \frac{2}{n} \\
		& = \frac{1}{2} - \left(\frac{1}{2} - \frac{1}{n-2}\right)\left(1 - \frac{2}{n}\right)^{r},
	\end{align*}
	as desired.
\end{proof}

We believe that the lower bound in Proposition~\ref{prop:S2} is in fact tight, which we now state formally as a conjecture.
\begin{conj}\label{conj:S2}
	Let $n \geq 3$ and let $P_{n,r}^{-}$ be the projection defined by Equations~\eqref{eq:proj_def}. Then
	\begin{align*}
		\big\|P_{n,r}^{-}\big\|_{S(2)} & = \frac{1}{2} - \left(\frac{1}{2} - \frac{1}{n-2}\right)\left(1 - \frac{2}{n}\right)^r.
	\end{align*}
\end{conj}
Conjecture~\ref{conj:S2}, if true, immediately implies that NPPT bound entangled states exist via Theorem~\ref{thm:bound_ent_s2}. In the $r = 1$ case, the conjecture reduces to the statement $\big\|P_{n,1}^{-}\big\|_{S(2)} = \frac{2}{n}$, which was proved in Example~\ref{exam:MatrixNormChoi}. When $r = 2$, the conjecture says that $\big\|P_{n,2}^{-}\big\|_{S(2)} = \frac{3n-4}{n^2}$, which was proved in \cite[Proposition~6]{PPHH10} under the additional assumption that the supremum that defines $\big\|P_{n,2}^{-}\big\|_{S(2)}$ is attained by a vector $\ket{v}$ such that ${\rm mat}(\ket{v})$ is normal. The conjecture in general remains open for all $r \geq 2$.

For the remainder of this section we consider concrete consequences of Inequality~\ref{eq:S2_upbound}. In particular, we provide a range of values for $\alpha$ and $n$ so that $\rho_\alpha \in M_n \otimes M_n$ is $r$-copy undistillable (note that $\alpha$ and $n$ both depend on $r$, so we do not provide a single state that is $r$-copy undistillable for all $r$).
\begin{thm}\label{thm:main}
	Let $n,r \in \mathbb{N}$ be such that $p := \frac{(n-2)^r}{n^r - (n-2)^r} \geq 1$. If $r$ is odd and $\alpha \leq \frac{1}{n}(\sqrt[r]{p} + 1)$ or if $r$ is even and $\alpha \leq \frac{1}{n}(\sqrt[r-1]{p} + 1)$, then the Werner state $\rho_\alpha \in M_n \otimes M_n$ is $r$-copy undistillable.
\end{thm}
\begin{proof}
	Let $p$ be as in the statement of the theorem and let $\alpha = \frac{1}{n}(\sqrt[2\lceil r/2\rceil - 1]{p} + 1)$. The eigenvalues of $(\rho_\alpha^{\otimes r})^\Gamma$ are
\begin{align*}
	(1 - \alpha n)^m	\quad \text{for $m = 0, 1, \ldots, r$}.
\end{align*}
Well, $p \geq 1$ implies that $\alpha \geq \frac{2}{n}$, so the minimal positive eigenvalue $\lambda_{\textup{min}}^{+}$ of $(\rho_\alpha^{\otimes r})^\Gamma$ is $1$, and its maximal (in absolute value) negative eigenvalue $\lambda_{\textup{max}}^{-}$ is $(1 - \alpha n)^{2\lceil r/2\rceil - 1}$. We thus have
\begin{align*}
	\lambda_{\textup{min}}^{+} = 1 & = (\alpha n - 1)^{2\lceil r/2\rceil - 1}\left(\Big(\frac{n}{n-2}\Big)^r - 1\right) \geq \lambda_{\textup{max}}^{-}\frac{\| P_{n,r}^{-} \|_{S(2)}}{1 - \| P_{n,r}^{-} \|_{S(2)}},
\end{align*}
where the second equality comes from the fact that $\alpha = \frac{1}{n}\Big(\sqrt[2\lceil r/2\rceil - 1]{\frac{(n-2)^r}{n^r - (n-2)^r}} + 1\Big)$, and the final inequality comes from Inequality~\eqref{eq:S2_upbound}. Now by condition (b) of Theorem~\ref{thm:kposSpectral}, we have that $(\rho_\alpha^{\otimes r})^\Gamma$ is $2$-block positive, and hence the result follows.
\end{proof}

Note that the value $p$ of Theorem~\ref{thm:main} is such that $p \geq 1$ if and only if $n \geq \frac{2\sqrt[r]{2}}{\sqrt[r]{2}-1}$. Thus, for any $r \geq 1$, there is always some non-PPT Werner state that is $r$-copy undistillable as long as the dimension $n$ is large enough. In fact, the dimension grows roughly linearly: $\frac{2\sqrt[r]{2}}{\sqrt[r]{2}-1}$ is asymptotic to $\frac{2}{\ln(2)}r + 1$. Also, if $p \geq 1$ then the result immediately implies that the $\alpha = 2/n$ Werner state is $r$-copy undistillable. It is not difficult to see that if $\rho_{\alpha} \in M_n \otimes M_n$ is $r$-copy undistillable then $\rho_\alpha \in M_m \otimes M_m$ is also $r$-copy undistillable for any $m \leq n$, so we then immediately arrive at the following slightly weaker (but much simpler) corollary of Theorem~\ref{thm:main}.
\begin{cor}
	If $\alpha \leq \min\big\{2/n, \ln(2)/(r + 3\ln(2) - 1)\big\}$ then $\rho_{\alpha} \in M_n \otimes M_n$ is $r$-copy undistillable.
\end{cor}
Similar results about $r$-copy undistillability of Werner states have appeared in the literature in the past \cite{BR03}. Notably, in \cite{LBCKKSST00} it was shown that, for any fixed $n \geq 3$, there exist NPPT Werner states that are $r$-copy undistillable, though the region that was shown to be $r$-copy undistillable shrinks exponentially with $r$. Our result is stronger in that our regions shown to be $r$-copy undistillable shrink only like $1/r$. On the other hand, for each fixed $n$ our result only gives a region of NPPT $r$-copy undistillability for $r \leq \ln(2)(n - 3) + 1$.

\section{Minimum Gate Fidelity}\label{sec:min_gate_fid}

If a quantum channel $\cl{U}$ satisfies $\cl{U}(\rho) = U\rho U^\dagger$ for some unitary operator $U$, $\cl{U}$ is called a \emph{unitary channel}. Unitary channels are exactly the channels that do not introduce mixedness (i.e., decoherence) into states and thus they are often the types of channels that are meant to be implemented in experimental settings. However, no implementation of a channel is perfect -- errors are introduced that cause the channel that is implemented to not actually be unitary. The \emph{gate fidelity} is a useful tool for comparing how well the implemented quantum channel $\cl{E}$ approximates the desired unitary channel $\cl{U}$. Gate fidelity is a function defined on pure states as follows:
\begin{align*}
	\cl{F}_{\cl{E},\cl{U}}(\ket{v}) := \cl{F}(\cl{E}(\ketbra{v}{v}),\cl{U}(\ketbra{v}{v})) = \bra{v}U^\dagger \cl{E}(\ketbra{v}{v})U\ket{v},
\end{align*}
where we recall from Section~\ref{sec:fidelity} that $\cl{F}(\cdot,\cdot)$ is the fidelity between states.
	
Without loss of generality, we can assume $U = I$ by noting that
\begin{align*}
	\bra{v}U^\dagger \cl{E}(\ketbra{v}{v})U\ket{v} = \bra{v} (\cl{U}^\dagger \circ \cl{E})(\ketbra{v}{v})\ket{v},
\end{align*}
so $\cl{F}_{\cl{E},\cl{U}} = \cl{F}_{\cl{U}^\dagger \circ \cl{E},id_n}$. For brevity, we will use the shorthand $\cl{F}_{\cl{E}} := \cl{F}_{\cl{E},id_n}$, which can be thought of as measuring how much noise the channel $\cl{E}$ introduces to a given pure state.
	
The two most well-studied distance measures based on the gate fidelity are the \emph{average gate fidelity} $\overline{\cl{F}_{\cl{E}}}$ \cite{HHH99,Nie02,BOSBJ02,GLN05,EAZ05} and the \emph{minimum gate fidelity} \cite{NC00,GLN05}
\begin{align}
	\cl{F}_{\cl{E}}^{\textup{min}} := \min_{\ket{v}} \cl{F}_{\cl{E}} (\ket{v}),
\end{align}
which are obtained by either averaging (via the Fubini-Study measure \cite{BZ06}) or minimizing over all pure states $\ket{v}$, respectively. The minimum gate fidelity has the interpretation as the most noise that $\cl{E}$ can introduce into a quantum system. It makes sense then that one might want instead to minimize $\cl{F}(\cl{E}(\rho),\rho)$ over all mixed states $\rho$. The reason we minimize over pure states is that joint concavity of fidelity implies that minimizing over mixed states $\rho$ gives the exact same quantity $\cl{F}_{\cl{E}}^{\textup{min}}$ as minimizing over pure states $\ket{v}$ -- proofs of this fact are contained in \cite[Section 9.3]{NC00} and \cite[Section IV.C]{GLN05}.

It is well-known that the average gate fidelity is easily-computable, and many formulas for computing it have appeared over the years \cite{HHH99,JK11b,Nie02}. Similarly, formulas for the variance and higher-order moments of the gate fidelity have been derived \cite{MBE11,PMM11}. However, computing the minimum gate fidelity has proved to be much more difficult -- some partial results are known \cite{KBO09,LRKKB11}, but no easy method of calculation is known in general. In this section, we make significant progress on this problem by showing that the minimum gate fidelity can be written in terms of the $S(1)$-norm, which allows all of the computational methods and inequalities already presented to apply in this setting. Furthermore, we show that computing the minimum gate fidelity of an arbitrary quantum channel is NP-hard.

Our starting point is the following simple lemma, which shows how minimum gate fidelity is related to separable states. Note that for this result, and the remainder of this section, we assume that the transpose in the partial transposition map is applied to the first subsystem (i.e., $X^\Gamma = (T \otimes id)(X)$). This assumption just simplifies some algebra -- most of our results do not depend on which subsystem the transpose is applied to.
\begin{lemma}\label{lem:fidel_choi}
	Let $\Phi : M_m \rightarrow M_n$ be a linear map and let $\ket{v} \in \mathbb{C}^m$, $\ket{w} \in \mathbb{C}^n$. Then
	\begin{align*}
		\bra{x}\Phi(\ketbra{v}{w})\ket{y} = \bra{wx}C_{\Phi}^\Gamma\ket{vy}.
	\end{align*}	
\end{lemma}
\begin{proof}
	The proof is by simple algebra.
	\begin{align*}
		\bra{wx}C_{\Phi}^\Gamma\ket{vy} & = \sum_{i,j=1}^{m}\bra{wx}(T(\ketbra{i}{j}) \otimes \Phi(\ketbra{i}{j}))\ket{vy} \\
		& = \sum_{i,j=1}^{m}\braket{w}{j}\braket{i}{v} \bra{x}\Phi(\ketbra{i}{j})\ket{y} \\
		& = \bra{x}\Phi\big((\sum_{i=1}^{m}\braket{i}{v}\ket{i}) (\sum_{j=1}^{m}\braket{w}{j} \bra{j})\big)\ket{y} \\
		& = \bra{x}\Phi(\ketbra{v}{w})\ket{y}.
	\end{align*}
\end{proof}
\noindent In particular, Lemma~\ref{lem:fidel_choi} says that $\cl{F}_{\cl{E}} (\ket{v}) = \bra{vv}C_{\cl{E}}^\Gamma\ket{vv}$.

\subsection{Connection with the S(1)-Norm}\label{sec:min_gate_fid_connection}

We now demonstrate how the minimum gate fidelity can be written in terms of the $S(1)$-operator norm. Recall from Section~\ref{sec:symmetric_sub} that $P_{\cl{S}}$ is the orthogonal projection onto the symmetric subspace of $\bb{C}^n \otimes \bb{C}^n$. 
\begin{thm}\label{thm:min_fid}
	Let $\cl{E} : M_n \rightarrow M_n$ be a quantum channel and let $\lambda_{\textup{max}}$ be the maximal eigenvalue of $P_{\cl{S}}C_{\cl{E}}^\Gamma P_{\cl{S}}$. Then
	\begin{align*}
		\cl{F}_{\cl{E}}^{\textup{min}} = \lambda_{\textup{max}} - \big\| P_{\cl{S}}(\lambda_{\textup{max}} I - C_{\cl{E}}^\Gamma)P_{\cl{S}} \big\|_{S(1)}.
	\end{align*}
\end{thm}
\begin{proof}
	Using Lemma~\ref{lem:fidel_choi} with $\Phi := \cl{E}$ reveals that
	\begin{align}\label{eq:min_formula}
		\cl{F}_{\cl{E}}^{\textup{min}} = \min_{\ket{v}}\big\{ \bra{vv}C_{\cl{E}}^\Gamma \ket{vv} \big\} = \lambda_{\textup{max}} - \max_{\ket{v}}\big\{ \bra{vv}(\lambda_{\textup{max}} I - C_{\cl{E}}^\Gamma)\ket{vv} \big\}.
	\end{align}
	For convenience, define $X := P_{\cl{S}}(\lambda_{\textup{max}} I - C_{\cl{E}}^\Gamma)P_{\cl{S}}$. Notice that $X$ is positive semidefinite. Also note that
	\begin{align*}
		\max_{\ket{v}}\big\{ \bra{vv}(\lambda_{\textup{max}} I - C_{\cl{E}}^\Gamma)\ket{vv} \big\} \leq \max_{\ket{v},\ket{w}}\big\{ \bra{vw}X\ket{vw} \big\} = \big\| X \big\|_{S(1)},
	\end{align*}
	where the equality follows from Proposition~\ref{prop:MatMultDiffVectors}. To see that the opposite inequality holds as well (and hence complete the proof), suppose $\ket{w} \neq \ket{v}$ and observe that $P_{\cl{S}}\ket{vw} = \frac{1}{2}(\ket{vw} + \ket{wv})$ is a scalar multiple of a symmetric state with Schmidt rank $2$. It follows via the Takagi factorization that we can write $P_{\cl{S}}\ket{vw} = \alpha\ket{xx} + \beta\ket{yy}$ for some $\ket{x},\ket{y} \in \bb{C}^n$ and $\alpha,\beta \geq 0$. Suppose without loss of generality that
	\begin{align*}
		\bra{xx}X\ket{xx} \geq \bra{yy}X\ket{yy}.
	\end{align*}
	
	Now write $X$ in its Spectral Decomposition as $X = \sum_i \lambda_i \ketbra{v_i}{v_i}$ and define the $i^{th}$ component of two vectors $x^\prime$ and $y^\prime$ by $x_i^\prime := \sqrt{\lambda_i}|\braket{v_i}{xx}|$ and $y_i^\prime := \sqrt{\lambda_i}|\braket{yy}{v_i}|$. Applying the Cauchy--Schwarz inequality to $x^\prime$ and $y^\prime$ shows
	\begin{align*}
		|\bra{yy}X\ket{xx}| \leq \sqrt{\bra{xx}X\ket{xx}}\sqrt{\bra{yy}X\ket{yy}} \leq \bra{xx}X\ket{xx}.
	\end{align*}
	Putting all of this together shows that
	\begin{align*}
		\bra{vw}X\ket{vw} & = (\alpha\bra{xx} + \beta\bra{yy})X(\alpha\ket{xx} + \beta\ket{yy}) \\
		& = \alpha^2\bra{xx}X\ket{xx} + \alpha\beta(\bra{xx}X\ket{yy} + \bra{yy}X\ket{xx}) + \beta^2\bra{yy}X\ket{yy} \\
		& \leq (\alpha^2 + \beta^2)\bra{\rho\rho}X\ket{\rho\rho} + \alpha\beta(|\bra{\rho\rho}X\ket{\sigma\sigma}| + |\bra{\sigma\sigma}X\ket{\rho\rho}|) \\
		& \leq (\alpha^2 + 2\alpha\beta + \beta^2)\bra{xx}X\ket{xx} \\
		& = (\alpha + \beta)^2\bra{xx}X\ket{xx}.
	\end{align*}
	
	Thus, if we can prove that $\alpha + \beta \leq 1$ then we are done. To this end, first note that without loss of generality we can assume that $\braket{x}{y}$ is real, simply by adjusting the global phase between $\ket{x}$ and $\ket{y}$ appropriately. Now recall from the Takagi factorization that $\alpha$ and $\beta$ are the square roots of the eigenvalues of the matrix
	\begin{align*}
		AA^\dagger & := \frac{1}{4}\big(\ket{x}\overline{\bra{y}} + \ket{y}\overline{\bra{x}}\big)\big(\overline{\ket{y}}\bra{x} + \overline{\ket{x}}\bra{y}\big) \\
		& = \frac{1}{4}\big(\ketbra{x}{x} + \braket{x}{y} (\ketbra{y}{x} + \ketbra{x}{y}) + \ketbra{y}{y}\big).
	\end{align*}
	It is easily verified that eigenvectors of $AA^\dagger$ are $\ket{x} \pm \ket{y}$ and the associated eigenvalues are $\frac{1}{4}\big( \braket{x}{y} \pm 1 \big)^2$. If we add the square roots of these eigenvalues, we get
	\begin{align*}
		\alpha + \beta = \frac{1}{2}\big|\braket{x}{y} + 1\big| + \frac{1}{2}\big|\braket{x}{y} - 1\big| = 1,
	\end{align*}
	where the final equality follows from the fact that $-1 \leq \braket{x}{y} \leq 1$.
\end{proof}

The inequalities of Section~\ref{sec:MatrixNormInequalities} and the computational methods from earlier in this chapter now all immediately apply to the minimum gate fidelity. We present a brief selection of these results here for completeness.
\begin{cor}
	Let $\cl{E} : M_n \rightarrow M_n$ be a quantum channel. Denote the eigenvalues of $P_{\cl{S}}C_{\cl{E}}^\Gamma P_{\cl{S}}$ supported on $P_{\cl{S}}$ by $\lambda_1 \geq \lambda_2 \geq \cdots \geq \lambda_{n(n+1)/2}$ (i.e., these are the eigenvalues of $P_{\cl{S}}C_{\cl{E}}^\Gamma P_{\cl{S}}$ with $n(n-1)/2$ zero eigenvalues removed). Let $\alpha_j$ be the maximal Schmidt coefficient of the eigenvector corresponding to $\lambda_j$. Then
	\begin{align*}
		\max_{j}\{ (\lambda_1 - \lambda_j) \alpha_j^2 \} \leq \lambda_1 - \cl{F}_{\cl{E}}^{\textup{min}} \leq \min\big\{\lambda_1 - \lambda_{n(n+1)/2}, \sum_j (\lambda_1 - \lambda_j) \alpha_j^2\big\}.
	\end{align*}
\end{cor}
\begin{proof}
	The fact that $\lambda_1 - \cl{F}_{\cl{E}}^{\textup{min}} \leq \lambda_1 - \lambda_{n(n+1)/2}$ follows immediately from Theorem~\ref{thm:min_fid} and the fact that $\|\cdot\|_{S(1)} \leq \|\cdot\|$. The other upper bound of $\lambda_1 - \cl{F}_{\cl{E}}^{\textup{min}}$ follows from Theorem~\ref{thm:sk_vector_norm} and Proposition~\ref{prop:matrixVectorNorms}. The lower bound can be derived by using the spectral decomposition to write
	\begin{align*}
		P_{\cl{S}}C_{\cl{E}}^\Gamma P_{\cl{S}} = \sum_{j} \lambda_j \ketbra{v_j}{v_j}.
	\end{align*}
	If $\ket{v} \in \bb{C}^n \otimes \bb{C}^n$ is the separable state corresponding to the maximal Schmidt coefficient $\alpha_j$ of $\ket{v_j}$ then
	\begin{align*}
		\bra{v}P_{\cl{S}}(\lambda_1 I - C_{\cl{E}}^\Gamma)P_{\cl{S}}\ket{v} & = \sum_i (\lambda_1 - \lambda_i) |\braket{v_i}{v}|^2 \\
		& = (\lambda_1 - \lambda_j) \alpha_j^2 + \sum_{i\neq j} (\lambda_1 - \lambda_i) |\braket{v_i}{v}|^2 \\
		& \geq (\lambda_1 - \lambda_j) \alpha_j^2.
	\end{align*}
	The corresponding lower bound follows by letting $j$ range from $1$ to $n(n+1)/2$.
\end{proof}

When $n = 2$, the $S(1)$-norm can be efficiently computed to any desired accuracy via the semidefinite programs of either Section~\ref{sec:SP_pos_map} or~\ref{sec:SP_shareable}. As a corollary of this fact, we now have a semidefinite program for efficiently computing the minimum gate fidelity of qubit channels $\cl{E} : M_2 \rightarrow M_2$. The primal and dual forms of the semidefinite program in question are as follows:
\begin{align*}
\begin{matrix}
\begin{tabular}{r l}
\multicolumn{2}{c}{{\bf Primal problem}} \\
\text{minimize:} & $\lambda_1 - \Tr\big(P_{\cl{S}}(\lambda_1 I - C_{\cl{E}}^\Gamma) P_{\cl{S}}\rho\big)$ \\
\text{subject to:} & $\rho, \rho^\Gamma \geq 0$ \\
\ & $\Tr(\rho) \leq 1$ \\
\ & \ \\
\multicolumn{2}{c}{{\bf Dual problem}} \\
\text{maximize:} & $\lambda_1 - \big\|Y^\Gamma + P_{\cl{S}}(\lambda_1 I - C_{\cl{E}}^\Gamma)P_{\cl{S}}\big\|$ \\
\text{subject to:} & $Y \geq 0$
\end{tabular}
\end{matrix}
\end{align*}

By using MATLAB to solve this semidefinite program, we are able to approximate the distribution of the minimum gate fidelity when $n = 2$. Figure~\ref{fig:min_fid} shows the distribution of $\cl{F}_{\cl{E}}^{\textup{min}}$ and $\overline{\cl{F}_{\cl{E}}}$ when the quantum channel $\cl{E}$ is chosen by picking a Haar-uniform unitary $U \in M_4 \otimes M_2$ and setting $\cl{E}(\rho) \equiv \Tr_1(U (\ketbra{1}{1} \otimes \rho) U^{\dagger})$.
\begin{figure}[ht]
\begin{center}
\includegraphics[width=0.9\textwidth]{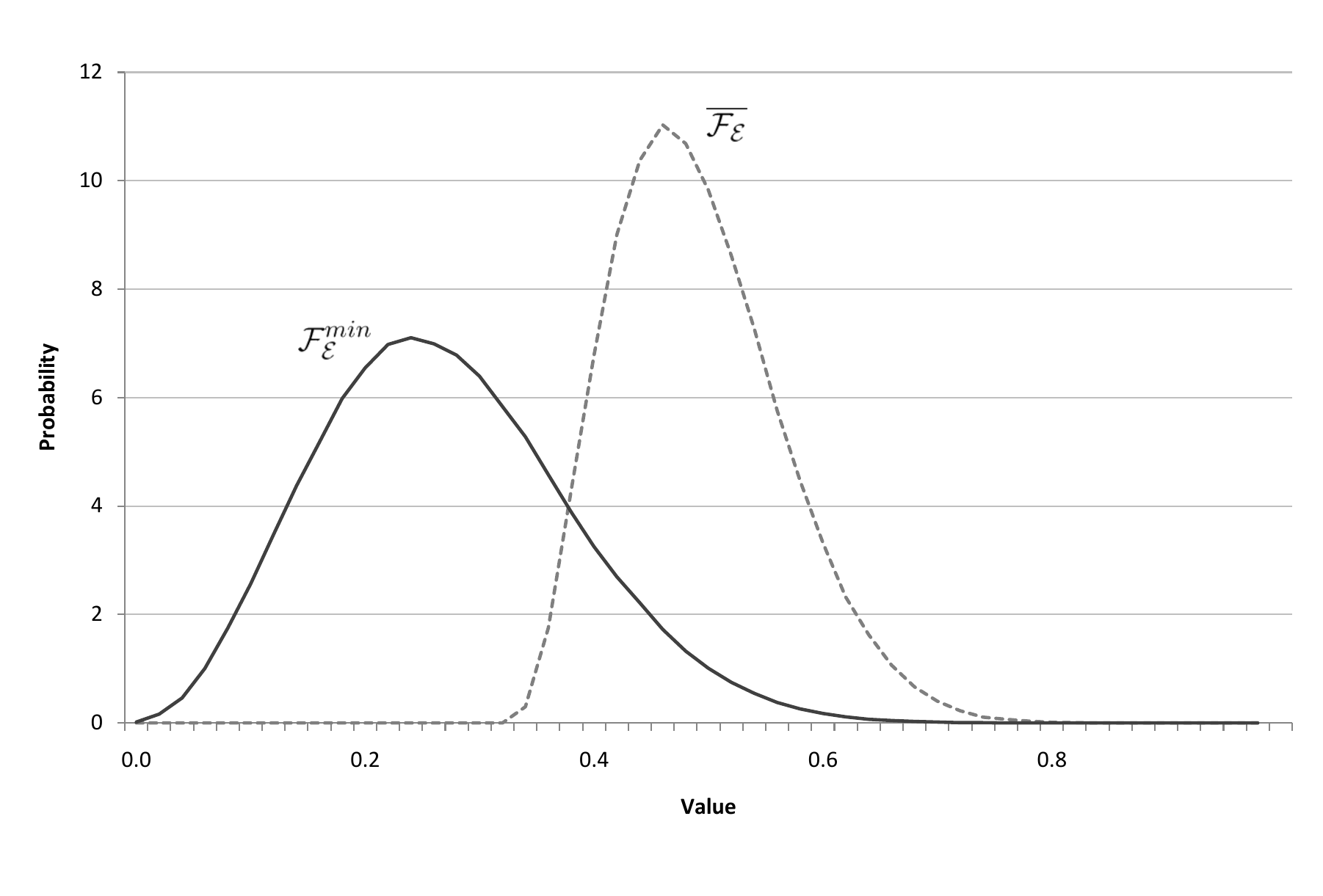}
\end{center}
\caption[Approximate distributions of minimum and average gate fidelity of qubit channels]{\hsp Approximate distributions of $\cl{F}_{\cl{E}}^{\textup{min}}$ and $\overline{\cl{F}_{\cl{E}}}$ when $n = 2$, based on $5 \cdot 10^5$ randomly-generated qubit channels.}\label{fig:min_fid}
\end{figure}

When $n \geq 3$, we no longer have a single semidefinite program that computes $\cl{F}_{\cl{E}}^{\textup{min}}$, but rather we have to use the entire hierarchy of semidefinite programs introduced in Section~\ref{sec:SP_shareable}.

\subsection{Computational Complexity}\label{sec:min_gate_fid_complexity}

In light of Theorem~\ref{thm:min_fid}, it is perhaps not surprising that computing minimum gate fidelity is NP-hard, since computing the $S(1)$-norm in general is NP-hard. We now show that computing minimum gate fidelity is indeed also NP-hard. In fact, we show that computing minimum gate fidelity is NP-hard even for the relatively small class of channels that are entanglement-breaking and self-dual.
\begin{thm}\label{thm:nphard_min_fid}
	The problem of computing $\cl{F}_{\cl{E}}^{{\rm min}}$ is NP-hard, even given the promise that $\cl{E}$ is entanglement-breaking and $\cl{E} = \cl{E}^\dagger$.
\end{thm}
\begin{proof}
	It was noted in \cite{Ioa07} that computing
	\begin{align}\label{eq:rsdf}
		\max_{x \in \mathbb{R}^n, \|x\|=1} \big\{\sum_{i,j=1}^n x_i^2 x_j^2 a_{ij}\big\}
	\end{align}
	is NP-hard, even given the promise that $A = (a_{ij})$ is a symmetric traceless $0$--$1$ matrix. Assume for the remainder of the proof that $A \in M_n$ is such a matrix.
	
	Let $\cl{D} : M_n \rightarrow M_n$ be the completely depolarizing map defined by $\cl{D}(X) = \frac{\Tr(X)}{n}I$ and let $\cl{S}_A : M_n \rightarrow M_n$ be the Schur map defined by $cl{S}_A(X) = A * X$, where $*$ denotes the Schur (i.e., entrywise) product. Define $\cl{E} := \cl{D} - \frac{1}{n^2(n-1)}\cl{S}_A$. The map $\cl{E}$ is easily seen to be trace-preserving since $\cl{D}$ is trace-preserving and $\Tr(A) = 0$. Also, $\cl{E}$ is clearly self-dual since each of $\cl{D}$ and $\cl{S}_A$ are self-dual. Its Choi matrix is easily seen to be separable by \cite[Theorem~1]{GB02}. The map $\cl{E}$ is thus a self-dual entanglement-breaking quantum channel, and its minimum gate fidelity is 
	\begin{align*}
		\cl{F}^{\textup{min}}_{\cl{E}} & = \min_{\ket{v}} \big\{ \bra{v} \cl{E}(\ketbra{v}{v})\ket{v} \big\} \\
		& = \frac{1}{n}\left[1 - \frac{1}{n(n-1)}\max_{\ket{v}} \big\{ \bra{v} \big(\sum_{i,j=1}^n v_i\overline{v_j} a_{ij} \ketbra{i}{j}\big)\ket{v} \big\} \right] \\
		& = \frac{1}{n}\left[1 - \frac{1}{n(n-1)}\max_{\ket{v}} \big\{ \sum_{i,j=1}^n|v_i|^2|v_j|^2 a_{ij} \big\} \right] \\
		& = \frac{1}{n}\left[1 - \frac{1}{n(n-1)}\max_{x \in \mathbb{R}^n, \|x\|=1} \big\{\sum_{i,j=1}^n x_i^2 x_j^2 a_{ij} \big\} \right].
	\end{align*}
	Since performing the maximization on the right is NP-hard, so is computing $\cl{F}^{\textup{min}}_{\cl{E}}$.
\end{proof}

It is worth briefly dwelling on the fact that Theorem~\ref{thm:nphard_min_fid} implies that optimizations of the form
\begin{align*}
	\sup_{\ket{v}}\left\{\bra{vv}X\ket{vv}\right\}
\end{align*}
are NP-hard to approximate. This demonstrates that determining whether or not an operator is a so-called \emph{symmetric witness} \cite{TG10} (i.e., an entanglement witness for symmetric states) is also NP-hard.

\section{Superoperator Norms and Maximum Output Purity}\label{sec:max_output_purity}

In this section we establish a connection between the $S(k)$-operator norms and the induced Schatten superoperator norms $\|\cdot\|_{q\rightarrow p}$ in the $q = 1$, $p = \infty$ case. First, we recall that for a linear map $\Phi : M_m \rightarrow M_n$, we define
\begin{align*}
	\big\|\Phi\big\|_{1\rightarrow \infty} := \sup_{X} \Big\{ \big\| \Phi(X) \big\| : \big\|X\big\|_{tr} = 1 \Big\}.
\end{align*}
It has been shown \cite{Wat05} that if $\Phi$ is completely positive (or even just $2$-positive \cite{Aud09,Sza10}) then this supremum is attained by a positive semidefinite matrix $X$, in which case the supremum can be rephrased as an optimization over density matrices or over pure states:
\begin{align*}
	\big\|\Phi\big\|_{1\rightarrow \infty} = \sup_{\rho} \Big\{ \big\| \Phi(\rho) \big\| \Big\} = \sup_{\ket{v}} \Big\{ \big\| \Phi(\ketbra{v}{v}) \big\| \Big\},
\end{align*}
where the second equality follows easily from convexity of the operator norm. If $\Phi$ is a quantum channel, then this quantity is known as the \emph{maximum output purity} of $\Phi$ \cite{AHW00,DR05}, a term that can be motivated by observing that for quantum channels, $\big\|\Phi\big\|_{1\rightarrow \infty} \leq 1$ always and $\big\|\Phi\big\|_{1\rightarrow \infty} = 1$ if and only if there exists a density matrix $\rho$ such that $\Phi(\rho)$ is a pure state. The maximum output purity also equals the $p = \infty$ case of the maximal $p$-norm of a quantum channel \cite{Dat04,Kin03,KNR05}.

Maximum output purity and maximal $p$-norms have received a lot of attention lately because they are closely related to several important additivity conjectures in quantum information theory \cite{AB04,AH03,Hay07,Hol06,HW08,Kin02,Kin03,WH02,Win07}. It is known \cite{Wat05} that for any linear map and any integer $k \geq 1$, $\big\|\Phi\big\|_{1\rightarrow \infty} = \big\|id_k \otimes \Phi\big\|_{1\rightarrow \infty}$, so the norm $\|\cdot\|_{1\rightarrow\infty}$ is equal to its completely bounded counterpart. However, there is another completely bounded version of this norm that in general is not equal to $\|\cdot\|_{1\rightarrow\infty}$ itself and is also connected to some additivity conjectures \cite{DJKR06}.

For a completely positive map $\Phi : M_m \rightarrow M_n$ and an integer $k \geq 1$, consider the following norm:
\begin{align*}
	\big\|\Phi\big\|_{k,1\rightarrow\infty}^\prime := \sup_{\ket{v}} \left\{ \frac{\big\| (id_k \otimes \Phi)(\ketbra{v}{v}) \big\|}{\big\| \Tr_2(\ketbra{v}{v}) \big\|} : \ket{v} \in \mathbb{C}^k \otimes \mathbb{C}^m \right\}.
\end{align*}
Observe in particular that $\|\cdot\|_{1,1\rightarrow\infty}^\prime = \|\cdot\|_{1\rightarrow\infty}^\prime$. We also define $\big\|\Phi\big\|_{cb,1\rightarrow\infty}^\prime := \sup_{k \geq 1} \left\{ \big\|\Phi\big\|_{k,1\rightarrow\infty} \right\}$ and note that this norm stabilizes in the sense that $\|\cdot\|_{cb,1\rightarrow\infty}^\prime = \|\cdot\|_{\min\{m,n\},1\rightarrow\infty}^\prime$. Our main result of this section says that the norm $\big\|\Phi\big\|_{k,1\rightarrow \infty}^\prime$ is actually very familiar for us -- it simply equals the $S(k)$-operator norm of the Choi matrix of $\Phi$.
\begin{thm}\label{thm:max_output_purity}
	Let $\Phi : M_m \rightarrow M_n$ be a completely positive linear map and let $1 \leq k \leq \min\{m,n\}$. Then
	\begin{align*}
		\big\|\Phi\big\|_{k,1\rightarrow \infty}^\prime = \big\|C_\Phi\big\|_{S(k)}.
	\end{align*}
\end{thm}
\begin{proof}
	Begin by writing
	\begin{align}\begin{split}\label{eq:one_inf_norm}
		\big\|\Phi\big\|_{k,1\rightarrow \infty}^\prime & = \sup_{\ket{v}} \left\{ \frac{\| (id_k \otimes \Phi)(\ketbra{v}{v}) \|}{\|\Tr_2(\ketbra{v}{v})\|} \right\} \\
		& = \sup_{\ket{v},\ket{w}} \left\{ \frac{1}{\alpha_1^2} \bra{w} (id_k \otimes \Phi)(\ketbra{v}{v}) \ket{w} \right\},
	\end{split}\end{align}
	where the supremums are taken over pure states $\ket{v} \in \mathbb{C}^k \otimes \mathbb{C}^m$ and $\ket{w} \in \mathbb{C}^k \otimes \mathbb{C}^n$, and $\alpha_1$ is the maximal Schmidt coefficient of $\ket{v}$. In the second equality, we used the fact that if $\ket{v} = \sum_{i=1}^k \alpha_i \ket{x_i} \otimes \ket{v_i}$ is a Schmidt decomposition of $\ket{v}$, then $\Tr_2(\ketbra{v}{v}) = \sum_{i=1}^k \alpha_i^2 \ketbra{x_i}{x_i}$, so $\big\|\Tr_2(\ketbra{v}{v})\big\| = \alpha_1^2$.
	
	Now we can write $\ket{w} = \sum_{i=1}^k \beta_i \ket{x_i} \otimes \ket{w_i}$ with each $\beta_i$ real and non-negative -- observe that we have chosen the vectors on the first subsystem to be the same as those in the Schmidt decomposition of $\ket{v}$. In this case the normalization condition $\big\|\ket{w}\big\| = 1$ implies that $\sum_{i=1}^k \beta_i^2 = 1$, but this decomposition in general will not be a Schmidt decomposition, as the set of vectors $\big\{\ket{w_i}\big\}$ in general will not be orthonormal. Carrying on from Equation~\eqref{eq:one_inf_norm} now gives
	\begin{align*}
		\big\|\Phi\big\|_{k,1\rightarrow \infty}^\prime & = \sup_{\ket{v},\ket{w}} \left\{ \frac{1}{\alpha_1^2} \sum_{i,j,r,s=1}^k \alpha_r \alpha_s \beta_i \beta_j \bra{x_i w_i} (id_k \otimes \Phi)(\ketbra{x_r v_r}{x_s v_s}) \ket{x_j w_j} \right\}, \\
		& = \sup_{\ket{v},\ket{w}} \left\{ \frac{1}{\alpha_1^2} \sum_{i,j=1}^k \alpha_i \alpha_j \beta_i \beta_j \bra{w_i}\Phi(\ketbra{v_i}{v_j}) \ket{w_j} \right\}.
	\end{align*}
	Now we can use Lemma~\ref{lem:fidel_choi} to see that
	\begin{align*}
		\big\|\Phi\big\|_{k,1\rightarrow \infty}^\prime = \sup_{\ket{v},\ket{w}} \left\{ \frac{1}{\alpha_1^2} \sum_{i,j=1}^k \alpha_i \alpha_j \beta_i \beta_j \bra{\overline{v_i}w_i}C_{\Phi} \ket{\overline{v_j}w_j} \right\} = \sup_{\ket{v},\ket{w}} \big\{ \bra{\tilde{v}} C_{\Phi} \ket{\tilde{v}} \big\},
	\end{align*}
	where $\ket{\tilde{v}} := \sum_{i=1}^k \frac{\alpha_i \beta_i}{\alpha_1} \ket{\overline{v_i}w_i}$. It is clear that $\ket{\tilde{v}}$ is a (not necessarily normalized) vector with Schmidt rank no larger than $k$. We can see that $\big\| \ket{\tilde{v}} \big\| \leq 1$ by defining two vectors $\alpha, \beta \in \mathbb{R}^k$ as follows:
	\begin{align*}
		\alpha := \frac{1}{\alpha_1^2}\big(\alpha_1^2,\ldots,\alpha_k^2\big)^T, \quad \beta := \big( \beta_1^2,\ldots,\beta_k^2 \big)^T.
	\end{align*}
	Then $\|\alpha\|_{\infty} = \|\beta\|_1 = 1$, so H\"{o}lder's inequality tells us that $\big\| \ket{\tilde{v}} \big\| = \sum_{i=1}^k \frac{\alpha_i^2 \beta_i^2}{\alpha_1^2} = \alpha^\dagger \beta \leq 1$. It follows that $\big\|\Phi\big\|_{k,1\rightarrow \infty}^\prime \leq \big\| C_{\Phi} \big\|_{S(k)}$.
	
	To see the other inequality, choose $\ket{v}$ so that $\alpha_1 = \cdots = \alpha_k = 1/\sqrt{k}$. Then we have $\ket{\tilde{v}} = \sum_{i=1}^k \beta_i \ket{\overline{v_i}w_i}$, which is a general pure state with $SR(\ket{\tilde{v}}) \leq k$, which shows that $\big\|\Phi\big\|_{k,1\rightarrow \infty}^\prime \geq \big\| C_{\Phi} \big\|_{S(k)}$ and completes the proof.
\end{proof}

In the $k = 1$ case, Theorem~\ref{thm:max_output_purity} says that the maximum output purity of a quantum channel is equal to the $S(1)$-norm of that channel's Choi matrix -- a result that was originally observed in \cite{NOP09}. The other extreme is also a known result -- Theorem~10 of \cite{DJKR06} showed that if $\Phi$ is completely positive then $\big\|\Phi\big\|_{cb,1\rightarrow\infty}^\prime = \big\|C_\Phi\big\|$, which is exactly the $k = \min\{m,n\}$ case of Theorem~\ref{thm:max_output_purity}.

It is also worth pointing out that another simple method of calculating $\big\|\Phi\big\|_{cb,1\rightarrow \infty}^\prime$ follows from the main result of \cite{Jen06}, where it was shown that if $\Phi^C$ is the complementary channel of $\Phi$ then $\big\|\Phi\big\|_{cb,1\rightarrow \infty}^\prime = \big\|\Phi^{C}\big\|$. By \cite[Proposition 3.6]{P03}, it then follows that $\big\|\Phi\big\|_{cb,1\rightarrow \infty}^\prime = \big\|\Phi^C(I)\big\|$. As a corollary of this, we see that $\big\|C_\Phi\big\| = \big\|\Phi^C(I)\big\|$.

As for calculating $\big\|\Phi\big\|_{k,1\rightarrow \infty}^\prime$ when $k < \min\{m,n\}$, we can now make use of the semidefinite programming techniques from Sections~\ref{sec:SP_pos_map} and~\ref{sec:SP_shareable}, much like we did for computing minimum gate fidelity in Section~\ref{sec:min_gate_fid}. Many other corollaries follow easily as well, such as NP-hardness of computing the norm $\|\cdot\|_{1\rightarrow \infty}$, as well as the following inequalities:
\begin{cor}\label{cor:1inf_norm_kraus}
	Let $\Phi : M_m \rightarrow M_n$ be a completely positive linear map with canonical Kraus representation $\Phi(\rho) = \sum_i A_i \rho A_i^\dagger$, with the set of operators $\big\{ A_i \big\}$ forming an orthogonal set in the Hilbert--Schmidt inner product. Then
	\begin{align*}
		\big\|\Phi\big\|_{k,1\rightarrow \infty}^\prime \leq \sum_i \big\| A_i \big\|_{(k,2)}^2.
	\end{align*}
	Furthermore, if $\Phi$ has just one Kraus operator then equality holds.
\end{cor}
\begin{proof}
	By using Theorem~\ref{thm:max_output_purity} and then Proposition~\ref{prop:matrixVectorNorms}, we see that if $\{\lambda_i\}$ is the set of eigenvalues of $C_\Phi$ with associated eigenvectors $\{\ket{v_i}\}$, then
	\begin{align*}
		\big\|\Phi\big\|_{k,1\rightarrow \infty}^\prime = \big\|C_\Phi\big\|_{S(k)} \leq \sum_i |\lambda_i|\big\|\ket{v_i}\big\|^2_{s(k)}.
	\end{align*}
	Now recall from Section~\ref{sec:vector_operator_isomorphism} that the canonical Kraus operators $\big\{A_i\big\}$ of $\Phi$ are the matricization of $\sqrt{\lambda_i}\ket{v_i}$, from which we have $|\lambda_i|\big\|\ket{v_i}\big\|^2_{s(k)} = \big\| A_i \big\|_{(k,2)}^2$. The desired inequality follows immediately.
	
	To see the final claim, simply note that $\Phi$ can be represented with a single Kraus operator if and only if $C_\Phi$ has rank one. Proposition~\ref{prop:rankOneNorm} then shows that equality is attained.
\end{proof}

In the $k = 1$ case, Corollary~\ref{cor:1inf_norm_kraus} simply says that $\big\|\Phi\big\|_{1\rightarrow \infty}^\prime \leq \sum_i \big\| A_i \big\|^2$, which follows easily from the definition of $\big\|\Phi\big\|_{1\rightarrow \infty}^\prime$. Indeed, for any particular $\rho \in M_m$, we have
\begin{align*}
	\big\|\Phi(\rho)\big\| = \left\| \sum_i A_i \rho A_i^\dagger \right\| \leq \sum_i \big\|A_i\big\|\big\|\rho\big\|\big\|A_i^\dagger\big\| \leq \sum_i \big\| A_i \big\|^2,
\end{align*}
from which the desired inequality follows easily.

\section{Tripartite and Quadripartite Geometric Measure of Entanglement}\label{sec:tri_geo_ent}

Recall the geometric measure of entanglement, introduced in Section~\ref{sec:geom_meas_ent}, which is defined in terms of the maximal overlap between a given pure state and a separable state. In the bipartite case, the geometric measure of entanglement is easy to calculate, as it is essentially the $s(1)$-vector norm. In particular, if $\ket{v} \in \mathbb{C}^m \otimes \mathbb{C}^n$ has Schmidt coefficients $\alpha_1 \geq \alpha_2 \geq \cdots \geq 0$ then
\begin{align*}
	E(\ket{v}) = 1 - \big\| \ket{v} \big\|_{s(1)}^2 = \sum_{i=2}^{\min\{m,n\}} \alpha_i^2.
\end{align*}

We now consider the tripartite case (i.e., the case of three subsystems) and show that it has a similar relationship with the $S(1)$-operator norm. Let $\ket{v} \in \mathbb{C}^{n_1} \otimes \mathbb{C}^{n_2} \otimes \mathbb{C}^{n_3}$ be a pure state. Then
\begin{align}\label{eq:tri_ent_geom}
	E(\ket{v}) = 1 - \sup_{\ket{w_1},\ket{w_2},\ket{w_3}}\Big\{ \big|(\bra{w_1} \otimes \bra{w_2} \otimes \bra{w_3})\ket{v}\big|^2 \Big\}.
\end{align}
Let's now write $\ket{v}$ in its Schmidt decomposition $\ket{v} = \sum_{i} \alpha_i \ket{a_i} \otimes \ket{b_i}$, where we decompose it over the first tensor product (so $\ket{a_i} \in \mathbb{C}^{n_1}$ and $\ket{b_i} \in \mathbb{C}^{n_2} \otimes \mathbb{C}^{n_3}$). If we fix $\ket{w_2}$ and $\ket{w_3}$ and only take the supremum over $\ket{w_1}$ in Equation~\eqref{eq:tri_ent_geom}, we see that the supremum is attained when $\ket{w_1} = (I_{n_1} \otimes \bra{w_2} \otimes \bra{w_3})\ket{v} = \sum_{i} \big( \alpha_i (\bra{w_2} \otimes \bra{w_3})\ket{b_i} \big) \ket{a_i}$. Plugging this into~\eqref{eq:tri_ent_geom} gives
\begin{align*}
	E(\ket{v}) & = 1 - \sup_{\ket{w_2},\ket{w_3}}\left\{ \left| \sum_i \alpha_i \braket{w_1}{a_i} (\bra{w_2} \otimes \bra{w_3})\ket{b_i} \right|^2 \right\} \\
	& = 1 - \sup_{\ket{w_2},\ket{w_3}}\left\{ \left| \sum_{i,j} \alpha_i \alpha_j \bra{b_j}(\ket{w_2} \otimes \ket{w_3})\braket{a_j}{a_i} (\bra{w_2} \otimes \bra{w_3})\ket{b_i} \right|^2 \right\} \\
	& = 1 - \sup_{\ket{w_2},\ket{w_3}}\left\{ \left| (\bra{w_2} \otimes \bra{w_3}) \left(\sum_{i} \alpha_i^2 \ketbra{b_i}{b_i}\right) (\ket{w_2} \otimes \ket{w_3}) \right|^2 \right\} \\
	& = 1 - \big\| \Tr_1(\ketbra{v}{v}) \big\|_{S(1)}^2.
\end{align*}

Of course, there is nothing special about the first subsystem -- we could just as easily have let either $\ket{w_2}$ or $\ket{w_3}$ vary, in which case we would have ended up tracing out the second or third subsystem, respectively. We state this result as the following theorem, which also appeared in \cite{NOP09}:
\begin{thm}\label{thm:tri_geom_ent}
	Let $\ket{v} \in \mathbb{C}^{n_1} \otimes \mathbb{C}^{n_2} \otimes \mathbb{C}^{n_3}$. Then
	\begin{align*}
		1 - E(\ket{v}) = \big\| \Tr_1(\ketbra{v}{v}) \big\|_{S(1)}^2 = \big\| \Tr_2(\ketbra{v}{v}) \big\|_{S(1)}^2 = \big\| \Tr_3(\ketbra{v}{v}) \big\|_{S(1)}^2.
	\end{align*}
\end{thm}

As was the case with maximum output purity and minimum gate fidelity, we can now use the semidefinite programming techniques of Sections~\ref{sec:SP_pos_map} and~\ref{sec:SP_shareable} to compute the tripartite geometric measure of entanglement. A particularly interesting corollary is that even though the geometric measure of entanglement is easily-computable in the bipartite case, computing the geometric measure of entanglement of a state $\ket{v} \in (\mathbb{C}^n)^{\otimes p}$ when $p \geq 3$ is an NP-hard problem.

In fact, we could extend the methods that were used to prove Theorem~\ref{thm:tri_geom_ent} to the case of more than three subsystems -- as before, we just fix all but one of the vectors that we optimize over. Working through the calculation reveals that if $\ket{v} \in \mathbb{C}^{n_1} \otimes \cdots \otimes \mathbb{C}^{n_p}$, then
\begin{align}\label{eq:ent_geom_multi_form}
	E(\ket{v}) = 1 - \sup_{\ket{w_i} \in \mathbb{C}^{n_i}} \left\{ (\bra{w_2} \otimes \cdots \otimes \bra{w_p}) \Tr_1(\ketbra{v}{v}) (\ket{w_2} \otimes \cdots \otimes \ket{w_p}) \right\}.
\end{align}
By using the natural multipartite extension of symmetric extensions \cite{DPS05}, this quantity can be computed using semidefinite programming techniques that are similar to those of Section~\ref{sec:SP_shareable}.

In the quadripartite case (i.e., the case of four subsystems), we would compute $E(\ket{v})$ by using symmetric extensions to perform the optimization~\eqref{eq:ent_geom_multi_form} over tripartite separable states. However we can actually reduce the problem slightly further, just like the tripartite geometric measure of entanglement, to the problem of computing the $S(1)$-norm.
\begin{thm}\label{thm:quad_geo_meas_ent}
	Let $\ket{v} \in \mathbb{C}^{m} \otimes \mathbb{C}^{m} \otimes \mathbb{C}^{n} \otimes \mathbb{C}^{n}$ and let $A_{\ket{v}} \in M_{m} \otimes M_{n}$ be the operator associated with $\ket{v}$ via the linear isomorphism that maps the vector $\ket{x_1} \otimes \ket{x_2} \otimes \ket{y_1} \otimes \ket{y_2}$ to the operator $\ket{x_2}\overline{\bra{x_1}} \otimes \ket{y_2}\overline{\bra{y_1}}$. Then
	\begin{align*}
		E(\ket{v}) = 1 - \big\|A_{\ket{v}}\big\|_{S(1)}^2.
	\end{align*}
\end{thm}
\begin{proof}
	First note that the isomorphism described by the theorem can be seen as a quadripartite version of the vector-operator isomorphism, and in fact is the exact same isomorphism that was used in proof of the $k = 1$ case of Proposition~\ref{prop:sepLPP}. If we write
	\begin{align*}
		\ket{v} = \sum_i c_i \ket{v_i} \otimes \ket{x_i} \otimes \ket{y_i} \otimes \ket{z_i},
	\end{align*}
	then for any $\ket{w_1} \otimes \ket{w_2} \otimes \ket{w_3} \otimes \ket{w_4} \in \mathbb{C}^{m} \otimes \mathbb{C}^{m} \otimes \mathbb{C}^{n} \otimes \mathbb{C}^{n}$ we have
	\begin{align*}
		(\bra{w_1} \otimes \bra{w_2} \otimes \bra{w_3} \otimes \bra{w_4})\ket{v} & = \sum_i c_i \braket{w_1}{v_i} \braket{w_2}{x_i} \braket{w_3}{y_i} \braket{w_4}{z_i} \\
		& = \sum_i c_i \braket{w_2}{x_i} \overline{\braket{v_i}{w_1}} \braket{w_4}{z_i} \overline{\braket{y_i}{w_3}} \\
		& = (\bra{w_2} \otimes \bra{w_4}) \left(\sum_i c_i \ket{x_i}\overline{\bra{v_i}} \otimes \ket{z_i}\overline{\bra{y_i}}\right) (\overline{\ket{w_1}} \otimes \overline{\ket{w_3}}) \\
		& = (\bra{w_2} \otimes \bra{w_4}) A_{\ket{v}} (\overline{\ket{w_1}} \otimes \overline{\ket{w_3}}).
	\end{align*}
	The result follows easily by taking the absolute value and then the supremum over $\ket{w_1}, \ket{w_2}, \ket{w_3}$, and $\ket{w_4}$.
\end{proof}

Recall that many of our results on the $S(k)$-norms, including the semidefinite program methods for their computation, only hold in the case when the operator under consideration is positive semidefinite. In Theorem~\ref{thm:quad_geo_meas_ent}, the operator $A_{\ket{v}}$ in general is not positive semidefinite and so we are better off using Equation~\eqref{eq:ent_geom_multi_form} to compute the quadripartite geometric measure of entanglement. Nevertheless, there is one important special case that is of interest -- the case when $\ket{v}$ is real and symmetric (i.e., $\ket{v} \in \cl{S}$, where $\cl{S}$ is the symmetric subspace of $\mathbb{C}^{n} \otimes \mathbb{C}^{m} \otimes \mathbb{C}^{n} \otimes \mathbb{C}^{n}$).

The geometric measure of entanglement of symmetric states has been extensively studied recently \cite{HMMOV08,HMMOV09,WG03,WS10}. In particular, the question of whether or not the state that optimizes the geometric measure of entanglement can be chosen to be symmetric when $\ket{v}$ is symmetric was an open question that was recently solved in the affirmative in \cite{HKWGG09}. By using the isomorphism of Theorem~\ref{thm:quad_geo_meas_ent} we then arrive at the following simple corollary, which can be thought of as a symmetric version of Proposition~\ref{prop:MatMultDiffVectors}.
\begin{cor}
	Let $X \in M_n \otimes M_n$ be such that $X = X^T = X^\Gamma = SX$, where $S$ is the swap operator. Then
	\begin{align*}
		\big\|X\big\|_{S(1)} = \sup_{\ket{v}}\big\{ (\overline{\bra{v}} \otimes \overline{\bra{v}})X(\ket{v} \otimes \ket{v}) \big\}.
	\end{align*}
\end{cor}

\chapter{Connections with Operator Theory}\label{ch:optheory}

In this chapter we link central areas of study in operator theory with the various norms and cones investigated throughout the previous chapters. More specifically, we connect recent investigations in operator space and operator system theory \cite{P03,Pis03} with the norms of Chapter~\ref{chap:Sk_norms} and the cones of separable and block positive operators. We also connect all of these areas of study with mapping cones \cite{S86} and right CP-invariant cones. As benefits of this combined perspective, we obtain new results and new elementary proofs in all of these areas.

We begin by investigating abstract operator spaces (i.e., matricially normed spaces) on complex matrices. We show that a well-known family of operator spaces, which are referred to as the ``$k$-minimal'' and ``$k$-maximal'' operator spaces \cite{OR04}, give rise to the $S(k)$-operator norms introduced in Chapter~\ref{chap:Sk_norms}. We use this connection to finally derive an expression for the dual of the $S(k)$-norm, and we see that this dual norm exactly characterizes Schmidt number. We introduce completely bounded norms between arbitrary operator spaces, and we show that the completely bounded norm on the $k$-minimal operator space stabilizes in a manner very similar to the standard completely bounded norm.

We then investigate abstract operator spaces (i.e., matricially ordered spaces) on complex matrices and see that many analogous results hold. We show that the ``$k$-super maximal'' and ``$k$-super minimal'' operator systems on $M_n$ \cite{Xthesis,Xha11} give rise to the (unnormalized) states with Schmidt number no larger than $k$ and the $k$-block positive operators, respectively. We also show that the completely positive maps between these different operator systems are simply the $k$-positive and $k$-superpositive maps. Furthermore, we connect the dual of a version of the completely bounded minimal operator space norm to the separability problem and extend recent results about how trace-contractive maps can be used to detect entanglement. We see that the maps that serve to detect quantum entanglement via norms are roughly the completely contractive maps on the minimal operator space on $M_n$.

Finally, we finish by considering the relationships between right CP-invariant cones, mapping cones, semigroup cones, and abstract operator systems. We show that every abstract operator system gives a natural right CP-invariant cone, and conversely every right CP-invariant cone gives an abstract operator system. In the case of mapping cones, we show that the associated operator systems have a property that we call ``super-homogeneity'', and we also provide an analogous result for semigroup cones. We present some simple consequences of these results, including an abstract operator system based on anti-degradable maps and shareable operators.

\section{Operator Spaces on Complex Matrices}\label{sec:op_space}

An (abstract) \emph{operator space} on $M_n$ is a family of norms $\big\{\|\cdot\|_{m}\big\}$ on $M_m \otimes M_n$ ($m \geq 1$) that satisfy two conditions:
\begin{enumerate}[(1)]
	\item If $A, B \in M_{r,m}$, $X \in M_m \otimes M_n$, and $\|\cdot\|$ denotes the operator norm, then
	\begin{align*}
		\big\|(A \otimes I) X (B^\dagger \otimes I)\big\|_{r} \leq \big\|A\big\|\big\|X\big\|_{m}\big\|B\big\|; and
	\end{align*}
	\item[($2,\infty$)] $\big\| X \oplus Y \big\|_{m+r} = \max\big\{ \|X\|_{m},\|Y\|_{r}\big\}$ for all $X \in M_m \otimes M_n, Y \in M_r \otimes M_n$, where we have associated $(M_m \otimes M_n) \oplus (M_r \otimes M_n)$ with $M_{m+r} \otimes M_n$ in the natural way.
\end{enumerate}

Property~(1) above ensures that the norms $\|\cdot\|_m$ ``behave well'' with each other. For example, it ensures that if we embed $X \in M_m \otimes M_n$ as $\tilde{X} \in M_{m+1} \otimes M_n$ by adding rows and columns of zeroes then $\big\|X\big\|_m = \big\|\tilde{X}\big\|_{m+1}$. Property~($2,\infty$), which is called the \emph{$L^\infty$ condition}, ensures that each of these norms behave ``like'' the standard operator norm in some sense. It is sometimes desirable to consider families of norms that instead satisfy property~(1) together with the following condition for some $1 \leq p < \infty$ \cite{ER88,R88}:
\begin{enumerate}[(1)]
	\item[($2,p$)] $\big\| X \oplus Y \big\|_{m+r} = \sqrt[p]{ \|X\|_{m}^p + \|Y\|_{r}^p }$ for all $X \in M_m \otimes M_n$ and $Y \in M_r \otimes M_n$.
\end{enumerate}
A family of norms $\|\cdot\|_{m}$ on $M_m \otimes M_n$ satisfying properties~(1) and~($2,p$) is said to be a family of \emph{$L^p$-matrix norms} and we see in the limit as $p \rightarrow \infty$ that we obtain an abstract operator space.

More generally, one can define abstract operator spaces and $L^p$-matrix norms by replacing $M_n$ by an arbitrary vector space $V$ throughout the preceding paragraphs. In this more general setting, the reason for the terminology ``abstract operator space'' becomes more clear: a theorem of Ruan \cite{R88} says that a matrix normed space is completely isometric with a concrete operator space (i.e., a subspace of the bounded operators on a Hilbert space) if and only if it is an $L^\infty$-matrix normed space. However, for us it is enough to consider abstract operator spaces and matrix norms on $M_n$. For a more detailed introduction to abstract operator spaces, the interested reader is directed to \cite[Chapter~13]{P03}.

Throughout this section, we use $M_n$ itself to denote the ``standard'' operator space structure on $M_n$ that is obtained by associating $M_m \otimes M_n$ with $M_{mn}$ in the natural way and using the operator norm. Similarly, we use $M_{n,tr}$ to denote the $L^1$-matrix normed space that arises from using the trace norm on $M_m \otimes M_n$ for all $m$ (the fact that the trace norm satisfies the matrix norm property~(1) follows from Proposition~\ref{prop:unitary_invar_symmetric}).

All other matrix normed spaces that we will consider will have their first norm, $\|\cdot\|_1$, equal to either the operator norm of the trace norm on $M_n$. In the former case, we will denote it by something like $V(M_n)$. In the latter case, we will use notation like $V(M_{n,tr})$. In the case when we do not specify what the first norm is, we will simply denote the operator space by $V$. When referring to the $m$-th norm of a family of matrix norms $V$, if there is a possibility for confusion we denote it by $\|\cdot\|_{V_m}$.

\subsection{Minimal and Maximal Operator Spaces}\label{sec:min_max_op_space}

Some particularly important operator spaces for us are the \emph{$k$-minimal operator space} $MIN^k(M_n)$ and the \emph{$k$-maximal operator space} $MAX^k(M_n)$ \cite{OR04}, defined respectively via the following families of norms on $M_m \otimes M_n$:
\begin{align}\label{eq:k_min_space}
	\big\|X\big\|_{MIN^k_m(M_n)} & := \sup_\Phi \Big\{ \big\| (id_m \otimes \Phi)(X) \big\| : \Phi : M_n \rightarrow M_k, \big\| \Phi \big\|_{cb} \leq 1 \Big\} \text{ and} \\ \label{eq:k_max_space}
	\big\|X\big\|_{MAX^k_m(M_n)} & := \sup_{r,\Phi} \Big\{ \big\| (id_m \otimes \Phi)(X) \big\| : \Phi : M_n \rightarrow M_r, \big\| id_k \otimes \Phi \big\| \leq 1 \Big\}.
\end{align}

The names of these operator spaces come from the facts that if $V(M_n)$ is any operator space on $M_n$ such that $\|\cdot\|_{V_m(M_n)}$ simply equals the operator norm on $M_m \otimes M_n$ for $1 \leq m \leq k$, then $\|\cdot\|_{MIN^k_m(M_n)} \leq \|\cdot\|_{V_m(M_n)} \leq \|\cdot\|_{MAX^k_m(M_n)}$ for all $m > k$. In the $k = 1$ case, these operator spaces are exactly the minimal and maximal operator space structures that are fundamental in operator space theory \cite[Chapter 14]{P03}. The interested reader is directed to \cite{OR04} and the references therein for further properties of $MIN^k(M_n)$ and $MAX^k(M_n)$ when $k \geq 2$.
\begin{figure}[ht]
\begin{center}
\includegraphics[width=\textwidth]{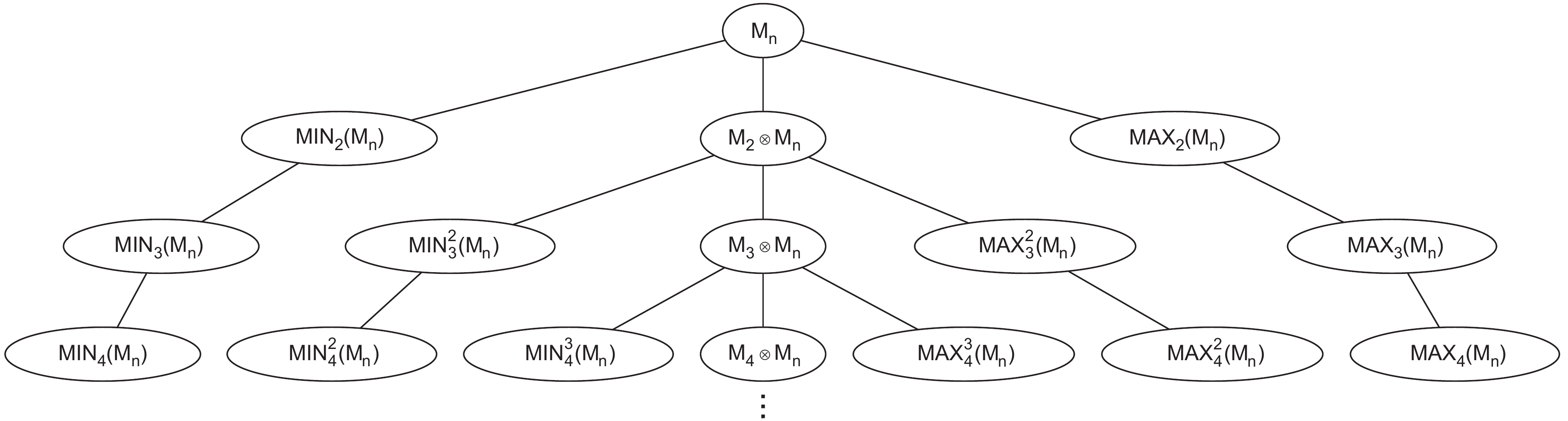}
\end{center}\vspace{-0.3in}
\caption[A representation of the $k$-minimal and $k$-maximal operator spaces]{\hsp A representation of the $k$-minimal and $k$-maximal operator spaces on $M_n$. The $m$-th row of the tree shows the various norms on $M_m \otimes M_n$ and each path starting from the root corresponds to one of the $k$-minimal or $k$-maximal operator spaces. The leftmost path represents $MIN(M_n)$ and the rightmost path represents $MAX(M_n)$. The path down the centre represents the ``naive'' operator space $M_n$ itself.}\label{fig:opspacetree}
\end{figure}

One of the primary reasons for our interest in the $k$-minimal operator spaces is the following result, which says that the norms of $MIN^k(M_n)$ are exactly the $S(k)$-operator norm of Section~\ref{sec:sk_operator_norm}.
\begin{thm}\label{thm:MINkChar}
	Let $X \in M_m \otimes M_n$. Then $\big\|X\big\|_{MIN^k_m(M_n)} = \big\|X\big\|_{S(k)}$.
\end{thm}
\begin{proof}
	Recall from Theorem~\ref{thm:cb_norm_gen_kraus} that any completely bounded map $\Phi : M_n \rightarrow M_k$ has a representation of the form
	\begin{align}\label{eq:cbform}
		\Phi(Y) = \sum_{i=1}^{nk} A_i Y B_i^\dagger \quad \text{with } A_i, B_i \in M_{k,n} \, \text{ and } \, \Big\| \sum_{i=1}^{nk}A_i A_i^\dagger\Big\|\Big\| \sum_{i=1}^{nk}B_i B_i^\dagger\Big\| = \big\|\Phi\big\|_{cb}^2.
	\end{align}
	By using the fact that $\Phi$ is completely contractive in the definition of $\big\|X\big\|_{MIN^k_m(M_n)}$ and a rescaling of the operators $\big\{A_i\big\}$ and $\big\{B_i\big\}$ we have
	\begin{align*}
		\big\|X\big\|_{MIN^k_m(M_n)} = \sup \Big\{ \big\| \sum_{i=1}^{nk} (I \otimes A_i) X (I \otimes B_i^\dagger) \big\| : \big\| \sum_{i=1}^{nk}A_i A_i^\dagger\big\| = \big\| \sum_{i=1}^{nk}B_i B_i^\dagger\big\| = 1 \Big\},
	\end{align*}
	where the supremum is taken over all families of operators $\big\{ A_i \big\}, \big\{ B_i \big\} \subset M_{k,n}$ satisfying the normalization condition. Now define $\alpha_{ij}\ket{a_{ij}} := A_i^\dagger\ket{j}$ and $\beta_{ij}\ket{b_{ij}} := B_i^\dagger\ket{j}$, and let $\ket{v} = \sum_{j=1}^{k} \gamma_j \ket{c_j} \otimes \ket{j}, \ket{w} = \sum_{j=1}^{k} \delta_j \ket{d_j} \otimes \ket{j} \in \bb{C}^m \otimes \bb{C}^k$ be arbitrary unit vectors. Then simple algebra reveals
	\begin{align*}
		\nu_i\ket{v_i} & := (I_m \otimes A_i^\dagger)\ket{v} = \sum_{j=1}^{k} \alpha_{ij} \gamma_j \ket{c_j} \otimes \ket{a_{ij}} \, \, \text{ and } \\
		\mu_i\ket{w_i} & := (I_m \otimes B_i^\dagger)\ket{w} = \sum_{j=1}^{k} \beta_{ij} \delta_j \ket{d_j} \otimes \ket{b_{ij}}.
	\end{align*}
	In particular, $SR(\ket{v_i}),SR(\ket{w_i}) \leq k$ for all $i$. Furthermore, by the normalization condition on $\big\{A_i\big\}$ and $\big\{B_i\big\}$ we have that
	\begin{align}\label{eq:muNormalize}
		\bra{v}(I_m \otimes \sum_{i=1}^{nk} A_i A_i^\dagger)\ket{v} = \sum_{i=1}^{nk} \nu_i^2 \leq 1 \, \, \, \text{ and } \, \, \, \bra{w}(I_m \otimes \sum_{i=1}^{nk} B_i B_i^\dagger)\ket{w} = \sum_{i=1}^{nk} \mu_i^2 \leq 1.
	\end{align}
	
	We can thus write
	\begin{align}\label{eq:cbschmidt}
		\left|\sum_{i=1}^{nk}\bra{v}(I_m \otimes A_i)(X)(I_m \otimes B_i^\dagger)\ket{w}\right| = \left|\sum_{i=1}^{nk} \nu_i \mu_i \bra{v_i} X \ket{w_i}\right| \leq \sum_{i=1}^{nk} \nu_i \mu_i \big| \bra{v_i} X \ket{w_i}\big|.
	\end{align}
	The normalization condition~\eqref{eq:muNormalize} and the Cauchy--Schwarz inequality tell us that there is a particular $i^\prime$ such that the sum~\eqref{eq:cbschmidt} $\leq \left| \bra{v_{i^\prime}} X \ket{w_{i^\prime}}\right|$. Taking the supremum over all vectors $\ket{v}$ and $\ket{w}$ gives the ``$\leq$'' inequality.
	
	The ``$\geq$'' inequality can be seen by noting that if we have two vectors in their Schmidt decompositions $\ket{v} = \sum_{i=1}^k \alpha_i \ket{c_i} \otimes \ket{a_i}$ and $\ket{w} = \sum_{i=1}^k \beta_i \ket{d_i} \otimes \ket{b_i}$, then we can define operators $A,B \in M_{k,n}$ by setting their $i^{th}$ row in the standard basis to be $\bra{a_i}$ and $\bra{b_i}$, respectively. Because the rows of $A$ and $B$ form orthonormal sets, $\big\|A\big\| = \big\|B\big\| = 1$. Additionally, if we define $\ket{v^\prime} = \sum_{i=1}^{k} \alpha_i \ket{c_i} \otimes \ket{i}$ and $\ket{w^\prime} = \sum_{i=1}^{k} \beta_i \ket{d_i} \otimes \ket{i}$, then
	\begin{align*}
		\big\| (I_m \otimes A)(X)(I_m \otimes B^\dagger) \big\| \geq \big| \bra{v^\prime}(I_m \otimes A)(X)(I_m \otimes B^\dagger)\ket{w^\prime} \big| = \big| \bra{v}X\ket{w} \big|.
	\end{align*}
	Taking the supremum over all vectors $\ket{v},\ket{w}$ with $SR(\ket{v}),SR(\ket{w}) \leq k$ gives the result.
\end{proof}

When thinking of $M_n$ as an operator system (instead of an operator space), it is more natural to define the norm~\eqref{eq:k_min_space} by taking the supremum over all completely positive unital maps $\Phi: M_n \rightarrow M_k$ rather than all complete contractions (similarly, to define the norm~\eqref{eq:k_max_space} one would take the supremum over all $k$-positive unital maps rather than $k$-contractive maps). In this case, the $k$-minimal norm no longer coincides with the $S(k)$-norm on $M_m(M_n)$ but rather has the following slightly different form:
	\begin{align}\begin{split}\label{eq:k_min_system}
		\big\|X\big\|_{OMIN^k_m(M_n)} = \sup_{\ket{v},\ket{w}}\Big\{ \big| \bra{v} X \ket{w} \big| : & \ SR(\ket{v}),SR(\ket{w}) \leq k \text{ and} \\
			& \ \exists \, P \in M_m \text{ s.t. } (P \otimes I_n)\ket{v} = \ket{w} \Big\},
	\end{split}\end{align}
where the notation $OMIN^k(M_n)$ refers to a new operator system structure that is being assigned to $M_n,$ which we discuss in detail in Section~\ref{sec:kMinOpSys}.
	
	Intuitively, this norm has the same interpretation as the $S(k)$-operator norm except with the added restriction that the vectors $\ket{v}$ and $\ket{w}$ look the same on the second subsystems. We will examine this norm in more detail in Section~\ref{sec:OpSysNorms}. In particular, we will see in Theorem~\ref{thm:k_order_norm} that the norm~\eqref{eq:k_min_system} is a natural norm on the $k$-super minimal operator system structure (to be defined in Section~\ref{sec:kMinOpSys}), which plays an analogous role to the $k$-minimal operator space structure.

Now that we have characterized the norms of $MIN^k(M_n)$ in a fairly concrete way, we turn our attention to the $k$-maximal norm. The following result is directly analogous to a corresponding known characterization of the $MAX(V)$ norm \cite[Theorem 14.2]{P03}.
\begin{thm}\label{thm:kMaxNormChar}
	Let $V$ be an operator space and let $X \in M_m \otimes M_n$. Then
	\begin{align*}
		\big\| X \big\|_{MAX^k_m(M_n)} = \inf\Big\{ \big\|A\big\| \big\|B\big\| : & \ A,B \in M_{m,rk}, x_i \in M_k \otimes M_n, \|x_i\| \leq 1 \text{ with} \\
		& \ X = (A \otimes I_n) {\rm diag}(x_1,\ldots,x_r) (B^\dagger \otimes I_n) \Big\},
	\end{align*}
	where we consider ${\rm diag}(x_1,\ldots,x_r) \in M_{rk} \otimes M_n$ in the natural way and the infimum is taken over all such decompositions of $X$.
\end{thm}
\begin{proof}
	The ``$\leq$'' inequality follows simply from the axioms of an operator space: if $X = (A \otimes I_n) {\rm diag}(x_1,\ldots,x_r) (B^\dagger \otimes I_n) \in M_m \otimes M_n$ then
	\begin{align*}
		(id_m \otimes \Phi)(X) = (A \otimes I_n) {\rm diag}((id_{k} \otimes \Phi)(x_1),\ldots,(id_{k} \otimes \Phi)(x_r) (B^\dagger \otimes I_n).
	\end{align*}
	Thus
	\begin{align*}
		\big\|(id_m \otimes \Phi)(X)\big\| \leq \big\|A\big\| \big\|B\big\| \max\big\{\|(id_k \otimes \Phi)(x_1)\|, \ldots, \|(id_k \otimes \Phi)(x_r)\|\big\}.
	\end{align*}
	By taking the supremum over maps $\Phi$ with $\|id_k \otimes \Phi\| \leq 1$, the ``$\leq$'' inequality follows.
	
	We will now show that the infimum on the right is an $L^{\infty}$-matrix norm that coincides with the operator norm $\|\cdot\|$ for $1 \leq m \leq k$. The ``$\geq$'' inequality will then follow from the fact that $\|\cdot\|_{MAX^k_m(M_n)}$ is the maximal such norm.
	
	First, denote the infimum on the right by $\big\|X\big\|_{m,\textup{inf}}$ and fix some $1 \leq m \leq k$. Then the inequality $\big\|X\big\| \leq \big\|X\big\|_{m,\textup{inf}}$ follows immediately by picking any particular decomposition $X = (A \otimes I_n) {\rm diag}(x_1,\ldots,x_r) (B^\dagger \otimes I_n)$ and using the axioms of an operator space to see that
	\begin{eqnarray*}
		\big\| X = (A \otimes I_n) {\rm diag}(x_1,\ldots,x_r) (B^\dagger \otimes I_n) \big\| & \leq & \big\|A\big\| \big\|B\big\| \max\big\{\|x_1\|, \ldots, \|x_r\|\big\} \\ & \leq & \big\|A\big\| \big\|B\big\| \\ & \leq & \big\|X\big\|_{m,\textup{inf}}.
	\end{eqnarray*}
	The fact that equality is attained by some decomposition of $X$ comes simply from writing letting $A = \big\| X \big\|I_k$, $B = I_k$, $r = 1$, and $x_1 = (X \oplus 0_{k-m})/\big\|X\big\|$. It follows that $\|\cdot\|_{m,\textup{inf}} = \|\cdot\|$ for $1 \leq m \leq k$.
	
	All that remains to be proved is that $\|\cdot\|_{m,\textup{inf}}$ is an $L^\infty$-matrix norm, which we omit as it is directly analogous to the proof of \cite[Theorem 14.2]{P03}.
\end{proof}

Theorem~\ref{thm:kMaxNormChar} was proved more generally for arbitrary operator spaces $V$ in \cite{JKPP11}, but the result as given is enough for our purposes. As one final note, observe that we can obtain lower bounds of the $k$-minimal and $k$-maximal operator space norms simply by choosing particular maps $\Phi$ that satisfy the normalization condition of their definition. Upper bounds of the $k$-maximal norms can be obtained from Theorem~\ref{thm:kMaxNormChar}. Upper bounds of the $k$-minimal norms can be obtained via the methods of Section~\ref{sec:semidefProgramMNorm}.

\subsection{Completely Bounded k-Minimal Norms}\label{sec:cb_norm_op_space}

We now investigate the completely bounded version of the $k$-minimal operator space norms that have been introduced. In Section~\ref{sec:contrac_sep_crit} we will use the ideas presented here to show that completely bounded norms can be used to provide a characterization of Schmidt number analogous to its more well-known characterization in terms of $k$-positive maps.

Given operator spaces $V$ and $W$, the completely bounded (CB) norm from $V$ to $W$ is defined by
\begin{align*}
	\big\|\Phi\big\|_{CB(V,W)} := \sup_{m \geq 1}\Big\{ \big\|(id_m \otimes \Phi)(X)\big\|_{M_m(W)} : X \in M_m(V) \text{ with } \big\|X\big\|_{M_m(V)} \leq 1 \Big\}.
\end{align*}
This quantity clear reduces to the ``standard'' completely bounded norm of $\Phi$ in the case when $V = M_r$ and $W = M_n$. We will now characterize this norm in the case when $V = M_r$ and $W = MIN^k(M_n)$. In particular, we will see that the $k$-minimal completely bounded norm of $\Phi$ is equal to the perhaps more familiar operator norm $\big\|id_k \otimes \Phi\big\|$ -- that is, the CB norm in this case stabilizes in much the same way that the standard CB norm stabilizes (indeed, in the $k = n$ case we get exactly the standard CB norm). This result was originally proved in \cite{OR04}, but we prove it here using elementary means for completeness and clarity, and also because we will subsequently need the operator system version of the result, which can be proved in the same way.
\begin{thm}\label{thm:mainCB}
	Let $\Phi : M_r \rightarrow M_n$ be a linear map and let $1 \leq k \leq n$. Then
	\begin{align*}
		\big\|id_k \otimes \Phi\big\| & = \big\|\Phi\big\|_{CB(M_r,MIN^k(M_n))}.
	\end{align*}
\end{thm}
\begin{proof}
	To see the ``$\leq$'' inequality, simply notice that $\big\|Y\big\|_{M_k(MIN^k(M_n))} = \big\|Y\big\|_{M_k(M_n)}$ for all $Y \in M_k(M_n)$. We thus just need to show the ``$\geq$'' inequality, which we do in much the same manner as Smith's original proof that the standard CB norm stabilizes \cite{S83}.
	
	First, use Theorem~\ref{thm:MINkChar} to write
	\begin{align}\label{eq:cb_k_sup}
		\big\|\Phi\big\|_{CB(M_r,MIN^k(M_n))} = \sup_{m \geq 1}\Big\{ \big\|(id_m \otimes \Phi)(X)\big\|_{S(k)} : \big\|X\big\| \leq 1 \Big\}.
	\end{align}
	Now fix $m \geq k$ and a pure state $\ket{v} \in \bb{C}^m \otimes \bb{C}^n$ with $SR(\ket{v}) \leq k$. We begin by showing that there exists an isometry $V : \bb{C}^k \rightarrow \bb{C}^m$ and a state $\ket{\tilde{v}} \in \bb{C}^k \otimes \bb{C}^n$ such that $(V \otimes I_n)\ket{\tilde{v}} = \ket{v}$. To this end, write $\ket{v}$ in its Schmidt decomposition $\ket{v} = \sum_{i=1}^{k} \alpha_i \ket{a_i} \otimes \ket{b_i}$. Because $k \leq m$, we may define an isometry $V : \bb{C}^k \rightarrow \bb{C}^m$ by $V\ket{i} = \ket{a_i}$ for $1 \leq i \leq k$. If we define $\ket{\tilde{v}} := \sum_{i=1}^{k} \alpha_i \ket{i} \otimes \ket{b_i}$ then $(V \otimes I_n)\ket{\tilde{v}} = \ket{v}$, as desired.
	
	Now choose $\tilde{X} \in M_m(M_r)$ such that $\big\|\tilde{X}\big\| \leq 1$ and the supremum~\eqref{eq:cb_k_sup} (holding $m$ fixed) is attained by $\tilde{X}$. Then choose vectors $\ket{v},\ket{w} \in \bb{C}^m \otimes \bb{C}^n$ with $SR(\ket{v}),SR(\ket{w}) \leq k$ such that
	\begin{align*}
		\big\| (id_m \otimes \Phi)(\tilde{X}) \big\|_{S(k)} = \big|\bra{v} (id_m \otimes \Phi)(\tilde{X}) \ket{w}\big|.
	\end{align*}
	As we saw earlier, there exist isometries $V, W : \bb{C}^k \rightarrow \bb{C}^m$ and unit vectors $\ket{\tilde{v}},\ket{\tilde{w}} \in \bb{C}^k \otimes \bb{C}^n$ such that $(V \otimes I_n)\ket{\tilde{v}} = \ket{v}$ and $(W \otimes I_n)\ket{\tilde{w}} = \ket{w}$. Thus
	\begin{align*}
		\big\| (id_m \otimes \Phi)(\tilde{X}) \big\|_{S(k)} & = \big|\bra{\tilde{v}}(V^\dagger \otimes I_n) (id_m \otimes \Phi)(\tilde{X}) (W \otimes I_n)\ket{\tilde{w}} \big| \\
		& = \big|\bra{\tilde{v}} (id_k \otimes \Phi)((V^\dagger \otimes I_r)\tilde{X}(W \otimes I_r))\ket{\tilde{w}} \big| \\
		& \leq \big\| (id_k \otimes \Phi)((V^\dagger \otimes I_r)\tilde{X}(W \otimes I_r)) \big\| \\
		& \leq \sup \Big\{ \big\| (id_k \otimes \Phi)(X) \big\| : X \in M_k(M_r) \text{ with } \big\|X\big\| \leq 1 \Big\},
	\end{align*}
	where the final inequality comes from the fact that $\big\|(V^\dagger \otimes I_r)\tilde{X}(W \otimes I_r)\big\| \leq 1$. The desired inequality follows, completing the proof.
\end{proof}

\subsection{The Dual of the S(k)-Operator Norm}\label{sec:dual_op_norm}

Using the techniques and results of the previous sections, we are finally in a position to explore the dual of the $S(k)$-operator norm, which we denote $\|\cdot\|_{S(k)}^\circ$. The key idea is that, because the $S(k)$-operator norm is the minimal $L^\infty$-matrix norm on $M_n$, the dual of the $S(k)$-norm is the maximal $L^1$-matrix norm on $M_{n,tr}$ (recall that $M_{n,tr}$ denotes $M_n$ equipped with the trace norm).

Before proceeding to the statement of the theorem, it is worth having another look at Theorem~\ref{thm:sk_dual_groth}, which provides the corresponding result for the $s(k)$-vector norms. With that result in mind, the following characterization is exactly what might be expected.
\begin{thm}\label{thm:dual_norm_char}
	Let $Y \in M_m \otimes M_n$. Then
	\begin{align*}
		\big\| Y \big\|_{S(k)}^{\circ} & = \inf\Big\{ \sum_i |c_i| : Y = \sum_i c_i\ketbra{v_i}{w_i} \text{ with } SR(\ket{v_i}),SR(\ket{w_i}) \leq k \ \forall \, i \Big\}, 
	\end{align*}
	where the infimum is taken over all decompositions of $Y$ of the given form.
\end{thm}
\begin{proof}
	We first recall \cite[Theorem 5.1]{R88}, which says that the collection of dual norms of any $L^\infty$-matrix normed space on $M_n$ defines an $L^1$-matrix normed space on $M_{n,tr}$. We can thus use Theorem~\ref{thm:MINkChar} to see that the norms $\| \cdot \|_{S(k)}^{\circ}$ define an $L^1$-matrix normed space on $M_{n,tr}$. Furthermore, because $\|\cdot\|_a \leq \|\cdot\|_b$ implies that $\|\cdot\|_a^\circ \geq \|\cdot\|_b^\circ$, it follows that $\| \cdot \|_{S(k)}^{\circ}$ is the largest $L^1$-matrix norm on $M_{n,tr}$ that is equal to $\|\cdot\|_{tr}$ on $M_m \otimes M_n$ for $1 \leq m \leq k$.
	
	Throughout this proof, we denote the given infimum by $\|\cdot\|_{k,\textup{inf}}$ for simplicity. To prove the result, we first show that $\|\cdot\|_{S(k)^\circ} \leq \|\cdot\|_{k,\textup{inf}}$ -- the opposite inequality comes from showing that $\|\cdot\|_{k,\textup{inf}}$ is also an $L^1$-matrix norm on $M_{n,tr}$ that is equal to $\|\cdot\|_{tr}$ on $M_m \otimes M_n$ for $1 \leq m \leq k$.
	
	To see that $\big\|Y\big\|_{S(k)^\circ} \leq \big\|Y\big\|_{k,\textup{inf}}$, let $X$ be such that $\big\|X\big\|_{S(k)} \leq 1$ and write $Y = \sum_i c_i\ketbra{v_i}{w_i}$ with $SR(\ket{v_i}),SR(\ket{w_i}) \leq k$ for all $i$. Then
	\begin{align*}
		\big|\langle X | Y \rangle\big| & = \Big|\sum_i c_i \big\langle X \big| \ketbra{v_i}{w_i}\big\rangle\Big| \\
		& \leq \sum_i |c_i| \big| \bra{w_i} X^\dagger \ket{v_i} \big| \\
		& \leq \sum_i |c_i|.
	\end{align*}
	By taking the supremum over operators $X$ with $\big\|X\big\|_{S(k)} \leq 1$ and the infimum over decompositions of $Y$ of the given form, the ``$\leq$'' inequality follows.
	
	We now show that $\|\cdot\|_{k,\textup{inf}}$ is an $L^1$-matrix norm that coincides with $\|\cdot\|_{tr}$ for $1 \leq m \leq k$. To this end, note that we already showed that $\|\cdot\|_{tr} \leq \|\cdot\|_{S(k)}^\circ \leq \|\cdot\|_{k,\textup{inf}}$ for all $m$. To see that the opposite inequality holds when $1 \leq m \leq k$, note that we can simply write $Y$ in its singular value decomposition $Y = \sum_i \alpha_i \ketbra{a_i}{b_i}$. Then $\big\|Y\big\|_{tr} = \sum_i \alpha_i \geq \big\|Y\big\|_{k,\textup{inf}}$. 
	
	The remainder of the proof is devoted to showing that $\|\cdot\|_{k,\textup{inf}}$ is an $L^1$-matrix norm. To see that it is a norm, note that the properties $\big\|\lambda Y\big\|_{k,\textup{inf}} = |\lambda|\big\|Y\big\|_{k,\textup{inf}}$ and $\big\|Y\big\|_{k,\textup{inf}} = 0$ if and only if $Y = 0$ both follow trivially from the definition of $\|\cdot\|_{k,\textup{inf}}$. To see the triangle inequality, fix $\varepsilon > 0$ and let $Y_1 = \sum_i c_i\ketbra{v_i}{w_i}$, $Y_2 = \sum_i d_i\ketbra{x_i}{y_i}$ be decompositions of $Y_1, Y_2$ with $SR(\ket{v_i}),SR(\ket{w_i}),SR(\ket{x_i}),SR(\ket{y_i}) \leq k$ for all $i$ such that $\sum_i |c_i| \leq \big\|Y_1\big\|_{k,\textup{inf}} + \varepsilon$ and $\sum_i |d_i| \leq \big\|Y_2\big\|_{k,\textup{inf}} + \varepsilon$. Then we can decompose $Y_1 + Y_2$ as
	\begin{align*}
		Y_1 + Y_2 = \sum_i c_i\ketbra{v_i}{w_i} + \sum_i d_i\ketbra{x_i}{y_i},
	\end{align*}
	so
	\begin{align*}
		\big\|Y_1 + Y_2\big\|_{k,\textup{inf}} \leq \sum_i |c_i| + \sum_i |d_i| \leq \big\|Y_1\big\|_{k,\textup{inf}} + \big\|Y_2\big\|_{k,\textup{inf}} + 2\varepsilon.
	\end{align*}
	Since $\varepsilon > 0$ was arbitrary, the triangle inequality follows and $\|\cdot\|_{k,\textup{inf}}$ is a norm.
	
	To see that $\|\cdot\|_{k,\textup{inf}}$ satisfies the matrix norm property (i.e., property~(1) in Section~\ref{sec:op_space}), write $Y = c_i\ketbra{v_i}{w_i}$. Define $a_i \ket{\tilde{v_i}} := (A \otimes I_n)\ket{v_i}$ and $b_i \ket{\tilde{w_i}} := (B \otimes I_n)\ket{w_i}$. Then
	\begin{align*}
		(A \otimes I_n)Y(B^\dagger \otimes I_n) & = \sum_i c_i(A \otimes I_n)\ketbra{v_i}{w_i}(B^\dagger \otimes I_n) \\
		& = \sum_i c_i a_i \overline{b_i} \ketbra{\tilde{v_i}}{\tilde{w_i}}.
	\end{align*}
	Because $|a_i| \leq \big\|A\big\|$ and $|b_i| \leq \big\|B\big\|$ for all $i$, it follows that
	\begin{align*}
		\big\|(A \otimes I_n)Y(B^\dagger \otimes I_n)\big\|_{k,\textup{inf}} \leq \sum_i |c_i a_i \overline{b_i}| \leq \big\|A\big\|\left( \sum_i |c_i| \right)\big\|B\big\|.
	\end{align*}
	The desired inequality now follows from taking the infimum over all decompositions of $Y$ of the desired form.
	
	Finally, to see that $\|\cdot\|_{k,\textup{inf}}$ satisfies the $L^1$ property, we show that $\big\|Y_1 \oplus Y_2\big\|_{k,\textup{inf}} = \big\|Y_1\big\|_{k,\textup{inf}} + \big\|Y_2\big\|_{k,\textup{inf}}$ for all $Y_1 \in M_m \otimes M_n$ and $Y_2 \in M_r \otimes M_n$. The ``$\leq$'' inequality follows immediately from the triangle inequality. To see the ``$\geq$'' inequality, define $P_1, P_2 \in M_m \oplus M_r$ by $P_1 := I_m \oplus 0_r$ and $P_2 := 0_m \oplus I_r$. Let $\Psi : (M_m \oplus M_r) \rightarrow (M_m \oplus M_r)$ be the completely positive map with $P_1$ and $P_2$ as its Kraus operators. If we write $Y_1 \oplus Y_2 = \sum_i c_i \ketbra{v_i}{w_i}$ then we have
	\begin{align}\label{eq:dual_oplus}
		Y_1 \oplus Y_2 & = (\Psi \otimes id_n)(Y_1 \oplus Y_2) \\
		& = \sum_i c_i ({\rm Ad}_{P_1} \otimes id_n)(\ketbra{v_i}{w_i}) + \sum_i c_i ({\rm Ad}_{P_2} \otimes id_n)(\ketbra{v_i}{w_i}),
	\end{align}
	where we recall the adjoint map ${\rm Ad}_{A}(X) = A X A^\dagger$. Define $a_{i,j} \ket{\tilde{v_{i,j}}} := (P_j \otimes I_n)\ket{v_i}$ and $b_{i,j} \ket{\tilde{w_{i,j}}} := (P_j \otimes I_n)\ket{w_i}$ for $j = 1,2$ so that $a_{i,j},b_{i,j} \geq 0$. Since $P_1$ and $P_2$ are mutually orthogonal projections, we have $a_{i,1}^2 + a_{i,2}^2 \leq 1$ and $b_{i,1}^2 + b_{i,2}^2 \leq 1$. As an aside for now, note that the Cauchy--Schwarz inequality tells us that
	\begin{align}\label{eq:sk_dual_cauchy}
		a_{i,1}b_{i,1} + a_{i,2}b_{i,2} \leq 1 \quad \forall \, i.
	\end{align}
	
	Continuing from Equation~\eqref{eq:dual_oplus} shows
	\begin{align*}
		Y_1 \oplus Y_2 & = \sum_i c_i a_{i,1} b_{i,1} \ketbra{\tilde{v_{i,1}}}{\tilde{w_{i,1}}} + \sum_i c_i a_{i,2} b_{i,2} \ketbra{\tilde{v_{i,2}}}{\tilde{w_{i,2}}},
	\end{align*}
	where the sums on the right are decompositions of $Y_1$ and $Y_2$, respectively. Thus
	\begin{align*}
		\big\|Y_1\big\|_{k,\textup{inf}} + \big\|Y_2\big\|_{k,\textup{inf}} & \leq \sum_i | c_i a_{i,1} b_{i,1} | + \sum_i | c_i a_{i,2} b_{i,2} | \\
		& = \sum_i | c_i | \big(a_{i,1}b_{i,1} + a_{i,2}b_{i,2}\big) \\
		& \leq \sum_i |c_i|,
	\end{align*}
	where the final line uses Inequality~\eqref{eq:sk_dual_cauchy}.
\end{proof}

Much like the $S(k)$-norm can be thought of as a ``$k$-local'' version of the operator norm, it is now clear how the dual of the $S(k)$-norm is similarly analogous to the trace norm. Indeed, in the $k = \min\{m,n\}$ case, Theorem~\ref{thm:dual_norm_char} tells us that
\begin{align*}
	\big\|Y\big\|_{S(\min\{m,n\})}^\circ = \inf\Big\{ \sum_i |c_i| : Y = \sum_i c_i\ketbra{v_i}{w_i} \Big\}, 
\end{align*}
where the infimum is taken over all decompositions of $Y$ into the sum of rank-$1$ operators. This is a well-known characterization of the trace norm, and the infimum is attained when the decomposition of $Y$ is chosen to be the singular value decomposition.

On the other extreme, this norm is also well-known in the $k = 1$ case. Indeed,
\begin{align*}
	\big\| Y \big\|_{S(1)}^{\circ} & = \inf\Big\{ \sum_i |c_i| : Y = \sum_i c_i\ketbra{v_i}{w_i} \otimes \ketbra{x_i}{y_i} \Big\} \\
	& = \inf\Big\{ \sum_i \big\|A_i\big\|_{tr} \big\|B_i\big\|_{tr} : Y = \sum_i A_i \otimes B_i \Big\},
\end{align*}
and it was shown in \cite{R00} (see also \cite{Rud01}) that this norm completely characterizes separability in the sense that if $\rho$ is a density operator then it is separable if and only if $\| \rho \|_{S(1)}^{\circ} = 1$. We now establish the natural generalization of this fact.
\begin{thm}\label{thm:dual_sep_equiv}
	Let $\rho \in M_m \otimes M_n$ be a density matrix. Then $SN(\rho) \leq k$ if and only if $\|\rho\|_{S(k)}^\circ = 1$.
\end{thm}
\begin{proof}
	Note that $\|\rho\|_{S(k)}^\circ \geq \|\rho\|_{tr} = 1$ for all $\rho$, so we only consider the opposite inequality. If $SN(\rho) \leq k$ then we can write $\rho = \sum_i p_i \ketbra{v_i}{v_i}$ with $SR(\ket{v_i}) \leq k$ for all $i$. Then $\|\rho\|_{S(k)}^\circ \leq \sum_i p_i = 1$, as desired. To see the converse, assume that $\|\rho\|_{S(k)}^\circ \leq 1$ and let $Y = Y^\dagger \in M_m \otimes M_n$ be $k$-block positive. Corollary~\ref{cor:kPosInf1} shows that if we write $Y = cI - X$ with $X \geq 0$ then $c \geq \big\|X\big\|_{S(k)}$. Without loss of generality, we scale $Y$ so that $c = 1$ (i.e., $Y = I - X$ with $\big\|X\big\|_{S(k)} \leq 1$). Then
	\begin{align*}
		\big\langle \rho | Y \big\rangle = \big\langle \rho | I - X \big\rangle = 1 - \big\langle \rho | X \big\rangle \geq 1 - \|\rho\|_{S(k)}^\circ \geq 0.
	\end{align*}
	Since $Y$ is (up to scaling) an arbitrary $k$-block positive operator, it follows that $SN(\rho) \leq k$, which completes the proof.
\end{proof}

\section{Operator Systems on Complex Matrices}\label{sec:op_system}

An (abstract) operator system on $M_n$ is a family of convex cones $\{C_m\}_{m=1}^{\infty} \subseteq M_m \otimes M_n$ that satisfy the following two properties:
\begin{itemize}
	\item for each $m_1,m_2 \in \bb{N}$ and $A \in M_{m_1,m_2}$ we have $({\rm Ad}_{A} \otimes id_n)(C_{m_1}) \subseteq C_{m_2}$; and
	\item $C_1 = M_n^{+}$, the cone of positive semidefinite elements of $M_n$.
\end{itemize}
Property~(1) above ensures that the cones $C_m$ ``behave well'' with each other. For example, it ensures that if we embed $X \in M_m \otimes M_n$ as $\tilde{X} \in M_{m+1} \otimes M_n$ by adding rows and columns of zeroes then $X \in C_m$ if and only if $\tilde{X} \in C_{m+1}$. It is also worth remarking at this point upon the similarity between the definition of an abstract operator space and that of an abstract operator system. In a sense, abstract operator systems do for cones what abstract operator spaces do for norms.

Abstract operator systems can be defined more generally on any Archimedean $*$-ordered vector space $V$, but the above definition with $V = M_n$ is much simpler and suited to our particular needs. The interested reader is directed to \cite[Chapter 13]{P03} for a more thorough treatment of general abstract operator systems. The fact that matrix ordered $*$-vector spaces can be thought of as operator systems follows from the work of Choi and Effros \cite{CE77}.

Abstract operator systems are typically defined with two additional requirements that we have not mentioned:
\begin{itemize}
	\item[(3)] $C_m \cap -C_m = \{0\}$ for each $m \in \bb{N}$; and
	\item[(4)] for every $m \in \bb{N}$ and $X = X^\dagger \in M_m \otimes M_n$, there exists $r > 0$ such that $rI + X \in C_m$.
\end{itemize}
Both of these conditions follow for free from the fact that, in our setting, $C_1 = M_n^{+}$.
	
	To see that property~(3) holds, notice that $C_1 \cap -C_1 = \{0\}$, and suppose that $X \in C_m \cap -C_m$ for some $m \geq 2$. Then $({\rm Ad}_{A} \otimes id_n)(X) \in C_1 \cap -C_1$ for any $A \in M_{m,1}$. Because $C_1 \cap -C_1 = \{0\}$, it follows that $\bra{vw}X\ket{vw} = 0$ for all $\ket{v}, \ket{w}$. It follows from Lemma~\ref{lem:sep_prod_zero} that $X = 0$, so $C_m \cap -C_m = \{0\}$ for all $m \in \bb{N}$.
	
	Property~(4) holds because the smallest family of cones on $M_n$ such that $({\rm Ad}_{A} \otimes id_n)(C_{m_1}) \subseteq C_{m_2}$ for all $m_1,m_2 \in \bb{N}$ are the cones of separable operators in $M_m \otimes M_n$ (this fact is easily-verified and is part of the statement of the upcoming Theorem~\ref{thm:opSysConeChar}). It is well-known that there always exists $r > 0$ such that $rI + X$ is separable \cite{GB02}, so the same $r$ ensures that $rI + X \in C_m$.

One particularly important operator system on $M_n$ is the one constructed by associating $M_m \otimes M_n$ with $M_{mn}$ in the natural way and letting $C_m \subseteq M_m \otimes M_n$ be the cones of positive semidefinite operators. We denote this operator system simply by $M_n$, and it will be clear from context whether we mean the operator system $M_n$, the operator space $M_n$, or simply the set $M_n$ without regard to any family of cones or norms. Other operator systems on $M_n$ are denoted like $V(M_n)$ (or simply $V$) in order to avoid confusion with the operator system $M_n$ itself.

If $V_1(M_n)$ and $V_2(M_n)$ are two operator systems defined by the cones $\{C_m\}_{m=1}^{\infty}$ and $\{D_m\}_{m=1}^{\infty}$ respectively, then a map $\Phi : M_n \rightarrow M_n$ is said to be \emph{completely positive from $V_1(M_n)$ to $V_2(M_n)$} if $(id_m \otimes \Phi)(C_m) \subseteq D_m$ for all $m \in \bb{N}$. The set of maps that are completely positive from $V_1(M_n)$ to $V_2(M_n)$ is denoted by $\cl{CP}(V_1(M_n),V_2(M_n))$, or simply $\cl{CP}(V(M_n))$ if the target operator system equals the source operator system. Note that in the case when $V_1(M_n)$ and $V_2(M_n)$ are both the standard operator system defined by the cones of positive semidefinite operators, this notation of complete positivity reduces to the standard notion of complete positivity introduced in Section~\ref{sec:cp_maps}.

\subsection{Minimal and Maximal Operator Systems}\label{sec:kMinOpSys}

A result of \cite{PTT11} shows that, much like there is a minimal and maximal abstract operator space on any normed vector space $V$, there is a minimal and maximal abstract operator system on any space $V$ satisfying certain (slightly technical) conditions. Importantly for us, there exist minimal and maximal operator systems on $M_n$, which we denote $OMIN(M_n)$ and $OMAX(M_n)$, respectively. That is, there exist particular families of cones $\{C_m^{\textup{min}}\}_{m=1}^\infty$ and $\{C_m^{\textup{max}}\}_{m=1}^\infty$ such that if $\{D_m\}_{m=1}^\infty$ are cones defining any other operator system on $M_n$ then $C_m^{\textup{max}} \subseteq D_m \subseteq C_m^{\textup{min}}$ for all $m \geq 1$. Notice that the inclusions are perhaps the opposite of what one might expect based on the names ``minimal'' and ``maximal'' -- the minimal operator system has the \emph{largest} family of cones and the maximal operator system has the \emph{smallest} family of cones. The names actually refer to the norms that they induce (see Section~\ref{sec:OpSysNorms}). The norm on the maximal operator system is the largest of any of the operator system norms, and the norm on the minimal operator system is the smallest of any operator system norm.

In \cite{Xthesis,Xha11} a generalization of these operator system structures, analogous to the $k$-minimal and $k$-maximal operator spaces presented in Section~\ref{sec:min_max_op_space}, was introduced. Given an operator system $V(M_n)$ (or even just cones $\big\{C_m\big\} \subseteq M_m \otimes M_n$ that satisfy the defining properties of an operator system for $1 \leq m \leq k$), the \emph{$k$-super minimal operator system} on $V$ and the \emph{$k$-super maximal operator system} of $V$, denoted $OMIN^k(V)$ and $OMAX^k(V)$ respectively, are defined via the following families of cones:
\begin{align*}
	C_m^{\textup{min},k}(V) & := \big\{ X : (id_m \otimes \Phi)(X) \geq 0 \ \ \forall \, \Phi \text{ with } (id_k \otimes \Phi)(C_k) \subseteq (M_k \otimes M_n)^+ \big\},	\\
	C_m^{\textup{max},k}(V) & := \big\{ \sum_i ({\rm Ad}_{A_i} \otimes id_n)(X_i) \in M_m \otimes M_n : A_i \in M_{m,k}, X_i \in C_k \ \forall \, i \big\}.
\end{align*}
We occasionally use the fact that the maps $\Phi$ in the definition of $C_m^{\textup{min},k}(V)$ can be chosen to be unital without loss of generality.

The interpretation of the $k$-super minimal and $k$-super maximal operator systems is completely analogous to the interpretation of $k$-minimal and $k$-maximal operator spaces. The positive cones $C_m^{\textup{min},k}(V)$ and $C_m^{\textup{max},k}(V)$ coincide with $C_m$ for $1 \leq m \leq k$, and out of all operator system structures with this property they are the largest (smallest, respectively) for $m > k$. For the remainder of this section, we restrict to the $V = M_n$ case, and in this case we denote these cones simply by $C_m^{\textup{min},k}$ and $C_m^{\textup{max},k}$.

Much like Theorem~\ref{thm:MINkChar} shows that the $k$-minimal operator spaces are very familiar to us, we now show that the $k$-minimal and $k$-maximal operator systems are familiar as well. In particular, we show that the cones $C_m^{\textup{min},k} \subseteq M_m \otimes M_n$ are exactly the cones of $k$-block positive operators, and the cones $C_m^{\textup{max},k} \subseteq M_m \otimes M_n$ are exactly the cones of (unnormalized) density operators $\rho$ with $SN(\rho) \leq k$. These facts have appeared implicitly in the past, but their importance merits making the details explicit:
\begin{thm}\label{thm:opSysConeChar}
	Let $X,\rho \in M_m \otimes M_n$. Then
	\begin{enumerate}[(a)]
		\item $X \in C_m^{\textup{min},k}$ if and only if $X$ is $k$-block positive; and
		\item $\rho \in C_m^{\textup{max},k}$ if and only if $SN(\rho) \leq k$.
	\end{enumerate}
\end{thm}
\begin{proof}
	To see (a), we will use techniques similar to those used in the proof of Theorem~\ref{thm:MINkChar}. Use the Kraus representation of completely positive maps so that $X \in C_m^{\textup{min},k}$ if and only if
	\begin{align*}
		\sum_{i=1}^{nk} (I_m \otimes A_i)X(I_m \otimes A_i^\dagger) \geq 0 \text{ for all } \big\{A_i\big\} \subset M_{k,n} \text{ with } \sum_{i=1}^{nk}A_i A_i^\dagger = I_k.
	\end{align*}
	Now define $\alpha_{ij}\ket{a_{ij}} := A_i^\dagger\ket{j}$ and let $\ket{v} = \sum_{j=1}^{k} \gamma_j \ket{c_j} \otimes \ket{j} \in \bb{C}^m \otimes \bb{C}^k$ be an arbitrary unit vector. Then some algebra reveals
	\begin{align*}
		\nu_i\ket{v_i} & := (I_m \otimes A_i^\dagger)\ket{v} = \sum_{j=1}^{k} \alpha_{ij} \gamma_j \ket{c_j} \otimes \ket{a_{ij}}.
	\end{align*}
	In particular, $SR(\ket{v_i}) \leq k$ for all $i$. Thus we can write
	\begin{align}\label{eq:cbschmidtSys}
		\sum_{i=1}^{nk}\bra{v}(I_m \otimes A_i)(X)(I_m \otimes A_i^\dagger)\ket{v} = \sum_{i=1}^{nk} \nu_i^2 \bra{v_i} X \ket{v_i} \geq 0.
	\end{align}
	Part (a) follows by noting that we can choose $\ket{v}$ and a completely positive map with one Kraus operator $A_1$ so that $(I_m \otimes A_1^\dagger)\ket{v}$ is any particular vector of our choosing with Schmidt rank no larger than $k$.
	
	To see the ``only if'' implication of (b), we could invoke various known duality results from operator theory and quantum information theory so that the result would follow from (a), but for completeness we will instead prove it using elementary means. To this end, suppose $\rho \in C_m^{\textup{max},k}$. Thus we can write $\rho = \sum_\ell ({\rm Ad}_{A_\ell} \otimes id_n)(X_\ell)$ for some $A_\ell \in M_{m,k}$ and $X_\ell \in (M_k \otimes M_n)^+$ for all $\ell$. Furthermore, write $X_\ell = \sum_h d_{\ell,h}\ketbra{v_{\ell,h}}{v_{\ell,h}}$ where $\ket{v_{\ell,h}} = \sum_{i=1}^k \ket{i} \otimes \ket{d_{\ell,h,i}}$. Then if we define $\alpha_{\ell,i}\ket{a_\ell,i} := A_{\ell}\ket{i}$, we have
	\begin{align*}
		\sum_\ell ({\rm Ad}_{A_\ell} \otimes id_n)(X_\ell) & = \sum_\ell \sum_{h=1}^{kn} d_{\ell,h} \sum_{ij=1}^k A_\ell\ketbra{i}{j} A_\ell^\dagger \otimes \ketbra{d_{\ell,h,i}}{d_{\ell,h,j}} \\
		& = \sum_{\ell}\sum_{h=1}^{kn} d_{\ell,h} \sum_{ij=1}^k \alpha_{\ell,i}\alpha_{\ell,j}\ketbra{a_{\ell,i}}{a_{\ell,j}} \otimes \ketbra{d_{\ell,h,i}}{d_{\ell,h,j}} \\
		& = \sum_{\ell}\sum_{h=1}^{kn} d_{\ell,h} \ketbra{w_{\ell,h}}{w_{\ell,h}},
	\end{align*}
	where
	\begin{align*}
		\ket{w_{\ell,h}} := \sum_{i=1}^k \alpha_{\ell,i}\ket{a_{\ell,i}} \otimes \ket{d_{\ell,h,i}}.
	\end{align*}
	Since $SR(\ket{w_{\ell,h}}) \leq k$ for all $\ell,h$, it follows that $SN(\rho) \leq k$ as well.
	
	For the ``if'' implication, we note that the above argument can easily be reversed.
\end{proof}

One of the useful consequences of Theorem~\ref{thm:opSysConeChar} is that we can now easily characterize completely positive maps between these various operator system structures. The following result characterizes the set of $k$-positive maps $\cl{P}_k$ and the set of $k$-superpositive maps $\cl{S}_k$ as completely positive maps between these $k$-super minimal and $k$-super maximal operator systems.
\begin{cor}\label{cor:kEntangleBreak}
	Let $\Phi : M_n \rightarrow M_n$ and let $k \leq n$. Then
	\begin{enumerate}[(a)]
		\item $\cl{CP}(OMIN^k(M_n), M_n) = \cl{S}_k$;
		\item $\cl{CP}(M_n, OMAX^k(M_n)) = \cl{S}_k$;
		\item $\cl{CP}(OMAX^k(M_n), M_n) = \cl{P}_k$;
		\item $\cl{CP}(M_n, OMIN^k(M_n)) = \cl{P}_k$;
		\item $\cl{CP}(OMIN^k(M_n), OMAX^k(M_n)) = \cl{S}_k$;
		\item $\cl{CP}(OMAX^k(M_n), OMIN^k(M_n)) = \cl{P}_k$;
		\item $\cl{CP}(OMIN^k(M_n)) = \cl{P}_k$; and
		\item $\cl{CP}(OMAX^k(M_n)) = \cl{P}_k$.
	\end{enumerate}
\end{cor}
\begin{proof}
	Facts (a), (b), (c), and (d) all follow immediately from Proposition~\ref{prop:right_cp_main}, the duality between the cones of $k$-positive and $k$-superpositive maps, and the Choi--Jamio{\l}kowski correspondences described in Section~\ref{sec:choi_jamiolkowski_schmidt}. Facts (e), (f), (g), and (h) similarly follow from Propositions~\ref{prop:semigroup} and~\ref{prop:semigroup2}, respectively.
\end{proof}

Most of the properties of Corollary~\ref{cor:kEntangleBreak} were originally proved in the $k = 1$ case in \cite{PTT11} and for arbitrary $k$ in \cite{Xthesis,Xha11}. Both of those proofs prove the result directly, without characterizing the cones $C_m^{\textup{min},k}$ and $C_m^{\textup{max},k}$ as in Theorem~\ref{thm:opSysConeChar}.

We close this section with a result that shows that the largest and smallest cones of completely positive maps between operator systems are the cones $\cl{P}(M_n)$ of positive maps and $\cl{S}(M_n)$ of superpositive maps, respectively.
\begin{cor}\label{cor:cp_sub_p}
	Let $V_1(M_n)$ and $V_2(M_n)$ be operator systems. Then $\cl{S}(M_n) \subseteq \cl{CP}(V_1(M_n),V_2(M_n)) \subseteq \cl{P}(M_n)$.
\end{cor}
\begin{proof}
	It is clear from the definitions of $OMIN(M_n)$ and $OMAX(M_n)$ that
	\begin{align*}
		\cl{CP}(OMIN(M_n),OMAX(M_n)) & \subseteq \cl{CP}(V_1(M_n),OMAX(M_n)) \\
		& \subseteq \cl{CP}(V_1(M_n),V_2(M_n)) \\
		& \subseteq \cl{CP}(V_1(M_n),OMIN(M_n)) \\
		& \subseteq \cl{CP}(OMAX(M_n),OMIN(M_n)).
	\end{align*}
	The result then follows from statements (e) and (f) of Corollary~\ref{cor:kEntangleBreak}.
\end{proof}

\subsection{Norms on Operator Systems}\label{sec:OpSysNorms}

Given an operator system defined by cones $\big\{C_m\big\}_{m=1}^\infty$, the \emph{matrix norm induced by the matrix order} $\big\{C_m\big\}_{m=1}^\infty$ is defined for $X \in M_m \otimes M_n$ to be
\begin{align}\label{eq:matrixNorm}\hsp
	\big\|X\big\|_{m} := \inf\left\{ r : \begin{bmatrix}r I & X \\ X^\dagger & r I \end{bmatrix} \in C_{2m} \right\}.
\dsp\end{align},
where we have identified $(M_m \otimes M_n) \oplus (M_m \otimes M_n)$ with $M_{2m} \otimes M_n$ in the natural way. In the particular case when the operator system under consideration is either $OMIN^k(M_n)$ or $OMAX^k(M_n)$, we denote the norm~\eqref{eq:matrixNorm} by $\big\|X\big\|_{OMIN^k_m(M_n)}$ or $\big\|X\big\|_{OMAX^k_m(M_n)}$, respectively. Our first result characterizes $\big\|X\big\|_{OMIN^k_m(M_n)}$ in terms of the Schmidt rank of pure states, much like Theorem~\ref{thm:MINkChar} characterized $\big\|X\big\|_{MIN^k_m(M_n)}$.
\begin{thm}\label{thm:k_order_norm}
	Let $X \in M_m \otimes M_n$. Then
	\begin{align*}
		\big\|X\big\|_{OMIN^k_m(M_n)} = \sup_{\ket{v},\ket{w}}\Big\{ \big| \bra{v} X \ket{w} \big| : & \ SR(\ket{v}),SR(\ket{w}) \leq k \text{ and} \\
			& \ \exists \, P \in M_m \text{ s.t. } (P \otimes I_n)\ket{v} = \ket{w} \Big\}.
	\end{align*}
\end{thm}
\begin{proof}
	Given $X \in M_m \otimes M_n$, consider the operator
	\begin{align*}\hsp
		\tilde{X} := \begin{bmatrix}r I & X \\ X^\dagger & r I \end{bmatrix} \in M_{2m} \otimes M_n.
	\dsp\end{align*}
	Then $\tilde{X} \in C^{\textup{min},k}_{2m}$ if and only if $\bra{v}\tilde{X}\ket{v} \geq 0$ for all $\ket{v} \in \bb{C}^{2m} \otimes \bb{C}^n$ with $SR(\ket{v}) \leq k$. If we multiply on the left and the right by a Schmidt-rank $k$ vector $\ket{v} := \sum_{i=1}^k \beta_i \ket{a_i} \otimes \ket{b_i}$, where $\ket{a_i} = \alpha_{i1}\ket{1} \otimes \ket{a_{i1}} + \alpha_{i2}\ket{2} \otimes \ket{a_{i2}} \in \bb{C}^2 \otimes \bb{C}^m \cong \bb{C}^{2m}$ and $\ket{b_i} \in \bb{C}^n$, we get
	\begin{align*}
		\bra{v}\tilde{X}\ket{v} & = \sum_{i=1}^k r\big(\beta_i^2\alpha_{i1}^2 + \beta_i^2\alpha_{i2}^2\big) + \sum_{ij=1}^k 2 \alpha_{i1}\alpha_{j2}\beta_i \beta_j{\rm Re}\big((\bra{a_{i1}} \otimes \bra{b_i})X(\ket{a_{j2}} \otimes \ket{b_j})\big)	\\
		& = r + \sum_{ij=1}^k 2\alpha_{i1}\alpha_{j2}\beta_i \beta_j{\rm Re}\big((\bra{a_{i1}} \otimes \bra{b_i})X(\ket{a_{j2}} \otimes \ket{b_j})\big) \\
		& = r + 2 c_1 c_2 {\rm Re}\big(\bra{v_1}X\ket{v_2}\big),
	\end{align*}
	where $c_1\ket{v_1} := \sum_{i=1}^k \alpha_{i1} \beta_i\ket{a_{i1}} \otimes \ket{b_i},c_2\ket{v_2} := \sum_{i=1}^k \alpha_{i2} \beta_i\ket{a_{i2}} \otimes \ket{b_i} \in \bb{C}^m \otimes \bb{C}^n$. Notice that the normalization of the Schmidt coefficients tells us that $c_1^2 + c_2^2 = 1$. Also notice that $\ket{v_1}$ and $\ket{v_2}$ can be written in this way using the same vectors $\ket{b_i}$ on the second subsystem if and only if there exists $P \in M_m$ such that $(P \otimes I_n)\ket{v_1} = \ket{v_2}$. Now taking the infimum over $r$ and requiring that the result be non-negative tells us that the quantity we are interested in is
	\begin{align*}
		\big\|X\big\|_{OMIN^k_m(M_n)} & = \sup\Big\{2c_1 c_2{\rm Re}\big(\bra{v_1}X\ket{v_2}\big) : SR(\ket{v_1}),SR(\ket{v_2}) \leq k, c_1^2 + c_2^2 = 1, \\
		& \quad \quad \quad \quad \quad \quad \quad \quad \quad \quad \quad \quad \quad \exists \, P \in M_m \text{ s.t. } (P \otimes I_n)\ket{v_1} = \ket{v_2} \Big\} \\
		& = \sup\Big\{\big|\bra{v_1}X\ket{v_2}\big| : SR(\ket{v_1}),SR(\ket{v_2}) \leq k \text{ and} \\
		& \quad \quad \quad \quad \quad \quad \quad \quad \quad \quad \quad \quad \quad \exists \, P \in M_m \text{ s.t. } (P \otimes I_n)\ket{v_1} = \ket{v_2} \Big\},
	\end{align*}
	where the final equality comes applying a complex phase to $\ket{v_1}$ so that $Re(\bra{v_1}X\ket{v_2}) = |\bra{v_1}X\ket{v_2}|$, and from H\"{o}lder's inequality telling us that the supremum is attained when $c_1 = c_2 = 1/\sqrt{2}$.
\end{proof}

The matrix norm induced by the matrix order is not the only way to define a norm on the various levels of an operator system $V(M_n)$. What is referred to as the \emph{order norm} of $X = X^\dagger \in M_m \otimes M_n$ \cite{PT09} is defined via
\begin{align}\label{eq:orderNorm}
	\big\|X\big\|_{V_m}^{or} := \inf\big\{ r \in \mathbb{R} : rI \pm X \in V_m(M_n) \big\}.
\end{align}
It is not difficult to see that the matrix norm induced by the matrix order~\eqref{eq:matrixNorm} coincides with the order norm~\eqref{eq:orderNorm} whenever $X = X^\dagger$. It was shown in \cite{PT09} how the order norm can be extended (non-uniquely) to a norm on the non-Hermitian elements of $M_m \otimes M_n$. Furthermore, there exists a minimal order norm $\|\cdot\|_{V_m}^{\textup{min}}$ and a maximal order norm $\|\cdot\|_{V_m}^{\textup{max}}$ satisfying $\|\cdot\|_{V_m}^{\textup{min}} \leq \|\cdot\|_{V_m}^{\textup{max}} \leq 2\|\cdot\|_{V_m}^{\textup{min}}$. We will now examine properties of these two norms as well as some other norms (all of which coincide with the order norm on Hermitian elements) on the $k$-super minimal operator system.

Let $X \in M_m \otimes M_n$, where the operator system on $M_n$ we are considering is $OMIN^k(M_n)$. Then we recall the minimal order norm $\|\cdot\|_{OMIN^{\textup{min},k}_m}$ and maximal order norm $\|\cdot\|_{OMIN^{\textup{max},k}_m}$ from
\cite{PT09}:
\begin{align*}
	\big\|X\big\|_{OMIN^{\textup{min},k}_m} & := \sup \Big\{ \big|f(X)\big| : f : OMIN^k_m(M_n) \rightarrow \bb{C} \text{ is positive and } f(I) = 1 \Big\}, \\
	\big\|X\big\|_{OMIN^{\textup{max},k}_m} & := \inf \left\{ \sum_{i=1}^r |\lambda_i| \big\| H_i \big\|_{OMIN^{or,k}_m} : X = \sum_{i=1}^r \lambda_i H_i, H_i = H_i^\dagger, \lambda_i \in \mathbb{C} \ \forall \, i \right\}.
\end{align*}
Our next result shows that the minimal order norm can be thought of in terms of vectors with Schmidt rank no greater than $k$, much like the norms $\|\cdot\|_{MIN^k_m(M_n)}$ and $\|\cdot\|_{OMIN^k_m(M_n)}$ introduced earlier.
\begin{thm}\label{thm:charMinNorm}
	Let $X \in M_m \otimes M_n$. Then
	\begin{align*}
		\big\|X\big\|_{OMIN^{\textup{min},k}_m} & = \sup_{\ket{v}}\Big\{ \big| \bra{v} X \ket{v} \big| : SR(\ket{v}) \leq k \Big\}.
	\end{align*}
\end{thm}
\begin{proof}
	Note that if we define a linear functional $f : M_m \otimes M_n \rightarrow \bb{C}$ by $f(X) = \bra{v} X \ket{v}$ for some fixed $\ket{v}$ with $SR(\ket{v}) \leq k$ then it is clear that $f(X) \geq 0$ whenever $X \in C^{\textup{min},k}_{m}$ (by definition of $k$-block positivity) and $f(I) = 1$. The ``$\geq$'' inequality follows immediately.
	
	To see the other inequality, we show that the given supremum is an order norm. The result then follows from minimality of $\|\cdot\|_{OMIN^{\textup{min},k}_m}$ among order norms. To this end, let $X = X^\dagger \in M_m \otimes M_n$. Then
\begin{align*}
	\big\|X\big\|_{OMIN^{or,k}_m} & = \inf\{ r \in \mathbb{R} : rI \pm X \text{ is $k$-block positive} \} \\
	& = \inf\{ r \in \mathbb{R} : \bra{v}(rI \pm X)\ket{v} \geq 0 \text{ for all } \ket{v} \text{ with } SR(\ket{v}) \leq k \} \\
	& = \inf\{ r \in \mathbb{R} : \big|\bra{v}X\ket{v}\big| \leq r \text{ for all } \ket{v} \text{ with } SR(\ket{v}) \leq k \} \\
	& = \sup_{\ket{v}}\Big\{ \big| \bra{v} X \ket{v} \big| : SR(\ket{v}) \leq k \Big\},
\end{align*}
	which completes the proof.
\end{proof}

The characterization of $\|\cdot\|_{OMIN^{\textup{min},k}_m}$ given by Theorem~\ref{thm:charMinNorm} can be thought of as in the same vein as \cite[Proposition 5.8]{PT09}, where it was shown that for a unital C$^*$-algebra, the minimal norm coincides with the numerical radius. In our setting, $\|\cdot\|_{OMIN^{\textup{min},k}_m}$ can be thought of as a bipartite analogue of the numerical radius, which has been studied in quantum information theory in the $k = 1$ case \cite{GPMSZ10,PGMSCZ11}.

We know in general that the minimal and maximal order norms can differ by at most a factor of two. We now present an example some of these norms and to demonstrate that in fact even $\|\cdot\|_{OMIN^{\textup{min},k}_m}$ and $\|\cdot\|_{OMIN^k_m(M_n)}$ can differ by a factor of two.
\begin{exam}\label{ex:min_op_compare}{\rm
	Consider the rank-$1$ operator $X := \ketbra{x}{y} \in M_n \otimes M_n$, where
	\begin{align*}
		\ket{x} := \frac{1}{\sqrt{n}}\sum_{i=1}^{n} \ket{i} \otimes \ket{i} \quad \text{ and } \quad \ket{y} := \frac{1}{\sqrt{n}}\sum_{i=1}^{n} \ket{i} \otimes \ket{i + 1},
	\end{align*}
	where the $i+1$ is understood in the sense that $n + 1 = 1$. It is easily verified that if $\ket{v} = \sum_{i=1}^k \alpha_i \ket{a_i} \otimes \ket{b_i}$ then
	\begin{align*}
		\big|\bra{v}X\ket{v}\big| & = \frac{1}{n} \bigg|\sum_{r,s=1}^k \sum_{i,j=1}^{n} \alpha_r \alpha_s \braket{a_r}{i} \braket{b_r}{i} \braket{j}{a_s} \braket{j + 1}{b_s}\bigg| \\
		& = \frac{1}{n} \bigg| \Tr\Big(\sum_{r=1}^k \alpha_r \overline{\ket{a_r}}\bra{b_r} \Big) \cdot \sum_{j=1}^{n} \bra{j} \Big(\sum_{r=1}^k \alpha_r \overline{\ket{a_r}}\bra{b_r} \Big) \ket{j + 1} \bigg|.
	\end{align*}
	In the final line above we have the trace of an operator with rank at most $k$, multiplied by the sum of the elements on the superdiagonal of the same operator, subject to the constraint that the Frobenius norm of that operator is equal to $1$. It follows that $\big|\bra{v}X\ket{v}\big| \leq \frac{k}{2n}$ and so $\big\|X\big\|_{OMIN^{\textup{min},k}_m} = \frac{k}{2n}$ (equality can be seen by taking $\ket{v} = \sum_{i=1}^{k} \frac{1}{\sqrt{2k}}\ket{i} \otimes (\ket{i} + \ket{i + 1})$).
	
	To see that $\|X\|_{OMIN^{k}_m}$ is twice as large, consider $\ket{v} = \frac{1}{\sqrt{k}}\sum_{i=1}^{k} \ket{i} \otimes \ket{i}$ and $\ket{w} = \frac{1}{\sqrt{k}}\sum_{i=1}^{k} \ket{i} \otimes \ket{i + 1}$. Then it is easily verified that $\bra{v}X\ket{w} = \frac{k}{n}$. Moreover, if $P \in M_m$ is the cyclic permutation matrix such that $P\ket{i} = \ket{i - 1}$ (in the modular arithmetic sense that $1 - 1 = n$) for all $i$ then $(P \otimes I_n)\ket{v} = \ket{w}$, showing that $\big\|X\big\|_{OMIN^{k}_m} \geq \frac{k}{n}$.
}\end{exam}

\subsection{Contractive Maps as Separability Criteria}\label{sec:contrac_sep_crit}

Recall Theorem~\ref{thm:sch_kpos_maps}, which characterizes Schmidt number in terms of $k$-positive maps. In light of Corollary~\ref{cor:kEntangleBreak}, this means that Schmidt number is characterized by maps that are completely positive from $M_n$ to $OMIN^k(M_n)$. In this section, we show that this characterization can be rephrased entirely in terms of norms: Schmidt number is characterized by maps that are completely contractive in operator system norm from $M_n$ to $OMIN^k(M_n)$. In the $k = 1$ case, our results reduce to those of \cite{HHH06}, which characterize separability via maps that are contractive in the trace norm on Hermitian operators.

We begin by considering an operator system version of the norms introduced in Section~\ref{sec:cb_norm_op_space}. Notice that if we consider the completely bounded norm from $M_r$ to the $k$-super minimal operator \emph{systems} on $M_n$, then a statement that is analogous to Theorem~\ref{thm:mainCB} holds. Its proof can be trivially modified to show that if $\Phi : M_r \rightarrow M_n$ and $1 \leq k \leq n$ then
\begin{align}\begin{split}\label{eq:opSysStab}
	& \sup\Big\{ \big|\bra{v}(id_k \otimes \Phi)(X)\ket{v}\big| : \big\|X\big\| \leq 1, X = X^\dagger \Big\} \\
	= & \sup_{m \geq 1}\Big\{ \big|\bra{v}(id_m \otimes \Phi)(X)\ket{v}\big| : \big\|X\big\| \leq 1, X = X^\dagger, SR(\ket{v}) \leq k \Big\}.
\end{split}
\end{align}

Equation~\eqref{eq:opSysStab} can be thought of as a stabilization result for the completely bounded version of the minimal order norm described by Theorem~\ref{thm:charMinNorm}. We could also have picked one of the other order norms on the $k$-super minimal operator systems to work with, but from now on we will be working exclusively with Hermiticity-preserving maps $\Phi$. By the fact that all of the operator system order norms are equal on Hermitian operators, it follows that these versions of their completely bounded norms are all equal as well.

Before proceeding, we will need to define some more notation. If $\Phi : M_n \rightarrow M_r$ is a linear map, then we define a Hermitian version of the induced trace norm of $\Phi$:
\begin{align}\label{eq:herm_tr_norm}
	\big\|\Phi\big\|_{tr}^H := \sup\Big\{ \big\| \Phi(X) \big\|_{tr} : \big\|X\big\|_{tr} \leq 1, X = X^\dagger \Big\}.
\end{align}

It is worth noting that the norm~\eqref{eq:herm_tr_norm} (and many other norms like it) were studied in depth in \cite{Wat05}. In particular, it is worth noting that removing the requirement that $X = X^\dagger$ above in general results in a different norm, even if $\Phi$ is Hermiticity-preserving. Nonetheless, because of convexity of the trace norm it is clear that the norm~\eqref{eq:herm_tr_norm} is unchanged in this case if instead of being restricted to Hermitian operators, the supremum is restricted to positive operators or even just projections. Now by taking the dual of the left and right norms described by Equation~\eqref{eq:opSysStab}, and using the fact that the operator norm is dual to the trace norm, we arrive at the following corollary:

\begin{cor}\label{cor:mainCB}
	Let $\Phi : M_n \rightarrow M_r$ be a Hermiticity-preserving linear map and let $1 \leq k \leq n$. Then
	\begin{align*}
		\big\|id_k \otimes \Phi\big\|^{H}_{tr} = \sup_{m \geq 1}\Big\{ \big\|(id_m \otimes \Phi)(\rho)\big\|_{tr} : \rho \in M_m \otimes M_n \text{ with } SN(\rho) \leq k \Big\}.
	\end{align*}
\end{cor}

We will now characterize the Schmidt number of a state $\rho$ in terms of maps that are contractive in the norm described by Corollary~\ref{cor:mainCB}. We begin with a simple lemma that will get us most of the way to the linear contraction characterization of Schmidt number. The $k = 1$ version of this lemma appeared as \cite[Lemma 1]{HHH06}, though our proof is more straightforward.

\begin{lemma}\label{lem:contrac}
	Let $\rho \in M_m \otimes M_n$ be a density operator. Then $SN(\rho) \leq k$ if and only if $(id_m \otimes \Phi)(\rho) \geq 0$ for all trace-preserving $k$-positive maps $\Phi : M_n \rightarrow M_{2n}$.
\end{lemma}
\begin{proof}
	The ``only if'' implication of the proof is clear, so we only need to establish that if $SN(\rho) > k$ then there is a trace-preserving $k$-positive map $\Phi : M_n \rightarrow M_{2n}$ such that $(id_m \otimes \Phi)(\rho) \not\geq 0$. To this end, let $\Psi : M_n \rightarrow M_n$ be a $k$-positive map such that $(id_m \otimes \Psi)(\rho) \not\geq 0$. Without loss of generality, $\Psi$ can be scaled so that $\big\|\Psi\big\|_{tr} \leq \frac{1}{n}$. Then if $\Omega : M_n \rightarrow M_n$ is the completely depolarizing channel defined by $\Omega(\rho) = \frac{1}{n}I_n$ for all $\rho \in M_n$, it follows that $(\Omega - \Psi)(\rho) \geq 0$ for all $\rho \geq 0$ and so the map $\Phi := \Psi \oplus (\Omega - \Psi) : M_n \rightarrow M_{2n}$ is $k$-positive (and easily seen to be trace-preserving). Because $(id_m \otimes \Psi)(\rho) \not\geq 0$, we have $(id_m \otimes \Phi)(\rho) \not\geq 0$ as well, completing the proof.
\end{proof}

We are now in a position to prove the main result of this section. Note that in the $k = 1$ case of the following theorem it is not necessary to restrict attention to Hermiticity-preserving linear maps $\Phi$ (and indeed this restriction was not made in \cite{HHH06}), but our proof for arbitrary $k$ does make use of Hermiticity-preservation.
\begin{thm}
	Let $\rho \in M_m \otimes M_n$ be a density operator. Then $SN(\rho) \leq k$ if and only if $\big\|(id_m \otimes \Phi)(\rho)\big\|_{tr} \leq 1$ for all Hermiticity-preserving linear maps $\Phi : M_n \rightarrow M_{2n}$ with $\big\| id_k \otimes \Phi \big\|_{tr}^{H} \leq 1$.
\end{thm}
\begin{proof}
	To see the ``only if'' implication, simply use Corollary~\ref{cor:mainCB} with $r = 2n$.
	
	For the ``if'' implication, observe that any positive trace-preserving map $\Psi$ is necessarily Hermiticity-preserving and has $\big\|\Psi\big\|_{tr}^H \leq 1$. Letting $\Psi = id_k \otimes \Phi$ then shows that any $k$-positive trace-preserving map $\Phi$ has $\big\|id_k \otimes \Phi\big\|_{tr}^H \leq 1$. Thus the set of Hermiticity-preserving linear maps $\Phi$ with $\big\|id_k \otimes \Phi\big\|_{tr}^H \leq 1$ contains the set of $k$-positive trace-preserving maps, so the ``if'' implication follows from Lemma~\ref{lem:contrac}.
\end{proof}

\subsection{Right CP-Invariant Cones as Operator Systems}\label{sec:opSysRight}

In this section we establish a tight link between right CP-invariant cones and operator systems. It is not difficult to verify that if $V(M_n)$ is any operator system, then $CP(M_n,V(M_n))$ is a right CP-invariant cone. Similarly, $CP(V(M_n),M_n)$ is easily seen to be a closed left CP-invariant cone. The main result of this section shows that these properties actually characterize the possible cones of completely positive maps to and from $M_n$. We begin with two simple lemmas.
\begin{lemma}\label{lem:right_cp_choi}
	Let $\Phi : M_m \rightarrow M_n$ and $A \in M_{r,m}$. Then
	\begin{align*}
		({\rm Ad}_A \otimes id_n)(C_{\Phi}) = C_{\Phi \circ {\rm Ad}_{A^T}}.
	\end{align*}
\end{lemma}
\begin{proof}
	The proof follows from simple algebraic manipulations:
	\begin{align*}
		({\rm Ad}_A \otimes id_n)(C_{\Phi}) & = m({\rm Ad}_A \otimes \Phi)(\ketbra{\psi_+}{\psi_+}) \\
		& = m(id_m \otimes \Phi)(S C_{{\rm Ad}_A} S^\dagger) \\
		& = m(id_m \otimes \Phi)(\overline{C_{{\rm Ad}_A}^\dagger}) \\
		& = C_{\Phi \circ {\rm Ad}_{A^T}},
	\end{align*}
	where $S$ is the swap operator and the third equality follows from Proposition~\ref{prop:left_right_choi}.
\end{proof}

For the following lemma, we use $S_k$ to denote the cone of (unnormalized) states $\rho$ with $SN(\rho) \leq k$ and $P_k$ to denote the cone of $k$-block positive operators.
\begin{lemma}\label{lem:cn_to_os}
	Let $C_k \subseteq M_k \otimes M_n$ be a cone such that $S_k \subseteq C_k \subseteq P_k$ and $({\rm Ad}_A \otimes id_n)(C_k) \subseteq C_k$ for all $A \in M_k$. Then there exists a family of cones $\{C_m\}_{m \neq k}$ such that $\{C_m\}_{m=1}^{\infty}$ defines an operator system on $M_k$, given by
	\begin{align*}
		C_m := \big\{ \sum_i({\rm Ad}_{A_i} \otimes id_n)(X_i) : A_i \in M_{m,k}, X_i \in C_k, \forall \, i \big\}.
	\end{align*}
	Furthermore, the cones $\{C_m\}$ are uniquely determined when $m \leq k$.
\end{lemma}
\begin{proof}
	We first prove that the family of cones given by the proposition do indeed define an operator system. We first show that $({\rm Ad}_B \otimes id_n)(Y) \in C_{m_2}$ for any $m_1,m_2 \in \bb{N}$, $Y \in C_{m_1}$, and $B \in M_{m_2,m_1}$. This is true from the definition of $C_m$ if $m_1 = k$. If $m_1 \neq n$ then write $Y = \sum_i({\rm Ad}_{A_i} \otimes id_n)(X_i)$ for some $\{X_i\} \subset C_k$ and $\{A_i\} \subset M_{m_1,k}$. Then $BA_i \in M_{m_2,k}$ for all $i$, so
	\begin{align*}
		({\rm Ad}_B \otimes id_n)(Y) = \sum_i({\rm Ad}_{B A_i} \otimes id_n)(X_i) \in C_{m_2}.
	\end{align*}
	
	We now show that $C_1 = M_n^{+}$. For any $\ket{v} \in \bb{C}^n$, note that $\ketbra{v}{v} \otimes X \in S_n$ if and only if $X \in M_n^{+}$, and similarly $\ketbra{v}{v} \otimes X \in P_n$ if and only if $X \in M_n^{+}$. It follows that $\ketbra{v}{v} \otimes X \in C_n$ if and only if $X \in M_n^{+}$. Then $C_1 \supseteq \{ ({\rm Ad}_{A} \otimes id_n)(\ketbra{v}{v} \otimes X) : A \in M_{1,n}, X \in M_n^{+} \big\} = M_n^{+}$, where we have identified $\bb{R}_+ \otimes M_n^+$ with $M_n^+$. The opposite inclusion follows simply from noting that if $X \in C_1$ and $\ket{v} \in \bb{C}^n$ then $\ketbra{v}{v} \otimes X \in C_n$, so $X \in M_n^+$. It follows that $C_1 \subseteq M_n^{+}$, so $C_1 = M_n^{+}$, so the cones $\{C_m\}_{m=1}^{\infty}$ define an operator system on $M_n$.
	
	To prove uniqueness of the cones $C_m$ when $m \geq k$, assume that there exists another family of cones $\{D_m\}_{m=1}^{\infty}$ that define an operator system such that $D_k = C_k$. It is clear that $C_m \subseteq D_m$ for all $m \in \bb{N}$, so we only need to prove the other inclusion. Fix $m \leq k$, let $X \in D_m$, and let $V : \bb{C}^m \rightarrow \bb{C}^k$ be an isometry (i.e., $V^{\dagger}V = I$). Then $Y := ({\rm Ad}_{V} \otimes id_n)(X) \in D_k = C_k$, so $X = ({\rm Ad}_{V^\dagger} \otimes id_n)(Y) \in C_m$. Thus $D_m \subseteq C_m$, so $D_m = C_m$ for $m \leq k$.
\end{proof}

Note that the operator system constructed in Lemma~\ref{lem:cn_to_os} is $OMAX^k(V)$, where $V$ is any operator system on $M_n$ whose $k$-th cone is $C_k$. For convenience, we denote this operator system simply by $OMAX^k(C_k)$. Similarly, we denote $OMIN^k(V)$ by $OMIN^k(C_k)$, and we note that the uniqueness property of Lemma~\ref{lem:cn_to_os} ensures that this notation is well-defined. Before stating our main result, we recall that $\cl{P}(M_n)$ denotes the cone of positive maps on $M_n$, $\cl{S}(M_n)$ denotes the cone of superpositive maps on $M_n$, and $C_{\cl{C}}$ denotes the cone of Choi matrices of maps from the cone $\cl{C}$.
\begin{thm}\label{thm:right_cp_invariant}
	Let $\cl{C} \subseteq \cl{L}(M_n)$ be a convex cone. The following are equivalent:
	\begin{enumerate}[(a)]
		\item the cone $\cl{C}$ is right CP-invariant with $\cl{S}(M_n) \subseteq \cl{C} \subseteq \cl{P}(M_n)$;
		\item there exists an operator system $V_1(M_n)$, defined by cones $\{C_m\}_{m=1}^{\infty}$, such that $C_{\cl{C}} = C_n$;
		\item there exists an operator system $V_2(M_n)$ such that $\cl{C} = \cl{CP}(M_n,V_2(M_n))$; and
		\item there exists an operator system $V_3(M_n)$ such that $(\cl{C}^{\circ})^{\dagger} = \cl{CP}(V_3(M_n),M_n)$.
	\end{enumerate}
	Furthermore, we can choose $V_1(M_n) = V_2(M_n) = OMIN^n(C_{\cl{C}})$ and $V_3(M_n) = OMAX^n(C_{\cl{C}})$.
\end{thm}
\begin{proof}
	We prove the result by showing that $(a) \Leftrightarrow (b)$, $(b) \Rightarrow (c)$, $(c) \Rightarrow (a)$, $(b) \Rightarrow (d)$, and $(d) \Rightarrow (a)$.
	
	To see that $(a) \Rightarrow (b)$, define $C_n := C_{\cl{C}}$. If $A \in M_n$ and $\Phi \in \cl{C}$ then Lemma~\ref{lem:right_cp_choi} tells us that
	\begin{align}\label{eq:cn_to_invariant}
		({\rm Ad}_A \otimes id_n)(C_{\Phi}) & = C_{\Phi \circ {\rm Ad}_{A^T}} \in C_n,
	\end{align}
	where the inclusion comes from the fact that $\cl{C}$ is right CP-invariant. The implication $(a) \Rightarrow (b)$ and the fact that we can choose $V_1(M_n) = OMAX^n(C_{\cl{C}})$ then follows from Lemma~\ref{lem:cn_to_os}. The reverse implication $(b) \Rightarrow (a)$ also follows from Equation~\eqref{eq:cn_to_invariant}, but this time we use the fact that $C_n$ is a cone defining an operator system to get the inclusion. The fact that $\cl{S}(M_n) \subseteq \cl{C} \subseteq \cl{P}(M_n)$ follows from the fact that for the minimal operator system on $M_n$, $C_n$ is the cone of block positive operators and for the maximal operator system on $M_n$, $C_n$ is the cone of separable operators (see Theorem~\ref{thm:opSysConeChar}).
	
	To see that $(b) \Rightarrow (c)$, let $V_2(M_n) = OMAX^n(C_{\cl{C}})$. We then have to show that if $C_{\cl{C}} = C_n$, then $\cl{C} = \cl{CP}(M_n,V_2(M_n))$. We already showed that $(b) \Rightarrow (a)$, so we know that $\cl{C}$ is right CP-invariant. If $\Phi \in \cl{C}$ then for any $X \in (M_n \otimes M_n)^{+}$ there exists $\Psi \in \cl{CP}$ such that
	\begin{align*}
		(id_n \otimes \Phi)(X) = C_{\Phi \circ \Psi} \in C_n,
	\end{align*}
	where the inclusion comes from $\cl{C}$ being right CP-invariant. It follows that via \cite[Proposition~2.3.3]{Xthesis} that $\Phi \in \cl{CP}(M_n,V_2(M_n))$, so $\cl{C} \subseteq \cl{CP}(M_n,V_2(M_n))$. To see the opposite inclusion, simply note that if $\Phi \in \cl{CP}(M_n,V_2(M_n))$ then, because $\ketbra{\psi_+}{\psi_+} \in (M_n \otimes M_n)^{+}$, we have $C_{\Phi} = m(id_n \otimes \Phi)(\ketbra{\psi_+}{\psi_+}) \in C_n = C_{\cl{C}}$, so $\Phi \in \cl{C}$. It follows that $\cl{C} = \cl{CP}(M_n,V_2(M_n))$.
	
	To prove $(c) \Rightarrow (a)$, simply note that $\cl{CP}(M_n,V_2(M_n))$ is trivially right CP-invariant. To see that $\cl{S}(M_n) \subseteq \cl{CP}(M_n,V_2(M_n)) \subseteq \cl{P}(M_n)$, we simply use Corollary~\ref{cor:cp_sub_p}.
	
	The proof that $(b) \Rightarrow (d)$ mimics the proof that $(b) \Rightarrow (c)$. Let $V_3(M_n) = OMAX^n(C_{\cl{C}})$. Then for any $\Psi \in \cl{C}^{\circ}$ and $\Phi \in \cl{C}$ we have $\Psi^{\dagger} \circ \Phi \in \cl{CP}$ (Proposition~\ref{prop:right_cp_main}), so $C_{\Psi^{\dagger} \circ \Phi} \in (M_n \otimes M_n)^{+}$. It follows that $(id_n \otimes \Psi^{\dagger})(C_n) \subseteq (M_n \otimes M_n)^{+}$. \cite[Proposition~2.3.7]{Xthesis} implies that $\Psi^{\dagger} \in \cl{CP}(V_3(M_n),M_n)$, so $(\cl{C}^{\circ})^{\dagger} \subseteq \cl{CP}(V_3(M_n),M_n)$. The opposite inclusion follows by simply reversing this argument.
	
	The implication $(d) \Rightarrow (a)$ follows similarly by the fact that $\cl{CP}(V_3(M_n),M_n)$ is trivially closed and left CP-invariant. To see that $\cl{S}(M_n) \subseteq (\cl{CP}(V_3(M_n),M_n)^\dagger)^\circ \subseteq \cl{P}(M_n)$, we again use Corollary~\ref{cor:cp_sub_p}.
\end{proof}

As a demonstration of Theorem~\ref{thm:right_cp_invariant}, we now recall a right CP-invariant cone that we have not yet considered in this chapter -- the cone of anti-degradable maps. In particular, we have the following result, which shows that the anti-degradable maps are exactly the completely positive maps into the operator system formed by the shareable operators.
\begin{thm}\label{thm:antidegrad_opsys}
	Let $H_m^2$ denote the cone of shareable operators in $M_m \otimes M_n$. Then the family of cones $\{H_m^2\}_{m=1}^\infty$ defines an operator system $V(M_n)$ such that $\cl{CP}(M_n,V(M_n)) = \cl{AD}$, the cone of anti-degradable maps.
\end{thm}
\begin{proof}
	We first show that the family of cones $\{H_m^2\}_{m=1}^\infty$ satisfies the two defining properties of operator systems on $M_n$. The cone $H_1^2$ of shareable operators in $M_1 \otimes M_n \cong M_n$ indeed satisfies $H_1^2 = M_n^+$ because if $X \in M_n^+$ is any positive semidefinite operator then $X \otimes X \in M_n \otimes M_n$ is a symmetric extension of it. To see that $({\rm Ad}_A \otimes id_n)(H_{m_1}^2) \subseteq H_{m_2}^2$ for all $m_1,m_2 \in \mathbb{N}$ and $A \in M_{m_2,m_1}$, simply note that if $X \in H_{M_1}^2$ is extended by the operator $\tilde{X} \in M_{m_1} \otimes (M_n \otimes M_n)$, then $({\rm Ad}_A \otimes id_n)(X)$ is extended by $({\rm Ad}_A \otimes id_n \otimes id_n)(\tilde{X})$. It follows that $\{H_m^2\}_{m=1}^\infty$ defines an operator system, which we denote $V(M_n)$.
	
	We now show that $\Phi \in \cl{AD}$ if and only if $\Phi$ is completely positive from $M_n$ to $V(M_n)$. If $\Phi \in \cl{CP}(M_n,V(M_n))$ then in particular $C_\Phi = m(id_n \otimes \Phi)(\ketbra{\psi_+}{\psi_+}) \in H_n^2$. We recall from Section~\ref{sec:choi_jamiolkowski_separable} that this implies $\Phi \in \cl{AD}$, so $\cl{CP}(M_n,V(M_n)) \subseteq \cl{AD}$. To see the opposite inclusion, suppose $\Phi \in \cl{AD}$. Theorem~\ref{thm:antidegrad_2broadcast} says that $\Phi$ is $2$-extendible, so there exists a map $\tilde{\Phi} : M_n \rightarrow (M_n \otimes M_n)$ such that $\Tr_1 \circ \tilde{\Phi} = \Tr_2 \circ \tilde{\Phi} = \Phi$. Then, for all $m \geq 1$ and all $X \in (M_m \otimes M_n)^+$ we have
	\begin{align*}
		(id_m \otimes \Phi)(X) = (id_m \otimes (\Tr_1 \circ \tilde{\Phi}))(X) = (id_m \otimes (\Tr_2 \circ \tilde{\Phi}))(X).
	\end{align*}
	It follows that $(id_m \otimes \Phi)(X)$ is shareable (indeed, it is extended by $(id_m \otimes \tilde{\Phi})(X)$), so $\Phi \in \cl{CP}(M_n,V(M_n))$.
\end{proof}

By recalling that the shareable operators and anti-degradable maps are naturally generalized by the $s$-shareable operators and $s$-extendible maps respectively, the following generalization of Theorem~\ref{thm:antidegrad_opsys} becomes clear (and hence we present it without proof).
\begin{thm}\label{thm:kextend_opsys}
	Let $H_m^s$ denote the cone of $s$-shareable operators in $M_m \otimes M_n$. Then the family of cones $\{H_m^s\}_{m=1}^\infty$ defines an operator system $V_s(M_n)$ such that $\cl{CP}(M_n,V_s(M_n)) = \cl{B}_s$, the cone of $s$-extendible maps.
\end{thm}

Because the cones of anti-degradable maps and $s$-extendible maps are not left CP-invariant, the operator systems of Theorems~\ref{thm:antidegrad_opsys} and~\ref{thm:kextend_opsys} do not fit into the framework of the next section.

\subsection{Mapping Cones as Operator Systems}\label{sec:map_cones_op_sys}

From now on, it will often be useful for us to consider operator systems $V(M_n)$ with the additional property that $(id_m \otimes {\rm Ad}_{B})(C_{m}) \subseteq C_{m}$ for each $m \in \bb{N}$ and $B \in M_{n}$ -- a property that is equivalent to the fact $\cl{CP}(M_n) \subseteq \cl{CP}(V(M_n))$. We call operator systems with this property \emph{super-homogeneous}.

The following result shows how Theorem~\ref{thm:right_cp_invariant} works when the right CP-invariant cone is in fact a mapping cone -- in this situation the associated operator system is super-homogeneous.
\begin{cor}\label{cor:mapping_cone_os}
	Let $\cl{C} \subseteq \cl{L}(M_n)$ be a closed, convex cone. The following are equivalent:
	\begin{enumerate}
		\item $\cl{C}$ is a mapping cone;
		\item there exists a super-homogeneous operator system $V_1(M_n)$, defined by cones\\
		$\{C_m\}_{m=1}^{\infty}$, such that $C_{\cl{C}} = C_n$;
		\item there exists a super-homogeneous operator system $V_2(M_n)$ such that\\
		$\cl{C} = \cl{CP}(M_n,V_2(M_n))$;
		\item there exists a super-homogeneous operator system $V_3(M_n)$ such that\\
		$(\cl{C}^{\circ})^{\dagger} = \cl{CP}(V_3(M_n),M_n)$; and
		\item there exist super-homogeneous operator systems $V_4(M_n)$ and $V_5(M_n)$ such that \\
		$\cl{C} = \cl{CP}(V_4(M_n),V_5(M_n))$.
	\end{enumerate}
	Furthermore, we can choose $V_1(M_n) = V_2(M_n) = OMIN^n(C_{\cl{C}})$ and $V_3(M_n) = OMAX^n(C_{\cl{C}})$.
\end{cor}
\begin{proof}
	The equivalence of $(a)$, $(b)$, $(c)$, and $(d)$ follows immediately from the corresponding statements of Theorem~\ref{thm:right_cp_invariant} and the fact that $\cl{C}$ is left CP-invariant if and only if $(id_n \otimes {\rm Ad}_B)(C_{\cl{C}}) \subseteq C_{\cl{C}}$, which then gives super-homogeneity of the corresponding operator system.
	
	Because $M_n$ is a super-homogeneous operator system, it is clear that $(c) \Rightarrow (e)$. All that remains to do is prove that $(e) \Rightarrow (a)$. To this end, simply notice that right CP-invariance of $\cl{CP}(V_4(M_n),V_5(M_n))$ follows from super-homogeneity of $V_4(M_n)$ and left CP-invariance of $\cl{CP}(V_4(M_n),V_5(M_n))$ follows from super-homogeneity of $V_5(M_n)$. The fact that $\cl{CP}(V_4(M_n),V_5(M_n)) \subseteq \cl{P}(M_n)$ and is nonzero follows from Corollary~\ref{cor:cp_sub_p}.
\end{proof}

It is natural at this point to consider well-known mapping cones and ask what operator systems give rise to them in the sense of Corollary~\ref{cor:mapping_cone_os}. The mapping cone of standard completely positive maps $\cl{CP}(M_n)$ appears when we let $V_1(M_n) = V_2(M_n) = M_n$ itself. Many other cases of interest come from Corollary~\ref{cor:kEntangleBreak}: if $\cl{C}$ is the mapping cone of $k$-positive maps, we can choose $V_1(M_n) = V_2(M_n) = OMIN^k(M_n)$, and if $\cl{C}$ is the mapping cone of $k$-superpositive maps, we can choose $V_1(M_n) = V_2(M_n) = OMAX^k(M_n)$. Finally, consider the mapping cone of completely co-positive maps $\{\Phi \circ T : \Phi \in \cl{CP}(M_n)\}$. It is not difficult to see that in this case we can choose $V_1(M_n) = V_2(M_n)$ to be the operator system defined by the cones of operators with positive partial transpose -- i.e., the operators $X \in M_m \otimes M_n$ such that $X^\Gamma \geq 0$.

\subsection{Semigroup Cones as Operator Systems}\label{sec:semigroups}

Theorem~\ref{thm:right_cp_invariant} and Corollary~\ref{cor:mapping_cone_os} provide characterizations of completely positive maps to and from $M_n$, and completely positive maps between two different super-homogeneous operator systems on $M_n$. However, they say nothing about completely positive maps from a super-homogeneous operator system back into itself. Toward deriving a characterization for this situation, we consider cones $\cl{C} \subseteq \cl{L}(M_n)$ that are \emph{semigroups} -- i.e., cones such that $\Phi \circ \Psi \in \cl{C}$ for all $\Phi,\Psi \in \cl{C}$. Notice that many of the standard examples of mapping cones, such as the $k$-positive maps and the $k$-superpositive maps, are semigroups (however, the cone of completely co-positive maps is not).

If $V(M_n)$ is an operator system defined by cones $\{C_m\}_{m=1}^{\infty}$, then the dual cones $\{C_m^{\circ}\}_{m=1}^{\infty}$ define an operator system as well, which we denote $V^{\circ}(M_n)$. For simplicity, we will only consider this operator system as a family of dual cones, in keeping with our focus throughout the preceding portion of this work, and not the associated dual operator space structure. The interested reader is directed to \cite{BM11} for a more thorough treatment of dual operator systems. It is easily verified that $V(M_n)$ is super-homogeneous if and only if $V^{\circ}(M_n)$ is super-homogeneous, and the ``naive'' operator system on $M_n$ is easily seen to be self-dual: $M_n^{\circ} = M_n$. By the duality of the cones of $k$-positive maps and $k$-superpositive maps we know that $OMIN_{k}^{\circ}(M_n) = OMAX_k(M_n)$ and $OMAX_{k}^{\circ}(M_n) = OMIN_k(M_n)$.

We now consider what types of cones can be completely positive from a super-homogeneous operator system back into itself. We already saw in Corollary~\ref{cor:kEntangleBreak} that $\cl{CP}(OMIN_k(M_n)) = \cl{CP}(OMAX_k(M_n)) = \cl{P}_k(M_n)$ -- a fact that we now see is related to the facts that $\cl{P}_k(M_n)$ is a semigroup and $OMIN_{k}^{\circ}(M_n) = OMAX_k(M_n)$.
\begin{thm}\label{thm:cp_semigroup}
	Let $\cl{C} \subseteq \cl{L}(M_n)$ be a convex cone. The following are equivalent:
	\begin{enumerate}[(a)]
		\item the cone $\cl{C}$ is a semigroup with $\cl{CP}(M_n) \subseteq \cl{C} \subseteq \cl{P}(M_n)$; and
		\item there exists a super-homogeneous operator system $V(M_n)$ with $\cl{C} = \cl{CP}(V(M_n))$.
	\end{enumerate}
	Furthermore, we can choose $V(M_n) = OMIN^n(C_{\cl{C}})$.
\end{thm}
\begin{proof}
	We first prove that $(b) \Rightarrow (a)$. Let $\{C_m\}_{m=1}^{\infty}$ be the cones associated with the operator system $V(M_n)$. If $X \in C_m$ and $\Phi,\Psi \in \cl{CP}(V(M_n))$ then $(id_m \otimes \Phi)(X) \in C_m$. But then applying $id_m \otimes \Psi$ shows $(id_m \otimes (\Psi \circ \Phi))(X) \in C_m$ as well, so it follows that $\Psi \circ \Phi \in \cl{CP}(V(M_n))$ and thus $\cl{CP}(V(M_n))$ is a semigroup. Because $V(M_n)$ is super-homogeneous, we know that ${\rm Ad}_B \in \cl{CP}(V(M_n))$ for all $B \in M_n$, and so $\cl{CP}(M_n) \subseteq \cl{CP}(V(M_n))$. To see that $\cl{CP}(V(M_n)) \subseteq \cl{P}(M_n)$, simply use Corollary~\ref{cor:cp_sub_p}.

	To see that $(a) \Rightarrow (b)$, we argue much as we did in Theorem~\ref{thm:right_cp_invariant}. It is clear, via the Choi--Jamio{\l}kowski isomorphism, that $S_n \subseteq C_{\cl{C}} \subseteq P_n$. Now note that $\cl{C}$ is left and right CP-invariant because $\Phi \circ \Psi \in \cl{C}$ for any $\Phi \in \cl{C}$ and $\Psi \in \cl{CP}(M_n) \subseteq \cl{C}$ (and similarly for composition on the left by $\Psi \in \cl{CP}(M_n)$). Thus, if $A \in M_m$, $B \in M_n$ and $\Phi \in \cl{C}$ then
	\begin{align*}
		({\rm Ad_{A} \otimes {\rm Ad}_{B}})(C_{\Phi}) & = (id_m \otimes {\rm Ad}_{B})(C_{\Phi \circ {\rm Ad}_{A^T}}) = C_{{\rm Ad}_{B} \circ \Phi \circ {\rm Ad}_{A^T}} \in C_{\cl{C}},
	\end{align*}
	where the first equality comes from Lemma~\ref{lem:right_cp_choi}. Lemma~\ref{lem:cn_to_os} then tells us that $V(M_n) = OMIN^n(C_{\cl{C}})$ is an operator system, and it is easily seen to be super-homogeneous. Because $\cl{C}$ is a semigroup, it follows that $C_{\Phi \circ \Psi} \in C_\cl{C}$ for any $\Phi,\Psi \in \cl{C}$. Then $(id_n \otimes \Phi)(C_\Psi) \in C_\cl{C}$, so $(id_n \otimes \Phi)(C_\cl{C}) \subseteq C_\cl{C}$, which implies $\cl{C} \subseteq \cl{CP}(V(M_n))$ by \cite[Proposition~2.3.3]{Xthesis}. To see the other inclusion, note that $id_n \in \cl{CP}(M_n)$, so $id_n \in \cl{C}$. It follows that $\ketbra{\psi_+}{\psi_+} \in C_\cl{C}$. Thus, if $\Phi \in \cl{CP}(V(M_n))$ then $(id_n \otimes \Phi)(\ketbra{\psi_+}{\psi_+}) \in C_{\cl{C}}$, so $\Phi \in \cl{C}$, which implies that $\cl{C} = \cl{CP}(V(M_n))$.
\end{proof}

It is worth noting that if $\cl{C}$ is closed and condition (a) of Theorem~\ref{thm:cp_semigroup} holds, then $\cl{C}$ is necessarily a mapping cone. It follows that if $V(M_n)$ is a super-homogeneous operator system defined by closed cones then $\cl{CP}(V(M_n))$ is always a mapping cone (which can also be seen from Corollary~\ref{cor:mapping_cone_os}), although the converse does not hold. That is, there exist mapping cones $\cl{C}$ such that there is no operator system $V(M_n)$ with $\cl{C} = \cl{CP}(V(M_n))$ -- the simplest example being the mapping cone of completely co-positive maps.

\clearpage
\addcontentsline{toc}{chapter}{Bibliography}
\bibliographystyle{alpha}

\begin{thebibliography}{SZWGC09}

\bibitem[AB04]{AB04}
K.~M.~R. Audenaert and S.~L. Braunstein.
\newblock On strong superadditivity of the entanglement of formation.
\newblock {\em Comm. Math. Phys.}, 246:443--452, 2004.

\bibitem[ABLS01]{ABLS01}
A.~Ac\'{i}n, D.~Bru\ss, M.~Lewenstein, and A.~Sanpera.
\newblock Classification of mixed three-qubit states.
\newblock {\em Phys. Rev. Lett.}, 87:040401, 2001.

\bibitem[AH02]{AH03}
G.~G. Amosov and A.~S. Holevo.
\newblock On the multiplicativity hypothesis for quantum communication
  channels.
\newblock {\em Theory Probab. Appl.}, 47:123--127, 2002.

\bibitem[AHW00]{AHW00}
G.~G. Amosov, A.~S. Holevo, and R.~F. Werner.
\newblock On some additivity problems in quantum information theory.
\newblock {\em Probl. Inf. Transm.}, 36:25--34, 2000.

\bibitem[AKMS05]{AKMS05}
M.~Asorey, A.~Kossakowski, G.~Marmo, and E.~C.~G. Sudarshan.
\newblock Relations between quantum maps and quantum states.
\newblock {\em Open Syst. Inf. Dyn.}, 12:319--329, 2005.

\bibitem[Ali95]{Al95}
F.~Alizadeh.
\newblock Interior point methods in semidefinite programming with applications
  to combinatorial optimization.
\newblock {\em SIAM J. Optim.}, 5:13--51, 1995.

\bibitem[And04]{And04}
T.~Ando.
\newblock Cones and norms in the tensor product of matrix spaces.
\newblock {\em Linear Algebra Appl.}, 379:3--41, 2004.

\bibitem[AP04]{AP04}
P.~Arrighi and C.~Patricot.
\newblock On quantum operations as quantum states.
\newblock {\em Ann. Physics}, 311:26--52, 2004.

\bibitem[AS72]{AS72}
M.~Abramowitz and I.~A. Stegun.
\newblock {\em Handbook of Mathematical Functions with Formulas, Graphs, and
  Mathematical Tables}.
\newblock Dover Publications, 1972.

\bibitem[AS10]{AS10}
E.~Alfsen and F.~Shultz.
\newblock Unique decompositions, faces, and automorphisms of separable states.
\newblock {\em J. Math. Phys.}, 51:052201, 2010.

\bibitem[ASW10]{ASW10}
G.~Aubrun, S.~Szarek, and E.~Werner.
\newblock Non-additivity of {R}enyi entropy and {D}voretzky's theorem.
\newblock {\em J. Math. Phys.}, 51:022102, 2010.

\bibitem[Aud09]{Aud09}
K.~M.~R. Audenaert.
\newblock On the $p \rightarrow q$ norms of $2$-positive maps.
\newblock {\em Linear Algebra Appl.}, 430:1436--1440, 2009.

\bibitem[BBC{\etalchar{+}}93]{BBCJPW93}
C.~H. Bennett, G.~Brassard, C.~Cr\'{e}peau, R.~Jozsa, A.~Peres, and W.~K.
  Wootters.
\newblock Teleporting an unknown quantum state via dual classical and
  {E}instein-{P}odolsky-{R}osen channels.
\newblock {\em Phys. Rev. Lett.}, 70:1895--1899, 1993.

\bibitem[BBPS96]{BBPS96}
C.~H. Bennett, H.~J. Bernstein, S.~Popescu, and B.~Schumacher.
\newblock Concentrating partial entanglement by local operations.
\newblock {\em Phys. Rev. A}, 53:2046--2052, 1996.

\bibitem[BCF{\etalchar{+}}96]{BCFJS96}
H.~Barnum, C.~M. Caves, C.~A. Fuchs, R.~Jozsa, and B.~Schumacher.
\newblock Noncommuting mixed states cannot be broadcast.
\newblock {\em Phys. Rev. Lett.}, 76:2818--2821, 1996.

\bibitem[BDF{\etalchar{+}}99]{BDFMRSSW99}
C.~H. Bennett, D.~P. DiVincenzo, C.~A. Fuchs, T.~Mor, E.~Rains, P.~W. Shor,
  J.~A. Smolin, and W.~K. Wootters.
\newblock Quantum nonlocality without entanglement.
\newblock {\em Phys. Rev. A}, 59:1070--1091, 1999.

\bibitem[BE08]{BE08}
F.~G.~S.~L. Brand{\~ a}o and J.~Eisert.
\newblock Correlated entanglement distillation and the structure of the set of
  undistillable states.
\newblock {\em J. Math. Phys.}, 49:042102, 2008.

\bibitem[Bea88]{B88}
L.~Beasley.
\newblock Linear operators on matrices: The invariance of rank-$k$ matrices.
\newblock {\em Linear Algebra Appl.}, 107:161--167, 1988.

\bibitem[BFP04]{BFP04}
F.~Benatti, R.~Floreanini, and M.~Piani.
\newblock Non-decomposable quantum dynamical semigroups and bound entangled
  states.
\newblock {\em Open Syst. Inf. Dyn.}, 11:325--338, 2004.

\bibitem[Bha97]{Bha97}
R.~Bhatia.
\newblock {\em Matrix analysis}.
\newblock Springer, 1997.

\bibitem[BL90]{BL90}
L.~Beasley and T.~Laffey.
\newblock Linear operators on matrices: The invariance of rank-$k$ matrices.
\newblock {\em Linear Algebra Appl.}, 133:175--184, 1990.

\bibitem[BL01]{BL01}
H.~Barnum and N.~Linden.
\newblock Monotones and invariants for multi-particle quantum states.
\newblock {\em J. Phys. A: Math. Gen.}, 34:6787--6805, 2001.

\bibitem[BM11]{BM11}
D.~Blecher and B.~Magajna.
\newblock Dual operator systems.
\newblock {\em Bull. Lond. Math. Soc.}, 43:311--320, 2011.

\bibitem[BOS{\etalchar{+}}02]{BOSBJ02}
M.~D. Bowdreya, D.~K.~L. Oia, A.~J. Shorta, K.~Banaszeka, and J.~A. Jones.
\newblock Fidelity of single qubit maps.
\newblock {\em Phys. Lett. A}, 294:258--260, 2002.

\bibitem[BP00]{BP00}
D.~Bru\ss and A.~Peres.
\newblock Construction of quantum states with bound entanglement.
\newblock {\em Phys. Rev. A}, 61:030301(R), 2000.

\bibitem[BR03]{BR03}
S.~Bandyopadhyay and V.~Roychowdhury.
\newblock Classes of $n$-copy undistillable quantum states with negative
  partial transposition.
\newblock {\em Phys. Rev. A}, 68:022319, 2003.

\bibitem[Bra09]{B09}
F.~G.~S.~L. Brand{\~ a}o.
\newblock Quantum hypothesis testing of non-i.i.d. states and its connection to
  reversible resource theories.
\newblock Talk at the Operator Structures in Quantum Information Workshop at
  the Fields Institute, 2009.

\bibitem[Bre06]{Bre06}
H.-P. Breuer.
\newblock Optimal entanglement criterion for mixed quantum states.
\newblock {\em Phys. Rev. Lett.}, 97:080501, 2006.

\bibitem[Bru00]{OpenProb2}
D.~Bru\ss.
\newblock Undistillability implies ppt?
\newblock Published electronically at
  \url{http://qig.itp.uni-hannover.de/qiproblems/2}, 2000.

\bibitem[Bur69]{B69}
D.~Bures.
\newblock An extension of {K}akutani's theorem on infinite product measures to
  the tensor product of semifinite $w^*$-algebras.
\newblock {\em Trans. Amer. Math. Soc.}, 135:199--212, 1969.

\bibitem[BV04]{BV04}
S.~Boyd and L.~Vandenberghe.
\newblock {\em Convex optimization}.
\newblock Cambridge University Press, 2004.

\bibitem[BW92]{BW92}
C.~H. Bennett and S.~J. Wiesner.
\newblock Communication via one- and two-particle operators on
  {E}instein-{P}odolsky-{R}osen states.
\newblock {\em Phys. Rev. Lett.}, 69:2881--2884, 1992.

\bibitem[B{\.Z}06]{BZ06}
I.~Bengtsson and K.~{\.Z}yczkowski.
\newblock {\em Geometry of quantum states: An introduction to quantum
  entanglement}.
\newblock Cambridge University Press, 2006.

\bibitem[CAG99]{CAG99}
N.~J. Cerf, C.~Adami, and R.~M. Gingrich.
\newblock Reduction criterion for separability.
\newblock {\em Phys. Rev. A}, 60:898--909, 1999.

\bibitem[Cav06]{Cav06}
D.~Cavalcanti.
\newblock Connecting the generalized robustness and the geometric measure of
  entanglement.
\newblock {\em Phys. Rev. A}, 73:044302, 2006.

\bibitem[CDKL01]{CDKL01}
J.~I. Cirac, W.~D\"{u}r, B.~Kraus, and M.~Lewenstein.
\newblock Entangling operations and their implementation using a small amount
  of entanglement.
\newblock {\em Phys. Rev. Lett.}, 86:544--547, 2001.

\bibitem[CE77]{CE77}
M.-D. Choi and E.~G. Effros.
\newblock Injectivity and operator spaces.
\newblock {\em J. Funct. Anal.}, 24:156--209, 1977.

\bibitem[Cho72]{Cho72}
M.-D. Choi.
\newblock Positive linear maps on ${C}^*$-algebras.
\newblock {\em Canad. J. Math.}, 24:520--529, 1972.

\bibitem[Cho75a]{C75}
M.-D. Choi.
\newblock Completely positive linear maps on complex matrices.
\newblock {\em Linear Algebra Appl.}, 10:285--290, 1975.

\bibitem[Cho75b]{Cho75}
M.-D. Choi.
\newblock Positive semidefinite biquadratic forms.
\newblock {\em Linear Algebra Appl.}, 12:95--100, 1975.

\bibitem[CJK09]{CJK09}
M.-D. Choi, N.~Johnston, and D.~W. Kribs.
\newblock The multiplicative domain in quantum error correction.
\newblock {\em J. Phys. A: Math. Theor.}, 42:245303, 2009.

\bibitem[CK06]{CK06}
D.~Chru\'{s}ci\'{n}ski and A.~Kossakowski.
\newblock On partially entanglement breaking channels.
\newblock {\em Open Syst. Inf. Dyn.}, 13:17--26, 2006.

\bibitem[CK09]{CK09}
D.~Chru\'{s}ci\'{n}ski and A.~Kossakowski.
\newblock Spectral conditions for positive maps.
\newblock {\em Comm. Math. Phys.}, 290:1051--1064, 2009.

\bibitem[CK11]{CK11}
D.~Chru\'{s}ci\'{n}ski and A.~Kossakowski.
\newblock Spectral conditions for positive maps and entanglement witnesses.
\newblock {\em J. Phys.: Conf. Ser.}, 284:012017, 2011.

\bibitem[CKS09]{CKS09}
D.~Chru\'{s}ci\'{n}ski, A.~Kossakowski, and G.~Sarbicki.
\newblock Spectral conditions for entanglement witnesses versus bound
  entanglement.
\newblock {\em Phys. Rev. A}, 80:042314, 2009.

\bibitem[Cla05]{C05}
L.~Clarisse.
\newblock Characterization of distillability of entanglement in terms of
  positive maps.
\newblock {\em Phys. Rev. A}, 71:032332, 2005.

\bibitem[CMW08]{CW08}
T.~S. Cubitt, A.~Montanaro, and A.~Winter.
\newblock On the dimension of subspaces with bounded {S}chmidt rank.
\newblock {\em J. Math. Phys.}, 49:022107, 2008.

\bibitem[Cor04]{Cor04}
J.~Cortese.
\newblock {H}olevo-{S}chumacher-{W}estmoreland channel capacity for a class of
  qudit unital channels.
\newblock {\em Phys. Rev. A}, 69:022302, 2004.

\bibitem[CRS08]{CRS08}
T.~S. Cubitt, M.~B. Ruskai, and G.~Smith.
\newblock The structure of degradable quantum channels.
\newblock {\em J. Math. Phys.}, 49:102104, 2008.

\bibitem[CS06]{CS06}
I.~Chattopadhyay and D.~Sarkar.
\newblock {NPT} bound entanglement- the problem revisited.
\newblock E-print: arXiv:quant-ph/0609050, 2006.

\bibitem[CW03]{CW03}
K.~Chen and L.-A. Wu.
\newblock A matrix realignment method for recognizing entanglement.
\newblock {\em Quantum Inf. Comput.}, 3:193--202, 2003.

\bibitem[CZW11]{CZW11}
L.~Chen, H.~Zhu, and T.-C. Wei.
\newblock Connections of geometric measure of entanglement of pure symmetric
  states to quantum state estimation.
\newblock {\em Phys. Rev. A}, 83:012305, 2011.

\bibitem[Dat04]{Dat04}
N.~Datta.
\newblock Multiplicativity of maximal $p$-norms in {W}erner-{H}olevo channels
  for $1 \leq p \leq 2$.
\newblock E-print: arXiv:quant-ph/0410063, 2004.

\bibitem[DBE95]{DBE95}
D.~Deutsch, A.~Barenco, and A.~Ekert.
\newblock Universality in quantum computation.
\newblock {\em Proc. R. Soc. Lond. Ser. A}, 449:669--677, 1995.

\bibitem[DCLB00]{DCL00}
W.~D\"{u}r, J.~I. Cirac, M.~Lewenstein, and D.~Bru\ss.
\newblock Distillability and partial transposition in bipartite systems.
\newblock {\em Phys. Rev. A}, 61:062313, 2000.

\bibitem[DFS08]{DGFS08}
J.~Diestel, J.~H. Fourie, and J.~Swart.
\newblock {\em The Metric Theory of Tensor Products: {G}rothendieck's
  R\'{e}sum\'{e} Revisited}.
\newblock American Mathematical Society, 2008.

\bibitem[DJKR06]{DJKR06}
I.~Devetak, M.~Junge, C.~King, and M.~B. Ruskai.
\newblock Multiplicativity of completely bounded $p$-norms implies a new
  additivity result.
\newblock {\em Comm. Math. Phys.}, 266:37--63, 2006.

\bibitem[dK02]{Kl02}
E.~de~Klerk.
\newblock {\em Aspects of semidefinite programming: Interior point algorithms
  and selected applications}.
\newblock Kluwer Academic Publishers, 2002.

\bibitem[dP67]{P67}
J.~de~Pillis.
\newblock Linear transformations which preserve {H}ermitian and positive
  semidefinite operators.
\newblock {\em Pacific J. Math.}, 23:129--137, 1967.

\bibitem[DPS02]{DPS02}
A.~C. Doherty, P.~A. Parrilo, and F.~M. Spedalieri.
\newblock Distinguishing separable and entangled states.
\newblock {\em Phys. Rev. Lett.}, 88:187904, 2002.

\bibitem[DPS04]{DPS04}
A.~C. Doherty, P.~A. Parrilo, and F.~M. Spedalieri.
\newblock A complete family of separability criteria.
\newblock {\em Phys. Rev. A}, 69:022308, 2004.

\bibitem[DPS05]{DPS05}
A.~C. Doherty, P.~A. Parrilo, and F.~M. Spedalieri.
\newblock Detecting multipartite entanglement.
\newblock {\em Phys. Rev. A}, 71:032333, 2005.

\bibitem[DR05]{DR05}
N.~Datta and M.~B. Ruskai.
\newblock Maximal output purity and capacity for asymmetric unital qudit
  channels.
\newblock {\em J. Phys. A: Math. Gen.}, 3:9785--9802, 2005.

\bibitem[DSS{\etalchar{+}}00]{DSSTT00}
D.~P. DiVincenzo, P.~W. Shor, J.~A. Smolin, B.~M. Terhal, and A.~V. Thapliyal.
\newblock Evidence for bound entangled states with negative partial transpose.
\newblock {\em Phys. Rev. A}, 61:062312, 2000.

\bibitem[EA{\.Z}05]{EAZ05}
J.~Emerson, R.~Alicki, and K.~{\.Z}yczkowski.
\newblock Scalable noise estimation with random unitary operators.
\newblock {\em J. Opt. B}, 7:S347--S352, 2005.

\bibitem[EPR35]{EPR35}
A.~Einstein, B.~Podolsky, and N.~Rosen.
\newblock Can quantum-mechanical description of physical reality be considered
  complete?
\newblock {\em Phys. Rev.}, 47:777--780, 1935.

\bibitem[ER88]{ER88}
E.~G. Effros and Z.-J. Ruan.
\newblock On matricially normed spaces.
\newblock {\em Pacific J. Math.}, 132:243--264, 1988.

\bibitem[Fan51]{F51}
Ky~Fan.
\newblock Maximum properties and inequalities for the eigenvalues of completely
  continuous operators.
\newblock {\em Proc. Natl. Acad. Sci. USA}, 37:760--766, 1951.

\bibitem[FLJS06]{FLS06}
S.-M. Fei, X.~Li-Jost, and B.-Z. Sun.
\newblock A class of bound entangled states.
\newblock {\em Phys. Lett. A}, 352:321--325, 2006.

\bibitem[FLPS11]{FLPS11}
S.~Friedland, C.-K. Li, Y.-T. Poon, and N.-S. Sze.
\newblock The automorphism group of separable states in quantum information
  theory.
\newblock {\em J. Math. Phys.}, 52:042203, 2011.

\bibitem[FLV88]{FLV88}
M.~Fannes, J.~T. Lewis, and A.~Verbeure.
\newblock Symmetric states of composite systems.
\newblock {\em Lett. Math. Phys.}, 15:255--260, 1988.

\bibitem[Fuc96]{Fuc96}
C.~A. Fuchs.
\newblock {\em Distinguishability and Accessible Information in Quantum
  Theory}.
\newblock PhD thesis, University of New Mexico, 1996.

\bibitem[GB02]{GB02}
L.~Gurvits and H.~Barnum.
\newblock Largest separable balls around the maximally mixed bipartite quantum
  state.
\newblock {\em Phys. Rev. A}, 66:062311, 2002.

\bibitem[GG08]{GG08}
V.~Gheorghiu and R.~B. Griffiths.
\newblock Separable operations on pure states.
\newblock {\em Phys. Rev. A}, 78:020304(R), 2008.

\bibitem[Gha10]{G10}
S.~Gharibian.
\newblock Strong {NP}-hardness of the quantum separability problem.
\newblock {\em Quantum Inf. Comput.}, 10:343--360, 2010.

\bibitem[Ghe10]{Ghe10}
V.~Gheorghiu.
\newblock {\em Separable operations, graph codes and the location of quantum
  information}.
\newblock PhD thesis, Carnegie Mellon University, 2010.

\bibitem[GLN05]{GLN05}
A.~Gilchrist, N.~K. Langford, and M.~A. Nielsen.
\newblock Distance measures to compare real and ideal quantum processes.
\newblock {\em Phys. Rev. A}, 71:062310, 2005.

\bibitem[GLS93]{GLS93}
M.~Gr\"{o}tschel, L.~Lov\'{a}sz, and A.~Schrijver.
\newblock {\em Geometric algorithms and combinatorial optimization}.
\newblock Springer-Verlag, 1993.

\bibitem[GL{\v S}00]{GLS00}
A.~Guterman, C.-K. Li, and P.~{\v S}emrl.
\newblock Some general techniques on linear preserver problems.
\newblock {\em Linear Algebra Appl.}, 315:61--81, 2000.

\bibitem[GM77]{GM77}
R.~Grone and M.~Marcus.
\newblock Isometries of matrix algebras.
\newblock {\em J. Algebra}, 47:180--189, 1977.

\bibitem[GPM{\etalchar{+}}10]{GPMSZ10}
P.~Gawron, Z.~Pucha{\l}a, J.~A. Miszczak, {\L}.~Skowronek, and
  K.~{\.Z}yczkowski.
\newblock Restricted numerical range: A versatile tool in the theory of quantum
  information.
\newblock {\em J. Math. Phys.}, 51:102204, 2010.

\bibitem[Gro53]{Gro53}
A.~Grothendieck.
\newblock R\'{e}sum\'{e} de la th\'{e}orie m\'{e}trique des produits tensoriels
  topologiques.
\newblock {\em Bol. Soc. Mat. Sao Paulo}, 8:1--79, 1953.

\bibitem[Gur03]{G03}
L.~Gurvits.
\newblock Classical deterministic complexity of {E}dmonds' problem and quantum
  entanglement.
\newblock In {\em Proceedings of the Thirty-Fifth Annual ACM Symposium on
  Theory of Computing}, pages 10--19, 2003.

\bibitem[GY02]{GY02}
C.~J. Goh and X.~Q. Yang.
\newblock {\em Duality in optimization and variational inequalities}.
\newblock Taylor \& Francis, 2002.

\bibitem[Hal06]{Hal06}
W.~Hall.
\newblock A new criterion for indecomposability of positive maps.
\newblock {\em J. Phys. A: Math. Gen.}, 39:14119, 2006.

\bibitem[Hay07]{Hay07}
P.~Hayden.
\newblock The maximal $p$-norm multiplicativity conjecture is false.
\newblock E-print: arXiv:0707.3291
  [quant-ph], 2007.

\bibitem[HH99]{HH99}
M.~Horodecki and P.~Horodecki.
\newblock Reduction criterion of separability and limits for a class of
  distillation protocols.
\newblock {\em Phys. Rev. A}, 59:4206--4216, 1999.

\bibitem[HHH96]{HHH96}
M.~Horodecki, P.~Horodecki, and R.~Horodecki.
\newblock Separability of mixed states: Necessary and sufficient conditions.
\newblock {\em Phys. Lett. A}, 223:1--8, 1996.

\bibitem[HHH97]{HHH97}
M.~Horodecki, P.~Horodecki, and R.~Horodecki.
\newblock Inseparable two spin-1/2 density matrices can be distilled to a
  singlet form.
\newblock {\em Phys. Rev. Lett.}, 78:574--577, 1997.

\bibitem[HHH98]{HHH98}
M.~Horodecki, P.~Horodecki, and R.~Horodecki.
\newblock Mixed-state entanglement and distillation: Is there a ``bound''
  entanglement in nature?
\newblock {\em Phys. Rev. Lett.}, 80:5239--5242, 1998.

\bibitem[HHH99]{HHH99}
M.~Horodecki, P.~Horodecki, and R.~Horodecki.
\newblock General teleportation channel, singlet fraction, and
  quasi-distillation.
\newblock {\em Phys. Rev. A}, 60:1888, 1999.

\bibitem[HHH00]{HHH00}
P.~Horodecki, M.~Horodecki, and R.~Horodecki.
\newblock Binding entanglement channels.
\newblock {\em J. Modern Opt.}, 47:347--354, 2000.

\bibitem[HHH06]{HHH06}
M.~Horodecki, P.~Horodecki, and R.~Horodecki.
\newblock Separability of mixed quantum states: Linear contractions and
  permutation criteria.
\newblock {\em Open Syst. Inf. Dyn.}, 13:103--111, 2006.

\bibitem[HHHH09]{HHH09}
R.~Horodecki, P.~Horodecki, M.~Horodecki, and K.~Horodecki.
\newblock Quantum entanglement.
\newblock {\em Rev. Mod. Phys.}, 81:865--942, 2009.

\bibitem[Hil73]{Hil73}
R.~D. Hill.
\newblock Linear transformations which preserve {H}ermitian matrices.
\newblock {\em Linear Algebra Appl.}, 6:257--262, 1973.

\bibitem[HJ85]{HJ85}
R.~A. Horn and C.~R. Johnson.
\newblock {\em Matrix analysis}.
\newblock Cambridge University Press, 1985.

\bibitem[HJ91]{HJ91}
R.~A. Horn and C.~R. Johnson.
\newblock {\em Topics in matrix analysis}.
\newblock Cambridge University Press, 1991.

\bibitem[HKL04]{HKL04}
J.~A. Holbrook, D.~W. Kribs, and R.~Laflamme.
\newblock Noiseless subsystems and the structure of the commutant in quantum
  error correction.
\newblock {\em Quantum Inf. Process.}, 2:381--419, 2004.

\bibitem[HKW{\etalchar{+}}09]{HKWGG09}
R.~H\"{u}bener, M.~Kleinmann, T.-C. Wei, C.~Gonz\'{a}lez-Guill\'{e}n, and
  O.~G\"{u}hne.
\newblock The geometric measure of entanglement for symmetric states.
\newblock {\em Phys. Rev. A}, 80:032324, 2009.

\bibitem[HLVC00]{HLVC00}
P.~Horodecki, M.~Lewenstein, G.~Vidal, and I.~Cirac.
\newblock Operational criterion and constructive checks for the separability of
  low-rank density matrices.
\newblock {\em Phys. Rev. A}, 62:032310, 2000.

\bibitem[HLW06]{HLW06}
P.~Hayden, D.~W. Leung, and A.~Winter.
\newblock Aspects of generic entanglement.
\newblock {\em Comm. Math. Phys.}, 265:95--117, 2006.

\bibitem[HMM{\etalchar{+}}08]{HMMOV08}
M.~Hayashi, D.~Markham, M.~Murao, M.~Owari, and S.~Virmani.
\newblock Entanglement of multiparty-stabilizer, symmetric, and antisymmetric
  states.
\newblock {\em Phys. Rev. A}, 77:012104, 2008.

\bibitem[HMM{\etalchar{+}}09]{HMMOV09}
M.~Hayashi, D.~Markham, M.~Murao, M.~Owari, and S.~Virmani.
\newblock The geometric measure of entanglement for a symmetric pure state with
  positive amplitudes.
\newblock {\em J. Math. Phys.}, 50:122104, 2009.

\bibitem[HN03]{HN03}
A.~W. Harrow and M.~A. Nielsen.
\newblock Robustness of quantum gates in the presence of noise.
\newblock {\em Phys. Rev. A}, 68:012308, 2003.

\bibitem[Hol98]{Hol98}
A.~S. Holevo.
\newblock Coding theorems for quantum channels.
\newblock {\em Russian Math. Surveys}, 53:1295--1331, 1998.

\bibitem[Hol06]{Hol06}
A.~S. Holevo.
\newblock Multiplicativity of $p$-norms of completely positive maps and the
  additivity problem in quantum information theory.
\newblock {\em Russian Math. Surveys}, 61:301--339, 2006.

\bibitem[Hol08]{Hol08}
A.~S. Holevo.
\newblock Entanglement-breaking channels in infinite dimensions.
\newblock {\em Probl. Inf. Transm.}, 44:3--18, 2008.

\bibitem[Hor97]{H97}
P.~Horodecki.
\newblock Separability criterion and inseparable mixed states with positive
  partial transposition.
\newblock {\em Phys. Lett. A}, 232:333--339, 1997.

\bibitem[Hou10]{Hou10}
J.~Hou.
\newblock A characterization of positive linear maps and criteria of
  entanglement for quantum states.
\newblock {\em J. Phys. A: Math. Theor.}, 43:385201, 2010.

\bibitem[HPS{\etalchar{+}}06]{HPSSSL06}
F.~Hulpke, U.~V. Poulsen, A.~Sanpera, A.~Sen(de), U.~Sen, and M.~Lewenstein.
\newblock Unitarity as preservation of entropy and entanglement in quantum
  systems.
\newblock {\em Found. Phys.}, 36:477--499, 2006.

\bibitem[HSR03]{HSR03}
M.~Horodecki, P.~W. Shor, and M.~B. Ruskai.
\newblock General entanglement breaking channels.
\newblock {\em Rev. Math. Phys.}, 15:629--641, 2003.

\bibitem[Hua06]{Hua06}
S.~Huang.
\newblock Schmidt number for quantum operations.
\newblock {\em Phys. Rev. A}, 73:052318, 2006.

\bibitem[HW08]{HW08}
P.~Hayden and A.~Winter.
\newblock Counterexamples to the maximal $p$-norm multiplicativity conjecture
  for all $p > 1$.
\newblock {\em Comm. Math. Phys.}, 284:263--280, 2008.

\bibitem[Ioa07]{Ioa07}
L.~M. Ioannou.
\newblock Computational complexity of the quantum separability problem.
\newblock {\em Quantum Inf. Comput.}, 7:335--370, 2007.

\bibitem[Jam72]{J72}
A.~Jamio{\l}kowski.
\newblock Linear transformations which preserve trace and positive
  semidefiniteness of operators.
\newblock {\em Rep. Math. Phys.}, 3:275--278, 1972.

\bibitem[Jam74]{Jam74}
A.~Jamio{\l}kowski.
\newblock An effective method of investigation of positive maps on the set of
  positive definite operators.
\newblock {\em Rep. Math. Phys.}, 5:415--424, 1974.

\bibitem[Jen06]{Jen06}
A.~Jen\v{c}ov\'{a}.
\newblock A relation between completely bounded norms and conjugate channels.
\newblock {\em Comm. Math. Phys.}, 266:65--70, 2006.

\bibitem[JJUW10]{JJUW10}
R.~Jain, Z.~Ji, S.~Upadhyay, and J.~Watrous.
\newblock {QIP} = {PSPACE}.
\newblock In {\em Proceedings of the Forty-Second Annual ACM Symposium on
  Theory of Computing}, pages 573--582, 2010.

\bibitem[JK10]{JK10}
N.~Johnston and D.~W. Kribs.
\newblock A family of norms with applications in quantum information theory.
\newblock {\em J. Math. Phys.}, 51:082202, 2010.

\bibitem[JK11a]{JK11}
N.~Johnston and D.~W. Kribs.
\newblock A family of norms with applications in quantum information theory
  {II}.
\newblock {\em Quantum Inf. Comput.}, 11:104--123, 2011.

\bibitem[JK11b]{JK11c}
N.~Johnston and D.~W. Kribs.
\newblock Generalized multiplicative domains and quantum error correction.
\newblock {\em Proc. Amer. Math. Soc.}, 139:627--639, 2011.

\bibitem[JK11c]{JK11b}
N.~Johnston and D.~W. Kribs.
\newblock Quantum gate fidelity in terms of {C}hoi matrices.
\newblock {\em J. Phys. A: Math. Theor.}, 44:495303, 2011.

\bibitem[JKP09]{JKP09}
N.~Johnston, D.~W. Kribs, and V.~I. Paulsen.
\newblock Computing stabilized norms for quantum operations.
\newblock {\em Quantum Inf. Comput.}, 9:16--35, 2009.

\bibitem[JKPP11]{JKPP11}
N.~Johnston, D.~W. Kribs, V.~I. Paulsen, and R.~Pereira.
\newblock Minimal and maximal operator spaces and operator systems in
  entanglement theory.
\newblock {\em J. Funct. Anal.}, 260:2407--2423, 2011.

\bibitem[Joh11]{Joh11}
N.~Johnston.
\newblock Characterizing operations preserving separability measures via linear
  preserver problems.
\newblock {\em Linear and Multilinear Algebra}, 59:1171--1187, 2011.

\bibitem[Joz94]{Joz94}
R.~Jozsa.
\newblock Fidelity for mixed quantum states.
\newblock {\em J. Modern Opt.}, 41:2315--2323, 1994.

\bibitem[JRC10]{JRC10}
J.~Jurkowski, A.~Rutkowski, and D.~Chru\'{s}ci\'{n}ski.
\newblock Local numerical range for a class of $2 \otimes d$ {H}ermitian
  operators.
\newblock {\em Open Syst. Inf. Dyn.}, 17:347--359, 2010.

\bibitem[JS12]{JS11}
N.~Johnston and E.~St{\o}rmer.
\newblock Mapping cones are operator systems.
\newblock \emph{Bulletin of the London Mathematical Society}. doi: 10.1112/blms/bds006, 2012.

\bibitem[KAB{\etalchar{+}}08]{KABLA08}
J.~K. Korbicz, M.~L. Almeida, J.~Bae, M.~Lewenstein, and A.~Ac\'{i}n.
\newblock Structural approximations to positive maps and entanglement-breaking
  channels.
\newblock {\em Phys. Rev. A}, 78:062105, 2008.

\bibitem[Kad51]{K51}
R.~V. Kadison.
\newblock Isometries of operator algebras.
\newblock {\em Ann. of Math.}, 54:325--338, 1951.

\bibitem[Ke05]{OpenProbGen}
O.~Krueger and R.~F.~Werner (editors).
\newblock Some open problems in quantum information theory.
\newblock E-print: arXiv:quant-ph/0504166, 2005.

\bibitem[KGBO09]{KBO09}
T.~Karasawa, J.~Gea-Banacloche, and M.~Ozawa.
\newblock Gate fidelity of arbitrary single-qubit gates constrained by
  conservation laws.
\newblock {\em J. Phys. A: Math. Theor.}, 42:225303, 2009.

\bibitem[Kin02]{Kin02}
C.~King.
\newblock Maximization of capacity and $p$-norms for some product channels.
\newblock {\em J. Math. Phys.}, 43:1247--1260, 2002.

\bibitem[Kin03]{Kin03}
C.~King.
\newblock Maximal $p$-norms of entanglement breaking channels.
\newblock {\em Quantum Inf. Comput.}, 3:186--190, 2003.

\bibitem[Kit97]{Kit97}
A.~Yu. Kitaev.
\newblock Quantum computations: Algorithms and error correction.
\newblock {\em Russian Math. Surveys}, 52:1191--1249, 1997.

\bibitem[KLPL06]{KLPL06}
D.~W. Kribs, R.~Laflamme, D.~Poulin, and M.~Lesosky.
\newblock Operator quantum error correction.
\newblock {\em Quantum Inf. Comput.}, 6:383--399, 2006.

\bibitem[Kni03]{OpenProb15}
E.~Knill.
\newblock Separability from spectrum.
\newblock Published electronically at
  \url{http://qig.itp.uni-hannover.de/qiproblems/15}, 2003.

\bibitem[KNR05]{KNR05}
C.~King, M.~Nathanson, and M.~B. Ruskai.
\newblock Multiplicativity properties of entrywise positive maps.
\newblock {\em Linear Algebra Appl.}, 404:367--379, 2005.

\bibitem[KR01]{KR01}
C.~King and M.~B. Ruskai.
\newblock Minimal entropy of states emerging from noisy quantum channels.
\newblock {\em IEEE Trans. Inform. Theory}, 47:192--209, 2001.

\bibitem[KR04]{KR04}
C.~King and M.~B. Ruskai.
\newblock Comments on multiplicativity of maximal $p$-norms when $p = 2$.
\newblock {\em Quantum Inf. Comput.}, 4:500--512, 2004.

\bibitem[Kra71]{K71}
K.~Kraus.
\newblock General state changes in quantum theory.
\newblock {\em Ann. Physics}, 64:311--335, 1971.

\bibitem[Kra83]{Kra83}
K.~Kraus.
\newblock {\em States, effects, and operations: Fundamental notions of quantum
  theory}.
\newblock Springer-Verlag, 1983.

\bibitem[Kri03]{Kri03}
D.~W. Kribs.
\newblock Quantum channels, wavelets, dilations and representations of $o_n$.
\newblock {\em Proc. Edinb. Math. Soc.}, 46:421--433, 2003.

\bibitem[KS05]{KS05}
A.-M. Kuah and E.~C.~G. Sudarshan.
\newblock Schmidt states and positivity of linear maps.
\newblock E-print: arXiv:quant-ph/0506095, 2005.

\bibitem[KS06]{KS06}
D.~W. Kribs and R.~W. Spekkens.
\newblock Quantum error correcting subsystems are unitarily recoverable
  subsystems.
\newblock {\em Phys. Rev. A}, 74:042329, 2006.

\bibitem[LBC{\etalchar{+}}00]{LBCKKSST00}
M.~Lewenstein, D.~Bru\ss, J.~I. Cirac, B.~Kraus, M.~Kus, J.~Samsonowicz,
  A.~Sanpera, and R.~Tarrach.
\newblock Separability and distillability in composite quantum systems -a
  primer-.
\newblock {\em J. Modern Opt.}, 47:2481--2499, 2000.

\bibitem[Li94]{L94}
C.-K. Li.
\newblock Some aspects of the theory of norms.
\newblock {\em Linear Algebra Appl.}, 212--213:71--100, 1994.

\bibitem[Li00]{Li00}
C.-K. Li.
\newblock Norms, isometries, and isometry groups.
\newblock {\em Amer. Math. Monthly}, 107:334--340, 2000.

\bibitem[Loe89]{L89}
R.~Loewy.
\newblock Linear transformations which preserve or decrease rank.
\newblock {\em Linear Algebra Appl.}, 121:151--161, 1989.

\bibitem[L{\"o}f04]{YALMIP}
J.~L{\"o}fberg.
\newblock {YALMIP}: A toolbox for modeling and optimization in {MATLAB}.
\newblock In {\em Proceedings of the International Symposium on Computer-Aided
  Control System Design}, 2004.
\newblock Software available at \url{http://users.isy.liu.se/johanl/yalmip/}.

\bibitem[Lov03]{Lo03}
L.~Lov\'{a}sz.
\newblock Semidefinite programs and combinatorial optimization.
\newblock Chapter in \emph{Recent Advances in Algorithms and Combinatorics},
  2003.

\bibitem[LP01]{LP01}
C.-K. Li and S.~Pierce.
\newblock Linear preserver problems.
\newblock {\em Amer. Math. Monthly}, 108:591--605, 2001.

\bibitem[LRK{\etalchar{+}}11]{LRKKB11}
H.-T. Lim, Y.-S. Ra, Y.-S. Kim, Y.-H. Kim, and J.~Bae.
\newblock Gate fidelities, quantum broadcasting, and assessing experimental
  realization.
\newblock E-print: arXiv:1106.5873 [quant-ph], 2011.

\bibitem[LS93]{LS93}
L.~J. Landau and R.~F. Streater.
\newblock On {B}irkhoff's theorem for doubly stochastic completely positive
  maps of matrix algebras.
\newblock {\em Linear Algebra Appl.}, 193:107--127, 1993.

\bibitem[LT90]{LT90}
C.-K. Li and N.~K. Tsing.
\newblock Linear operators preserving unitarily invariant norms on matrices.
\newblock {\em Linear and Multilinear Algebra}, 26:119--132, 1990.

\bibitem[LT92]{LT92}
C.-K. Li and N.-K. Tsing.
\newblock Linear preserver problems: A brief introduction and some special
  techniques.
\newblock {\em Linear Algebra Appl.}, 162--164:217--235, 1992.

\bibitem[Maj11]{Maj11}
W.~A. Majewski.
\newblock On the structure of positive maps; finite dimensional case.
\newblock E-print: arXiv:1005.3949 [math-ph], 2011.

\bibitem[Mar59]{M59}
M.~Marcus.
\newblock All linear operators leaving the unitary group invariant.
\newblock {\em Duke Math. J.}, 26:155--163, 1959.

\bibitem[MBKE11]{MBE11}
E.~Magesan, R.~Blume-Kohout, and J.~Emerson.
\newblock Gate fidelity fluctuations and quantum process invariants.
\newblock {\em Phys. Rev. A}, 84:012309, 2011.

\bibitem[MCL06]{MCL06}
T.~Moroder, M.~Curty, and N.~L\"{u}tkenhaus.
\newblock One-way quantum key distribution: Simple upper bound on the secret
  key rate.
\newblock {\em Phys. Rev. A}, 74:052301, 2006.

\bibitem[Men08]{Men08}
C.~B. Mendl.
\newblock \emph{Unital quantum channels}.
\newblock Diploma thesis, Technische Universit\"{a}t M\"{u}nchen, 2008.

\bibitem[MF85]{MF85}
G.~S. Mudholkar and M.~Freimer.
\newblock A structure theorem for the polars of unitarily invariant norms.
\newblock {\em Proc. Amer. Math. Soc.}, 95:331--337, 1985.

\bibitem[Mil71]{M71}
V.~Milman.
\newblock A new proof of the theorem of {A}.~{D}voretzky on sections of convex
  bodies.
\newblock {\em Funct. Anal. Appl.}, 5:28--37, 1971.

\bibitem[ML09]{ML09}
G.~O. Myhr and N.~L\"{u}tkenhaus.
\newblock Spectrum conditions for symmetric extendible states.
\newblock {\em Phys. Rev. A}, 79:062307, 2009.

\bibitem[MM59]{MM59}
M.~Marcus and B.~N. Moyls.
\newblock Transformations on tensor product spaces.
\newblock {\em Pacific J. Math.}, 9:1215--1221, 1959.

\bibitem[MM01]{MM01}
W.~A. Majewski and M.~Marciniak.
\newblock On a characterization of positive maps.
\newblock {\em J. Phys. A: Math. Gen.}, 34:5863--5874, 2001.

\bibitem[MW09]{MW09}
C.~B. Mendl and M.~M. Wolf.
\newblock Unital quantum channels - convex structure and revivals of
  {B}irkhoff's theorem.
\newblock {\em Comm. Math. Phys.}, 289:1057--1096, 2009.

\bibitem[NC00]{NC00}
M.~A. Nielsen and I.~L. Chuang.
\newblock {\em Quantum computation and quantum information}.
\newblock Cambridge University Press, 2000.

\bibitem[NDD{\etalchar{+}}03]{NDDGMOBHH03}
M.~A. Nielsen, C.~M. Dawson, J.~L. Dodd, A.~Gilchrist, D.~Mortimer, T.~J.
  Osborne, M.~J. Bremner, A.~W. Harrow, and A.~Hines.
\newblock Quantum dynamics as a physical resource.
\newblock {\em Phys. Rev. A}, 67:052301, 2003.

\bibitem[Nie98]{Nie98}
M.~A. Nielsen.
\newblock {\em Quantum information theory}.
\newblock PhD thesis, University of New Mexico, 1998.

\bibitem[Nie02]{Nie02}
M.~Nielsen.
\newblock A simple formula for the average gate fidelity of a quantum dynamical
  operation.
\newblock {\em Phys. Lett. A}, 303:249--252, 2002.

\bibitem[NOP09]{NOP09}
M.~Navascu\'{e}s, M.~Owari, and M.~B. Plenio.
\newblock Power of symmetric extensions for entanglement detection.
\newblock {\em Phys. Rev. A}, 80:052306, 2009.

\bibitem[OR04]{OR04}
T.~Oikhburg and E.~Ricard.
\newblock Operator spaces with few completely bounded maps.
\newblock {\em Math. Ann.}, 328:229--259, 2004.

\bibitem[OS{\.Z}10]{OSZ10}
V.~Osipov, H.-J. Sommers, and K.~{\.Z}yczkowski.
\newblock Random {B}ures mixed states and the distribution of their purity.
\newblock {\em J. Phys. A: Math. Theor.}, 43:055302, 2010.

\bibitem[Pau03]{P03}
V.~I. Paulsen.
\newblock {\em Completely bounded maps and operator algebras}.
\newblock Cambridge University Press, 2003.

\bibitem[PBHS11]{PBHS11}
L.~Pankowski, F.~G.~S.~L. Brand{\~ a}o, M.~Horodecki, and G.~Smith.
\newblock Entanglement distillation by means of $k$-extendible maps.
\newblock E-print: arXiv:1109.1779 [quant-ph], 2011.

\bibitem[Per96]{P96}
A.~Peres.
\newblock Separability criterion for density matrices.
\newblock {\em Phys. Rev. Lett.}, 77:1413--1415, 1996.

\bibitem[PGM{\etalchar{+}}11]{PGMSCZ11}
Z.~Pucha{\l}a, P.~Gawron, J.~A. Miszczak, {\L}.~Skowronek, M.-D. Choi, and
  K.~{\.Z}yczkowski.
\newblock Product numerical range in a space with tensor product structure.
\newblock {\em Linear Algebra Appl.}, 434:327--342, 2011.

\bibitem[PH81]{PH81}
J.~A. Poluikis and R.~D. Hill.
\newblock Completely positive and {H}ermitian-preserving linear
  transformations.
\newblock {\em Linear Algebra Appl.}, 35:1--10, 1981.

\bibitem[Pis03]{Pis03}
G.~Pisier.
\newblock {\em Introduction to operator space theory}.
\newblock Cambridge University Press, 2003.

\bibitem[PM07]{PM07}
M.~Piani and C.~E. Mora.
\newblock Class of positive-partial-transpose bound entangled states associated
  with almost any set of pure entangled states.
\newblock {\em Phys. Rev. A}, 75:012305, 2007.

\bibitem[PMM11]{PMM11}
L.~H. Pedersen, N.~M. M{\o}ller, and K.~M{\o}lmer.
\newblock The distribution of quantum fidelities.
\newblock {\em Phys. Lett. A}, 372:7028--7032, 2011.

\bibitem[PPHH10]{PPHH10}
L.~Pankowski, M.~Piani, M.~Horodecki, and P.~Horodecki.
\newblock A few steps more towards {NPT} bound entanglement.
\newblock {\em IEEE Trans. Inf. Theory}, 56:4085--4100, 2010.

\bibitem[P{\v{R}}04]{PR04}
M.~Paris and J.~{\v{R}}eh{\'a}{\v{c}}ek.
\newblock {\em Quantum state estimation}.
\newblock Springer, 2004.

\bibitem[PT09]{PT09}
V.~I. Paulsen and M.~Tomforde.
\newblock Vector spaces with an order unit.
\newblock {\em Indiana Univ. Math. J.}, 58:1319--1359, 2009.

\bibitem[PTT11]{PTT11}
V.~I. Paulsen, I.~Todorov, and M.~Tomforde.
\newblock Operator system structures on ordered spaces.
\newblock {\em Proc. Lond. Math. Soc.}, 102:25--49, 2011.

\bibitem[PV07]{PV07}
M.~B. Plenio and S.~Virmani.
\newblock An introduction to entanglement measures.
\newblock {\em Quantum Inf. Comput.}, 7:1--51, 2007.

\bibitem[QH11]{QH11}
X.~Qi and J.~Hou.
\newblock Positive finite rank elementary operators and characterizing
  entanglement of states.
\newblock {\em J. Phys. A: Math. Theor.}, 44:215305, 2011.

\bibitem[RA07]{RA07}
K.~S. Ranade and M.~Ali.
\newblock The {J}amio{\l}kowski isomorphism and a simplified proof for the
  correspondence between vectors having {S}chmidt number $k$ and $k$-positive
  maps.
\newblock {\em Open Syst. Inf. Dyn.}, 14:371--378, 2007.

\bibitem[Rai97]{Rai97}
E.~M. Rains.
\newblock Entanglement purification via separable superoperators.
\newblock E-print: arXiv:quant-ph/9707002, 1997.

\bibitem[Roc97]{R97}
R.~Rockafellar.
\newblock {\em Convex analysis}.
\newblock Princeton University Press, 1997.

\bibitem[Rua88]{R88}
Z.-J. Ruan.
\newblock Subspaces of {C}$^*$-algebras.
\newblock {\em J. Funct. Anal.}, 76:217--230, 1988.

\bibitem[Rud00]{R00}
O.~Rudolph.
\newblock A separability criterion for density operators.
\newblock {\em J. Phys. A: Math. Gen.}, 33:3951--3955, 2000.

\bibitem[Rud01]{Rud01}
O.~Rudolph.
\newblock A new class of entanglement measures.
\newblock {\em J. Math. Phys.}, 42:5306--5314, 2001.

\bibitem[Rud03]{R03}
O.~Rudolph.
\newblock Some properties of the computable cross norm criterion for
  separability.
\newblock {\em Phys. Rev. A}, 67:032312, 2003.

\bibitem[Rus03]{Rus03}
M.~B. Ruskai.
\newblock Qubit entanglement breaking channels.
\newblock {\em Rev. Math. Phys.}, 15:643--662, 2003.

\bibitem[RW89]{RW89}
G.~A. Raggio and R.~F. Werner.
\newblock Quantum statistical mechanics of general mean field systems.
\newblock {\em Helv. Phys. Acta}, 62:980--1003, 1989.

\bibitem[Sar08]{Sar08}
G.~Sarbicki.
\newblock Spectral properties of entanglement witnesses.
\newblock {\em J. Phys. A: Math. Theor.}, 41:375303, 2008.

\bibitem[Sch25]{S25}
I.~Schur.
\newblock Einige bemerkungen zur determinanten theorie.
\newblock {\em Sitzungsber. Preuss. Akad. Wiss. Berlin}, 25:454--463, 1925.

\bibitem[Sch35]{Sch35}
E.~Schr\"{o}dinger.
\newblock Discussion of probability relations between separated systems.
\newblock {\em Math. Proc. Cambridge Philos. Soc.}, 31:555--563, 1935.

\bibitem[Sch60]{Sch60}
R.~Schatten.
\newblock {\em Norm ideals of completely continuous operators}.
\newblock Springer-Verlag, Berlin, 1960.

\bibitem[Shi95]{Shi95}
A.~Shimony.
\newblock Degree of entanglement.
\newblock {\em Ann. N.Y. Acad. Sci.}, 755:675--679, 1995.

\bibitem[Sho02]{Sho02}
P.~Shor.
\newblock Additivity of the classical capacity of entanglement-breaking
  channels.
\newblock {\em J. Math. Phys.}, 43:4334--4340, 2002.

\bibitem[Sko08]{Sko08}
{\L}.~Skowronek.
\newblock {\em Quantum entanglement and certain problems in mathematics}.
\newblock PhD thesis, Krakow, 2008.

\bibitem[Sko10]{Sko10}
{\L}.~Skowronek.
\newblock Dualities and positivity in the study of quantum entanglement.
\newblock {\em Int. J. Quantum Inf.}, 8:721--754, 2010.

\bibitem[Sko11]{Sko11}
{\L}.~Skowronek.
\newblock Cones with a mapping cone symmetry in the finite-dimensional case.
\newblock {\em Linear Algebra Appl.}, 435:361--370, 2011.

\bibitem[Smi83]{S83}
R.~R. Smith.
\newblock Completely bounded maps between {C}$^*$-algebras.
\newblock {\em J. Lond. Math. Soc.}, 27:157--166, 1983.

\bibitem[SMR61]{SMR61}
E.~C.~G. Sudarshan, P.~M. Mathews, and J.~Rau.
\newblock Stochastic dynamics of quantum-mechanical systems.
\newblock {\em Phys. Rev.}, 121:920--924, 1961.

\bibitem[Sou81]{S81}
A.~R. Sourour.
\newblock Isometries of norm ideals of compact operators.
\newblock {\em J. Funct. Anal.}, 43:69--77, 1981.

\bibitem[SS12a]{SS10}
{\L}.~Skowronek and E.~St{\o}rmer.
\newblock Choi matrices, norms and entanglement associated with positive maps
  on matrix algebras.
\newblock {\em J. Funct. Anal.}, 262:639--647, 2012.

\bibitem[SS12b]{SS12}
G.~Smith and J.~A. Smolin.
\newblock Detecting incapacity of a quantum channel.
\newblock {\em Phys. Rev. Lett.}, 108:230507, 2012.

\bibitem[SSGF05]{SSF05}
D.~Salgado, J.~L. S{\'a}nchez-G{\'o}mez, and M.~Ferrero.
\newblock A simple proof of the {J}amio{\l}kowski criterion for complete
  positivity of linear maps.
\newblock {\em Open Syst. Inf. Dyn.}, 12:55--64, 2005.

\bibitem[SST01]{SST01}
P.~W. Shor, J.~A. Smolin, and B.~M. Terhal.
\newblock Nonadditivity of bipartite distillable entanglement follows from a
  conjecture on bound entangled {W}erner states.
\newblock {\em Phys. Rev. Lett.}, 86:2681--2684, 2001.

\bibitem[SS{\.Z}09]{SSZ09}
{\L}.~Skowronek, E.~St{\o}rmer, and K.~{\.Z}yczkowski.
\newblock Cones of positive maps and their duality relations.
\newblock {\em J. Math. Phys.}, 50:062106, 2009.

\bibitem[Ste03]{Ste03}
M.~Steiner.
\newblock Generalized robustness of entanglement.
\newblock {\em Phys. Rev. A}, 67:054305, 2003.

\bibitem[Sti55]{Sti55}
W.~F. Stinespring.
\newblock Positive functions on ${C}^*$-algebras.
\newblock {\em Proc. Amer. Math. Soc.}, 6:211--216, 1955.

\bibitem[St{\o}63]{S63}
E.~St{\o}rmer.
\newblock Positive linear maps of operator algebras.
\newblock {\em Acta Math.}, 110:233--278, 1963.

\bibitem[St{\o}86]{S86}
E.~St{\o}rmer.
\newblock Extension of positive maps into ${B}({H})$.
\newblock {\em J. Funct. Anal.}, 66:235--254, 1986.

\bibitem[St{\o}09]{S09}
E.~St{\o}rmer.
\newblock Duality of cones of positive maps.
\newblock {\em M\"{u}nster J. Math.}, 2:299--310, 2009.

\bibitem[St{\o}11a]{S11}
E.~St{\o}rmer.
\newblock Mapping cones of positive maps.
\newblock {\em Math. Scand.}, 108:223--232, 2011.

\bibitem[St{\o}11b]{Sto11}
E.~St{\o}rmer.
\newblock Tensor products of positive maps of matrix algebras.
\newblock E-print: arXiv:1101.2114
  [math.OA], 2011.

\bibitem[Stu99]{SeDuMi}
J.~F. Sturm.
\newblock {S}e{D}u{M}i 1.02, a {MATLAB} toolbox for optimization over symmetric
  cones.
\newblock {\em Optim. Methods Softw.}, 11--12:625--653, 1999.
\newblock Software available at
  \url{http://sedumi.ie.lehigh.edu/}.

\bibitem[SV11]{SV11}
J.~Sperling and W.~Vogel.
\newblock Determination of the {S}chmidt number.
\newblock {\em Phys. Rev. A}, 83:042315, 2011.

\bibitem[SW{\.Z}08]{SWZ08}
S.~Szarek, E.~Werner, and K.~{\.Z}yczkowski.
\newblock Geometry of sets of quantum maps: A generic positive map acting on a
  high-dimensional system is not completely positive.
\newblock {\em J. Math. Phys.}, 49:032113, 2008.

\bibitem[Sza10]{Sza10}
S.~Szarek.
\newblock On norms of completely positive maps.
\newblock {\em Oper. Theory Adv. Appl.}, 202:535--538, 2010.

\bibitem[SZWGC09]{SZG09}
Z.~Shun, Z.~Zheng-Wei, and G.~Guang-Can.
\newblock Separability of bipartite superoperator based on witness.
\newblock {\em Chinese Phys. Lett.}, 26:020304, 2009.

\bibitem[Tak24]{Tak24}
T.~Takagi.
\newblock On an algebraic problem related to an analytic theorem of
  {C}arath\'{e}odory and {F}ej\'{e}r and on an allied theorem of {L}andau.
\newblock {\em Jpn. J. Math.}, 1:83--93, 1924.

\bibitem[Ter00]{Ter00}
B.~M. Terhal.
\newblock A family of indecomposable positive linear maps based on entangled
  quantum states.
\newblock {\em Linear Algebra Appl.}, 323:61--73, 2000.

\bibitem[TG10]{TG10}
G.~T\'{o}th and O.~G\"{u}hne.
\newblock Separability criteria and entanglement witnesses for symmetric
  quantum states.
\newblock {\em Appl. Phys. B}, 98:617--622, 2010.

\bibitem[TH00]{TH00}
B.~M. Terhal and P.~Horodecki.
\newblock Schmidt number for density matrices.
\newblock {\em Phys. Rev. A}, 61:040301(R), 2000.

\bibitem[Tom85]{Tom85}
J.~Tomiyama.
\newblock On the geometry of positive maps in matrix algebras {II}.
\newblock {\em Linear Algebra Appl.}, 69:169--177, 1985.

\bibitem[TT83]{TT83}
K.~Takesaki and J.~Tomiyama.
\newblock On the geometry of positive maps in matrix algebras.
\newblock {\em Math. Z.}, 184:101--108, 1983.

\bibitem[Uhl76]{Uhl76}
A.~Uhlmann.
\newblock The ``transition probability'' in the state space of a *-algebra.
\newblock {\em Rep. Math. Phys.}, 9:273--279, 1976.

\bibitem[Vai94]{Vai94}
L.~Vaidman.
\newblock Teleportation of quantum states.
\newblock {\em Phys. Rev. A}, 49:1473--1476, 1994.

\bibitem[VAM01]{VAD01}
F.~Verstraete, K.~Audenaert, and B.~De Moor.
\newblock Maximally entangled mixed states of two qubits.
\newblock {\em Phys. Rev. A}, 64:012316, 2001.

\bibitem[VB94]{VB94}
L.~Vandenberghe and S.~Boyd.
\newblock Semidefinite programming.
\newblock {\em SIAM Review}, 38:49--95, 1994.

\bibitem[VD06]{VD06}
R.~O. Vianna and A.~C. Doherty.
\newblock Distillability of {W}erner states using entanglement witnesses and
  robust semidefinite programs.
\newblock {\em Phys. Rev. A}, 74:052306, 2006.

\bibitem[Vid00]{Vid00}
G.~Vidal.
\newblock Entanglement monotones.
\newblock {\em J. Modern Opt.}, 47:355--376, 2000.

\bibitem[VP98]{VP98}
V.~Vedral and M.~B. Plenio.
\newblock Entanglement measures and purification procedures.
\newblock {\em Phys. Rev. A}, 57:1619--1633, 1998.

\bibitem[VV03]{VV03}
F.~Verstraete and H.~Verschelde.
\newblock On quantum channels.
\newblock E-print:
  arXiv:quant-ph/0202124, 2003.

\bibitem[Wat04]{Wat04}
J.~Watrous.
\newblock Theory of quantum information lecture notes.
\newblock Published electronically at
  \url{http://www.cs.uwaterloo.ca/~watrous/lecture-notes.html}, 2004.

\bibitem[Wat05]{Wat05}
J.~Watrous.
\newblock Notes on super-operator norms induced by {S}chatten norms.
\newblock {\em Quantum Inf. Comput.}, 5:58--68, 2005.

\bibitem[Wat09]{Wa09}
J.~Watrous.
\newblock Semidefinite programs for completely bounded norms.
\newblock {\em Theory Comput.}, 5:217--238, 2009.

\bibitem[WEGM04]{WEGM04}
T.-C. Wei, M.~Ericsson, P.~M. Goldbart, and W.~J. Munro.
\newblock Connections between relative entropy of entanglement and geometric
  measure of entanglement.
\newblock {\em Quantum Inf. Comput.}, 4:252--272, 2004.

\bibitem[Wer89a]{Wer89}
R.~F. Werner.
\newblock An application of {B}ell's inequalities to a quantum state extension
  problem.
\newblock {\em Lett. Math. Phys.}, 17:359--363, 1989.

\bibitem[Wer89b]{W89}
R.~F. Werner.
\newblock Quantum states with {E}instein-{P}odolsky-{R}osen correlations
  admitting a hidden-variable model.
\newblock {\em Phys. Rev. A}, 40:4277--4281, 1989.

\bibitem[Wes67]{W67}
R.~Westwick.
\newblock Transformations on tensor spaces.
\newblock {\em Pacific J. Math.}, 23:613--620, 1967.

\bibitem[WG03]{WG03}
T.-C. Wei and P.~M. Goldbart.
\newblock Geometric measure of entanglement and applications to bipartite and
  multipartite quantum states.
\newblock {\em Phys. Rev. A}, 68:042307, 2003.

\bibitem[WH02]{WH02}
R.~F. Werner and A.~S. Holevo.
\newblock Counterexample to an additivity conjecture for output purity of
  quantum channels.
\newblock {\em J. Math. Phys.}, 43:4353--4357, 2002.

\bibitem[Win07]{Win07}
A.~Winter.
\newblock The maximum output $p$-norm of quantum channels is not multiplicative
  for any $p > 2$.
\newblock E-print: arXiv:0707.0402
  [quant-ph], 2007.

\bibitem[Wor76]{W76}
S.~L. Woronowicz.
\newblock Positive maps of low dimensional matrix algebras.
\newblock {\em Rep. Math. Phys.}, 10:165--183, 1976.

\bibitem[WPG07]{WPG07}
M.~M. Wolf and D.~P\'{e}rez-García.
\newblock Quantum capacities of channels with small environment.
\newblock {\em Phys. Rev. A}, 75:012303, 2007.

\bibitem[WS10]{WS10}
T.-C. Wei and S.~Severini.
\newblock Matrix permanent and quantum entanglement of permutation invariant
  states.
\newblock {\em J. Math. Phys.}, 51:092203, 2010.

\bibitem[WSV00]{WSV00}
H.~Wolkowicz, R.~Saigal, and L.~Vandenberghe.
\newblock {\em Handbook of semidefinite programming: Theory, algorithms, and
  applications}.
\newblock Springer, 2000.

\bibitem[WW01]{WW01}
R.~F. Werner and M.~M. Wolf.
\newblock Bound entangled {G}aussian states.
\newblock {\em Phys. Rev. Lett.}, 86:3658--3661, 2001.

\bibitem[Xha09]{Xthesis}
B.~Xhabli.
\newblock {\em Universal operator system structures on ordered spaces and their
  applications}.
\newblock PhD thesis, University of Houston, 2009.

\bibitem[Xha12]{Xha11}
B.~Xhabli.
\newblock The super operator system structures and their applications in
  quantum entanglement theory.
\newblock {\em J. Funct. Anal.}, 262:1466--1497, 2012.

\bibitem[Yan06]{Y06}
D.~Yang.
\newblock A simple proof of monogamy of entanglement.
\newblock {\em Phys. Lett. A}, 360:249--250, 2006.

\bibitem[YlL05]{YL05}
S.~Yu and N.~l.~Liu.
\newblock Entanglement detection by local orthogonal observables.
\newblock {\em Phys. Rev. Lett.}, 95:150504, 2005.

\bibitem[{\.Z}B04]{ZB04}
K.~{\.Z}yczkowski and I.~Bengtsson.
\newblock On duality between quantum maps and quantum states.
\newblock {\em Open Syst. Inf. Dyn.}, 11:3--42, 2004.

\end{thebibliography}
\newcommand{\etalchar}[1]{$^{#1}$}

\end{document}